\documentclass[12pt]{article}

\pdfoutput=1
 
\usepackage{macros}
\usepackage{fullpage}
 
\allowdisplaybreaks

\numberwithin{equation}{section}

\begin{document}

\title{\vspace{-.3in}\textbf{Multivariate Fidelities}}

\author{Theshani Nuradha\thanks{School of Electrical and Computer Engineering, Cornell University, Ithaca, NY 14850, USA
  (pt388@cornell.edu (corresponding author), hemant.mishra@cornell.edu, wilde@cornell.edu).}
\and  Hemant K.~Mishra\footnotemark[1]
\and Felix Leditzky\thanks{Department of Mathematics \&
Illinois Quantum Information Science and Technology (IQUIST) Center, University of Illinois Urbana-Champaign, Urbana, IL 61801, USA
  (leditzky@illinois.edu)}  
\and Mark M.~Wilde\footnotemark[1]
}

\date{ }

\maketitle

\begin{abstract}
The bivariate classical fidelity is a widely used measure of the similarity of two probability distributions. 
There exist a few extensions of the notion of the bivariate classical fidelity to more than two probability distributions; herein we call these extensions multivariate classical fidelities, with some examples being the Matusita multivariate fidelity and the average pairwise fidelity. Hitherto, quantum generalizations of multivariate classical fidelities have not been systematically explored, even though there are several well known generalizations of the bivariate classical fidelity to quantum states, such as the Uhlmann and Holevo fidelities. 
The main contribution of our paper is to introduce a number of multivariate quantum fidelities and show that they satisfy several desirable properties that are natural extensions of those of the Uhlmann and Holevo fidelities. We propose several variants that reduce to the average pairwise fidelity for  commuting states, including the average pairwise $z$-fidelities, the multivariate semi-definite programming (SDP) fidelity, and  a multivariate fidelity inspired by an existing secrecy measure. The second one is obtained by extending the SDP formulation of the Uhlmann fidelity to more than two states. All of these variants satisfy the following properties: 
(i) reduction to multivariate classical fidelities for commuting states, (ii)  the data-processing inequality, (iii) invariance under  permutations of the  states, (iv) its values are in the interval~$[0,1]$; they are faithful, that is, their values are equal to one if and only if all the  states are equal, and they satisfy orthogonality, that is their values are equal to zero if and only if the  states are mutually orthogonal to each other, (v)  direct-sum property, (vi)  joint concavity, and (vii) uniform continuity bounds under certain conditions. 
Furthermore, we establish inequalities relating these different variants, indeed clarifying that all these definitions coincide with the average pairwise fidelity for commuting states. We also introduce another multivariate fidelity called multivariate log-Euclidean fidelity, which is a quantum generalization of the Matusita multivariate fidelity. We also show that it satisfies most of the desirable properties listed above, it is a function of a multivariate log-Euclidean divergence, and it has an operational interpretation in terms of quantum hypothesis testing with an arbitrarily varying null hypothesis. Lastly, we propose multivariate generalizations of Matsumoto's geometric fidelity and establish several properties of them.
\end{abstract}

\setcounter{tocdepth}{2}
\tableofcontents

\section{Introduction}

Distinguishability and similarity are essential concepts that hold significance across all scientific disciplines. The fundamental toolkit for understanding these concepts revolves around measures of distinguishability and  similarity. 
This has led to several measures of distinguishability between probability distributions in classical information theory, such as the Kullback--Leibler divergence~\cite{kl1951} and R\'enyi relative entropy~\cite{renyi1961measures}, and similarity measures such as the Bhattacharyya overlap~\cite{Bhattacharyya1946}.
In quantum information theory, we have distinguishability measures between quantum states such as quantum relative entropy~\cite{umegaki1962conditional}, Petz--R\'enyi relative entropy~\cite{P86}, and sandwiched R\'enyi relative entropy~\cite{muller2013quantum, wilde2014strong}, and information quantities derived from these measures, such as quantum mutual information~\cite{stratonovich1965information}.  Moreover, the Uhlmann~\cite{uhlmann1976transition} and Holevo~\cite{holevo1972} fidelities  are similarity measures between two quantum states. 
All of the aforementioned distinguishability and similarity measures take into account two probability distributions in the classical case and two quantum states in the quantum case.

There exist a few extensions of the notion of the Bhattacharyya overlap or  bivariate classical fidelity to more than two probability distributions. Examples include the Matusita affinity~\cite{Matusita1967notion} and the average pairwise classical fidelity; herein we call all such extensions  multivariate classical fidelities.
Hitherto, quantum generalizations of multivariate classical fidelities have not been systematically explored. 
Existing multivariate measures that do not reduce to classical multivariate fidelities are as follows. The Holevo information of a tuple of states 
with uniform prior characterizes the distinguishability of the states (see~\eqref{eq:generalized_Holevo} for a precise definition).
Indeed, the Holevo information being approximately zero implies that the states in this tuple are similar, while it being far from zero indicates that they are distinguishable. The multivariate Chernoff divergence is defined in \cite[Definition~3]{mnw2023} for commuting states, and various generalizations to non-commuting states were postulated there. 
Moreover, multivariate generalizations of R\'enyi divergences have been explored in~\cite{mosonyi2022geometric} for general states and in~\cite{farooq2023asymptotic} for commuting states.

The main contribution of our paper is to introduce quantum generalizations of the Matusita affinity and average pairwise classical fidelity; here we call them multivariate quantum fidelities. We show that these multivariate fidelities satisfy several desirable properties that are natural generalizations of those of bivariate fidelities. Consequently, the multivariate quantum fidelities are bona fide measures of similarity between multiple states. See \Cref{sec:paper-contribs} for a more detailed discussion of our contributions.

\subsection{Classical and quantum bivariate fidelities}

\label{sec:c-q-bivariate-fid-review}

We first provide a brief overview of classical and quantum fidelities.

The bivariate classical fidelity quantifies the similarity of two discrete probability distributions $p$ and $q$ on a finite set $\cX$; it is defined as 
\begin{equation}\label{eq:bc}
    F(p,q) \coloneqq \sum_{x\in \cX} \sqrt{p(x) q(x)}.
\end{equation}
It satisfies $0 \leq F(p,q) \leq 1$ for all distributions $p$ and $q$,  it is equal to one if and only if $p=q$, and it is equal to zero if and only if the supports of $p$ and $q$ are disjoint. The above formulation can be naturally extended to commuting states as follows: for commuting states~$\rho$ and $\sigma$, define
\begin{equation}\label{eq:commuting_bivariate}
    F(\rho,\sigma) \coloneqq \sum_x \sqrt{\rho(x) \sigma(x)},
\end{equation}
where $(\rho(x))_x$ and $(\sigma(x))_x$ are the tuples of eigenvalues of $\rho$ and $\sigma$, respectively, in their common eigenbasis.

In the same spirit of distinguishability measures, there is an infinite number  of similarity measures that generalize the bivariate classical fidelity to quantum states. Here we call them $z$-fidelities and  define them for $z \geq 1/2 $ and quantum states $\rho$ and $\sigma$ as
\begin{equation}
    F_z(\rho,\sigma) \coloneqq \Tr\!\left[  \left( \sigma^{\frac{1}{4z}} \rho^{\frac{1}{2z}} \sigma^{\frac{1}{4z}} \right)^z \right] = \left\| \rho^{\frac{1}{4z}} \sigma^{\frac{1}{4z}} \right\|_{2z}^{2z}.
    \label{eq:z-fid-def}
\end{equation}
For every positive integer $z$, cyclicity of trace leads to the following expression:
\begin{equation}\label{eq:decompose_z}
    F_z(\rho, \sigma) = \operatorname{Tr}\!\left[\left(\rho^{1/2z} \sigma^{1/2z}\right)^z\right].
\end{equation}
{As indicated in \cite[Eq.~(4)]{audenaert2015alpha}, this identity in fact holds for all $z>0$. See also \cite{li2015some,baldwin2023efficiently} for the special case $z=1/2$.}

The $z$-fidelities are obtained by fixing $\alpha=1/2$ in the $\alpha$-$z$ R\'enyi relative entropies~\cite{audenaert2015alpha}. They reduce to the classical fidelity for commuting states, and each of them satisfies the data-processing inequality \cite[Theorem~1.2]{Z20}; i.e., for all $z\geq 1/2$ and for every quantum channel~$\cN$ and pair of states $\rho$ and $\sigma$, the following inequality holds:
\begin{equation}
    F_z(\rho,\sigma) \leq F_z\!\left( \cN(\rho), \cN(\sigma) \right).
    \label{eq:z-fid-DP}
\end{equation}
The $z$-fidelities are monotonically decreasing and continuous in~$z$ \cite[Proposition~6]{lin2015investigating}.
Notably, for $z=1/2$, the $z$-fidelity reduces to the Uhlmann fidelity~\cite{uhlmann1976transition}: 
\begin{equation}
    F_{1/2}(\rho,\sigma)=F(\rho,\sigma) \coloneqq \left\| \sqrt{\rho} \sqrt{\sigma} \right\|_1,
\end{equation}
which was expounded upon in~\cite{J94fid}. 
For $z=1$, the $z$-fidelity reduces to the Holevo fidelity~\cite{holevo1972}: 
\begin{equation}
    F_1(\rho,\sigma) = F_H(\rho, \sigma) \coloneqq \Tr\!\left[\sqrt{\rho} \sqrt{\sigma}\right].
\end{equation}
The log-Euclidean fidelity is defined as the $z\to \infty$ limit of~\eqref{eq:z-fid-def}, for which it is possible to obtain a closed-form expression by an application of the Lie--Trotter product formula. Indeed, for $\varepsilon>0$, define $\rho(\varepsilon)\coloneqq
\rho+\varepsilon I$ and $\sigma(\varepsilon)\coloneqq
\sigma+\varepsilon I$. Then
\begin{equation}
    \begin{aligned}
   F_{\flat}(\rho,\sigma)   \coloneqq  \lim_{z\to \infty} F_z(\rho,\sigma) 
    &= \inf_{\varepsilon >0} \operatorname{Tr}\!\left[\exp\!\left(\frac{1}{2} (\ln \rho(\varepsilon) + \ln \sigma(\varepsilon))\right)\right] \\
   & = \lim_{\varepsilon \to  0^+} \operatorname{Tr}\!\left[\exp\!\left(\frac{1}{2} (\ln \rho(\varepsilon) + \ln \sigma(\varepsilon))\right)\right].
\end{aligned}
\label{eq:log-euc-fid-limit}
\end{equation}
See \cref{sec:proof-log-euclid-fid} for a short proof of~\eqref{eq:log-euc-fid-limit}, and see \cite[Section~4]{audenaert2015alpha} and \cite[Eq.~(17)]{mosonyi2017strong} for a quantity that generalizes the log-Euclidean fidelity. 

For two pure states $|\psi\rangle\!\langle \psi|$ and $|\phi\rangle\!\langle \phi|$, the $z$-fidelity reduces to
\begin{equation}
    F_z(|\psi\rangle\!\langle \psi|, |\phi\rangle\!\langle \phi|) = | \langle \psi | \phi \rangle |^{2z}.
\end{equation}
Operationally, this is the $z$th power of the probability for the state $|\psi\rangle\!\langle \psi|$ to pass a test for being the state $|\phi\rangle\!\langle \phi|$.

Another generalization of the classical bivariate fidelity is Matsumoto's geometric fidelity \cite{matsumoto2010reverse}, defined as
\begin{equation}
\label{eq:intro-geo-fid-def}
     F_G(\rho,\sigma)\coloneqq \inf_{\varepsilon >0 } \Tr[ \rho(\varepsilon) \# \sigma(\varepsilon)],
\end{equation}
where $A \# B$ denotes the matrix geometric mean of the positive definite operators $A$ and $B$:
\begin{equation}
    A \# B \coloneqq A^{1/2}\left( A^{-1/2} B A^{-1/2} \right)^{1/2} A^{1/2}, \label{eq:geometric-mean}
\end{equation}
$\rho(\varepsilon) \coloneqq \rho + \varepsilon I$,
and $\sigma(\varepsilon) \coloneqq \sigma + \varepsilon I$. This quantity obeys several natural properties desired for a quantum fidelity and was studied further in \cite{cree2020fidelity}.

The bivariate setting focuses on the similarity of two states. Going forward, the main goal of our paper is to identify measures for quantifying the similarity of a given \textit{tuple of states}.
We focus on two main approaches to arrive at multivariate quantum fidelities, as shown in \cref{fig:boxes_fidelities}. The first approach is to generalize multivariate classical fidelities, which reduce to the bivariate classical fidelity for two commuting states, to a tuple of non-commuting states. The second is to generalize bivariate quantum fidelities to multivariate quantum fidelities that reduce to the bivariate setting when two states are considered. We also note that some generalizations can be obtained by following either of the aforementioned approaches.

\begin{figure}
\centering
\begin{tikzpicture}
[squarednode/.style={rectangle, draw=black!60, very thick, minimum height=3em,minimum width=3em,align=center},distance=4cm]
      \node [squarednode,text width=5cm] (model) {\textbf{Multivariate classical} \\ \textbf{fidelities} \\[.25em] 
    $F$, $F_r$, $F_{k,r}$};
    \node [squarednode, text width=5.5cm, right=of model.0] (model2) {\textbf{Multivariate quantum} \\  \textbf{fidelities} \\[.25em] $F_z$, $F_{\operatorname{SDP}}$, $F_S$, $  F^{\flat}_r$, $F_{k,r}^{\flat}$, $F_G$, $F^G_{\operatorname{SDP}}$};
        \node [squarednode, below=of model, text width=5cm] (model3){\textbf{Bivariate classical fidelity} \\[.25em] $F$};
    \node [squarednode,right=of model3.0, below=of model2, text width=5.5cm] (model4) {\textbf{Bivariate quantum fidelities} \\[.25em] $F_z$, $F_G$ 
    };
    \draw [-latex,very thick] (model) -- (model2); 
    \draw[-latex,very thick] (model3.north) -- (model.south);
    \draw [-latex,very thick] (model3) -- (model4); 
    \draw[-latex,very thick] (model4.north) -- (model2.south);
\end{tikzpicture}
\caption{\small Two main approaches employed in this work to arrive at multivariate quantum fidelities are as follows. Our first approach is to generalize multivariate classical fidelities to a tuple of non-commuting states. These include the average pairwise fidelity $F$ in~\eqref{eq:average_pairwise_commuting}, the Matusita multivariate fidelity $F_r$ in~\eqref{eq:Matusita_fid_commuting}, and the average $k$-wise fidelities in~\eqref{eq:def-avg-k-wise}, all of which reduce to the bivariate classical fidelity in~\eqref{eq:bc} for two commuting states. Our second approach is to generalize  bivariate quantum fidelities ($z$-fidelity $F_z$ in~\eqref{eq:z-fid-def} or geometric fidelity $F_G$ in~\eqref{eq:intro-geo-fid-def}) to multivariate quantum fidelities, each of which reduces to a bivariate quantum fidelity when two general states are considered. With these two approaches, we present several multivariate quantum fidelities:  average pairwise $z$-fidelity $F_z$ (\cref{subS:avg_pairwise});   multivariate SDP fidelity $F_{\operatorname{SDP}}$ (\cref{subS:SDP_F});   secrecy-based multivariate fidelity  $F_S$ (\cref{subS:secrecy_based_F}); multivariate log-Euclidean fidelity~$  F^{\flat}_r$ (\cref{subS:loog_eucledian_F}); average $k$-wise log-Euclidean fidelity $F_{k,r}^{\flat}$ (\cref{subS:average_k_log_eucli_F}); average pairwise geometric fidelity $F_G$  and multivariate geometric SDP fidelity $F^G_{\operatorname{SDP}}$ (\cref{sec:multiv-geo-fids}).}
\label{fig:boxes_fidelities}
\end{figure}
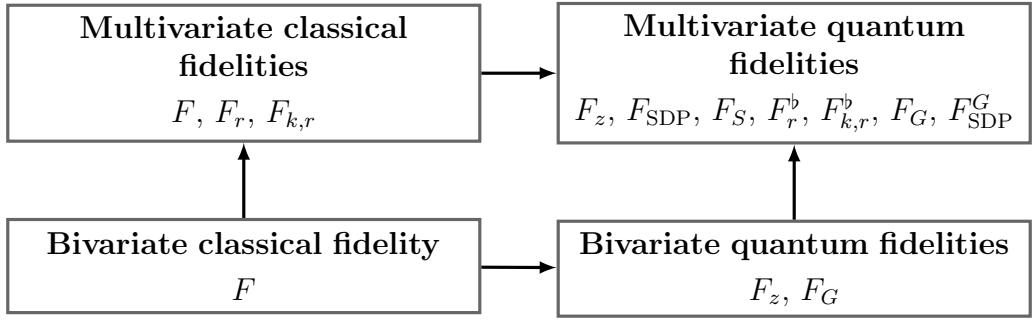

\subsection{Contributions}

\label{sec:paper-contribs}

In this paper, we comprehensively study multivariate generalizations of bivariate fidelity, while establishing connections between different formulations. 

First, we recall some multivariate classical fidelities for commuting states, namely, the Matusita multivariate fidelity and average pairwise fidelity, and then we establish an inequality relating them (\cref{prop:r_root_to_commuting_fidelity}). Thereafter we generalize these quantities to average $k$-wise fidelities (\cref{def:average_k_wise_classical}) and prove that they are ordered (\cref{prop:ineqs-avg-k-wise-classical}). 

Next, we introduce quantum generalizations,  mainly focusing on three variants that reduce to average pairwise fidelity for commuting states: average pairwise fidelity using $z$-fidelities for $z \geq  1/2$  (\cref{def:average_pairwise_z_fidelity}); multivariate semi-definite programming (SDP) fidelity (\cref{def:multivar-fid}); and secrecy-based multivariate fidelity (\cref{def:secrecy_based_multi}). The multivariate SDP fidelity is obtained by generalizing the SDP of Uhlmann fidelity from~\cite{watrous2012simpler} to multiple states. The secrecy-based multivariate fidelity is inspired by an existing secrecy measure from~\cite[Eq.~(19)]{konig2009operational}---the average fidelity between each state of the tuple and a fixed state, where it is maximized over all such fixed states (c.f.,~\eqref{eq:secrecy_measure} for the definition). 
With these definitions, we show that both SDP fidelity  and secrecy-based multivariate fidelity are sandwiched between the average pairwise Holevo fidelity $(z=1)$ and the average pairwise Uhlmann fidelity $(z=1/2)$, implying that all of these variants converge to average pairwise fidelities in the commuting case (see \cref{thm:SDP_fidelity_pairwise_upper_and_lower}).

\cref{Prop:properties_average_pairwise_z}, \cref{thm:properties_SDP_fidelity},  and \cref{thm:properties_secrecy_multi_F} assert that all three of the aforementioned quantum variants satisfy a number of properties desired for a multivariate fidelity: (i) reduction to multivariate classical fidelities for commuting states; (ii) data processing; (iii) permutation invariance; (iv) faithfulness (i.e., it is equal to one if and only if the states are the same); (v) orthogonality (i.e., it is equal to zero if and only if the tuple of states forms an  orthogonal set); (vi) direct-sum property; and (vii) joint concavity. In addition, the average pairwise fidelity satisfies super-multiplicativity, whereas the SDP fidelity and secrecy-based multivariate fidelity do not satisfy this property in general. \cref{thm:uniform_cont_SDP_F} establishes that the multivariate SDP fidelity satisfies a uniform continuity bound, which follows by means of an alternative formulation for the SDP fidelity in \cref{prop:another_formulation_multi_SDP_fid}. \cref{Prop:uniform_cont_average_pairwise} and \cref{prop:uniform_cont_secrecy_multi_F} state that the average pairwise Uhlmann and Holevo fidelities and secrecy-based multivariate fidelity also satisfy uniform continuity bounds, respectively. In addition, we define maximal and minimal extensions of multivariate classical fidelities (\cref{def:maximal_multi_fidelity} and \cref{def:minimal_multi_F}) and analyse their properties. 

Furthermore, we explore a quantum generalization of the Matusita multivariate fidelity, and we show that it is a special case of a multivariate log-Euclidean divergence. We call this variant the multivariate log-Euclidean fidelity.  We show that this satisfies most of the desirable properties of a multivariate fidelity in \cref{thm:properties_quantum_Matusita_multi}.
We give an operational interpretation of all multivariate log-Euclidean divergences in terms of quantum hypothesis testing with an arbitrarily varying null hypothesis (\cref{Cor:log_Euclidean_operational}), by making use of the main result of~\cite{Notzel_2014}. In addition, we define the \textit{oveloh information} in~\eqref{eq:oveloh-info} and show that it is related to the multivariate log-Euclidean fidelity in~\eqref{eq:quantum_Matusita_oveloh}. In passing, we also derive a connection between the Holevo information and oveloh information in \cref{Cor:Holevo_and_Oveloh}. We also define average $k$-wise log-Euclidean fidelities (\cref{def:average_k_log_Eucl}) and show that they are ordered (\cref{prop:avg-k-wise-log-Euc-ordered}), as a generalization of the aforementioned average $k$-wise classical fidelities.

Lastly, we define multivariate geometric fidelities (\cref{def:avg_geometric} and \cref{def:SDP_geometric_multi}) by generalizing geometric fidelity (also known as Matsumoto fidelity)~\cite{matsumoto2010reverse,cree2020fidelity} with the same procedures used in generalizing Uhlmann fidelity, and we analyse their properties.

\subsection{Paper organization}

The rest of our paper is organized as follows.  In Section~\ref{Sec:Background}, we introduce notation and background needed to understand the rest of the paper. Section~\ref{Sec:BivariateFidelity} reviews and presents different formulations of the bivariate Uhlmann fidelity. Multivariate classical fidelities are defined and studied in Section~\ref{Sec:classical_multivariate}. In Section~\ref{Sec:QuantumMultivariate_Fidelities}, we introduce proposals for multivariate quantum fidelity,  focusing on 
their properties, relationships between different formulations, and operational interpretations of some of them.  
In the appendices, we provide mathematical proofs of many results stated in the paper.
Finally, Section~\ref{Sec:Conclusion} provides concluding remarks and future directions.

\section{Notations and background} \label{Sec:Background}

We begin by reviewing basic concepts from quantum information theory and refer the reader to~\cite{khatri2020principles} for more details. A quantum system~$R$ is identified with a finite-dimensional complex Hilbert space~$\mathcal{H}_R$ with inner product denoted in bra-ket notation as $\langle\psi|\phi\rangle$ for $|\psi\rangle,|\phi\rangle\in\cH_R$. We denote the set of linear operators acting on $\mathcal{H}_R$ by $\mathscr{L}_R$. The support of a linear operator $X \in \mathscr{L}_R$ is defined to be the orthogonal complement of its kernel, and we denote it by $\operatorname{supp}(X)$. For a probability distribution $p = (p_1,\dots,p_n)$ the support is defined as $\supp(p) = \lbrace i\colon p_i>0\rbrace$. We denote the {\it Hermitian conjugate} or {\it adjoint} of $X$ by $X^{\dagger}$, which is the unique linear operator acting on $\mathcal{H}_R$ that satisfies 
$\langle \psi| (X^\dagger |\phi\rangle) = (X |\psi\rangle)^\dagger |\phi\rangle$
for all $| \psi \rangle, | \phi \rangle \in \mathcal{H}_R$.
Here, $|\psi\rangle^\dagger \equiv \langle\psi|$ is the linear functional $\mathcal{H}_R\to\mathbb{C}$ given by $ |\phi\rangle\mapsto \langle\psi|\phi\rangle$.
The set $\mathscr{L}_R$ is a linear space and has a Hilbert-space structure given by the Hilbert--Schmidt inner product, defined as $\langle X, Y \rangle \coloneqq \Tr[X^\dagger Y]$ for all $X, Y \in \mathscr{L}_R$.
Given two quantum systems $A$ and $B$, with respective Hilbert spaces $\mathcal{H}_A$ and $\mathcal{H}_B$, the Hilbert space of the composite system $AB$ is given by $\mathcal{H}_{AB}\coloneqq\mathcal{H}_A \otimes \mathcal{H}_B$. For $K_{AB} \in \mathscr{L}_{AB}$, let $\Tr\!\left[K_{AB} \right]$ denote the trace of~$K_{AB}$, and let $\Tr_A \!\left[K_{AB}\right]$ denote the partial trace of $K$ over the system~$A$.
We use the standard notation $K_A \equiv \Tr_{B}[K_{AB}]$ and $K_B \equiv \Tr_{A}[K_{AB}]$ to denote the marginals of $K_{AB}$. The trace norm  of an operator $K$ is defined as $\left\|K\right\|_1 \coloneqq \Tr[\sqrt{K^\dagger K} ]$. For Hermitian operators 
$K$ and $L$, the notation $K \geq L$ indicates that $K-L$ is a positive semi-definite (PSD) operator, while $K > L$ indicates that $K-L$ is a positive definite operator. Given a complex number $z$, we denote the real part of $z$ by $\mathfrak{R}[z]$.

A quantum state of a system $R$ is identified with a density operator $\rho_R \in \mathscr{L}_R$, which is a PSD operator of unit trace. 
We denote the set of all density operators acting on $\mathcal{H}_R$ by~$\mathscr{D}_R$. We also use the notations $\mathscr{D}$ and $\mathscr{L}$ to denote the sets of density operators and linear operators, respectively, when there is no ambiguity regarding the underlying quantum system. A quantum state $\rho_R$  is called a pure state if its rank is equal to one, and in this case, there exists a state vector $| \psi \rangle_R \in \mathcal{H}_R$ such that $\rho_R= | \psi \rangle\!\langle \psi |_R $. Otherwise, $\rho_R$ is called a mixed state. By the spectral decomposition theorem, every state can be written as a convex combination of pure and mutually orthogonal states. 
A quantum channel from system $A$ to system $B$ is a linear, 
completely positive and trace-preserving (CPTP) map from $\mathscr{L}_A$ to $\mathscr{L}_B$. 
A measurement of a quantum system $R$ is described by a
positive operator-valued measure (POVM), which is defined as a tuple of PSD operators $(M_y)_{y \in \mathcal{Y}} $ satisfying $\sum_{y \in \mathcal{Y}} M_y= I_{R}$, where $I_{R}$ is the identity operator acting on~$\mathcal{H}_R$ and $\mathcal{Y}$ is a finite alphabet. The Born rule asserts that, when applying the above POVM to a state $\rho$, the probability of observing the outcome $y$ is given by $\Tr\!\left[M_y \rho \right]$~\cite{born1926quantum}.  Associated with any POVM $(M_y)_{y \in \mathcal{Y}}$ is a measurement or quantum-to-classical channel $\mathcal{M}$ described as follows. Let $\mathcal{K}$ be a complex Hilbert space of dimension $|\mathcal{Y}|$, and let $\{|y \rangle: y \in \mathcal{Y}\}$ be a fixed orthonormal basis of $\mathcal{K}$. Given  an  input state $\omega$, the output of the measurement channel is given by
\begin{equation}\label{eq:measurement_channel}
    \mathcal{M}(\omega) \coloneqq \sum_{y \in \mathcal{Y}} \Tr\!\left[ M_y \omega \right]   |y\rangle\!\langle y|.
\end{equation}

Finally, throughout our paper, we adopt the shorthand $[r] \coloneqq \{1, \ldots, r\}$, and we let $S_r$ denote the permutation group on $[r]$.

\section{Many forms of Uhlmann fidelity and its properties} \label{Sec:BivariateFidelity}

In this section, we discuss several equivalent formulations for the Uhlmann fidelity, along with its various properties. This section also serves as a foundation for subsequent sections, in which we extend the Uhlmann fidelity to multivariate fidelities.

We adopt the following notations here: $\mathscr{U}_R$ denotes the set of unitary operators acting on~$\mathcal{H}_R$;
$\textnormal{POVM}_A$ denotes the set of POVMs acting on system $A$; and $\mathscr{L}_A^{+} \subset \mathscr{L}_A$ denotes the set of positive semi-definite operators.

Let us begin by recalling that the Uhlmann fidelity has several equivalent formulations, as stated in the following proposition. This is one of the reasons that it is the {most widely used} bivariate quantum fidelity.

\begin{proposition}[Equivalent formulations of Uhlmann fidelity]\label{prop:bivariate_fidelity_formulas}
Let $\rho$ and $\sigma$ be quantum states of a system $A$. Let $R$ be a reference system of the same dimension as $A$, and let {$| \psi^\rho \rangle_{AR}$ and $| \psi^\sigma \rangle_{AR}$} be purifications of $\rho$ and $\sigma$, respectively.
 The Uhlmann fidelity $F(\rho,\sigma)$ between $\rho$ and $\sigma$ is equal to any one of the following expressions:
  \begin{align}
      \notag F(\rho,\sigma) &\coloneqq \left\| \sqrt{\rho} \sqrt{\sigma} \right\|_1 \\
      &=\sup_{U_R \in \mathscr{U}_R }\left| \langle \psi^\rho \middle | U_R \otimes I_A \middle| \psi^\sigma\rangle \right| \label{eq:uhlmann}\\ 
      &=\inf_{(\Lambda_x)_x \in \textnormal{ POVM}_A} \sum_x \sqrt{\Tr[\Lambda_x \rho] \Tr[\Lambda_x \sigma]} \label{eq:bivariate_measured} \\ 
        &= \frac{1}{2} \inf_{Y_1, Y_2 \geq 0} \left\{ \Tr[Y_1 \rho] + \Tr[Y_2 \sigma]:  \begin{bmatrix}
        Y_1 & -I \\
        -I & Y_2
        \end{bmatrix} \geq 0  \right\} \label{eq:bivariate_fid_primal}\\ 
          &=\sup_{X \in \mathscr{L}_A} \left\{ \mathfrak{R}\!\left[ \Tr[X] \right] :   \begin{bmatrix}
        \rho & X \\
        X^\dagger & \sigma
        \end{bmatrix} \geq 0  \right\}. \label{eq:bivariate_fid_dual} \\
       & = {\Tr[\sqrt{\rho \sigma}] } \label{eq:simplified_bi_fid}
  \end{align}
  {Furthermore, when $\rho$ and $\sigma$ are invertible, the Uhlmann fidelity can also be written as follows:
  \begin{equation} \label{eq:invertible_fidelity}
      F(\rho,\sigma)= \Tr\!\left[  \left(\rho^{-1} \#\sigma \right) \rho \right],
  \end{equation}
where $A \# B$ denotes the geometric mean of the positive definite operators $A$ and $B$, as defined in~\eqref{eq:geometric-mean}. 
  }\end{proposition}

 Uhlmann's theorem~\eqref{eq:uhlmann} was established in~\cite{uhlmann1976transition}. The fact that the fidelity is achieved by a quantum measurement~\eqref{eq:bivariate_measured} was realized in \cite[Eq.~(7)]{fuchs1995mathematical}. The SDPs~\eqref{eq:bivariate_fid_primal} and~\eqref{eq:bivariate_fid_dual} for bivariate fidelity were established in~\cite{watrous2012simpler}. We refer to~\cite{khatri2020principles} for proofs of the expressions above (i.e., Theorem~6.8, Theorem~6.12, and Proposition~6.6 therein). We also note that the supremum and infimum in the SDP formulations are achievable and can be replaced with a maximum and minimum, respectively. The simplified expression in~\eqref{eq:simplified_bi_fid} follows from~\cite[Eq.~(4)]{audenaert2015alpha} and \cite{li2015some,baldwin2023efficiently}.
  For invertible states, \eqref{eq:invertible_fidelity} follows by identifying that $Y= \rho^{-1}\# \sigma$ is the optimal choice for the SDP in~\eqref{eq:bivariate_fid_primal}, which follows due to Alberti Theorem as shown in~\cite{alberti1983stochastic}. 
In addition to these expressions, a corollary of \Cref{prop:another_formulation_multi_SDP_fid} in \Cref{Sec:QuantumMultivariate_Fidelities} is a new formulation for bivariate fidelity of two states $\rho$ and $\sigma$, which we state as \Cref{Cor:another_formulation_bivariate}.

In the following proposition, we list several properties of the Uhlmann fidelity that we use as a guide for defining multivariate quantum fidelities. 
\begin{proposition}[Properties of Uhlmann fidelity]\label{prop:properties_bivariate_F}
    The Uhlmann fidelity satisfies the following properties for  states $\rho$ and $\sigma$:
    \begin{enumerate}[label={(\roman*)}]
        \item Reduction to classical fidelity: If the states $\rho$ and $\sigma$ commute, then
        \begin{equation}
            F(\rho,\sigma) = \sum_{x} \sqrt{\rho(x) \sigma(x)},
        \end{equation}
        where $(\rho(x))_x$ and $(\sigma(x))_x$ are the spectra of $\rho$ and $\sigma$, respectively, in their common eigenbasis.
        \item Data-processing inequality: For every quantum channel $\cN$, we have
        \begin{equation}
            F(\rho,\sigma) \leq F\!\left( \cN(\rho), \cN(\sigma) \right).
        \end{equation}
        \item Symmetry: Uhlmann fidelity is independent of the order of the quantum states, i.e.,
        \begin{equation}
            F(\rho,\sigma)= F(\sigma,\rho).
        \end{equation}
    \item Faithfulness and orthogonality: Uhlmann fidelity satisfies the inequalities $0\leq F(\rho, \sigma) \leq 1$. Furthermore, $F(\rho,\sigma)=1$ if and only if $\rho=\sigma$, and 
         $F(\rho,\sigma)=0$ if and only if $\rho$ and~$\sigma$ are orthogonal to each other, i.e., $\rho \sigma=0$.
         
        \item Direct-sum property: Let $(\rho_x)_{x \in \cX},$ and $(\sigma_x)_{x\in \cX}$ be tuples of quantum states, where $\cX$ is a finite alphabet. For the classical-quantum states formed using those tuples, the following equality holds: 
    \begin{equation}\label{eq:CQequality}
         F\! \left( \sum_{x \in \cX} p(x) |x\rangle\!\langle x| \otimes \rho_x, \sum_{x \in \cX} p(x)  |x\rangle\!\langle x| \otimes \sigma_x\right) = \sum_{x \in \cX} p(x)  F(\rho_x, \sigma_x),
    \end{equation} 
    where $p$ is an arbitrary probability distribution on $\cX$.
    
    \item Joint concavity: Let $(\rho_x)_{x \in \cX},$ and $(\sigma_x)_{x\in \cX}$ be tuples of quantum states for a finite alphabet $\cX$, and let $(p(x))_{x\in \cX}$ be a probability distribution. Then,
    \begin{equation}
        F\! \left( \sum_{x \in \cX} p(x) \rho_x, \sum_{x \in \cX} p(x) \sigma_x\right) \geq \sum_{x \in \cX} p(x)  F(\rho_x, \sigma_x).
    \end{equation}
  
    \end{enumerate}
\end{proposition}
  We note that the joint concavity of Uhlmann fidelity is a consequence of its properties {$(ii)$}~data-processing inequality under the partial trace channel and {$(v)$}~direct-sum property. We refer the reader to~\cite{khatri2020principles} for proofs of these properties.\footnote{In particular, we refer to Theorem~6.9 for (ii), Theorem~6.7 for (iv), and Theorem~6.11 for (vi) in~\cite{khatri2020principles}, while (i) and (iii) follow by the definition of fidelity.}
  
\section{Multivariate classical fidelities} 
\label{Sec:classical_multivariate}

The basic requirement that we set for a multivariate classical fidelity is that it should reduce to the bivariate classical fidelity in~\eqref{eq:bc}, when only two probability distributions are being considered. Further requirements include satisfying the properties outlined in \cref{prop:properties_bivariate_F}, such as the data-processing inequality, symmetry under exchange of the probability distributions, faithfulness, orthogonality, and the direct-sum property. As such, there are many possible functions that satisfy these requirements, and 
the aims of this section are to outline several variants of multivariate classical fidelity and to establish relationships between them.
In subsequent sections, 
we show how quantum generalizations reduce to these definitions for commuting states (states represented using classical distributions), which is one of the desired properties for a quantum generalization of a classical measure. 

\subsection{Matusita multivariate fidelity}

Matusita  introduced a generalization of the bivariate classical fidelity of two probability distributions to several probability distributions~\cite{Matusita1967notion}, and therein it was called the affinity of several distributions. Here we call it the Matusita multivariate fidelity, and we recall its definition now.

\begin{definition}[Matusita multivariate fidelity] \label{def:Matusita_multi_classical}
For $r\in\mathbb{N}$, let $p_1,\ldots, p_r$ be probability distributions on a finite set $\cX$. 
    The Matusita multivariate fidelity is defined as 
\begin{equation}\label{eq:Matusita_fid_probability}
        F_r(p_1, \ldots, p_r) \coloneqq \sum_{x \in \cX} \left(p_1(x) \cdots p_r(x)\right)^{\frac{1}{r}}.
    \end{equation}
Furthermore, for a tuple of commuting states $(\rho_i)_{i=1}^r$, the Matusita multivariate fidelity is defined as 
\begin{equation}\label{eq:Matusita_fid_commuting}
        F_r(\rho_1, \ldots, \rho_r) \coloneqq \sum_{x} \left( \rho_1(x) \cdots \rho_r(x)\right)^{\frac{1}{r}}
    \end{equation}
 where, for $i\in [r]$, $\left(\rho_i(x)\right)_x$ is the spectrum of the state $\rho_i$  in the common eigenbasis of the states.
\end{definition}

The Matusita multivariate fidelity satisfies the data-processing inequality, symmetry under exchange of the probability distributions, and the direct-sum property. It also satisfies $0 \leq F_r(p_1, \ldots, p_r) \leq 1$, and the equality $F_r(p_1, \ldots, p_r)=1$ holds if and only if $p_1=\cdots = p_r$.  See Theorem~1 and Corollary~3 of~\cite{Matusita1967notion}. Also, if there exists at least one pair of distributions $p_i$ and $p_j$ with $i\neq j \in [r]$ that have disjoint support, then  $F_r(p_1, \ldots, p_r)=0$. However,  $F_r(p_1, \ldots, p_r)=0$ does not imply that at least one pair of distributions is disjoint. 
For example, the probability vectors
\begin{equation}
    p_1=(0,1/2, 1/2), \qquad  p_2=(1/2, 0, 1/2), \qquad p_3=(1/2, 1/2, 0)
\label{eq:counterexample-orthogonality-Matusita}
\end{equation}
satisfy $F_r(p_1,p_2,p_3)=0$, but no two distributions have disjoint support.

The Matusita multivariate fidelity can be written in terms of the classical relative entropy 
\begin{equation}
D(p\Vert q) \coloneqq \begin{cases}
    \sum_x p(x) \ln\!\left(\frac{p(x)}{q(x)}\right) & \text{ if } \text{supp}(p)\subseteq \text{supp}(q) \\
    +\infty & \text{ else }
\end{cases}
\end{equation}
as follows:
\begin{equation}
    F_r(p_1, \ldots, p_r) = \exp\!\left(-\inf_{w} \frac{1}{r} \sum\nolimits_i D(w\Vert p_i)\right),
    \label{eq:rel-ent-Matusita}
\end{equation}
where the infimum is over every probability distribution~$w$ with support contained in the intersection of the supports of $p_1, \ldots, p_r$.
If no such distribution exists, then the relative entropy is equal to $+\infty$, and we set $F_r = 0$.
We review this representation of the Matusita fidelity in the more general quantum case in \Cref{subS:loog_eucledian_F}, 
where we also provide an operational interpretation of it. Since the relative entropy is well known to obey the data-processing inequality, the formulation of the Matusita fidelity in~\eqref{eq:rel-ent-Matusita} makes it easier to see that it obeys the data-processing inequality as well.

\begin{remark}[Reduction to R\'enyi relative entropy \& Hellinger transform]\ \\
    The R\'enyi relative entropy between two probability distributions $p$ and $q$  is defined  for $\alpha \in (0,1) \cup (1, \infty)$ as~\cite{renyi1961measures}
    \begin{equation}
D_\alpha(p\Vert q) \coloneqq \begin{cases}
   \frac{1}{\alpha-1} \ln \!\left(\sum_x p(x) ^\alpha  q(x)^{1-\alpha} \right) & \text{ if } \alpha \in (0,1) \vee (\alpha > 1 \wedge \operatorname{supp}(p)\subseteq \operatorname{supp}(q)) \\
    +\infty & \text{ else }
\end{cases}.
\end{equation}
Let $\alpha \in (0,1)$ be rational such that $\alpha=t/r$ with $t<r$. Then, for such $\alpha$ we have that 
\begin{equation}
    D_\alpha(p\Vert q) =  \frac{1}{\alpha-1} \ln \!\left(F_r(p, \ldots,p,q,\ldots,q) \right),
\end{equation}
where, on the right-hand side, $p$ occurs $t$ times and $q$ occurs $r-t$ times. Thus, the R\'enyi relative entropy of rational order is a special case of the Matusita fidelity. 

For a tuple of probability distributions $p_1, \ldots, p_r$ and a probability vector $s=(s_1,\dots,s_r)$, the Hellinger transform is defined as 
(see \cite[page~189]{Matusita1967notion} and \cite[Eq.~(34)]{toussaint1974some}) 
\begin{equation}
H_s(p_1, ..., p_r) \coloneqq  \sum_x \prod_{i=1}^r  p_i(x)^{s_i}.
\label{eq:hellinger-transform-def}
\end{equation}
If every $s_i$ is rational such that $s_i=t_i/t$, then we have 
\begin{equation}
    H_s(p_1, ..., p_r)= F_t(p_1,\ldots,p_1, p_2, \ldots,p_2,\ldots, p_r, \ldots, p_r),
\end{equation}
where, on the right-hand side, $p_i$ appears $t_i$ times, for all $i \in [r]$. As such, the Hellinger transform for rational $s$ is a special case of the Matusita fidelity.
\end{remark}

Next, we state that the Matusita multivariate fidelity obeys a uniform continuity bound, which quantifies its deviation in terms of the deviation of corresponding probability distributions in two different tuples of probability distributions. To quantify the deviation of corresponding probability distributions, we employ the Hellinger distance, defined for two probability distributions $p$ and $q$ as 
\begin{equation} \label{eq:Helinger_classical}
    d_H (p, q)  \coloneqq \left\| \sqrt{p} - \sqrt{q} \right\|_2 = \left(\sum_{x} \left( \sqrt{p(x)} - \sqrt{q(x)} \right)^2\right)^{\frac{1}{2}} .
\end{equation}

\begin{proposition}
[Uniform continuity of Matusita multivariate fidelity]
\label{prop:unif-cont-Matusita}\ \\
Let $\left(  p_{1},\ldots,p_{r}\right)  $ and
$\left(  q_{1},\ldots,q_{r}\right)  $ be two
tuples of probability distributions, and let $\varepsilon\geq 0$ be such that
$\frac{1}{r}\left(  \sum_{i=1}^{r}  d_H (p_i, q_i)^2 \right)  \leq\varepsilon$, where $ d_H (\cdot, \cdot) $ is defined in~\eqref{eq:Helinger_classical}. Then,
\begin{equation}
\left\vert F_{r}(p_{1},\ldots,p_{r})-F_{r}(q_{1},\ldots,q_{r})\right\vert \leq {r \sqrt[r]{\varepsilon}}.
\end{equation}

\end{proposition}

\begin{proof}
    See \Cref{sec:proof-unif-cont-Matusita}.
\end{proof}

\subsection{Average pairwise fidelity}

The average pairwise Bhattacharyya overlap was defined in \cite[Eq.~(7)]{averageBattacharya} and shown to be an upper bound on the error probability of an uncoded transmission of classical data through polarized channels. We also call this quantity the  average pairwise fidelity, and note here that it has been considered previously in a similar context for the quantum case \cite[Section~II]{nasser2018polar}.

Here we recall the definition of the average pairwise fidelity for classical probability distributions and  for commuting states more generally.

\begin{definition}[Average pairwise fidelity] \label{def:classical_average_fidelity}
Let $p_1,\ldots, p_r$ be probability distributions on a finite set $\cX$. 
The average pairwise fidelity is defined as 
\begin{align}
    F(p_1,\ldots, p_r) & \coloneqq \frac{2}{r(r-1)}  \sum_{i <j} F (p_i, p_j) \\
    & = \frac{2}{r(r-1)}  \sum_{i <j} \sum_{x \in \cX} \sqrt{p_i(x) p_j(x)}.
\end{align}
Furthermore, for  commuting states $\rho_1, \ldots, \rho_r$, the average pairwise fidelity is defined as 
    \begin{align}
        F(\rho_1, \ldots, \rho_r) & \coloneqq \frac{2}{r(r-1)} \sum_{i <j} F(\rho_i,\rho_j) \label{eq:average_pairwise_commuting} \\
        &=  \frac{2}{r(r-1)}  \sum_{i <j} \sum_x \sqrt{\rho_i(x) \rho_j(x)} ,
    \end{align}
where $(\rho_i(x))_x$ is the tuple of eigenvalues of $\rho_i$.
\end{definition}

\begin{remark}[Uniform continuity of average pairwise fidelity]
The average pairwise fidelity of commuting states obeys a uniform continuity bound, which follows as an immediate consequence of \Cref{Prop:uniform_cont_average_pairwise} below, the latter holding for general quantum states.    
\end{remark}

\begin{proposition}[Inequalities relating  multivariate classical fidelities]
\label{prop:r_root_to_commuting_fidelity}
  For a tuple of commuting states $(\rho_i)_{i=1}^r$, the following inequality holds, relating the average pairwise fidelity $F(\rho_1, \ldots, \rho_r)$ and the Matusita fidelity $F_r(\rho_1, \ldots, \rho_r)$:
  \begin{equation}
  \label{eq:rel-Matusita-average-fid}
        F(\rho_1, \ldots, \rho_r) \geq F_r(\rho_1, \ldots, \rho_r).
   \end{equation}
\end{proposition}

\begin{proof}
By \cref{def:classical_average_fidelity}, we have
\begin{align}
        F(\rho_1, \ldots, \rho_r)  & =\frac{2}{r(r-1)} \sum_{i<j} \sum_{x} \sqrt{\rho_i(x) \rho_j(x)} \\
        & = \sum_{x}
 \frac{2}{r(r-1)} \sum_{i<j} \sqrt{\rho_i(x) \rho_j(x)},
    \end{align}
where $(\rho_i(x))_x$ is the tuple of eigenvalues of $\rho_i$, for $i\in [r]$.
For each $x$, the right-hand side of the above expression is an average of $r(r-1)/2$ terms.
Applying the inequality of arithmetic and geometric means then gives
\begin{align}
    \frac{2}{r(r-1)} \sum_{i<j} \sqrt{\rho_i(x) \rho_j(x)} & \geq  \prod_{i<j}  \left(\sqrt{\rho_i(x) \rho_j(x)} \ \right)^{\frac{2}{r(r-1)}}\\
    & = \prod_{i<j}  \left({\rho_i(x) \rho_j(x)}\right)^{\frac{1}{r(r-1)}}.
    \label{eq:am-gm-inequality-fidelity}
\end{align}
The desired inequality~\eqref{eq:rel-Matusita-average-fid} then follows from the above inequality because each $\rho_i(x)$ appears exactly $r-1$ times in the  product in~\eqref{eq:am-gm-inequality-fidelity}.
\end{proof}
\medskip

Later on in \Cref{prop:upper_multi_log_eucli}, we generalize the inequality in~\eqref{eq:rel-Matusita-average-fid} to the log-Euclidean class of quantum fidelities.

\subsection{Average $k$-wise fidelities}

\label{sec:avg-classical-k-wise-fids}

As a generalization of the average pairwise fidelity and to interpolate between this quantity and the Matusita fidelity of order $r$, we define the average $k$-wise fidelities of a tuple of $r$~commuting states as follows:
\begin{definition}[Average $k$-wise fidelities] \label{def:average_k_wise_classical}
    For $r\in \{3, 4, \ldots\}$, let $\rho_1, \ldots, \rho_r$ be a tuple of commuting states, where $(\rho_i(x))_x$ is the tuple of eigenvalues of $\rho_i$. For $k \in \{2, \ldots, r\}$, we define the average $k$-wise fidelity of $\rho_1, \ldots, \rho_r$  as
    \begin{equation}
        F_{k,r}(\rho_1, \ldots, \rho_r) \coloneqq \binom{r}{k}^{-1} \sum_{i_1 < i_2 < \cdots <i_k} F_k(\rho_{i_1}, \rho_{i_2}, \ldots, \rho_{i_k}),
        \label{eq:def-avg-k-wise}
    \end{equation}
    where $F_k(\rho_{i_1}, \rho_{i_2}, \ldots, \rho_{i_k}) = \sum_x (\rho_{i_1}(x) \cdots \rho_{i_k}(x))^{\frac{1}{k}}$ is the Matusita fidelity of order $k$.
\end{definition}

As a consequence of the properties of Matusita fidelity, all of these quantities obey the data-processing inequality, symmetry under exchange of the states, faithfulness, equal to zero for orthogonal states, and the direct-sum property. They are all also jointly concave, as a consequence of the data-processing inequality and the direct-sum property.

Observe that $F(\rho_1, \ldots, \rho_r) = F_{2,r}(\rho_1, \ldots, \rho_r)$ and $F_r(\rho_1, \ldots, \rho_r) = F_{r,r}(\rho_1, \ldots, \rho_r)$, so that both the average pairwise fidelity and the Matusita fidelity of order $r$ are special cases of the average $k$-wise fidelities.

As a refinement of \Cref{prop:r_root_to_commuting_fidelity}, the average $k$-wise fidelities are sorted in a descending order, which again follows from an application of the inequality of arithmetic and geometric means, as well as basic combinatorial reasoning:

\begin{proposition}[Inequalities relating average $k$-wise fidelities]
\label{prop:ineqs-avg-k-wise-classical}
For $r \in \{3,4, \ldots\}$, let $\rho_1, \ldots, \rho_r$ be a tuple of commuting states. Then
\begin{equation}
    F_{2,r} \geq F_{3,r} \geq \cdots \geq F_{r-1,r}
     \geq F_{r,r},
\end{equation}
where, for brevity, we have suppressed the dependence of each quantity $F_{k,r}$ on $\rho_1, \ldots, \rho_r$.    
\end{proposition}

\begin{proof}
    See \Cref{sec:proof-avg-k-wise-order-classical}.
\end{proof}

\section{Proposed quantum generalizations} 

\label{Sec:QuantumMultivariate_Fidelities}

In this section, we first introduce several desirable properties of multivariate quantum fidelity, which are generalizations of the properties satisfied by the Uhlmann fidelity in \cref{prop:properties_bivariate_F}.
Then we propose four main generalizations of the bivariate quantum fidelities from \Cref{sec:c-q-bivariate-fid-review}:
\begin{enumerate}
    \item The first generalization is the average pairwise $z$-fidelity, which is the simplest generalization of the classical average pairwise fidelity in \cref{def:classical_average_fidelity}.

    \item The next generalization, called the multivariate SDP fidelity, is based on the SDP formulation of the Uhlmann fidelity presented in \Cref{prop:bivariate_fidelity_formulas}.

\item The third generalization, called the secrecy-based multivariate fidelity, is inspired by an existing secrecy measure from~\cite{konig2009operational}.

\item The fourth generalization, called the multivariate log-Euclidean fidelity, is defined through the log-Euclidean divergences given in \cref{def:multi-var-log-euclid}.
\end{enumerate}
We also show that the second and third generalizations are  quantum generalizations of the classical average pairwise fidelity introduced in \Cref{def:classical_average_fidelity}. Furthermore, we show that the last one is a quantum generalization of the Matusita multivariate fidelity given in \cref{def:Matusita_multi_classical}.  In addition, we define maximal and minimal extensions of multivariate classical fidelities (see \cref{def:maximal_multi_fidelity} and \cref{def:minimal_multi_F}) and multivariate geometric fidelities (\cref{def:avg_geometric} and \cref{def:SDP_geometric_multi}) by generalizing geometric fidelity (also known as Matsumoto fidelity)~\cite{matsumoto2010reverse,cree2020fidelity}. We  show that all proposed multivariate quantum fidelities satisfy several desired properties  presented in \cref{def:properties_multi_fidelity}.
\begin{definition}[Desired properties of multivariate fidelity]
\label{def:properties_multi_fidelity}
    For quantum states $\rho_1,\ldots, \rho_r$,  the following are desirable properties of a multivariate fidelity quantity $\mathbf{F}(\rho_1, \ldots, \rho_r)$:
     \begin{enumerate}[label={(\roman*)}]
        \item Reduction to multivariate classical fidelity: For commuting states $\rho_1, \ldots, \rho_r$,
        \begin{align}
        \mathbf{F}(\rho_1, \ldots, \rho_r) & = \frac{2}{r(r-1)} \sum_{i <j} F(\rho_i,\rho_j) \\
        &=  \frac{2}{r(r-1)}  \sum_{i <j} \sum_x \sqrt{\rho_i(x) \rho_j(x)},
    \end{align}
or 
\begin{align}
    \mathbf{F}(\rho_1, \ldots, \rho_r) & =  F_r(\rho_1, \ldots, \rho_r) \\
    &=\sum_{x} \left( \rho_1(x) \cdots \rho_r(x)\right)^{\frac{1}{r}},
\end{align}
or, for some $k\in \{3,\ldots, r-1\}$,
\begin{align}
    \mathbf{F}(\rho_1, \ldots, \rho_r) & =  F_{k,r}(\rho_1, \ldots, \rho_r) \\
    &=\binom{r}{k}^{-1} \sum_{i_1 < \dots < i_k} \sum_{x} \left( \rho_{i_1}(x) \cdots \rho_{i_k}(x)\right)^{\frac{1}{k}},
\end{align}
where $(\rho_i(x))_x$ is the spectrum of the state $\rho_i$ in the common eigenbasis of all the states.

        \item Data processing: For a quantum channel $\cN$,
        \begin{equation}
           \mathbf{F}(\rho_1, \ldots, \rho_r) \leq  \mathbf{F}\!\left( \cN(\rho_1), \ldots, \cN(\rho_r) \right).
           \label{eq:data-proc-property}
        \end{equation}
        \item Symmetry: For every permutation  $\pi$ of  $[r]$,
        \begin{equation}
             \mathbf{F}(\rho_1,\ldots, \rho_r)= \mathbf{F}(\rho_{\pi(1)}, \ldots, \rho_{\pi(r)}).
        \end{equation}
        \item Faithfulness: $ \mathbf{F}(\rho_1, \ldots, \rho_r)=1$ if and only if all the states are the same, i.e., $\rho_i=\rho_j$ for all $i,j \in [r]$.
        
        \item Orthogonality: $ \mathbf{F}(\rho_1, \ldots, \rho_r)=0$ if and only if all states are orthogonal to each other; i.e., $\rho_i \rho_j=0$ for all $ i,j  \in [r]$ such that $i \neq j$.\footnote{Note that it is only possible for the Matusita fidelity to satisfy the following implication: if all states are orthogonal to each other, then $ \mathbf{F}(\rho_1, \ldots, \rho_r)=0$. For the Matusita fidelity, the example presented in~\eqref{eq:counterexample-orthogonality-Matusita} excludes the converse implication from holding. The same statements apply to the average $k$-wise fidelities for $r\geq 3$ and $3 \leq k \leq r$.}

        \item Direct-sum property: Let $\left(\rho_i^x\right)_x$ be tuples of quantum states for all $i \in [r]$. For a probability distribution $(p(x))_{x\in \cX}$ where $\cX$ is a finite alphabet and classical--quantum states formed using those tuples, the following holds: 
    \begin{equation}\label{eq:CQequality_multi_general}
          \mathbf{F}\! \left( \sum_{x \in \cX} p(x) |x\rangle\!\langle x| \otimes \rho_1^x, \ldots, \sum_{x \in \cX} p(x)  |x\rangle\!\langle x| \otimes \rho_r^x\right) \\
         = \sum_{x \in \cX} p(x)  \mathbf{F}(\rho_1^x, \ldots, \rho_r^x).
    \end{equation} 
    \end{enumerate}
\end{definition}

Another desirable property for a multivariate fidelity is joint concavity, stated as follows. Let $(\rho_i^x)_{x \in \mathcal{X}}$ be tuples of quantum states for all $i \in [r]$, and let $(p(x))_{x\in \mathcal{X}}$ be a probability distribution. Then
    \begin{equation}
         \mathbf{F}\! \left( \sum_{x \in \cX} p(x) \rho_1^x, \ldots, \sum_{x \in \cX} p(x) \rho_r^x\right) \geq \sum_{x \in \cX} p(x)   \mathbf{F}(\rho_1^x, \ldots, \rho_r^x).
    \end{equation}
This property is an immediate consequence of data processing and the direct-sum property. Indeed, if these latter two properties hold, then one obtains joint concavity by applying~\eqref{eq:CQequality_multi_general} and then~\eqref{eq:data-proc-property} with the channel $\mathcal{N}$ as the partial trace over the classical register.

\subsection{Average pairwise fidelities} \label{subS:avg_pairwise}

In this subsection, we generalize the  average pairwise classical fidelity in the simplest way possible by using the $z$-fidelity defined in~\eqref{eq:z-fid-def}.

\begin{definition}[Average pairwise $z$-fidelity] \label{def:average_pairwise_z_fidelity}
    For $z \geq 1/2$ and quantum states $\rho_1, \ldots, \rho_r$, 
   the average pairwise $z$-fidelity is defined as 
   \begin{equation}
       F_z(\rho_1, \ldots, \rho_r) \coloneqq \frac{2}{r(r-1)} \sum_{i <j} F_z(\rho_i,\rho_j).
   \end{equation}
\end{definition}

\begin{remark}[Average pairwise Holevo and Uhlmann fidelities]\label{rem:average-pairwise-fids}
    By fixing $z=1$, we obtain the average pairwise Holevo fidelity, denoted by
    \begin{equation}
        F_{H}\! \left(  \rho_{1},\ldots,\rho_{r}\right)   \coloneqq \frac{2}{r\left(  r-1\right)  }\sum
_{i<j}F_{H}(\rho_{i},\rho_{j}).
    \end{equation}
    For $z=1/2$, we obtain the average pairwise Uhlmann fidelity, denoted by
    \begin{equation}
        F_{U}\! \left(  \rho_{1},\ldots,\rho_{r}\right)    \coloneqq \frac{2}{r\left(  r-1\right)  }\sum
_{i<j}F(\rho_{i},\rho_{j}).
    \end{equation}
    As mentioned previously, the average pairwise Uhlmann fidelity has appeared in~\cite{nasser2018polar} as a measure of the reliability of a classical-quantum channel.
    Note that in the classical (viz.,~commuting states) case, the average pairwise $z$-fidelities are equal for all $z \geq 1/2$,  reducing to the multivariate classical fidelity given in \cref{def:classical_average_fidelity}.
\end{remark}

\subsubsection{Uniform continuity bound for average pairwise Uhlmann and Holevo fidelities}

Here we prove uniform continuity bounds for the Uhlmann and Holevo average pairwise fidelities.
Before doing so, let us define the Bures~\cite{Hel67b,Bur69} and Hellinger distances~\cite{jencova2004,LZ04} and recall that they obey the triangle
inequality:
\begin{align}
d_{B}(\rho,\sigma)  &  \coloneqq \sqrt{2\left[  1-F(\rho,\sigma)\right]
}, \label{eq:bures-dist-def} \\
d_{H}(\rho,\sigma)  &  \coloneqq \sqrt{2\left[  1-F_{H}(\rho,\sigma)\right]
}.
\end{align}

\begin{proposition}[Uniform continuity of average pairwise fidelities]
\label{Prop:uniform_cont_average_pairwise}
Let $\rho_1, \ldots, \rho_r, \sigma_1, \ldots, \sigma_r$ be quantum states, and let $\varepsilon\geq 0$ be such that $\frac{1}{r}\sum_{i=1}^r d_{B}(\rho_{i},\sigma_{i})\leq\varepsilon$.
Then
\begin{align}
&\!\!\!\! \left\vert F_{U}\! \left(  \rho_{1},\ldots,\rho_{r}\right)-F_{U} \!\left(  \sigma_{1},\ldots,\sigma_{r}\right)\right\vert 
\leq 2\sqrt{2}  \varepsilon. \label{eq:Uhlmann-average-fidelity-difference-mod-two-sets-quantum-states-2}
\end{align}
Similarly, if $\varepsilon\geq 0$ is such that $\frac{1}{r}\sum_{i=1}^r d_{H}(\rho_{i},\sigma_{i}
)\leq\varepsilon$,
then
\begin{align}
&\!\!\!\!\left\vert F_{H}\! \left(  \rho_{1},\ldots,\rho_{r}\right)-F_{H} \!\left(  \sigma_{1},\ldots,\sigma_{r}\right)\right\vert 
\leq 2\sqrt{2}  \varepsilon.
\label{eq:holevo-pairwise-unif-cont-2}
\end{align}
\end{proposition}
\begin{proof}
    See \Cref{Proof:uniform_cont_average_pairwise}.
\end{proof}

\subsubsection{Properties of average pairwise $z$-fidelities}

Here we show that the average pairwise $z$-fidelities satisfy all desirable properties from \cref{def:properties_multi_fidelity}. In addition, we show that they satisfy super-multiplicativity and coarse-graining. 

\begin{theorem}[Properties of average pairwise $z$-fidelity]\label{Prop:properties_average_pairwise_z}
    For all $z\in [1/2,+\infty)$, the average pairwise $z$-fidelity satisfies all desired properties of a multivariate fidelity, as listed in \cref{def:properties_multi_fidelity}.
\end{theorem}

\begin{proof}
    Proofs of these properties follow because the $z$-fidelity (for $r=2$) satisfies all the properties listed in \cref{prop:properties_bivariate_F}. In particular, reduction to the classical setting follows by plugging commuting states into~\eqref{eq:z-fid-def} and simplifying.  Faithfulness follows because, if the states are the same, then $F_z(\rho, \rho) = \big\| \rho^{\frac{1}{4z}} \rho^{\frac{1}{4z}} \big\|_{2z}^{2z}= \big\| \rho^{\frac{1}{2z}}  \big\|_{2z}^{2z} = 1$. If $F_z(\rho, \sigma) = 1$, then due to the fact that the $z$-fidelities are monotonically decreasing and continuous in~$z$ \cite[Proposition~6]{lin2015investigating}, we conclude that $F(\rho, \sigma)=1$, from which we conclude that $\rho = \sigma$. This also proves that $F_z(\rho_1, \ldots, \rho_r) = 1$ implies $\rho_1 = \cdots = \rho_r$; this is because the equality $F_z(\rho_1, \ldots, \rho_r) = 1$ implies that each of the pairwise fidelities, of the form $F_z(\rho_i,\rho_j)$, is equal to one. Orthogonality follows because, if $\rho \sigma = 0 $, then  $\rho^{\frac{1}{4z}} \sigma^{\frac{1}{4z}} = 0$, from which we conclude that $F_z(\rho, \sigma) = 0$. If $F_z(\rho, \sigma) = 0$, then $\rho^{\frac{1}{4z}} \sigma^{\frac{1}{4z}} = 0$ (from positive definiteness of the norm expression in~\eqref{eq:z-fid-def}), which implies that $\rho \sigma = 0 $. The same reasoning can be applied to the average pairwise $z$-fidelity.
    Moreover, data-processing follows from \cite[Theorem~1.2]{Z20} and the direct-sum property because
    \begin{multline}
   \Tr\!\left[ \left( \left( \rho_{XA} \right)^{\frac{1}{4z} } \left( \sum\nolimits_{x}  p(x)|x\rangle \!\langle x| \otimes \sigma_x\right)^{\frac{1}{2z} }\left( \rho_{XA}\right)^{\frac{1}{4z} }\right)^z\right] \\
     = \sum\nolimits_x p(x) \Tr\!\left[ \left(\rho_x^{{\frac{1}{4z} }}\sigma_x^{{\frac{1}{2z}} }\rho_x^{{\frac{1}{4z}} }\right)^z \right],
 \end{multline}
 where $\rho_{XA} \coloneqq  \sum\nolimits_{x} p(x) |x\rangle \!\langle x| \otimes \rho_x$,
 thus concluding the proof.
\end{proof}

\begin{proposition}[Super-multiplicativity of average pairwise fidelities]\label{prop:super_multiplicative_average_pairwise}
    For states \\ 
    $\rho_1, \ldots, \rho_r$ and $z \geq 1/2$, the average pairwise $z$-fidelity is super-multiplicative: 
    \begin{equation}
        \left(F_z(\rho_1, \ldots, \rho_r) \right)^n \leq  F_z\!\left(\rho_1^{\otimes n}, \ldots, \rho_r^{\otimes n} \right).
    \end{equation}
\end{proposition}

\begin{proof}
    Using \cref{def:average_pairwise_z_fidelity} for the average pairwise $z$-fidelity, consider the following steps:
    \begin{align}
        F_z\!\left(\rho_1^{\otimes n}, \ldots, \rho_r^{\otimes n} \right) 
        &= \frac{2}{r(r-1)} \sum_{i <j} F_z(\rho_i^{\otimes n},\rho_j^{\otimes n}) \\
        &=  \frac{2}{r(r-1)} \sum_{i <j}   \left(F_z(\rho_i,\rho_j)\right)^n \\
        & \geq \left( \frac{2}{r(r-1)} \sum_{i <j}   F_z(\rho_i,\rho_j)  \right)^n \\
        &=  \left(F_z(\rho_1, \ldots, \rho_r) \right)^n ,
    \end{align}
    where the second equality follows from the multiplicativity of $z$-fidelity.
    For the inequality, we use convexity of the function $x^n$ for $n\geq 1$. 
\end{proof}

\begin{proposition}[Coarse-graining property of average pairwise fidelities]
\label{prop:coarse_graining_avg_pairwise}\ \\ 
For states $\rho_1,\ldots,\rho_r, \dots, \rho_{r+m}$ and for all $z \geq 1/2$,
    the average pairwise $z$-fidelity satisfies the following inequality:
    \begin{equation}
        F_{z}(\rho_1,\ldots, \rho_r) \leq \frac{(r+m)(r+m-1)}{r(r-1)} F_{z}(\rho_1,\ldots,\rho_r, \dots, \rho_{r+m}).
    \end{equation}
\end{proposition}

\begin{proof}
Consider that
    \begin{equation}
   \sum_{i, j=1 : i \neq j}^{r} F_z(\rho_i,\rho_j) \leq       \sum_{i, j=1 : i \neq j}^{r+m} F_z(\rho_i,\rho_j)
    \end{equation}
    due to $0 \leq F_z(\rho,\sigma)$. Then, by rewriting the above terms by using \cref{def:average_pairwise_z_fidelity}, we arrive at the desired inequality.    
\end{proof}

\subsection{Multivariate semi-definite programming  fidelity}\label{subS:SDP_F}

Here we propose a multivariate generalization of Uhlmann fidelity by generalizing its SDP formulation  in~\eqref{eq:bivariate_fid_primal}:
\begin{definition}[Multivariate SDP fidelity]
\label{def:multivar-fid}
    Given quantum states $\rho_1, \ldots, \rho_r$, we define their multivariate SDP fidelity as 
\begin{equation}
        \label{eq:multivar-fid}
       F_{\operatorname{SDP}}(\rho_1,\ldots, \rho_r) \\
    \coloneqq \frac{1}{r (r-1)} \inf_{\substack{Y_1, \ldots, Y_r \geq 0 }} \left \{ \sum_{i=1}^r \Tr[ Y_i \rho_i] : \sum_{i=1}^r |i\rangle\!\langle i| \otimes Y_i \geq  \sum_{i \neq j} |i\rangle\!\langle j| \otimes I \right\}.
    \end{equation}
    \end{definition}

\begin{remark}[Multivariate SDP fidelity for PSD operators]
 \cref{def:multivar-fid} extends to general positive semi-definite (PSD) operators, simply by replacing the tuple of states $(\rho_i)_{i=1}^r$ with the tuple of PSD operators $(A_i)_{i=1}^r$.
\end{remark}

In view of the other SDP formulations 
of the Uhlmann fidelity in Proposition~\ref{prop:bivariate_fidelity_formulas}, it is natural to ask if analogous dual formulations can be constructed for the multivariate SDP fidelity. This is indeed the case, as stated in \cref{prop:dual_SDP_multi} below.
We also give an alternative dual formulation in \cref{prop:another_formulation_multi_SDP_fid}.

\begin{proposition}[Dual of multivariate SDP fidelity] \label{prop:dual_SDP_multi}
The multivariate SDP fidelity in \cref{def:multivar-fid} has the following dual formulation: 
      \begin{multline}
        F_{\operatorname{SDP}}(\rho_1, \ldots, \rho_r)= \\
         \frac{2}{r (r-1)}\sup_{\substack{(X_{ij})_{i \neq j} \textnormal{ s.t.} \\ X_{ji} =X^\dagger_{ij}} } \left \{  \sum_{i <j} \mathfrak{R}\!\left[ \Tr[X_{ij}]\right]: \sum_{i=1}^r |i\rangle\!\langle i| \otimes \rho_i + \sum_{i \neq j} |i\rangle\!\langle j| \otimes X_{ij} \geq 0 \right \}.         \label{eq:multivariate_fid_sdp_dual}
      \end{multline}
\end{proposition}
\begin{proof}
    See \Cref{Sec:Proof_dual_SDP}.
\end{proof}

\subsubsection{Uniform continuity bound for multivariate SDP fidelity}

The goal of this subsection is to prove a uniform continuity bound for the multivariate SDP fidelity. One important step in doing so is to introduce the following quantity, which we show later to be equal to the multivariate SDP fidelity.

\begin{definition}[$K^\star$-representation]
\label{def:K-star-rep}
    For a tuple $(\rho_i)_{i=1}^r$ of states, define 
\begin{equation}
\label{eq:K_star_F}
  F_{K^\star}(\rho_1, \ldots, \rho_r) \coloneqq \\
  \frac{1}{ r-1}\sup_{K \geq 0} \bigg\{  \langle \psi | K \otimes I_d | \psi \rangle -1:  K= I_r \otimes I_d+ \sum_{i \neq j} |i\rangle\!\langle j| \otimes K_{ij}     \bigg \},
\end{equation}
where
\begin{equation}
| \psi\rangle \coloneqq \frac{1}{\sqrt{r}} \sum_{i=1}^r |i \rangle | \phi^{\rho_i} \rangle  ,  
\end{equation}
with $| \phi^{\rho_i} \rangle$ being  a purification of $\rho_i$ for all $i \in [r]$ and $d$ the dimension of the Hilbert space of the  states.
\end{definition}

We call this the $K^\star$-representation, and later on,  Theorem~\ref{prop:another_formulation_multi_SDP_fid} states that it is equal to the multivariate SDP fidelity:
\begin{equation}
    F_{\operatorname{SDP}}(\rho_1, \ldots, \rho_r) 
 = F_{K^\star}(\rho_1, \ldots, \rho_r) .
\end{equation}

Let us note that the $K^\star$-representation can also be expressed as
\begin{equation}
     F_{K^\star}(\rho_1, \ldots, \rho_r) = \\
     \frac{1}{ r-1}\sup_{K' \in \mathrm{Herm}} \bigg\{  \langle \psi | K' \otimes I_d | \psi \rangle :    K' \coloneqq \sum_{i \neq j} |i\rangle\!\langle j| \otimes K_{ij}  \leq I_r \otimes I_d    \bigg \},
     \label{eq:alt-K}
\end{equation}
where {we denote the set of Hermitian matrices by $\mathrm{Herm}$.}

Next, we show that the $K^\star$-representation  satisfies a uniform continuity bound. This property finds use in the proof of \cref{prop:another_formulation_multi_SDP_fid}.

\begin{theorem}[Uniform continuity of $K^\star$-representation]\label{thm:uniform_cont_SDP_K}
Let $\rho_1$, \ldots, $\rho_r$, $\sigma_1$, \ldots, $\sigma_r$ be quantum states, and let $\varepsilon\in[0,1]$ be such that
\begin{equation}
    \frac{1}{r}\sum_{i=1}^{r}F(\rho_{i}
,\sigma_{i})\geq1-\varepsilon.
\label{eq:cont-SDP-assumption}
\end{equation}
Then,
\begin{equation}
\left\vert F_{K^\star}(\rho_1, \ldots, \rho_r)-F_{K^\star}(\sigma_1, \ldots, \sigma_r)\right\vert \leq
\frac{r}{r-1}\sqrt{\varepsilon\left(  2-\varepsilon\right)  }.
\end{equation}
\end{theorem}

\begin{proof}
    See \Cref{Proof:uniform_cont_SDP_F}.
\end{proof}
\bigskip

\Cref{prop:another_formulation_multi_SDP_fid} below provides another equivalent formulation for the multivariate SDP fidelity, stating that it is equivalent to $K^\star$-representation. 
This formulation will be useful in proving the lower bound stated in  \cref{thm:SDP_fidelity_pairwise_upper_and_lower} and some consequential properties in \cref{thm:properties_SDP_fidelity}.

\begin{theorem}[Multivariate SDP fidelity and $K^\star$-representation] \label{prop:another_formulation_multi_SDP_fid}
For 
quantum states $\rho_1,\ldots, \rho_r$, the multivariate SDP fidelity~\eqref{eq:multivar-fid} is equal to the $K^\star$-representation: 
\begin{align}\label{eq:sdp_fidelity_K_star_equivalence}
     F_{\operatorname{SDP}}(\rho_1, \ldots, \rho_r) = F_{K^\star}(\rho_1, \ldots, \rho_r) .
\end{align}
\end{theorem}
\begin{proof}
See \Cref{Sec:proof_another_formulation_multi}.
\end{proof}
\bigskip

By using \cref{thm:uniform_cont_SDP_K} and \cref{prop:another_formulation_multi_SDP_fid}, we arrive at the following uniform continuity bound for the multivariate SDP fidelity.

\begin{corollary}[Uniform continuity of multivariate SDP fidelity]\label{thm:uniform_cont_SDP_F}
Let $\rho_1$, \ldots, $\rho_r$, $\sigma_1$, \ldots, $\sigma_r$ be quantum states, and let $\varepsilon\in\left[  0,1\right]  $ be such that
\begin{equation}
    \frac{1}{r}\sum_{i=1}^{r}F(\rho_{i}
.\sigma_{i})\geq1-\varepsilon,
\end{equation}
Then,
\begin{equation}
\left\vert F_{\operatorname{SDP}}(\rho_1, \ldots, \rho_r)-F_{\operatorname{SDP}}(\sigma_1, \ldots, \sigma_r)\right\vert \leq
\frac{r}{r-1}\sqrt{\varepsilon\left(  2-\varepsilon\right)  }.
\end{equation}
\end{corollary}

\subsubsection{Properties of multivariate SDP fidelity}

Here we present several properties satisfied by the multivariate SDP fidelity. 

\begin{theorem}[Properties of multivariate SDP fidelity]\label{thm:properties_SDP_fidelity}
   The multivariate SDP fidelity satisfies all the desired properties of a multivariate fidelity, as listed in \cref{def:properties_multi_fidelity}. For commuting states, it reduces to the classical average pairwise fidelity.  
\end{theorem}
\begin{proof}
    See \Cref{Sec:proof_properties}.
\end{proof}

\begin{remark}[Properties for  a set of PSD operators]
  For a tuple of PSD operators $(A_i)_{i=1}^r$, the following properties hold.
  For $c >0$, by applying \Cref{def:multivar-fid} directly, we have the following scaling property: 
\begin{equation}\label{eq:scaling_multi_PSD}
    F_{\operatorname{SDP}}\!\left( cA_1, \ldots cA_r\right) = c F_{\operatorname{SDP}}(A_1, \ldots, A_r).
\end{equation}
Let $(A_i^x)_x$ be tuples of PSD operators for all $i \in [r]$. Then
\begin{equation}\label{eq:CQequality_multi_PSD}
         F_{\operatorname{SDP}}\! \left( \sum_{x \in \cX} |x\rangle\!\langle x| \otimes A_1^x, \ldots, \sum_{x \in \cX}  |x\rangle\!\langle x| \otimes A_r^x\right) = \sum_{x \in \cX}  F_{\operatorname{SDP}}(A_1^x, \ldots, A_r^x).
    \end{equation} 
    A proof of~\eqref{eq:CQequality_multi_PSD} follows similarly to the proof of the direct-sum property of multivariate SDP fidelity in \cref{thm:properties_SDP_fidelity}.
    To this end, note that property (vi) in \cref{thm:properties_SDP_fidelity} can be obtained as a special case of~\eqref{eq:CQequality_multi_PSD} by using~\eqref{eq:scaling_multi_PSD} and~\eqref{eq:CQequality_multi_PSD}.
The other properties (i)-(iii) 
in \cref{thm:properties_SDP_fidelity} also hold by following the same proof arguments, due to the fact that \cref{thm:SDP_fidelity_pairwise_upper_and_lower} holds even for PSD operators, among other things.
\end{remark}

\begin{proposition}[Coarse-graining property] \label{prop:coarse_graining}
    Let $r,m \in \mathbb{N}$, and let $\rho_1,\ldots, \rho_{r+m}$ be quantum states. We have the following inequality: 
    \begin{equation}
        F_{\operatorname{SDP}}(\rho_1,\ldots, \rho_r) \leq \frac{(r+m)(r+m-1)}{r(r-1)} F_{\operatorname{SDP}}(\rho_1,\ldots,\rho_r, \dots, \rho_{r+m}).
    \end{equation}    
\end{proposition}

\begin{proof}
    See \Cref{Proof:coarse_graining}.
\end{proof}

\begin{remark}[Super-multiplicativity of multivariate SDP fidelity] \label{rem:super_multiplic_SDP}
  Through numerical calculations, we found that there exists a tuple of states for which super-multiplicativity does not hold for the SDP fidelity, i.e., 
  $ \left(F_{\operatorname{SDP}}(\rho_1, \ldots, \rho_r) \right)^n >  F_{\operatorname{SDP}}\!\left(\rho_1^{\otimes n}, \ldots, \rho_r^{\otimes n} \right) $ for some $r$ and $(\rho_i)_{i=1}^r$. An example of a tuple of pure states (i.e., $\rho_i \coloneqq | \psi_i\rangle\!\langle \psi_i|$ for $i\in \{1,2,3\}$) for $r=3$, $n=2$, and $d=3$ is as follows (up to numerical approximations of MATLAB): 
  \begin{align}
      |\psi_1 \rangle &=\begin{pmatrix}-0.8954 + 0.2791 i\\0.2061-0.0805i \\ 0.2418+ 0.1135i \end{pmatrix}, \\
      |\psi_2 \rangle &=\begin{pmatrix}-0.2422 + 0.2315 i\\0.4386-0.4318i \\ -0.6928 - 0.1704i \end{pmatrix}, \\ 
      |\psi_3 \rangle &=\begin{pmatrix}0.2560+ 0.5837i\\0.4811-0.3057i \\ -0.5175- 0.314i \end{pmatrix}.
  \end{align}
For the above example, 
\begin{equation}
0.4075=\left(F_{\operatorname{SDP}}(\rho_1, \rho_2, \rho_3) \right)^2 >  F_{\operatorname{SDP}}\!\left(\rho_1^{\otimes 2}, \rho_2^{\otimes 2}, \rho_3^{\otimes 2} \right) =0.3820,
\end{equation}
which shows a violation of super-multiplicativity.\footnote{The MATLAB code used to find the given counterexample is available as arXiv ancillary files along with the arXiv posting of this paper.} 
\end{remark}

Although the multivariate SDP fidelity is not multiplicative, there is an interesting function of it that is. Let us define the following
secrecy measure based on the multivariate SDP\ fidelity:
\begin{equation}
S_{\operatorname{SDP}}(\rho_{1},\ldots,\rho_{r})\coloneqq \sqrt{ \frac{\left(  r-1\right)
F_{\operatorname{SDP}}(\rho_{1},\ldots,\rho_{r})+1}{r}}.
\label{eq:def-S-SDP-secrecy}
\end{equation}
We call it a secrecy measure in analogy with the secrecy measure introduced later on in~\eqref{eq:secrecy_measure}.

The quantity $S_{\operatorname{SDP}}$ has the following SDP representation, which we can use to prove that it is multiplicative:
\begin{proposition}
[SDP representation for $S_{\operatorname{SDP}}$]
\label{prop:SDP-secrecy-SDP}
Given $r$ states $\rho_{1},\ldots,\rho_{r}$ of dimension $d$, the square of the secrecy measure
$S_{\operatorname{SDP}}(\rho_{1},\ldots,\rho_{r})$ can be written as follows:
\begin{align}
S_{\operatorname{SDP}}(\rho_{1},\ldots,\rho_{r})^{2}  & =\frac{1}{r^{2}}\sup_{K\geq
0}\left\{  \langle\varphi|K\otimes I_{d}|\varphi\rangle : \left(  \Delta
_{r}\otimes\operatorname{id}_{d}\right)  \left(  K\right)  =I_{r} \otimes I_{d} \right\}
\label{eq:SDP-secrecy-measure}\\
& =\frac{1}{r^{2}}\inf_{Y\in\operatorname{Herm}}\left\{  \operatorname{Tr}[Y] : \left(
\Delta_{r}\otimes\operatorname{id}_{d}\right)  (Y)\geq\omega\right\}
,\label{eq:SDP-secrecy-measure-SDP-dual}
\end{align}
where
\begin{equation}
|\varphi\rangle_{XRS}\coloneqq \sum_{i=1}^{r}|i\rangle_{X}|\psi^{\rho_{i}}\rangle
_{RS},
\end{equation}
the state vector $|\psi^{\rho_{i}}\rangle$ is a purification of $\rho_{i}$,
$\Delta_{r}(\cdot)\coloneqq \sum_{i=1}^{r}|i\rangle\!\langle i|(\cdot)|i\rangle\!\langle
i|$ denotes the completely dephasing channel, and
\begin{equation}
\omega\coloneqq \operatorname{Tr}_{S}[|\varphi\rangle\!\langle\varphi|_{XRS}].
\end{equation}

\end{proposition}

\begin{proof}
    See \cref{app:SDP-secrecy-SDP}.
\end{proof}

Define the following notation:
\begin{equation}\label{eq:notation_tensor_two_states}
(\rho_{i}\otimes\sigma_{j})_{i\in\left[
r_{1}\right]  ,j\in\left[  r_{2}\right]  } \coloneqq (\rho_{1}\otimes\sigma_{1},\ldots,\rho_{1}\otimes\sigma_{r_{2}
},\ldots,\rho_{r_{1}}\otimes\sigma_{1},\ldots,\rho_{r_{1}}\otimes\sigma
_{r_{2}}).
\end{equation}
{ In words, the notation in~\eqref{eq:notation_tensor_two_states} refers to the tuple of $r_1 r_2$ states that consists of the states obtained by the tensor product of each element of the tuples of the states $(\rho_i)_{i\in [r_1]}$ and $(\sigma_j)_{j\in [r_2]}$.}
\begin{proposition}
[Multiplicativity of $S_{\operatorname{SDP}}$]
\label{prop:Multiplicativity-S-SDP}
The secrecy measure $S_{\operatorname{SDP}}$ is multiplicative in the following sense:
\begin{equation}
S_{\operatorname{SDP}}\!\left( (\rho_{i}\otimes\sigma_{j})_{i\in\left[
r_{1}\right]  ,j\in\left[  r_{2}\right]  }\right)
=S_{\operatorname{SDP}}(\rho_{1},\ldots,\rho_{r_{1}})\cdot S_{\operatorname{SDP}}(\sigma
_{1},\ldots,\sigma_{r_{2}}),
\end{equation}
where $\rho_{1},\ldots,\rho_{r_{1}}$ and $\sigma_{1},\ldots,\sigma_{r_{2}}$
are states.
\end{proposition}

\begin{proof}
    See \cref{app:Multiplicativity-S-SDP}.
\end{proof}

\subsubsection{Other reductions of multivariate SDP fidelity}

 \begin{remark}[Multivariate SDP fidelity as a pairwise quantity]
The development in the proof of \cref{prop:another_formulation_multi_SDP_fid} gives insight into how we can view the SDP\ fidelity as being somewhat like
an average pairwise quantity (see~\eqref{eq:pairwise_relation_fidelity}). Indeed, we see that
\begin{equation} \label{eq:SDP_fidelity_with_pairwise_structure}
F_{\operatorname{SDP}}(\rho_{1},\ldots,\rho_{r})=\\
\frac{2}{r\left(  r-1\right)  }\sup_{K\geq0}\left\{
\begin{array}
[c]{c}
\sum_{i<j}\mathfrak{R}\left[  \langle\psi^{\rho_{i}}|K_{ij}\otimes
I|\psi^{\rho_{j}}\rangle\right]  :\\
K=\sum_{i,j=1}^{r}|i\rangle\!\langle j|\otimes K_{ij},\ K_{ii}=I_{d}\ \forall
i\in\left[  r\right]
\end{array}
\right\}  .
\end{equation}
\end{remark}

By choosing $r=2$ in \cref{prop:another_formulation_multi_SDP_fid} and applying~\eqref{eq:alt-K}, we obtain yet another formulation for Uhlmann fidelity. 

\begin{corollary}[Another formulation for Uhlmann fidelity]\label{Cor:another_formulation_bivariate}
    For all  quantum states $\rho$ and~$\sigma$, the Uhlmann fidelity can be expressed as 
    \begin{equation}
        F(\rho,\sigma)=\sup_{K\in \operatorname{Herm}}\left\{
\begin{array}
[c]{c}
\langle\phi|K\otimes I_d|\phi\rangle:
I_2 \otimes I_d \geq
\begin{bmatrix}
0 & X\\
X^{\dagger} & 0
\end{bmatrix}
\eqqcolon K
\end{array}
\right\}
,
    \end{equation}
    where 
    \begin{equation}
        |\phi\rangle \coloneqq \frac{1}{\sqrt{2}}\left(  |0\rangle|\psi^{\rho}\rangle
+|1\rangle|\psi^{\sigma}\rangle\right)
    \end{equation}
  with   $ | \psi^\rho \rangle$ and $| \psi^\sigma \rangle $ being purifications of $\rho$ and $\sigma$, respectively.
\end{corollary}

We obtain a formulation for multivariate SDP fidelity of pure states by simplifying the $K^\star$-representation and applying  \cref{prop:another_formulation_multi_SDP_fid}.
\begin{corollary}
    [Multivariate SDP fidelity for pure states]
    \label{Prop:pure_states_SDP_F}
For  a tuple of pure states, $(  |\psi_{1}\rangle\!\langle\psi_{1}|,\ldots,|\psi_{r}\rangle
\!\langle\psi_{r}|)  $,
\begin{multline}
F_{\operatorname{SDP}}(|\psi_{1}\rangle\!\langle\psi_{1}|,\ldots,|\psi
_{r}\rangle\!\langle\psi_{r}|)=\label{eq:K-form-SDP-fid-pure}\\
\frac{2}{r\left(  r-1\right)  }\sup_{k_{ij}\in\mathbb{C}}\left\{  \sum
_{i<j}\mathfrak{R}[k_{ij}\langle\psi_{i}|\psi_{j}\rangle]:\sum
_{i,j=1}^{r}k_{ij}|i\rangle\!\langle j|\geq0,\ k_{ii}=1\ \forall i\in\left[
r\right]  \right\}  .
\end{multline}

\end{corollary}

\begin{proof}
    See \Cref{Proof:pure_states_SDP_F}.
\end{proof}

We end this section by noting the following alternative forms of $F_{\operatorname{SDP}}$ and $S_{\operatorname{SDP}}$, which follow from developments similar to those given in \cref{app:mult-geo-SDP-G-alt} and \cref{app:secrecy-meas-geo-SDP-SDP-dual}: Given $r$ states $\rho_{1},\ldots,\rho_{r}$ of dimension $d$, the multivariate
SDP fidelity $F_{\operatorname{SDP}}(\rho_{1},\ldots,\rho_{r})$ and the secrecy measure $S_{\operatorname{SDP}}(\rho_{1},\ldots,\rho_{r})$
can be written as follows:
\begin{align}
F_{\operatorname{SDP}}(\rho_{1},\ldots,\rho_{r}) & =\frac{1}{r\left(
r-1\right)  }\sup_{X\geq0}\left\{
\begin{array}
[c]{c}
\operatorname{Tr}[\left(  |+\rangle\!\langle+|\otimes I_d\right)  X]-r:\\
\left(  \Delta_{r}\otimes\operatorname{id}_{d}\right)  \left(  X\right)
=\sum_{i=1}^{r}|i\rangle\!\langle i|\otimes\rho_{i},
\end{array}
\right\}  , \\
S_{\operatorname{SDP}}(\rho_{1},\ldots,\rho_{r})^{2}
&  =\frac{1}{r^{2}}\sup_{X\geq0}\left\{
\begin{array}
[c]{c}
\operatorname{Tr}[\left(  |+\rangle\!\langle+|\otimes I_d\right)  X]:\\
\left(  \Delta_{r}\otimes\operatorname{id}_{d}\right)  \left(  X\right)
=\sum_{i=1}^{r}|i\rangle\!\langle i|\otimes\rho_{i}
\end{array}
\right\} \\
&  =\frac{1}{r^{2}} \inf_{Y\in\operatorname{Herm}}\left\{
\begin{array}
[c]{c}
\operatorname{Tr}\!\left[  Y\left(  \sum_{i=1}^{r}|i\rangle\!\langle
i|\otimes\rho_{i}\right)  \right]  :\\
|+\rangle\!\langle+|\otimes I_d  \leq\left(  \Delta_{r}\otimes\operatorname{id}_{d}\right)
\left(  Y\right)  
\end{array}
\right\}  .
\end{align}
where the vector $|+\rangle$ and the completely dephasing channel $\Delta_{r}$ are defined as
\begin{equation} 
|+\rangle   \coloneqq \sum_{i=1}^{r}|i\rangle,
\qquad
\Delta_{r}(\cdot)   \coloneqq \sum_{i=1}^{r}|i\rangle\!\langle
i|(\cdot)|i\rangle\! \langle i|.
\end{equation}

\subsection{Secrecy-based multivariate fidelity} \label{subS:secrecy_based_F}

In this subsection, we introduce another multivariate quantum fidelity inspired by an existing secrecy measure in~\eqref{eq:secrecy_measure} and analyse its properties.

Recall the following secrecy measure from~\cite[Eq.~(19)]{konig2009operational}, defined in terms of the average Uhlmann fidelity of a tuple of states to a reference state:
\begin{equation}
\label{eq:secrecy_measure}
    S(\rho_1, \ldots, \rho_r) \coloneqq  \sup_{\sigma  \in \mathscr{D}}\frac{1}{r}\sum_{i=1}^{r}F(  \rho_{i}, \sigma) =\sup_{\sigma \in \mathscr{D}} F(\rho_{XA}, \rho_X \otimes \sigma),
\end{equation}
where
\begin{equation}
\label{eq:cq-state-secrecy-meas}
\rho_{XA}\coloneqq \sum_{i=1}^{r} \frac{1}{r}|i\rangle\!\langle i|_{X}\otimes\rho_{i}.     
\end{equation}
Observe that $\rho_X = I/r$. The square of this  secrecy measure was given a direct operational interpretation in \cite[Theorem~6]{RASW23} as the maximum success probability, in a quantum interactive proof, that a prover could pass a test for the states being similar. See~\cite{afham2022quantum} for a comprehensive study on the optimization problem involved with this secrecy measure. In particular, by putting together Eqs.~(22)--(23), Definition~3, Lemma~4, and Theorems~5 and 6 of~\cite{afham2022quantum} and performing some other simplifications, we arrive at the following primal and dual SDP formulations of $S(\rho_1, \ldots, \rho_r)$:
\begin{align}
  & S(\rho_{1},\ldots,\rho_{r})  \notag \\
& =\frac{1}{r}\sup_{\substack{X_{i,j}=X_{j,i}^{\dag}\in \mathscr{L},\\ \forall i,j\in\{0,1, \ldots, r\}  ,\\\sigma\geq0}}\left\{
\begin{array}
[c]{c}
\sum_{i=1}^{r}\mathfrak{R}[\operatorname{Tr}[X_{i,0}
]] : \operatorname{Tr}[\sigma]=1,\\  |0\rangle\!\langle
0|\otimes\sigma
+\sum_{i=1}^{r}|i\rangle\!\langle i|\otimes\rho_{i}
+\sum_{i,j=0: i\neq j}^{r}|i\rangle\!\langle j|\otimes X_{i,j}\geq0
\end{array}
\right\}  \label{eq:sup_SDP_secrecy_measure}\\
& =\frac{1}{2r}\inf_{\substack{Y_{1},\ldots,Y_{r}\geq0,\\\mu\geq0}}\left\{
\begin{array}
[c]{c}
\mu+\sum_{i=1}^{r}\operatorname{Tr}[Y_{i}\rho_{i}]:\\ \begin{aligned}
    & |0\rangle\!\langle
0|\otimes\mu I_{d} +
\sum_{i=1}^{r}|i\rangle\!\langle i|\otimes Y_{i}
\geq\sum_{i=1}^{r}\left(  |i\rangle\!\langle 0|+|0\rangle\!\langle i|\right)
\otimes I_d
\end{aligned}
\end{array}
\right\} \label{eq:inf_SDP_secrecy_measure} .
\end{align}
Another SDP\ formulation of the secrecy measure $S(\rho_{1},\ldots,\rho_{r})$ is as follows:
\begin{align}
& S(\rho_{1},\ldots,\rho_{r})\nonumber\\
& =\frac{1}{2r}\sup_{\substack{X_{1},\ldots,X_{r}\in\mathscr{L},\\\sigma\geq
0}}\left\{  \sum_{i=1}^{r}\operatorname{Tr}[X_i]+\operatorname{Tr}[X_i^{\dag
}]:\operatorname{Tr}[\sigma]=1,\
\begin{bmatrix}
\rho_{i} & X_{i}\\
X_{i}^{\dag} & \sigma
\end{bmatrix}
\geq0\ \forall i\in\left[  r\right]  \right\}
\label{eq:primal-SDP-secrecy-measure}\\
& =\frac{1}{2r}\inf_{\substack{Y_{1},\ldots,Y_{r},\\Z_{1},\ldots,Z_{r}
\geq0,\lambda\geq0}}\left\{  \sum_{i=1}^{r}\operatorname{Tr}[Y_{i}\rho
_{i}]+\lambda:\sum_{i=1}^{r}Z_{i}\leq\lambda I,\
\begin{bmatrix}
Y_{i} & -I\\
-I & Z_{i}
\end{bmatrix}
\geq0\ \forall i\in\left[  r\right]  \right\}  .
\label{eq:dual-alt-sdp-secrecy}
\end{align}
The first SDP\ in \eqref{eq:primal-SDP-secrecy-measure} was noted in \cite[Appendix~A]{afham2022quantum}, where it was stated to be more numerically stable than that given in \eqref{eq:sup_SDP_secrecy_measure}--\eqref{eq:inf_SDP_secrecy_measure}. It
 follows from a direct application of the fidelity SDP in \eqref{eq:bivariate_fid_primal} to the
first formula in \eqref{eq:secrecy_measure}. In Appendix~\ref{app:dual-alt-sdp-secrecy}, we derive the dual SDP in \eqref{eq:dual-alt-sdp-secrecy} and prove that strong duality holds. Similar to \cref{prop:Multiplicativity-S-SDP}, the secrecy measure $S$ is multiplicative in the following sense:
\begin{equation}
S\!\left( (\rho_{i}\otimes\sigma_{j})_{i\in\left[
r_{1}\right]  ,j\in\left[  r_{2}\right]  }\right)
=S(\rho_{1},\ldots,\rho_{r_{1}})\cdot S(\sigma
_{1},\ldots,\sigma_{r_{2}}),
\label{eq:mult-secrecy-meas-S}
\end{equation}
where $\rho_{1},\ldots,\rho_{r_{1}}$ and $\sigma_{1},\ldots,\sigma_{r_{2}}$
are states. This follows as a direct consequence of the additivity claims of \cite[Section~I-C-2]{konig2009operational}. Indeed, it follows from \cite[Eq.~(19)]{konig2009operational} that $\ln (rS(\rho_1, \ldots, \rho_r)^2) = H_{\max}(X|A)$ with
\begin{equation}
    H_{\max}(A|B)_\rho \coloneqq \sup_{\sigma_B} \log F(\rho_{AB}, I_A \otimes \sigma_B),
\end{equation}
where $H_{\max}$ is evaluated with respect to the classical--quantum state in \eqref{eq:cq-state-secrecy-meas}, so that additivity of $H_{\max}$ is equivalent to the multiplicativity equality in \eqref{eq:mult-secrecy-meas-S}.

Using the secrecy measure, we define another multivariate fidelity as follows.
\begin{definition}[Secrecy-based multivariate fidelity] \label{def:secrecy_based_multi}
Let $\rho_1, \ldots, \rho_r$ be quantum states.
We define the multivariate secrecy   fidelity as 
\begin{equation}
    F_S(\rho_1, \ldots, \rho_r) \coloneqq \frac{1}{r-1} \left( r  S(\rho_1, \ldots, \rho_r)^2 -1 \right),
\end{equation}
   where $S(\rho_1,\ldots,\rho_r)$ is defined in~\eqref{eq:secrecy_measure}.
\end{definition}

\begin{proposition}[Pairwise formulation for secrecy measure and secrecy-based multivariate fidelity] \label{prop:another_rep_secrecy_based}
    For
quantum states $\rho_1,\ldots, \rho_r$, the following equalities hold:
\begin{align}
S(\rho_{1},\ldots,\rho_{r}) & = \frac{1}{r}
\sup_{\left(  V^{i}\right)  _{i}}
\left(\sum_{i,j=1}^{r}\langle\phi^{\rho_{j}}
|_{RA}\left(  V_{R}^{j}\right)  ^{\dag}V_{R}^{i}|\phi^{\rho_{i}}\rangle_{RA}\right)^{1/2} , \label{eq:secrecy-as-pairwise} \\
      F_S(\rho_1, \ldots, \rho_r) & = \frac{2}{r(r-1)}\sup_{ (V^i)_{i=1}^r } 
        \sum_{i<j} \mathfrak{R}\! \left[\left\langle \phi^{\rho_j} \right|_{RA}\left(V^j_{R }\right)^\dagger V^i_{R} \left| \phi^{\rho_i} \right \rangle_{RA}\right],
        \label{eq:secrecy-fid-as-pairwise} 
\end{align}
where $(V^i)_{i=1}^r$ is a tuple of unitaries (i.e., $V^i \in \mathscr{U}$ for all $i\in[r]$) and $\left| \phi^{\rho_i} \right \rangle $ is a purification of the state $\rho_i$ for all $i \in [r]$. 
\end{proposition}

\begin{proof}
See~\cref{proof_another_rep_secrecy_based}. 
\end{proof}

\begin{remark}[Optimal unitary for secrecy measure]
\label{rem:prover-opt-unitary}
In \cite[Theorem~6]{RASW23} and as mentioned above, the quantity 
$ S(\rho_1, \ldots, \rho_r)^2$ was given an operational interpretation as the maximum success probability with which a prover could convince a verifier that the states $\rho_1, \ldots, \rho_r$ are similar, by means of a particular quantum interactive proof given there.
    In \cite[Eqs.~(A70)--(A76)]{RASW23}, it was shown that  the following equality holds:
\begin{equation} \label{eq:equivalent_S_2}
      S(\rho_1, \ldots, \rho_r)^2 = \sup_{P_{T'RF \to T^{''}F'} \in \mathscr{U} } \frac{1}{r^2} \sum_{i,j=1}^r \left\langle \phi^{\rho_j} \right|_{RA}\left(P^j_{R \to F'}\right)^\dagger P^i_{R \to F'} \left| \phi^{\rho_i} \right \rangle_{RA},
\end{equation}
 where 
    \begin{equation}
        P^i_{R \to F'} \coloneqq \langle i|_{T^{''}} P_{T'RF \to T^{''}F'}|i \rangle_{T'} |0 \rangle_F ,
    \end{equation}
    with $ P_{T'RF \to T^{''}F'}$  a unitary applied by the prover, and $\left| \phi^{\rho_i} \right \rangle $ is a purification of the state $\rho_i$ for all $i \in [r]$. 

    It follows from our proof in \cref{proof_another_rep_secrecy_based} that it is in fact optimal for the prover to perform a controlled unitary of the form $P_{T'RF \to T^{''}F'} = \sum_i |i\rangle_{T''} \langle i|_{T'} \otimes V^i_{RF\to F'}$, where each $V^i_{RF\to F'}$ is a unitary.
\end{remark}

By observing that the reference system for pure states is a trivial one-dimensional system and each unitary acting on the reference system reduces to a complex phase in this case, we conclude from~\eqref{eq:secrecy-fid-as-pairwise} the following formula for the secrecy-based multivariate fidelity of pure states.
\begin{corollary}
    [Secrecy-based multivariate fidelity for pure states]
    \label{Prop:pure_states_secrecy_F}
For  a tuple of pure states, $(  |\psi_{1}\rangle\!\langle\psi_{1}|,\ldots,|\psi_{r}\rangle
\!\langle\psi_{r}|)  $,
\begin{equation}
F_{S}(|\psi_{1}\rangle\!\langle\psi_{1}|,\ldots,|\psi
_{r}\rangle\!\langle\psi_{r}|)=
\frac{2}{r\left(  r-1\right)  }\sup_{\substack{\phi_{i}\in[0,2\pi],\\ \forall i \in [r]}}\left\{  \sum
_{i<j}\mathfrak{R}[e^{i (\phi_j - \phi_i)}\langle\psi_{i}|\psi_{j}\rangle]  \right\}  .
\end{equation}

\end{corollary}

In the definition of the secrecy-based multivariate fidelity, the secrecy measure in~\eqref{eq:secrecy_measure} is defined in terms of the Uhlmann fidelity. By replacing the Uhlmann fidelity therein with the Holevo fidelity, the resulting Holevo secrecy measure is a function of the average pairwise Holevo fidelity, as stated in \cref{rem:holevo-pairwise-secrecy}.

\begin{remark}[Average pairwise Holevo fidelity as a secrecy-based multivariate fidelity]
\label{rem:holevo-pairwise-secrecy}
For quantum states $\rho_1, \ldots, \rho_r$,  a secrecy measure based on the Holevo fidelity is defined as
\begin{equation}\label{eq:secrecy_measure_holevo}
    S_H(\rho_1, \ldots, \rho_r) \coloneqq  \sup_{\sigma  \in \mathscr{D}}\frac{1}{r}\sum_{i=1}^{r}F_H(  \rho_{i}, \sigma) =\sup_{\sigma \in \mathscr{D}} F_H(\rho_{XA}, \rho_X \otimes \sigma),
\end{equation}
where $\rho_{XA}\coloneqq \sum_{i=1}^{r} \frac{1}{r}|i\rangle\!\langle i|_{X}\otimes\rho_{i}$.
With the use of~\eqref{eq:relation_to_fidelity_d_1/2} in \cref{proof:SDP_fidelity_bounds_with_pairwise}, together with the fact that 
\begin{equation}
    \inf_{\sigma_{A}}{D}_{\frac{1}{2}}(\rho_{XA}\Vert\rho_{X}
\otimes\sigma_{A})  = -2 \ln S_H(\rho_1, \ldots, \rho_r),
\end{equation}
where $D_{1/2}$ is the Petz--R\'enyi relative entropy of order $1/2$ defined in~\eqref{eq:petz renyi},
the following equality holds:
\begin{equation}
    F_H(\rho_1, \ldots, \rho_r) = \frac{1}{r-1} \left( r  S_H(\rho_1, \ldots, \rho_r)^2 -1 \right),
\end{equation}
thus relating this secrecy measure to the average pairwise Holevo fidelity $F_H(\rho_1, \ldots, \rho_r)$.
Similar to \cref{prop:Multiplicativity-S-SDP}, the secrecy measure $S_H$ is multiplicative in the following sense:
\begin{equation}
S_H\!\left( (\rho_{i}\otimes\sigma_{j})_{i\in\left[
r_{1}\right]  ,j\in\left[  r_{2}\right]  }\right)
=S_H(\rho_{1},\ldots,\rho_{r_{1}})\cdot S_H(\sigma
_{1},\ldots,\sigma_{r_{2}}),
\label{eq:mult-secrecy-meas-S-H}
\end{equation}
where $\rho_{1},\ldots,\rho_{r_{1}}$ and $\sigma_{1},\ldots,\sigma_{r_{2}}$
are states. This follows as a direct consequence of the explicit form for $S_H$ given in \eqref{eq:ln-SH-explicit-form}.
\end{remark}

\subsubsection{Uniform continuity bound for secrecy-based multivariate fidelity}

\begin{proposition}[Uniform continuity] \label{prop:uniform_cont_secrecy_multi_F}
Let $\rho_1$, \ldots, $\rho_r$, $\sigma_1$, \ldots, $\sigma_r$ be quantum states, and let $\varepsilon\geq 0  $ be such that
\begin{equation}
    \frac{1}{r} \sum_{i=1}^r d_{B}(\rho_{i},\sigma_{i})\leq\varepsilon,
\end{equation}
where the Bures distance $d_{B}$ is defined in~\eqref{eq:bures-dist-def}.
Then,
\begin{equation}
    \left| F_S(\rho_1, \ldots, \rho_r) - F_S(\sigma_1, \ldots, \sigma_2)  \right| \leq 2 \sqrt{2}\left(\frac{ r \  }{r-1}\right)\varepsilon.
\end{equation}
\end{proposition}
\begin{proof}
    See \Cref{proof:uniform_cont_secrecy_multi_F}.
\end{proof}

\subsubsection{Properties of secrecy-based multivariate fidelity}

Next we prove that secrecy-based multivariate fidelity also satisfies the desirable properties of a multivariate fidelity from \cref{def:properties_multi_fidelity}.

\begin{theorem}[Properties of secrecy-based multivariate fidelity] \label{thm:properties_secrecy_multi_F}
    The secrecy-based multivariate fidelity satisfies all the desired properties of a multivariate fidelity listed in \cref{def:properties_multi_fidelity}. For commuting states, it reduces to the classical average pairwise fidelity.
\end{theorem}

\begin{proof}
  See \cref{proof:Properties_secrecy_fidelity}.
\end{proof}
\medskip

Similar to the average pairwise fidelity and multivariate SDP fidelity, the secrecy-based multivariate fidelity $F_S$ also satisfies the following coarse-graining property. 
\begin{proposition}[Coarse-graining property] \label{prop:coarse_graining_secrecy_F}
   Let $r,m \in \mathbb{N}$, and let $\rho_1,\ldots, \rho_{r+m}$ be quantum states. Then the following inequality holds: 
    \begin{equation}
        F_{S}(\rho_1,\ldots, \rho_r) \leq \frac{(r+m)(r+m-1)}{r(r-1)} F_{S}(\rho_1,\ldots,\rho_r, \dots, \rho_{r+m}).
    \end{equation}  
\end{proposition}
\begin{proof}
    See \cref{proof:coarse_graining_secrecy_F}.
\end{proof}

\begin{remark}[Super-multiplicativity]
    Through numerical calculations, we found that there exists a tuple of states for which super-multiplicativity does not hold for the secrecy-based multivariate fidelity, i.e., 
  $ \left(F_{S}(\rho_1, \ldots, \rho_r) \right)^n >  F_{S}\!\left(\rho_1^{\otimes n}, \ldots, \rho_r^{\otimes n} \right) $ for some $r$ and $(\rho_i)_{i=1}^r$. The same example that violates super-multiplicativity of SDP multivariate fidelity, as mentioned in \cref{rem:super_multiplic_SDP}, also violates the super-multiplicativity of secrecy-based multivariate fidelity. 
\end{remark}

\subsection{Relations between proposed generalizations}

In this section, we present some inequalities relating the average pairwise Holevo fidelity, the secrecy-based multivariate fidelity, the multivariate SDP fidelity, and the average pairwise Uhlmann fidelity, demonstrating that they are ordered.

\begin{theorem}[Inequalities relating multivariate fidelities] \label{thm:SDP_fidelity_pairwise_upper_and_lower}
For a tuple $(\rho_i)_{i=1}^r$ of states, the following inequalities hold for  the average pairwise Holevo fidelity $F_H$, the secrecy-based multivariate fidelity $F_S$, 
the multivariate SDP fidelity $F_{\operatorname{SDP}}$,   and the 
average pairwise Uhlmann fidelity $F_U$:
\begin{multline}\label{eq:bounds_SDP_fidelity}
F_H(\rho_1, \ldots, \rho_r) \leq F_S(\rho_1, \ldots, \rho_r) \leq F_{\operatorname{SDP}}(\rho_1, \ldots, \rho_r) \\ \leq F_U(\rho_1, \ldots, \rho_r) \leq \sqrt{F_H(\rho_1, \ldots, \rho_r)}.
\end{multline}  
\end{theorem}

\begin{proof}
    See \Cref{proof:SDP_fidelity_bounds_with_pairwise}. 
\end{proof}
\medskip 

An immediate implication of \cref{thm:SDP_fidelity_pairwise_upper_and_lower} is the following ordering of the secrecy measures from \cref{rem:holevo-pairwise-secrecy}, \eqref{eq:secrecy_measure}, and \eqref{eq:def-S-SDP-secrecy}:
\begin{equation}
    S_H(\rho_1, \ldots, \rho_r) \leq S(\rho_1, \ldots, \rho_r) \leq S_{\operatorname{SDP}}(\rho_1, \ldots, \rho_r) \leq S_U(\rho_{1},\ldots,\rho_{r}),
\end{equation}
where
\begin{equation}
S_{U}(\rho_{1},\ldots,\rho_{r})\coloneqq \sqrt{\frac{\left( r-1\right)
F_{U}(\rho_{1},\ldots,\rho_{r})+1}{r}}.
\end{equation}

In \cref{thm:SDP_fidelity_pairwise_upper_and_lower} we showed that multivariate SDP fidelity is  bounded from above by the average pairwise Uhlmann fidelity. Next, we show that it is also  bounded from below by a function of pairwise Uhlmann fidelities (note that it is not an average).
\begin{proposition}[Lower bound on multivariate SDP fidelity]\label{prop:lower_bound_multi_SDP}
    The multivariate SDP fidelity is bounded from below by a multiple of the maximum of sums of pairwise fidelities of the given states. More precisely, given quantum states $\rho_1,\ldots, \rho_r$, we have
  \begin{equation}
        \frac{2}{r(r-1)} \sup_{\pi \in S_r} \sum_{i=1}^{\lfloor \frac{r}{2}\rfloor } F(\rho_{\pi(i)}, \rho_{\pi(i+\lfloor r/2 \rfloor)})  
    \leq F_{\operatorname{SDP}}(\rho_1, \ldots, \rho_r).\label{eq:sdp-fid-lower-bound}
  \end{equation}
\end{proposition}
\begin{proof}
    See \Cref{Sec:Proof_lower_bound_multi_SDP}.
\end{proof}

\subsection{Minimal and maximal extensions of multivariate classical fidelities}

One way to construct a multivariate quantum fidelity is by optimizing a classical multivariate fidelity, and there are at least two ways of doing so. The first way is to perform a measurement of the $r$ quantum states, resulting in $r$ probability distributions, and then to minimize a classical multivariate fidelity of the $r$ distributions over all possible measurements. The measurement is also called a quantum-to-classical transformation. The second way is to begin with $r$ probability distributions and act on them with a preparation channel that realizes the $r$ quantum states, and then to maximize a classical multivariate fidelity over all possible preparation channels. A preparation channel is also known as a classical-to-quantum transformation.
The former approach leads to a quantity that we call a maximal extension of multivariate fidelity (or measured multivariate fidelity), and the latter leads to a minimal extension of multivariate fidelity (or preparation multivariate fidelity).

These concepts have been introduced and studied extensively in prior work. The measured relative entropy was defined in~\cite{donald1986relative,hiai1991proper,P09} and its extension to Petz-R\'enyi divergences  in~\cite[Eqs.~(3.116)--(3.117)]{F96} (see also~\cite{Berta2015OnEntropies,mosonyi2022test}). The preparation relative entropy was defined in~\cite{matsumoto2010reverse_relent} and shown therein to be equal to the Belavkin--Staszewski relative entropy from~\cite{Belavkin1982}. 
The preparation fidelity and more general $f$-divergences were defined in~\cite{matsumoto2010reverse,Matsumoto2018} and studied further in~\cite{HM17,hiai19,katariya2021geometric}.
See also~\cite{gour2020optimal} for maximal and minimal extensions of resource measures, including generalized fidelity for sub-normalized states. Most recently, maximal and minimal extensions of multivariate Chernoff divergence were defined in~\cite[Section~V-A]{mnw2023}.

We  define the maximal and minimal extensions of the classical multivariate fidelity as follows.

\begin{definition}[Maximal extension of multivariate fidelity]
\label{def:maximal_multi_fidelity}
We define a maximal extension of multivariate fidelity (also called a measured multivariate fidelity) as 
\begin{equation}
   \overline{ \mathbf{F}}(\rho_1, \ldots, \rho_r) \coloneqq \inf_{( \Lambda_x)_x}  \mathbf{F}\!\left(\Lambda(\rho_1), \ldots, \Lambda(\rho_r)\right)
\end{equation}
where $\mathbf{F}$ is a multivariate fidelity that reduces to one of the multivariate classical fidelities given in \cref{def:Matusita_multi_classical}, \cref{def:classical_average_fidelity}, or \cref{def:average_k_wise_classical}
and $\Lambda(\rho_i)\coloneqq \sum_x \operatorname{Tr}[\Lambda_x \rho]|x\rangle\!\langle x|$ is the measurement channel associated with the POVM $( \Lambda_x)_x$ (see also~\eqref{eq:measurement_channel}).   
\end{definition}

\begin{definition}[Minimal extension of multivariate fidelity]
\label{def:minimal_multi_F}
    We define a minimal extension of multivariate fidelity (also called a preparation multivariate fidelity) as 
    \begin{equation}
        \underline{ \mathbf{F}}(\rho_1, \ldots, \rho_r) \coloneqq \sup_{\substack{\cP \in \operatorname{CPTP}, \, \omega_i \in \mathscr{D}_c \forall i \in [r]}} \left\{\mathbf{F}\!\left( \omega_1, \ldots, \omega_r \right) : \cP(\omega_i) =\rho_i \,\forall i \in [r]\right \},
    \end{equation}
where $\mathbf{F}$ is  a multivariate fidelity that reduces to one of the multivariate classical fidelities given in \cref{def:Matusita_multi_classical}, \cref{def:classical_average_fidelity}, or \cref{def:average_k_wise_classical} and $\cP$ is a classical-quantum channel satisfying $\cP(\omega_i) =\rho_i$ for all $i\in [r]$ and $\mathscr{D}_c$ denotes the set of all states diagonal in the standard basis (so that $\mathscr{D}_c$ is a set of commuting states).
\end{definition}

Note that, by the above definitions and data processing of multivariate fidelity (\cref{def:properties_multi_fidelity}), we arrive at the following inequalities.
\begin{proposition}[Inequalities between maximal and minimal extensions]\label{prop:ineq_maximal_minimal_multi}
   For $\rho_1, \ldots, \rho_r$ quantum states we have 
   \begin{equation}
     \underline{ \mathbf{F}}(\rho_1, \ldots, \rho_r) \leq \mathbf{F}(\rho_1,\ldots, \rho_r) \leq  \overline{ \mathbf{F}}(\rho_1, \ldots, \rho_r),
\end{equation}
where $\underline{ \mathbf{F}}$ and $\overline{ \mathbf{F}}$ are defined in \cref{def:minimal_multi_F} and \cref{def:maximal_multi_fidelity}, respectively, for a multivariate fidelity $\mathbf{F}$.
\end{proposition}
\begin{proof}
    Proofs for these inequalities follow the same approach used in \cite[Appendices~E \&~F]{mnw2023} along with the fact that both the maximal and minimal extensions also satisfy the data-processing inequality.
\end{proof}

\begin{remark}[Bivariate setting] \label{rem:bivariate_minimal}
In the bivariate setting, the maximal extension of the classical fidelity is equal to the Uhlmann fidelity, and the minimal extension is equal to the geometric fidelity, as defined in~\eqref{eq:geometric_fidelity}:
By \cite[Theorem~1]{matsumoto2010reverse} we have 
\begin{equation}
    \underline{ \mathbf{F}}(\rho,\sigma) \leq \mathbf{F}(\rho, \sigma) \leq \overline{ \mathbf{F}}(\rho,\sigma)= F(\rho,\sigma), 
    \label{eq:ordering-prep-meas-fids}
\end{equation}
where the equality follows from~\eqref{eq:bivariate_measured}.
Furthermore, \cite[Remark~3]{matsumoto2010reverse} shows that 
\begin{equation}\label{eq:geometric_fidelity}
     \underline{ \mathbf{F}}(\rho,\sigma)= \inf_{\varepsilon >0 } \Tr[ \rho(\varepsilon) \# \sigma(\varepsilon)],
\end{equation}
where $A \# B$ denotes the geometric mean of the positive definite operators $A$ and $B$:
\begin{equation}
    A \# B \coloneqq A^{1/2}\left( A^{-1/2} B A^{-1/2} \right)^{1/2} A^{1/2}, 
\end{equation}
$\rho(\varepsilon) \coloneqq \rho + \varepsilon I$,
and $\sigma(\varepsilon) \coloneqq \sigma + \varepsilon I$.
Due to this formulation, the minimal extension of the classical bivariate fidelity is also called the geometric fidelity. In~\cite{cree2020fidelity}, the geometric fidelity was called the Matsumoto fidelity, and several of its properties were established by means of semi-definite programming.
\end{remark}

\subsubsection{Measured multivariate fidelity}

In this section, we consider various properties satisfied by the measured multivariate fidelity defined in \cref{def:maximal_multi_fidelity}.

\begin{proposition}[Properties of measured multivariate fidelity]\label{prop:properties_measured_multi_G}
If the underlying classical multivariate fidelity satisfies data processing, symmetry, faithfulness, 
and the direct-sum property, then the maximal extension satisfies data processing, symmetry, faithfulness, 
and the direct-sum property.

Furthermore, if the underlying classical multivariate fidelity is the average pairwise fidelity in \cref{def:classical_average_fidelity}, then the maximal extension also satisfies orthogonality. If the underlying classical multivariate fidelity is the Matusita multivariate fidelity from \cref{def:Matusita_multi_classical} or an average $k$-wise fidelity from \cref{def:average_k_wise_classical}, the maximal extension satisfies  weak orthogonality only.
\end{proposition}

\begin{proof}
    See \cref{proof:properties_measured_multivariate_G}.
\end{proof}
\medskip

By choosing the underlying classical multivariate fidelity to be the average pairwise fidelity, here we define measured average pairwise fidelity.

\begin{definition}[Measured average pairwise fidelity] \label{def:measured_multivariate}
    For states $\rho_1, \ldots, \rho_r$, the measured average pairwise fidelity is defined as 
\begin{equation}
    F_\mathrm{M}(\rho_1, \ldots, \rho_r)\coloneqq  \frac{2}{r(r-1)} \inf_{( \Lambda_x)_x} 
     \sum_{i <j} F\big( \Lambda(\rho_i), \Lambda(\rho_j) \big)
\end{equation}
where $\Lambda(\rho_i)\coloneqq \sum_x \operatorname{Tr}[\Lambda_x \rho_i]|x\rangle\!\langle x|$ is the measurement channel induced by applying the POVM $( \Lambda_x)_x$ on the state $\rho_i$ (see also~\eqref{eq:measurement_channel}).
\end{definition}

Note that \cref{def:measured_multivariate} is the largest quantum generalization of the average  pairwise classical fidelity that satisfies data processing.

When $\mathbf{F}$ is chosen to be a multivariate fidelity that generalizes the average pairwise classical fidelity, we have that
\begin{equation}
    \overline{\mathbf{F}}(\rho_1, \ldots, \rho_r) = F_{\mathrm{M}}(\rho_1,\ldots, \rho_r). \label{eq:measured_equals_avgPairwiseF}
\end{equation}
Similarly, when $\mathbf{F}$ is chosen to be a multivariate fidelity that generalizes the Matusita multivariate fidelity, we have that 
\begin{equation}
    \overline{\mathbf{F}}(\rho_1, \ldots, \rho_r) = \inf_{(\Lambda_x)_x} F_r\!\left( \Lambda(\rho_1), \ldots, \Lambda(\rho_r) \right).
\end{equation}
The next proposition generalizes \cref{prop:r_root_to_commuting_fidelity} to the case of measured multivariate fidelities. 

\begin{proposition}[Measured Matusita fidelity and average pairwise fidelity]
The following inequality holds: 
  \begin{equation}
      F_{\mathrm{M}}(\rho_1,\ldots, \rho_r) \geq \inf_{(\Lambda_x)_x} F_r\!\left( \Lambda(\rho_1), \ldots, \Lambda(\rho_r) \right),
  \end{equation}  
where the Matusita multivariate fidelity $F_r$ is defined in~\eqref{eq:Matusita_fid_commuting}.
\end{proposition}
\begin{proof}
    The proof of the first statement follows from applying \cref{prop:r_root_to_commuting_fidelity} for each measurement channel $\Lambda$ and then taking the infimum over all measurement channels. 
\end{proof}

With the use of uniform continuity of the average pairwise Uhlmann fidelity and the multivariate SDP fidelity in \cref{Prop:uniform_cont_average_pairwise} and \cref{thm:uniform_cont_SDP_F}, respectively, we establish  uniform continuity of the measured average pairwise fidelity as the following corollary:
\begin{corollary}[Uniform continuity of measured average pairwise fidelity]\label{cor:uniform_cont_measured_multi_avg}\ \\
    Let $\rho_1, \ldots, \rho_r, \sigma_1, \ldots, \sigma_r$ be quantum states, and let $\varepsilon\geq 0$ be such that $\frac{1}{r}\sum_{i=1}^r d_{B}(\rho_{i},\sigma_{i})\leq\varepsilon$.
Then,
\begin{equation}
    \left\vert F_\mathrm{M} \! \left(  \rho_{1},\ldots,\rho_{r}\right)-F_\mathrm{M} \!\left(  \sigma_{1},\ldots,\sigma_{r}\right)\right\vert 
\leq 2\sqrt{2}  \varepsilon. 
\end{equation}
Furthermore, if $\varepsilon\in [0,1]$ is such that $ \frac{1}{r}\sum_{i=1}^{r}F(\rho_{i}
,\sigma_{i})\geq1-\varepsilon,$
then
\begin{equation}
  \left\vert F_\mathrm{M} \! \left(  \rho_{1},\ldots,\rho_{r}\right)-F_\mathrm{M} \!\left(  \sigma_{1},\ldots,\sigma_{r}\right)\right\vert  \leq
\frac{r}{r-1}\sqrt{\varepsilon\left(  2-\varepsilon\right)  }.
\end{equation}
\end{corollary}
\begin{proof}
    For $F_\mathrm{M} \! \left(  \rho_{1},\ldots,\rho_{r}\right)$,
    the underlying classical multivariate fidelity is the average pairwise fidelity. Considering the reduction to the classical case for the quantum generalizations we have proposed, the following equalities hold for every measurement channel $\Lambda$:
    \begin{equation}\label{eq:equality_classical}
        F\!\left( \Lambda(\rho_1), \ldots, \Lambda(\rho_r)\right)=  F_U\!\left( \Lambda(\rho_1), \ldots, \Lambda(\rho_r)\right)=  F_{\operatorname{SDP}} \!\left( \Lambda(\rho_1), \ldots, \Lambda(\rho_r)\right). 
    \end{equation}
Since the fidelity and Bures distance satisfy the data-processing inequality, we also have the following implications:
\begin{align}
    & \frac{1}{r}\sum_{i=1}^r d_{B}(\rho_{i},\sigma_{i})\leq\varepsilon \implies \frac{1}{r}\sum_{i=1}^r d_{B}\!\left( \Lambda(\rho_{i}),\Lambda(\sigma_{i}) \right)\leq\varepsilon, \\ 
    & \frac{1}{r}\sum_{i=1}^{r}F(\rho_{i}
,\sigma_{i})\geq1-\varepsilon \implies \frac{1}{r}\sum_{i=1}^{r}F\!\left(\Lambda(\rho_{i})
,\Lambda(\sigma_{i}) \right)\geq1-\varepsilon.
\end{align}
Now consider that
\begin{align}
     & F_\mathrm{M} \! \left(  \rho_{1},\ldots,\rho_{r}\right) -  F_\mathrm{M} \! \left(  \sigma_{1},\ldots,\sigma_{r}\right) \notag \\
     & =  \inf_{( \Lambda_x)_x}   F\!\left( \Lambda(\rho_1), \ldots, \Lambda(\rho_r)\right) -    \inf_{( \Lambda'_x)_x} F\!\left( \Lambda'(\sigma_1), \ldots, \Lambda'(\sigma_r)\right) \\ 
   & =  \sup_{( \Lambda'_x)_x}\left( \inf_{( \Lambda_x)_x}   F\!\left( \Lambda(\rho_1), \ldots, \Lambda(\rho_r)\right) -  F\!\left( \Lambda'(\sigma_1), \ldots, \Lambda'(\sigma_r)\right)  \right) \\
   & \leq \sup_{( \Lambda'_x)_x}\left(    F\!\left( \Lambda'(\rho_1), \ldots, \Lambda'(\rho_r)\right) -  F\!\left( \Lambda'(\sigma_1), \ldots, \Lambda'(\sigma_r)\right)  \right).
\end{align}
At this point, we conclude the two inequalities by applying \cref{Prop:uniform_cont_average_pairwise} and \cref{thm:uniform_cont_SDP_F} along with~\eqref{eq:equality_classical}.
\end{proof}

\begin{remark}[Uniform continuity of measured Matusita multivariate fidelity]
    Using a similar proof technique as given in the proof of \cref{cor:uniform_cont_measured_multi_avg} and using uniform continuity of the Matusita multivariate fidelity in \cref{prop:unif-cont-Matusita}, it is possible to obtain a uniform continuity bound for the measured Matusita multivariate fidelity as well, under the assumption that $\frac{1}{r} \left( \sum_{i=1}^r d_B(\rho_i,\sigma_i)^2\right) \leq \varepsilon$.
\end{remark} 

The following proposition establishes inequalities between the measured average pairwise fidelity and the average pairwise Uhlmann fidelity. 
\begin{proposition}[Average pairwise measured and Uhlmann fidelities] \label{prop:measured_unmeasured_ineq}
The following inequality holds for every tuple $(\rho_1,\dots,\rho_r)$ of states:
\begin{equation}
    F_\mathrm{M}(\rho_1, \ldots, \rho_r) \geq  F_{U}(\rho_1, \ldots, \rho_r). 
    \label{eq:F_M-to-F_U}
\end{equation}   
\end{proposition}

\begin{proof}
    Consider the following:
\begin{align}
    F_\mathrm{M}(\rho_1, \ldots, \rho_r) 
    &=\inf_{( \Lambda_x)_x} \frac{2}{r(r-1)} \sum_{i <j} F\big( \Lambda(\rho_i), \Lambda(\rho_j) \big) \\ 
    &\geq \frac{2}{r(r-1)} \sum_{i <j} \inf_{( \Lambda_x)_x} F\big( \Lambda(\rho_i), \Lambda(\rho_j) \big) \\
    &= \frac{2}{r(r-1)} \sum_{i <j} F(\rho_i,\rho_j),
\end{align}
where the first inequality follows from placing the infimum inside the sum, the second equality follows from the Fuchs--Caves characterization of Uhlmann fidelity in terms of the measured bivariate fidelity (see~\eqref{eq:bivariate_measured}). 
\end{proof}

\begin{remark}[Multivariate SDP fidelity and the measured version]
    When $r=2$, the Uhlmann fidelity is equal to the measured fidelity, as stated in~\eqref{eq:bivariate_measured}. However, when $r>2$, through numerical evaluations, it follows that there exists a tuple of states satisfying 
    $F_\mathrm{M}(\rho_1, \ldots, \rho_r) > F_{\operatorname{SDP}}(\rho_1,\ldots, \rho_r)$. This is numerically validated by showing that $F_U \geq F_{\operatorname{SDP}}$ in~\eqref{eq:bounds_SDP_fidelity} is a strict inequality for some tuples of states and then combining with~\eqref{eq:F_M-to-F_U}. One such tuple of states (up to numerical approximations of MATLAB)  
    is the following:
    \begin{align}
       \rho_1 &= \begin{bmatrix}
0.3465+0.0000i  & -0.2036+0.2643i \\
-0.2036-0.2643i &0.6535+0.0000i
\end{bmatrix}, \\
\rho_2 &= \begin{bmatrix}
    0.6546+0.0000i & -0.3308-0.2297i \\ -0.3308+0.2297i & 0.3454+0.0000i
\end{bmatrix},\\
\rho_3 &=\begin{bmatrix}
    0.6169+0.0000i & 0.0327-0.0321i \\  0.0327+0.0321i & 0.3831+0.0000i
\end{bmatrix}.
    \end{align}
The MATLAB code used to find the given counterexample is available as arXiv ancillary files along with the arXiv posting of this paper.

Furthermore, it is interesting to determine whether there exists a tuple of states for which $F_M(\rho_1, \ldots, \rho_r) > F_U(\rho_1,\ldots, \rho_r)$---we leave this as an open question.
\end{remark}

\subsubsection{Preparation multivariate fidelity}

The preparation multivariate fidelity, as outlined in \cref{def:minimal_multi_F}, is also known as the minimal extension of multivariate fidelity. Next, we discuss various properties satisfied by the preparation multivariate fidelity.

\begin{proposition}[Properties of preparation multivariate fidelity] \label{Prop:properties_minimal_extension}
The preparation multivariate fidelity in \cref{def:minimal_multi_F} satisfies data processing, symmetry, faithfulness, and the direct-sum property, as listed  in \cref{def:properties_multi_fidelity}.

Furthermore, if the states are orthogonal to each other, then $ \underline{ \mathbf{F}}(\rho_1, \ldots, \rho_r) =0$ (weak orthogonality).
\end{proposition}
\begin{proof}
    See \cref{proof:properties_minimal}.
\end{proof}

\begin{remark}[Lack of uniform continuity]
The minimal extension of the multivariate fidelity does not satisfy uniform continuity in general. To see this, consider the following counterexample: Fix $r=2$  and choose 
\begin{equation}
    \rho_1 = \rho_2= \sigma_1= |0 \rangle\!\langle 0|  \quad \textnormal{and}  \quad \sigma_2= |\psi\rangle\!\langle \psi|,
\end{equation}
where $|\psi\rangle \coloneqq \sqrt{1-\varepsilon} |0 \rangle  + \sqrt{\varepsilon} |1\rangle$ for $\varepsilon>0$.
Then, we obtain the following by using the properties of Uhlmann fidelity from \cref{prop:properties_bivariate_F}:
\begin{align}
    F(\rho_1,\sigma_1) & =1 ,\\
    F(\rho_2, \sigma_2) &= \sqrt{1-\varepsilon} ,\\ 
   \underline{ \mathbf{F}}(\rho_1, \rho_2) &=1 ,\\
   \underline{\mathbf{F}}(\sigma_1, \sigma_2) &= 0,
\end{align}
where the third equality follows by the faithfulness of the preparation fidelity and the last equality by the fact that the minimal extension of the fidelity between two pure states  evaluates to zero~\cite[Eq.~(23)]{matsumoto2010reverse} (see also \cite[Lemma~3.17 and Eq.~(3.59)]{cree2020fidelity}).
With the above evaluations we see that  $\frac{1}{2} \left( F(\rho_1,\sigma_1) + F(\rho_2,\sigma_2) \right) \geq 1- \varepsilon$ and yet 
\begin{equation}
    | \underline{\mathbf{F}}(\rho_1, \rho_2) - \underline{\mathbf{F}}(\sigma_1, \sigma_2) |=1,
\end{equation}
implying that the preparation multivariate fidelity does not satisfy uniform continuity in general. 
\end{remark}

\subsection{Multivariate log-Euclidean fidelity} 

\label{subS:loog_eucledian_F}

In this section, we recall the definition of multivariate log-Euclidean divergences~\cite{mosonyi2022geometric,mnw2023}, we show how one possible quantum Matusita fidelity is a special case, and we give operational interpretations of them in terms of quantum hypothesis testing with an arbitrarily varying null hypothesis, specifically making use of the main result of~\cite{Notzel_2014}. 

\begin{definition}[Multivariate log-Euclidean divergence]
\label{def:multi-var-log-euclid}
    Let $\rho_{1},\ldots,\rho_{r}$ be positive definite states, and let $s=\left(  s_{1},\ldots,s_{r}\right)  $ be a probability distribution.
    The multivariate
log-Euclidean divergence is defined as
\begin{equation}
H_{s}(\rho_{1},\ldots,\rho_{r})\coloneqq -\ln\operatorname{Tr}\!\left[  \exp\!\left(
\sum_{i=1}^{r}{s_{i}}\ln\rho_{i}\right)  \right]  .
\label{eq:multi-log-euc-PD-states}
\end{equation}
For general states, the multivariate log-Euclidean divergence is defined as
\begin{equation}
    H_{s}(\rho_{1},\ldots,\rho_{r}) \coloneqq \sup_{\varepsilon > 0 } H_{s}(\rho_{1}(\varepsilon),\ldots,\rho_{r}(\varepsilon)) = \lim_{\varepsilon \to 0^+ } H_{s}(\rho_{1}(\varepsilon),\ldots,\rho_{r}(\varepsilon)),
    \label{eq:multi-log-euc-extend-all-states}
\end{equation}
where $\rho_i(\varepsilon) \coloneqq \rho_i + \varepsilon I$ for all $i\in[r]$ and the rightmost equality follows from operator monotonicity of the logarithm and monotonicity of $(\cdot)\to \operatorname{Tr}[\exp(\cdot)]$. The argument of the trace in~\eqref{eq:multi-log-euc-PD-states} is known as the log-Euclidean Frechet mean \cite[Section~3.4]{arsigny2007geometric}.
\end{definition}

We show next that multivariate log-Euclidean divergence satisfies additivity. 
\begin{proposition}[Additivity of log-Euclidean divergence]\label{prop:additivity_log_Euclidean_divergence}
Let $s=\left(  s_{1},\ldots,s_{r}\right)  $ be a probability distribution,
and let $\rho_{1},\ldots,\rho_{r}$ and $\sigma_{1},\ldots,\sigma_{r}$ be states.
The multivariate log-Euclidean divergence is additive. That is,
\begin{equation}
H_{s}\!\left(  \rho_{1}\otimes\sigma_{1},\ldots,\rho_{r}\otimes\sigma
_{r}\right)  =H_{s}\!\left(  \rho_{1},\ldots,\rho_{r}\right)  +H_{s}\!\left(
\sigma_{1},\ldots,\sigma_{r}\right).
\end{equation}
\end{proposition}
\begin{proof}
    See~\cref{proof:additivity_log_eucli_divergence}.
\end{proof}

For simplicity, in most of what follows we focus on positive definite states and note that all of the statements can be extended as in~\eqref{eq:multi-log-euc-extend-all-states}.

Now we define the multivariate log-Euclidean fidelity, which extends~\eqref{eq:Matusita_fid_commuting} to the setting of general quantum states.

\begin{definition}[Multivariate log-Euclidean fidelity] \label{def:multi_log_eucli_fidelity}
Let $(\rho_i)_{i=1}^r$ be a tuple of states. We define the multivariate log-Euclidean fidelity for positive definite states as
    \begin{equation}
F^{\flat}_r(\rho_{1},\ldots,\rho_{r})\coloneqq \operatorname{Tr}\!\left[  \exp\!\left(
\frac{1}{r}\sum_{i=1}^{r}\ln\rho_{i}\right)  \right]  \label{eq:mult_log-Euc_fid_definition},
\end{equation}
and for the general case as
\begin{equation}
    F^{\flat}_r(\rho_{1},\ldots,\rho_{r})\coloneqq \inf_{\varepsilon >0 } F^{\flat}_r(\rho_{1}(\varepsilon),\ldots,\rho_{r}(\varepsilon)) = \lim_{\varepsilon \to 0^+ } F^{\flat}_r(\rho_{1}(\varepsilon),\ldots,\rho_{r}(\varepsilon)),
\end{equation}
where $\rho_{i}(\varepsilon)\coloneqq \rho_i + \varepsilon I$.
\end{definition}

Note that \cref{def:multi_log_eucli_fidelity} reduces to~\eqref{eq:log-euc-fid-limit} for the bivariate setting.

\begin{remark} \label{rem:special_cases_log_Euclidean}
   When $s$ is uniform (i.e., $s_i=1/r$ for all $i \in [r]$),  the multivariate log-Euclidean divergence is equal to the negative logarithm of the multivariate log-Euclidean fidelity in \cref{def:multi_log_eucli_fidelity}:
\begin{equation} \label{eq:quantum_Matusita_fidelity_log_Eucl}
-\ln \! \left(F^{\flat}_r(\rho_{1},\ldots,\rho_{r}) \right) = -\ln\operatorname{Tr}\!\left[  \exp\!\left(  \sum_{i=1}^{r}\frac{1}{r}\ln\rho_{i}
\right)  \right].
\end{equation}
For the special case of commuting states, the multivariate log-Euclidean divergence reduces to the negative logarithm of the Hellinger transform from~\eqref{eq:hellinger-transform-def}:
\begin{equation}
H_{s}(\rho_{1},\ldots,\rho_{r})= -\ln\operatorname{Tr}\!\left[  \prod
\limits_{i=1}^{r}\rho_{i}^{s_{i}}\right].  
\end{equation}
\end{remark}

\subsubsection{Connection to oveloh information}
The generalized Holevo information of a tuple of states $(
\rho_{i})  _{i}$ with uniform prior distribution is defined as
\begin{equation}\label{eq:generalized_Holevo}
\inf_{\sigma_{A} \in \mathscr{D}}\boldsymbol{D}(\rho_{XA}\Vert\rho_{X}\otimes\sigma_{A}),
\end{equation}
where $\boldsymbol{D}$ is a generalized divergence that satisfies data processing under quantum channels~\cite{SW12,leditzky2018approaches},
\begin{equation} \label{eq:rho_XB}
\rho_{XA}\coloneqq \frac{1}{r}\sum_{i=1}^{r}|i\rangle\!\langle i|_{X}\otimes\rho_{i},
\end{equation}
{and $\rho_X \coloneqq \Tr_A[\rho_{XA}]$.}

Fix $\alpha \in (0,1) \cup (1, \infty)$. The Petz--R\'enyi relative entropy of a state $\rho$ and a PSD operator~$\sigma$ is defined as
\cite{P85,P86}
\begin{equation}
D_\alpha(\rho\Vert \sigma) \coloneqq
\begin{cases}\frac{1}{\alpha-1} \ln\Tr[ \rho^\alpha \sigma^{1- \alpha}], & \mbox{if } \alpha \in (0,1) \lor
(\alpha > 1 \land \supp(\rho)\subseteq\supp(\sigma))\\
+\infty, & \mbox{otherwise}
\end{cases}\label{eq:petz renyi}
\end{equation}
It is a generalized divergence for $\alpha \in [0,1) \cup (1,2]$~\cite{P86}.
Taking the limit $\alpha\to 1$ gives the quantum relative entropy defined as ~\cite{umegaki1962conditional}
\begin{equation}
\label{eq:q-rel-entr}
    D(\rho\| \sigma) \coloneqq 
        \begin{cases}
            \Tr[ \rho ( \log \rho -\log \sigma)], & \operatorname{supp}(\rho) \subseteq 
                \operatorname{supp}(\sigma), \\
            +\infty, & \text{otherwise}.
        \end{cases}
\end{equation}

In the spirit of the definition of \textit{lautum information} (reverse of mutual information)~\cite{Lautum}, we define the \textit{oveloh information} (reverse of Holevo information): 
\begin{equation}
\mathscr{O}(X;A)_\rho \coloneqq 
\inf_{\sigma\in\mathscr{D}}D(\rho_{X}\otimes\sigma\Vert\rho_{XA}) =\inf_{\sigma\in\mathscr{D}}\frac{1}{r}\sum_{i=1}^{r}D(\sigma\Vert\rho
_{i}),
\label{eq:oveloh-info}
\end{equation}
where $\rho_{XA}$ is given in~\eqref{eq:rho_XB}.
Also note that oveloh information is a special case of log-Euclidean divergence with the choice of $s$ being a uniform distribution (see~\eqref{eq:hellinger-to-optimized-rel-ent}). 
To this end, 
$\exp\left[- \mathscr{O}(X;A)_\rho \right]$ is precisely the  multivariate log-Euclidean fidelity in \cref{def:multi_log_eucli_fidelity}: 
\begin{equation} \label{eq:quantum_Matusita_oveloh}
    F^{\flat}_r (\rho_1, \ldots, \rho_r)= \exp\left[- \mathscr{O}(X;A)_\rho \right] =\exp\!\left(-\inf_{\sigma\in\mathscr{D}}\frac{1}{r}\sum_{i=1}^{r}D(\sigma\Vert\rho
_{i})\right).
\end{equation}

With this connection, we establish a relationship between Holevo and oveloh information as given below.

\begin{corollary}[Holevo and oveloh information] \label{Cor:Holevo_and_Oveloh}
    Let $(\rho_i)_i$ be a tuple of 
    states. Then the following inequality holds: 
    \begin{equation}
        \exp\!\left[ - \inf_{\sigma \in\mathscr{D}} {D}_{\frac{1}{2}}(\rho_{XA}\Vert\rho
_{X} \otimes \sigma) \right] \geq \frac{1}{r} + \frac{r-1}{r} \exp\!\left[  -\mathscr{O}(X;A)_\rho\right],
    \end{equation}
    where $\rho_{XA}$ is given in~\eqref{eq:rho_XB} and 
    ${D}_{1/2}(\cdot\| \cdot )$ is the Petz--R\'enyi relative entropy of order $1/2$ in~\eqref{eq:petz renyi}. 
\end{corollary}
\begin{proof}

By using~\eqref{eq:quantum_Matusita_oveloh} together with~\eqref{eq:oveloh-info}, we have 
\begin{align}
   \exp\!\left[  -\inf_{\sigma\in\mathscr{D}}D(\rho_{X}\otimes\sigma\Vert\rho
_{XA})\right]   
& = F^{\flat}_r(\rho_1, \ldots, \rho_r) \\ 
& \leq \frac{2}{r(r-1)} \sum_{i<j} F_\flat(\rho_i, \rho_j), 
\end{align}
where the inequality follows from \cref{prop:upper_multi_log_eucli}.
Then, by recalling that $z$-fidelities are antimonotone  for $z \geq 1/2$ together with~\eqref{eq:log-euc-fid-limit}, we have $ F_\flat(\rho,\sigma) \leq F_H(\rho,\sigma)$ leading to 
\begin{equation}
    \exp\!\left[  -\inf_{\sigma\in\mathscr{D}}D(\rho_{X}\otimes\sigma\Vert\rho
_{XA})\right]    \leq \frac{2}{r(r-1)} \sum_{i<j} F_H(\rho_i, \rho_j).\label{eq:upper_bound_flat_hol}
\end{equation}
Also, using~\eqref{eq:relation_to_fidelity_d_1/2}, we obtain the following: 
\begin{equation}
 \exp\!\left[  -\inf_{\sigma\in\mathscr{D}}D_{\frac{1}{2}}(\rho_{XA}\Vert\rho_X \otimes \sigma
)\right]   = \frac{1}{r} +\frac{r-1}{r} \left[ \frac{2}{r(r-1)} \sum_{i<j} F_H(\rho_i, \rho_j)\right]. 
\end{equation}
Combining~\eqref{eq:upper_bound_flat_hol} and the previous equality 
completes the proof. 
\end{proof}

\medskip

Note that the inequality stated above is tighter than that which one would obtain by using the facts that 
\begin{align}
\inf_{\sigma \in\mathscr{D}} {D}_{\frac{1}{2}}(\rho_{XA}\Vert\rho
_{X} \otimes \sigma) &= \inf_{\sigma \in\mathscr{D}} {D}_{\frac{1}{2}}(\rho
_{X} \otimes \sigma\Vert \rho_{XA}) , \\
\inf_{\sigma \in\mathscr{D}} {D}_{\frac{1}{2}}(\rho
_{X} \otimes \sigma\Vert \rho_{XA}) &\leq \inf_{\sigma \in\mathscr{D}} {D}(\rho
_{X} \otimes \sigma\Vert \rho_{XA}).
\end{align}

\subsubsection{Properties of log-Euclidean fidelity}

\begin{theorem}[Properties of multivariate log-Euclidean fidelity] \label{thm:properties_quantum_Matusita_multi}
     The multivariate log-Euclidean fidelity in \cref{def:multi_log_eucli_fidelity} satisfies several desirable properties of a multivariate fidelity listed in \cref{def:properties_multi_fidelity}, including reduction to classical Matusita multivariate fidelity in~\eqref{eq:Matusita_fid_commuting}, data processing, symmetry, faithfulness, direct-sum property, and continuity.

     If all states are orthogonal to each other, i.e., $\rho_i \rho_j=0$ for $i \neq j, \ i,j  \in [r]$, then $ F^{\flat}_r (\rho_1, \ldots, \rho_r)=0$ (weak orthogonality). However, $ F^{\flat}_r (\rho_1, \ldots, \rho_r)=0$ does not imply that the states are orthogonal to each other.
\end{theorem}

\begin{proof}
    See \cref{proof:properties_quantum_Matusita_multi} for a proof of the properties.
    As a counterexample for one direction of the orthogonality property mentioned above, consider the diagonal states formed by the probability vectors given in~\eqref{eq:counterexample-orthogonality-Matusita}. 
\end{proof}
\medskip 

Additionally, the multivariate log-Euclidean fidelity is multiplicative, as a consequence of \cref{prop:additivity_log_Euclidean_divergence}. 
\begin{corollary}[Multiplicativity of log-Euclidean fidelity]
    Let $\rho_1,\ldots,\rho_r$ and $\sigma_1, \ldots, \sigma_r$ be states. The multivariate log-Euclidean fidelity satisfies multiplicativity; i.e., 
    \begin{equation}
         F^{\flat}_r(\rho_1 \otimes \sigma_1, \ldots, \rho_r \otimes \sigma_r) = F^{\flat}_r(\rho_1, \ldots, \rho_r)  F^{\flat}_r(\sigma_1, \ldots, \sigma_r).
    \end{equation}
\end{corollary}

\begin{proof}
The proof follows by substituting  $s$ to be uniform (i.e., $s_i=1/r$ for all $i \in [r]$) in~\cref{prop:additivity_log_Euclidean_divergence} together with~\eqref{eq:quantum_Matusita_fidelity_log_Eucl}.   
\end{proof}

Next, we establish a quantum generalization of the inequality from \cref{prop:r_root_to_commuting_fidelity}.

\begin{proposition}[Multivariate log-Euclidean fidelity upper bound] \label{prop:upper_multi_log_eucli}
Let $(\rho_i)_{i=1}^r$ be a tuple of states. Then, we have
\begin{equation}
    F^{\flat}_r(\rho_1, \ldots, \rho_r) \leq \frac{2}{r(r-1)} \sum_{i<j} F_{\flat}(\rho_i,\rho_j).
\end{equation}   
\end{proposition}
\begin{proof}
For positive definite states $\rho_{1}$, \ldots, $\rho_{r}$, consider that
\begin{align}
& \frac{2}{r\left(  r-1\right)  }\sum_{i<j}F_{\flat}(\rho_{i},\rho
_{j})\nonumber\\
& =\frac{2}{r\left(  r-1\right)  }\sum_{i<j}\exp\!\left(  -\inf_{\omega
\in\mathscr{D}}\left[  \frac{1}{2}D(\omega\Vert\rho_{i})+\frac{1}{2}
D(\omega\Vert\rho_{j})\right]  \right)  \\
& \geq\exp\!\left(  -\frac{2}{r\left(  r-1\right)  }\sum_{i<j}\inf_{\omega
\in\mathscr{D}}\left[  \frac{1}{2}D(\omega\Vert\rho_{i})+\frac{1}{2}
D(\omega\Vert\rho_{j})\right]  \right)  \\
& \geq\exp\!\left(  -\inf_{\omega\in\mathscr{D}}\frac{2}{r\left(  r-1\right)
}\sum_{i<j}\left[  \frac{1}{2}D(\omega\Vert\rho_{i})+\frac{1}{2}D(\omega
\Vert\rho_{j})\right]  \right)  \\
& =\exp\!\left(  -\inf_{\omega\in\mathscr{D}}\frac{1}{r}\sum_{i=1}^{r}
D(\omega\Vert\rho_{i})\right)  \\
& =F^{\flat}_r(\rho_{1},\ldots,\rho_{r}),
\end{align}
where the first equality follows from~\eqref{eq:quantum_Matusita_oveloh}; the first inequality  from convexity of the exponential
function; the second inequality follows because
\begin{multline}
\inf_{\omega\in\mathscr{D}}\frac{2}{r\left(  r-1\right)  }\sum_{i<j}\left[
\frac{1}{2}D(\omega\Vert\rho_{i})+\frac{1}{2}D(\omega\Vert\rho_{j})\right]  \\ 
\geq\frac{2}{r\left(  r-1\right)  }\sum_{i<j}\inf_{\omega\in\mathscr{D}
}\left[  \frac{1}{2}D(\omega\Vert\rho_{i})+\frac{1}{2}D(\omega\Vert\rho
_{j})\right];
\end{multline}
the penultimate equality follows from a counting argument, and the final equality
follows from~\eqref{eq:quantum_Matusita_oveloh}. The general case follows by a limiting argument.
\end{proof}

\subsubsection{Operational interpretation of multivariate log-Euclidean divergences}

\label{sec:op-int-mult-log-euc-fid}

The following formula is known (see~\cite{mosonyi2022geometric}, and \cite[Appendix~J]{mnw2023}, as well as \cite[Eq.~(47)]{bhatia2019matrix} for a related observation):
\begin{equation}
H_{s}(\rho_{1},\ldots,\rho_{r})=\inf_{\omega\in\mathscr{D}}\sum_{i=1}^{r}
s_{i}D(\omega\Vert\rho_{i})=\inf_{\omega\in\mathscr{D}}D(\pi(s)\otimes
\omega\Vert\rho(s)),\label{eq:hellinger-to-optimized-rel-ent}
\end{equation}
where $\mathscr{D}$ is the set of density matrices and the quantum relative entropy is given in~\eqref{eq:q-rel-entr}.
The second equality
follows from the direct-sum property of quantum relative entropy \cite[Eq.~(7.2.27)]{khatri2020principles}, along with
the following definitions:
\begin{equation}
\pi(s)   \coloneqq \sum_{i=1}^{r}s_{i}|i\rangle\!\langle i|,\qquad 
\rho(s)   \coloneqq \sum_{i=1}^{r}s_{i}|i\rangle\!\langle i|\otimes\rho_{i}.
\end{equation}

Now we present the hypothesis testing setting and result that provides an operational interpretation of the multivariate log-Euclidean divergence. The setting was considered in a general way in~\cite{Notzel_2014} and is  known as quantum hypothesis
testing with arbitrarily varying null hypothesis and simple alternative hypothesis. 
The null hypothesis for our case is that a state  is chosen adversarially from the following set:
\begin{equation}
\Omega^{(n)}_s\coloneqq \left\{  \bigotimes\limits_{j=1}^{n}\left(  \pi(s)\otimes
\omega_{j}\right)  :\omega_{j}\in\mathscr{D\ }\, \, \forall j\in\left[  n\right]
\right\}  .
\end{equation}
This set is an example of what is called an arbitrarily varying source in~\cite{Notzel_2014}.
The alternative hypothesis is that the state $\rho(s)^{\otimes n}$ is chosen.
The receiver of the unknown state is allowed to perform a measurement $\left\{  M^{(n)},I^{\otimes
n}-M^{(n)}\right\}  $ acting on all $n$ systems in order to figure out which
is the case, such that $0\leq M^{(n)}\leq I^{\otimes n}$. The outcome
$M^{(n)}$ is associated with choosing the null hypothesis, and the outcome
$I^{\otimes n}-M^{(n)}$ is associated with choosing the alternative
hypothesis. Formally, the type~I and type~II error probabilities are
respectively defined in this setting as
\begin{align}
\alpha(M^{(n)})  & \coloneqq \sup_{\sigma^{(n)}\in\Omega^{(n)}_s}\operatorname{Tr}
\!\left[  \left(  I^{\otimes n}-M^{(n)}\right)  \sigma^{(n)}\right]  ,\\
\beta(M^{(n)})  & \coloneqq \operatorname{Tr}[M^{(n)}\rho(s)^{\otimes n}].
\end{align}
The fact that the type~I\ error probability is maximized over arbitrary
$\sigma^{(n)}\in\Omega^{(n)}_s$ is what makes the
hypothesis testing scenario adversarial, as the state $\sigma^{(n)}$ is
selected from $\Omega^{(n)}_s$ in such a way, for a given measurement operator
$M^{(n)}$ {(i.e., $0\leq M^{(n)}\leq I^{\otimes n}$)}, to make the type~I error probability as large as possible. For
$\varepsilon\in\left[  0,1\right]  $ and $n \in \mathbb{N}$, the  optimal non-asymptotic type~II error exponent is
defined as
\begin{equation}
\frac{1}{n}D_{H}^{\varepsilon}(\Omega^{(n)}_s\Vert\rho(s)^{\otimes n}
)\coloneqq -\frac{1}{n}\ln\!\left(  \inf_{M^{(n)}\in\mathcal{M}^{(n)}}\left\{
\beta(M^{(n)}):\alpha(M^{(n)})\leq\varepsilon\right\}  \right)  ,
\label{eq:hypo-test-div-mod}
\end{equation}
where $\mathcal{M}^{(n)}$ denotes the set of all measurement operators:
\begin{equation}
\mathcal{M}^{(n)} \coloneqq \left \{ M^{(n)}: 0\leq M^{(n)}\leq I^{\otimes n}\right\}.
\end{equation}

\begin{corollary} \label{Cor:log_Euclidean_operational}
Let $s$ be a probability distribution on $[r]$, and let $\rho_1, \ldots, \rho_r$ be  states. For all $\varepsilon\in\left(  0,1\right)  $, the following equality holds:
\begin{equation}
\lim_{n\rightarrow\infty}\frac{1}{n}D_{H}^{\varepsilon}(\Omega^{(n)}_s\Vert
\rho(s)^{\otimes n})=H_{s}(\rho_{1},\ldots,\rho_{r}),
\end{equation}
where $\frac{1}{n} D_{H}^{\varepsilon}(\Omega^{(n)}_s\Vert
\rho(s)^{\otimes n})$ is defined in~\eqref{eq:hypo-test-div-mod} and the multivariate log-Euclidean divergence in \cref{def:multi-var-log-euclid}.
\end{corollary}

\begin{proof}
    The proof follows by applying \cite[Theorem~1]{Notzel_2014} and the identity in
\eqref{eq:hellinger-to-optimized-rel-ent}.
\end{proof}
\medskip

Thus, we obtain a fundamental operational meaning for $H_{s}(\rho_{1}
,\ldots,\rho_{r})$ in this hypothesis testing scenario, as the optimal asymptotic type~II error exponent. By picking $s$ to be
uniform, we obtain an operational meaning for the negative logarithm of the  multivariate log-Euclidean fidelity in~\eqref{eq:quantum_Matusita_fidelity_log_Eucl}.

\subsection{Average $k$-wise log-Euclidean fidelities} \label{subS:average_k_log_eucli_F}

As a generalization of the ideas presented in \Cref{sec:avg-classical-k-wise-fids} for the classical case, we define the average
$k$-wise log-Euclidean fidelities as follows:

\begin{definition}
[Average $k$-wise log-Euclidean fidelity]\label{def:average_k_log_Eucl}
For $r\in\left\{  3,4,\ldots\right\}
$, let $\rho_{1},\ldots,\rho_{r}$ be quantum states. For $k\in\left\{
2,3,\ldots,r\right\}  $, define the average $k$-wise log-Euclidean fidelity as
follows:
\begin{equation}
F_{k,r}^{\flat}(\rho_{1},\ldots,\rho_{r}):= \binom{r}{k}^{-1}\sum
_{i_{1}<\cdots<i_{k}}F_{k}^{\flat}(\rho_{i_{1}},\ldots,\rho_{i_{k}}).
\end{equation}

\end{definition}

We can establish an ordering relationship between them, generalizing that
found in \Cref{prop:ineqs-avg-k-wise-classical} for the classical case. The proof below is similar to previous ones, combining
elements of the proofs of \Cref{prop:ineqs-avg-k-wise-classical} and \Cref{prop:upper_multi_log_eucli}:

\begin{proposition}[Ordering of average $k$-wise log-Euclidean fidelities]
\label{prop:avg-k-wise-log-Euc-ordered}
For $r\in\left\{  3,4,\ldots\right\}  $, let $\rho_{1},\ldots,\rho_{r}$ be
quantum states. Then the following inequalities hold:
\begin{equation}
F_{2,r}^{\flat}\geq F_{3,r}^{\flat}\geq\cdots\geq F_{r-1,r}^{\flat}\geq
F_{r,r}^{\flat},
\end{equation}
where, for brevity, we have suppressed the dependence of $F_{k,r}^{\flat}$ on
$\rho_{1},\ldots,\rho_{r}$.
\end{proposition}

\begin{proof}
See \Cref{sec:avg-k-wise-log-Euc-ordered}.
\end{proof}

\subsection{Multivariate geometric fidelities}

\label{sec:multiv-geo-fids}

In this section, we propose multivariate generalizations of the geometric fidelity defined below in~\eqref{eq:geometric_fidelity_def}, by following the same procedures used in generalizing Uhlmann fidelity (c.f., \cref{def:average_pairwise_z_fidelity} and \cref{def:multivar-fid}). We also show that they satisfy several desired properties of a multivariate fidelity, as stated in \cref{def:properties_multi_fidelity}.

Recall the geometric fidelity, which is defined in~\eqref{eq:geometric_fidelity} as 
\begin{equation}
\label{eq:geometric_fidelity_def}
     F_G(\rho,\sigma) \coloneqq \inf_{\varepsilon >0 } \Tr[ \rho(\varepsilon) \# \sigma(\varepsilon)],
\end{equation}
where $A \# B$ denotes the geometric mean of the positive definite operators $A$ and $B$, as defined in~\eqref{eq:geometric-mean}, and we set
$\rho(\varepsilon) \coloneqq \rho + \varepsilon I$
and $\sigma(\varepsilon) \coloneqq \sigma + \varepsilon I$. 
The geometric fidelity (also known as Matsumoto fidelity) satisfies reduction to classical fidelity in~\eqref{eq:commuting_bivariate} for commuting states, symmetry, data processing, faithfulness, and direct-sum property, similar to the Uhlmann and Holevo fidelities. However, it does not satisfy both directions of the orthogonality property. The direction it satisfies is that $F_G(\rho,\sigma) =0$ for a pair of orthogonal states (i.e., satisfying $\rho\sigma = 0$). For proofs of these properties, refer to \cite[Section~8]{matsumoto2010reverse} and \cite[Section~3]{cree2020fidelity}.

We define the average pairwise geometric fidelity as follows: 
\begin{definition}[Average pairwise geometric fidelity] \label{def:avg_geometric}
For quantum states $\rho_1, \ldots, \rho_r$, 
   the average pairwise geometric fidelity is defined as 
   \begin{equation}
F_G(\rho_1, \ldots, \rho_r) \coloneqq \frac{2}{r(r-1)} \sum_{i <j} F_G(\rho_i, \rho_j).
   \end{equation}  
\end{definition}

\begin{proposition}[Properties of average pairwise geometric fidelity]
    The average pairwise geometric fidelity in~\cref{def:avg_geometric} satisfies reduction to average pairwise classical fidelity for commuting states, data processing, symmetry, faithfulness, and the direct-sum property, as stated in \cref{def:properties_multi_fidelity}. 
    In addition, $F_G(\rho_1, \ldots, \rho_r)$ satisfies super-multiplicativity and coarse-graining. 
\end{proposition}

\begin{proof}
    A proof follows from the fact that the bivariate geometric fidelity satisfies the said properties (see~\cite[Section~3]{cree2020fidelity}).

    Super-multiplicativity and coarse-graining follow from similar proof lines as given for \cref{prop:super_multiplicative_average_pairwise} and \cref{prop:coarse_graining_avg_pairwise}, along with the fact that $F_G\!\left(\rho^{\otimes n}, \sigma^{\otimes n}\right) =  F_G(\rho,\sigma)^n$ as stated in  \cite[Proposition~7]{matsumoto2010reverse}.
\end{proof}

The bivariate geometric fidelity has the following SDP formulation, which is given in \cite[Section~6]{matsumoto2014quantum} and \cite[Eq.~(1.11) and Lemma~3.20]{cree2020fidelity}:
\begin{align}
    F_G(\rho,\sigma) & =\sup_{W \in \operatorname{Herm}} \left\{ \Tr[W] 
     :   \begin{bmatrix}
        \rho & W \\
        W & \sigma
        \end{bmatrix} \geq 0\right\}.\\
        & = \frac{1}{2} \inf_{Y,Z,X} \left\{\Tr[Y \rho] + \Tr[Z\sigma] : \begin{bmatrix}
        Y & X^\dag \\
        X & Z
        \end{bmatrix} \geq 0, \, \frac{X+X^\dag}{2} = -I\right\}.
\end{align}
The last equality follows from the substitution $X \to -X$ when compared to \cite[Lemma~3.20]{cree2020fidelity}, as well as the fact that
\begin{equation}
    \begin{bmatrix}
        Y & X^\dag \\
        X & Z
        \end{bmatrix} \geq 0 \qquad \Longleftrightarrow \qquad \begin{bmatrix}
        Y & -X^\dag \\
       - X & Z
        \end{bmatrix} \geq 0.
\end{equation}
Indeed a multivariate generalization of the SDP above, which is similar to the multivariate SDP fidelity in~\cref{def:multivar-fid}, is as follows:

\begin{definition}[Multivariate SDP geometric fidelity] \label{def:SDP_geometric_multi}
For quantum states $\rho_1, \ldots, \rho_r$, 
  the multivariate SDP geometric fidelity is defined as 
    \begin{multline} F_{\operatorname{SDP}}^G(\rho_1, \ldots, \rho_r)  \coloneqq  \\ \frac{2}{r (r-1)}\sup_{\substack{(X_{ij})_{i \neq j} \textnormal{ s.t.} \\ X_{ji} =X_{ij} \in \operatorname{Herm}} } \left \{  \sum_{i <j}  \Tr[X_{ij}]: \sum_{i=1}^r |i\rangle\!\langle i| \otimes \rho_i + \sum_{i \neq j} |i\rangle\!\langle j| \otimes X_{ij} \geq 0 \right \}.
    \label{eq:SDP-multi-geometric-def}
\end{multline}
\end{definition}

\begin{proposition}[Dual of multivariate SDP geometric fidelity] \label{prop:dual_SDP_multi_geo}
The multivariate SDP geometric fidelity in \cref{def:SDP_geometric_multi} has the following dual formulation: 
      \begin{multline}
F^G_{\operatorname{SDP}}(\rho_1, \ldots, \rho_r)= \\
         \frac{1}{r (r-1)}\inf_{Y_{ij}}\left\{  \sum_{i}\operatorname{Tr}[Y_{ii}\rho_{i}]:\sum
_{i,j}|i\rangle\!\langle j|\otimes Y_{ij}\geq0,\ \frac{Y_{ij}+Y_{ji}^{\dag
}}{2}=-I_d\ \forall i\neq j\right\}.         \label{eq:multivariate_fid_sdp_dual_geo}
      \end{multline}
\end{proposition}

\begin{proof}
    See \Cref{Sec:Proof_dual_SDP_geo}.
\end{proof}

The following inequalities hold:
\begin{equation} F_{\operatorname{SDP}}^G(\rho_1, \ldots, \rho_r) \leq F_G(\rho_1, \ldots, \rho_r) \leq F_H(\rho_1, \ldots, \rho_r).
     \label{eq:ineq-FGSDP-FG-FH}
\end{equation}
The first inequality follows from a similar proof strategy used to prove the third inequality stated  in~\cref{thm:SDP_fidelity_pairwise_upper_and_lower} (i.e., $F_{\operatorname{SDP}}(\rho_1, \ldots, \rho_r) \leq F_U(\rho_1, \ldots, \rho_r)$). The second inequality follows from definitions of the average pairwise fidelities and because $F_G(\rho,\sigma) \leq F_H(\rho,\sigma)$ for all states $\rho$ and $\sigma$, as recalled in \eqref{eq:ordering-prep-meas-fids} and \eqref{eq:geometric_fidelity}. \cref{thm:SDP_fidelity_pairwise_upper_and_lower} then allows us to conclude that 
\begin{equation}
F_{\operatorname{SDP}}^G\leq F_G \leq F_H \leq F_S \leq F_{\operatorname{SDP}} \leq F_U.
\label{eq:upper_geometric_SDP}
\end{equation}
    
\begin{remark}[Multivariate geometric fidelities for pure states]
 For a tuple of orthogonal states $\left(|\phi_i\rangle\!\langle \phi_i| \right)_{i=1}^r $, we have that 
 \begin{equation}
    F_{\operatorname{SDP}}^G\!\left( |\phi_1\rangle\!\langle \phi_1|, \ldots, |\phi_r\rangle\!\langle \phi_r|\right) = F_G\!\left( |\phi_1\rangle\!\langle \phi_1|, \ldots, |\phi_r\rangle\!\langle \phi_r|\right) =0.
 \end{equation}
 This follows due to the fact that for two pure states $F_G\!\left( |\phi_1\rangle\!\langle \phi_1|, |\phi_2\rangle\!\langle \phi_2|\right) =0$ (~\cite[Eq.~(23)]{matsumoto2010reverse}) together with~\eqref{eq:ineq-FGSDP-FG-FH}.
\end{remark}

\begin{proposition}[Properties of multivariate SDP geometric fidelity]
\label{prop:SDP-geo-props}
    The multivariate SDP geometric fidelity in~\cref{def:SDP_geometric_multi} satisfies reduction to average pairwise classical fidelity for commuting states, data processing, symmetry, faithfulness, non-negativity, weak orthogonality (i.e., if $\rho_i\rho_j=0$ for all $i \neq j \in [r]$, then $F_{\operatorname{SDP}}^G(\rho_1, \ldots, \rho_r)=0$), and the direct-sum property.
    In addition multivariate SDP geometric fidelity is super-multiplicative: 
    \begin{align}
     \left( F^G_{\operatorname{SDP}}(\rho_1, \ldots, \rho_r)\right)^n &\leq F^G_{\operatorname{SDP}}\!\left(\rho_1^{\otimes n}, \ldots, \rho_r^{\otimes n}\right).  
    \end{align}
\end{proposition}

\begin{proof}
    See \cref{app:SDP-geo-props}.
\end{proof}

Similar to what we did previously in  \eqref{eq:def-S-SDP-secrecy}, \eqref{eq:secrecy_measure}, and \cref{rem:holevo-pairwise-secrecy}, we can define a secrecy measure
based on the multivariate SDP\ geometric fidelity, which itself has an
SDP\ representation. Let us define the secrecy measure $S^{G}
_{\operatorname{SDP}}$ as follows:
\begin{equation}
S^{G}_{\operatorname{SDP}}(\rho_{1},\ldots,\rho_{r})\coloneqq 
\sqrt{\frac{\left(  r-1\right)  F^{G}_{\operatorname{SDP}}(\rho_{1},\ldots,\rho
_{r})+1}{r}}. \label{eq:def-mult-secrecy-geo}
\end{equation}
One might think that this secrecy measure should be multiplicative in the sense shown in \cref{prop:Multiplicativity-S-SDP}, \eqref{eq:mult-secrecy-meas-S}, and \eqref{eq:mult-secrecy-meas-S-H}. However, we are only able to prove that it is super-multiplicative, and we suspect, based on the form of the dual SDP in \eqref{eq:secrecy-measure-G-dual} below, that sub-multiplicativity does not hold in general.

We first determine an alternative form for the multivariate SDP\ geometric
fidelity, which is similar in spirit to the $K^{\star}$ representation from \cref{def:K-star-rep}:

\begin{proposition}
\label{prop:mult-geo-SDP-G-alt}
Given $r$ states $\rho_{1},\ldots,\rho_{r}$ of dimension $d$, the multivariate
SDP\ geometric fidelity $F^{G}_{\operatorname{SDP}}(\rho_{1},\ldots,\rho_{r})$
can be written as follows:
\begin{equation}
F^{G}_{\operatorname{SDP}}(\rho_{1},\ldots,\rho_{r})=\frac{1}{r\left(
r-1\right)  }\sup_{X\geq0}\left\{
\begin{array}
[c]{c}
\operatorname{Tr}[\left(  |+\rangle\!\langle+|\otimes I_d\right)  X]-r:\\
\left(  \Delta_{r}\otimes\operatorname{id}_{d}\right)  \left(  X\right)
=\sum_{i=1}^{r}|i\rangle\!\langle i|\otimes\rho_{i},\\
\left(  T_r\otimes\operatorname{id}_d\right)  \left(  X\right)  =X
\end{array}
\right\}  , \label{eq:mult-geo-SDP-G-alt}
\end{equation}
where the vector $|+\rangle$, the completely dephasing channel $\Delta_{r}$,
and the transpose map $T_r$ are defined as
\begin{equation} 
|+\rangle   \coloneqq \sum_{i=1}^{r}|i\rangle,
\qquad
\Delta_{r}(\cdot)   \coloneqq \sum_{i=1}^{r}|i\rangle\!\langle
i|(\cdot)|i\rangle\! \langle i|,\qquad
T_r(\cdot)    \coloneqq \sum_{i,j=1}^{r}|i\rangle\!\langle j|(\cdot
)|i\rangle\!\langle j|.
\end{equation}

\end{proposition}

\begin{proof}
    See \cref{app:mult-geo-SDP-G-alt}.
\end{proof}

The quantity $S^{G}_{\operatorname{SDP}}$ has the following SDP\ representation:

\begin{proposition}
\label{prop:secrecy-meas-geo-SDP-SDP-dual}
Using the same notation from \cref{prop:mult-geo-SDP-G-alt}, the square of the secrecy measure
$S^{G}_{\operatorname{SDP}}(\rho_{1},\ldots,\rho_{r})$ can be written as
follows:
\begin{align}
 & S^{G}_{\operatorname{SDP}}(\rho_{1},\ldots,\rho_{r})^{2} \notag \\ 
&  =\frac{1}{r^{2}}\sup_{X\geq0}\left\{
\begin{array}
[c]{c}
\operatorname{Tr}[\left(  |+\rangle\!\langle+|\otimes I_d\right)  X]:\\
\left(  \Delta_{r}\otimes\operatorname{id}_{d}\right)  \left(  X\right)
=\sum_{i=1}^{r}|i\rangle\!\langle i|\otimes\rho_{i},\\
\left(  T_r\otimes\operatorname{id}_d\right)  \left(  X\right)  =X
\end{array}
\right\} \label{eq:secrecy-measure-G-primal}\\
&  =\frac{1}{r^{2}} \inf_{Y,Z\in\operatorname{Herm}}\left\{
\begin{array}
[c]{c}
\operatorname{Tr}\!\left[  Y\left(  \sum_{i=1}^{r}|i\rangle\!\langle
i|\otimes\rho_{i}\right)  \right]  :\\
|+\rangle\!\langle+|\otimes I_d +\left(  T_r\otimes\operatorname{id}_d\right)
\left(  Z\right)  \leq\left(  \Delta_{r}\otimes\operatorname{id}_{d}\right)
\left(  Y\right)  +Z
\end{array}
\right\}  . \label{eq:secrecy-measure-G-dual}
\end{align}

\end{proposition}

\begin{proof}
    See \cref{app:secrecy-meas-geo-SDP-SDP-dual}.
\end{proof}

\begin{proposition}
\label{prop:multi-geo-submult}
For states $\rho_{1},\ldots,\rho_{r_{1}}$ and $\sigma_{1},\ldots,\sigma
_{r_{2}}$, the secrecy measure $S_{\operatorname{SDP}}^{G}$ is
super-multiplicative in the following sense:
\begin{equation}
S_{\operatorname{SDP}}^{G}(\rho_{1},\ldots,\rho_{r_{1}})\cdot
S_{\operatorname{SDP}}^{G}(\sigma_{1},\ldots,\sigma_{r_{2}})\leq
S_{\operatorname{SDP}}^{G}((\rho_{i}\otimes\sigma_{j})_{i\in\left[
r_{1}\right]  ,j\in\left[  r_{2}\right]  }).\label{eq:multi-geo-submult}
\end{equation}

\end{proposition}

\begin{proof}
    See \cref{app:multi-geo-submult}.
\end{proof}

Due to the presence of the terms $Z$ and $\left(  T_r\otimes\operatorname{id}_d\right)
\left(  Z\right)$ in the constraint in \eqref{eq:secrecy-measure-G-dual}, it is not clear whether submultiplicativity of $S_{\operatorname{SDP}}^{G}$ holds. That is, it is unclear that the conventional strategy of taking the tensor products $Y_1 \otimes Y_2$ and $Z_1 \otimes Z_2$, where $Y_1$ and $Z_1$ are feasible for $S_{\operatorname{SDP}}^{G}(\rho_{1},\ldots,\rho_{r_{1}})$ and  $Y_2$ and $Z_2$ are feasible for $S_{\operatorname{SDP}}^{G}(\sigma_{1},\ldots,\sigma_{r_{2}})$, will work here. This is because these tensor products are not necessarily feasible for $S_{\operatorname{SDP}}^{G}((\rho_{i}\otimes\sigma_{j})_{i\in\left[
r_{1}\right]  ,j\in\left[  r_{2}\right]  })$. So it remains an open question to determine whether submultiplicativity of $S_{\operatorname{SDP}}^{G}$ holds.

\section{Concluding remarks and future directions} \label{Sec:Conclusion}

In this paper, we extended bivariate fidelities to multivariate fidelities, which are measures characterizing the similarity between multiple quantum states. In the classical case, we analysed the Matusita multivariate fidelity~\cite{Matusita1967notion}, average pairwise fidelity, and average $k$-wise fidelities. We proposed several variants for the quantum case that reduce to the average pairwise fidelity for commuting states, namely, average pairwise $z$-fidelities for $z \geq 1/2$, the multivariate SDP fidelity, inspired by the SDP of the Uhlmann fidelity, and the secrecy-based multivariate fidelity inspired by the secrecy measure proposed in ~\cite{konig2009operational}. We analysed their mathematical properties, such as data processing, faithfulness, uniform continuity bounds, and connections between these  variants. To this end, we proved that multivariate SDP fidelity and secrecy-based multivariate fidelity are sandwiched between the average pairwise $z$-fidelities at $z=1/2$ and $z=1$. 
We also explored a quantum extension of Matusita multivariate fidelity, which we called as multivariate log-Euclidean fidelity while analysing its properties and providing an operational interpretation to this in terms of quantum hypothesis testing with an arbitrarily varying null hypothesis. Finally, we generalized Matsumoto's geometric fidelity to the multivariate case in two different ways, and we established several of their properties.

Open questions for future work include the following (in no particular order):
\begin{enumerate}
    \item 
    What are similarity measures  for a tuple of channels, which generalize multivariate fidelity of states? One approach towards this is to extend the SDP of channel fidelity \cite[Proposition~55]{katariya2021geometric} from two channels to multiple channels, by using the same approach we followed here when defining multivariate SDP fidelity  (\cref{def:multivar-fid}). In particular, we can define multivariate channel SDP fidelity for a tuple of channels $\left( \cN_i \right)_{i=1}^r$ (each one of those is a CPTP map from $\mathscr{L}_A$ to $\mathscr{L}_B$) as follows:
\begin{align}
     & F_{\operatorname{SDP}}\!\left(\cN_1, \ldots, \cN_r \right) \notag \\
     & \coloneqq  \inf_{\rho_{RA} \in \mathscr{D}_{RA}}  F_{\operatorname{SDP}}\!\left( \cN_1(\rho_{RA}), \ldots, \cN_r(\rho_{RA}) \right)  \\
     & = \frac{1}{r (r-1)} \inf_{\substack{Y_{RB}^i \geq 0 \ \forall  i \in [r], \\ \rho_R \in \mathscr{D}}} \left \{
     \begin{array}{c}
     \sum_{i=1}^r \Tr\!\left[ Y_{RB}^i   \Gamma_{RB}^{\cN_i} \right] : \\
     \sum_{i=1}^r |i\rangle\!\langle i| \otimes Y_{RB}^i \geq  \sum_{i \neq j} |i\rangle\!\langle j| \otimes \rho_R \otimes I_B
     \end{array}\right\},
\end{align}
where $ \Gamma_{RB}^{\cN_i}$ is the Choi operator of channel $\cN_i$ for all $i\in[r]$ (i.e., $\Gamma_{RB}^{\cN_i} \coloneqq \sum_{k,\ell} |k\rangle\!\langle \ell| \otimes \mathcal{N}_i(|k\rangle\!\langle \ell|)$). The proof of the equality in the last line follows from the same reasoning used to prove \cite[Proposition~55]{katariya2021geometric}.

Furthermore, similar to the average pairwise $z$-fidelity for states in \cref{def:average_pairwise_z_fidelity}, we can define average pairwise channel $z$-fidelity as follows for $z \geq 1/2$:
\begin{align}
    F_z(\cN_1, \ldots, \cN_r) 
   & \coloneqq \inf_{\rho_{RA} \in \mathscr{D}_{RA}}  F_{z}\!\left( \cN_1(\rho_{RA}), \ldots, \cN_r(\rho_{RA}) \right) \\ 
   & = \inf_{\rho_{RA} \in \mathscr{D}_{RA}} \frac{2}{r(r-1)} \sum_{i <j} F_z (\cN_i(\rho_{RA}),\cN_{j}(\rho_{RA})).
\end{align}

One could also define multivariate geometric channel fidelities, generalizing the approach for states from \cref{sec:multiv-geo-fids}.

\item Are there information-theoretic tasks that provide operational interpretations to quantum generalizations of multivariate fidelity, including average pairwise $z$-fidelities and multivariate SDP fidelity? We already noted interpretations for the multivariate log-Euclidean fidelity in \cref{sec:op-int-mult-log-euc-fid} and for the secrecy-based multivariate fidelity in \cite[Theorem~6]{RASW23}. 

\item 
Is $\sqrt{2(1-F_z(\rho,\sigma))}$ a distance measure for $z \in (1/2,1)\cup(1,\infty)$? 
 For PSD matrices $A$ and $B$, this question refers to evaluating whether $\sqrt{ \Tr[A+B]- 2F_z(A,B)}$ is a distance measure. Knowing this in turn would provide uniform continuity bounds for average pairwise $z$-fidelities, as  \cref{Prop:uniform_cont_average_pairwise} does for the special cases $z \in \{1/2,1\}$. 

\item Another question that we pose is related to multivariate generalizations of the Holevo fidelity. It was shown in \cite[page 14]{matsumoto2014quantum} that the Holevo fidelity has the following SDP characterization:
\begin{equation}
    F_H(\rho,\sigma) = \sup_{C\geq 0} \left\{ \langle \Gamma | C |\Gamma\rangle: \begin{bmatrix}
        \rho \otimes I& C \\
        C & I \otimes \sigma^T
        \end{bmatrix} \geq 0\right\},
        \label{eq:holevo-fid-sdp}
\end{equation}
where $|\Gamma \rangle \coloneqq \sum_i |i\rangle|i\rangle$.
This characterization follows because an optimal choice for $C$ is the matrix geometric mean of the commuting operators $\rho \otimes I$ and $I \otimes \sigma^T$ (see, e.g., \cite[Theorem~4.1.3]{bhatia2009positive}), which is simply $(\rho \otimes \sigma^T)^{1/2}$, and because $\langle \Gamma | (\rho \otimes \sigma^T)^{1/2} |\Gamma\rangle = \operatorname{Tr}[\rho^{1/2}\sigma^{1/2}]$. The dual of the SDP in~\eqref{eq:holevo-fid-sdp} is given by
\begin{equation}
\inf_{\substack{Y,Z\geq0, \\ W\in\mathscr{L}}}\left\{
\operatorname{Tr}[Y(  \rho\otimes I)  ]+\operatorname{Tr}[Z(
I\otimes\sigma^{T})  ]:
W+W^{\dag}\geq|\Gamma\rangle\!\langle\Gamma|,
\begin{bmatrix}
Y & W\\
W^{\dag} & Z
\end{bmatrix}
\geq0
\right\} .
\label{eq:holevo-fid-sdp-dual}
\end{equation}
However, it is unclear how to generalize the constructions in~\eqref{eq:holevo-fid-sdp} and~\eqref{eq:holevo-fid-sdp-dual} to multiple states, given that it makes use of properties specific to bipartite systems. Interestingly, the constructions in~\eqref{eq:multivariate_fid_sdp_dual} and~\eqref{eq:SDP-multi-geometric-def} do not encounter this difficulty. 

\item Is there a possibility of
finding analytical or alternative variational expressions for minimal and maximal extensions of multivariate classical fidelity in \cref{def:minimal_multi_F} and \cref{def:maximal_multi_fidelity}, respectively?
\end{enumerate}

\section*{Acknowledgements}

We thank Komal Malik and Kaiyuan Ji for helpful discussions. TN  acknowledges support from the NSF under grant no.~2329662. HKM  acknowledges support from the NSF under grant no.~2304816 and AFRL under agreement no.~FA8750-23-2-0031.
FL acknowledges support from the Department of Mathematics at UIUC for a research visit to Cornell University, as well as the hospitality of Mark Wilde's research group during that stay.
MMW acknowledges support from the NSF under grants 2329662, 2315398, 1907615, 2304816. 

This material is based on research
sponsored by Air Force Research Laboratory under agreement number
FA8750-23-2-0031. The U.S.~Government is authorized to reproduce and
distribute reprints for Governmental purposes notwithstanding any copyright
notation thereon. The views and conclusions contained herein are those of the
authors and should not be interpreted as necessarily representing the official
policies or endorsements, either expressed or implied, of Air Force Research
Laboratory or the U.S.~Government. 

\bibliographystyle{alpha}
 \bibliography{reference}

\newcommand{\etalchar}[1]{$^{#1}$}
\begin{thebibliography}{MLDS{\etalchar{+}}13}

\bibitem[AD15]{audenaert2015alpha}
Koenraad M.~R. Audenaert and Nilanjana Datta.
\newblock $\alpha$-z-{R}{\'e}nyi relative entropies.
\newblock {\em Journal of Mathematical Physics}, 56(2):022202, 2015.
\newblock arXiv:1310.7178 [quant-ph].

\bibitem[AFPA07]{arsigny2007geometric}
Vincent Arsigny, Pierre Fillard, Xavier Pennec, and Nicholas Ayache.
\newblock Geometric means in a novel vector space structure on symmetric positive-definite matrices.
\newblock {\em SIAM Journal on Matrix Analysis and Applications}, 29(1):328--347, 2007.

\bibitem[AKF22]{afham2022quantum}
A.~Afham, Richard Kueng, and Chris Ferrie.
\newblock Quantum mean states are nicer than you think: fast algorithms to compute states maximizing average fidelity, 2022.
\newblock arXiv:2206.08183 [quant-ph].

\bibitem[ANSV08]{audenaert2008asymptotic}
Koenraad M.~R. Audenaert, Michael Nussbaum, Arleta Szko{\l}a, and Frank Verstraete.
\newblock Asymptotic error rates in quantum hypothesis testing.
\newblock {\em Communications in Mathematical Physics}, 279:251--283, 2008.
\newblock arXiv:0708.4282 [quant-ph].

\bibitem[AU83]{alberti1983stochastic}
Peter~M Alberti and Armin Uhlmann.
\newblock Stochastic linear maps and transition probability.
\newblock {\em Letters in Mathematical Physics}, 7:107--112, 1983.

\bibitem[BFT15]{Berta2015OnEntropies}
Mario Berta, Omar Fawzi, and Marco Tomamichel.
\newblock On variational expressions for quantum relative entropies.
\newblock {\em Letters in Mathematical Physics}, 107(12):2239--2265, December 2015.
\newblock arXiv:1512.02615 [quant-ph].

\bibitem[BGJ19]{bhatia2019matrix}
Rajendra Bhatia, Stephane Gaubert, and Tanvi Jain.
\newblock Matrix versions of the {H}ellinger distance.
\newblock {\em Letters in Mathematical Physics}, 109:1777--1804, 2019.
\newblock arXiv:1901.01378 [math-ph].

\bibitem[Bha46]{Bhattacharyya1946}
A.~Bhattacharyya.
\newblock On a measure of divergence between two multinomial populations.
\newblock {\em Sankhy$\bar{a}$: The Indian Journal of Statistics (1933-1960)}, 7(4):401--406, July 1946.

\bibitem[Bha07]{bhatia2009positive}
Rajendra Bhatia.
\newblock {\em Positive Definite Matrices}.
\newblock Princeton University Press, 2007.

\bibitem[BJ23]{baldwin2023efficiently}
Andrew~J Baldwin and Jonathan~A Jones.
\newblock Efficiently computing the uhlmann fidelity for density matrices.
\newblock {\em Physical Review A}, 107(1):012427, 2023.

\bibitem[Bor26]{born1926quantum}
Max Born.
\newblock Quantum mechanics of collision processes.
\newblock {\em Uspekhi Fizich}, June 1926.

\bibitem[BS82]{Belavkin1982}
V.~P. Belavkin and P.~Staszewski.
\newblock C*-algebraic generalization of relative entropy and entropy.
\newblock {\em Annales de l'I.H.P. Physique th\'eorique}, 37:51--58, 1982.

\bibitem[Bur69]{Bur69}
Donald Bures.
\newblock An extension of {Kakutani's} theorem on infinite product measures to the tensor product of semifinite {$\omega^*$}-algebras.
\newblock {\em Transactions of the American Mathematical Society}, 135(0):199--212, 1969.

\bibitem[CS20]{cree2020fidelity}
Sam Cree and Jamie Sikora.
\newblock A fidelity measure for quantum states based on the matrix geometric mean, 2020.
\newblock arXiv:2006.06918 [quant-ph].

\bibitem[DL14]{datta2014limit}
Nilanjana Datta and Felix Leditzky.
\newblock A limit of the quantum {R}{\'e}nyi divergence.
\newblock {\em Journal of Physics A}, 47(4):045304, 2014.
\newblock arXiv:1308.5961 [quant-ph].

\bibitem[Don86]{donald1986relative}
Matthew~J. Donald.
\newblock On the relative entropy.
\newblock {\em Communications in Mathematical Physics}, 105:13--34, 1986.

\bibitem[FC95]{fuchs1995mathematical}
Christopher~A. Fuchs and Carlton~M. Caves.
\newblock Mathematical techniques for quantum communication theory.
\newblock {\em Open Systems \& Information Dynamics}, 3(3):345--356, 1995.
\newblock arXiv:quant-ph/9604001.

\bibitem[FFHT23]{farooq2023asymptotic}
Muhammad~Usman Farooq, Tobias Fritz, Erkka Haapasalo, and Marco Tomamichel.
\newblock Asymptotic and catalytic matrix majorization, 2023.
\newblock arXiv:2301.07353 [math.ST].

\bibitem[Fuc96]{F96}
Christopher Fuchs.
\newblock {\em Distinguishability and Accessible Information in Quantum Theory}.
\newblock PhD thesis, University of New Mexico, December 1996.
\newblock arXiv:quant-ph/9601020.

\bibitem[GT20]{gour2020optimal}
Gilad Gour and Marco Tomamichel.
\newblock Optimal extensions of resource measures and their applications.
\newblock {\em Physical Review A}, 102(6):062401, December 2020.
\newblock arXiv:2006.12408 [quant-ph].

\bibitem[GW15]{gupta2015multiplicativity}
Manish~K. Gupta and Mark~M. Wilde.
\newblock Multiplicativity of completely bounded $p$-norms implies a strong converse for entanglement-assisted capacity.
\newblock {\em Communications in Mathematical Physics}, 334:867--887, 2015.
\newblock arXiv:1310.7028 [quant-ph].

\bibitem[Hel67]{Hel67b}
Carl~W. Helstrom.
\newblock Minimum mean-squared error of estimates in quantum statistics.
\newblock {\em Physics Letters A}, 25(2):101--102, July 1967.

\bibitem[Hia19]{hiai19}
Fumio Hiai.
\newblock {Quantum $f$-divergences in von {N}eumann algebras. II. Maximal $f$-divergences}.
\newblock {\em Journal of Mathematical Physics}, 60(1):012203, January 2019.
\newblock arXiv:1807.03118 [math-ph].

\bibitem[HM92]{BlockMatrix}
Roger~A. Horn and Roy Mathias.
\newblock Block-matrix generalizations of {S}chur's basic theorems on {H}adamard products.
\newblock {\em Linear Algebra and its Applications}, 172:337--346, 1992.

\bibitem[HM17]{HM17}
Fumio Hiai and Mil\'{a}n Mosonyi.
\newblock Different quantum $f$-divergences and the reversibility of quantum operations.
\newblock {\em Reviews in Mathematical Physics}, 29(07):1750023, 2017.
\newblock arXiv:1604.03089 [math-ph].

\bibitem[Hol72]{holevo1972}
Alexander~S. Holevo.
\newblock On quasiequivalence of locally normal states.
\newblock {\em Theoretical and Mathematical Physics}, 13(2):1071--1082, November 1972.

\bibitem[HP91]{hiai1991proper}
Fumio Hiai and D{\'e}nes Petz.
\newblock The proper formula for relative entropy and its asymptotics in quantum probability.
\newblock {\em Communications in Mathematical Physics}, 143:99--114, 1991.

\bibitem[Jen04]{jencova2004}
Anna Jencova.
\newblock Geodesic distances on density matrices.
\newblock {\em Journal of Mathematical Physics}, 45(5):1787--1794, May 2004.
\newblock arXiv:math-ph/0312044.

\bibitem[Joz94]{J94fid}
Richard Jozsa.
\newblock Fidelity for mixed quantum states.
\newblock {\em Journal of Modern Optics}, 41(12):2315--2323, 1994.

\bibitem[KL51]{kl1951}
S.~Kullback and R.~A. Leibler.
\newblock {On Information and Sufficiency}.
\newblock {\em The Annals of Mathematical Statistics}, 22(1):79--86, March 1951.

\bibitem[KP85]{unitaries_convex_combinations85}
Richard~V. Kadison and Gert~K. Pedersen.
\newblock Means and convex combinations of unitary operators.
\newblock {\em Mathematica Scandinavica}, 57(2):249--266, 1985.

\bibitem[KRS09]{konig2009operational}
Robert K\"{o}nig, Renato Renner, and Christian Schaffner.
\newblock The operational meaning of min-and max-entropy.
\newblock {\em IEEE Transactions on Information Theory}, 55(9):4337--4347, 2009.
\newblock arXiv:0807.1338 [quant-ph].

\bibitem[KW20]{khatri2020principles}
Sumeet Khatri and Mark~M. Wilde.
\newblock Principles of quantum communication theory: A modern approach, 2020.
\newblock arXiv:2011.04672v2 [quant-ph].

\bibitem[KW21]{katariya2021geometric}
Vishal Katariya and Mark~M. Wilde.
\newblock Geometric distinguishability measures limit quantum channel estimation and discrimination.
\newblock {\em Quantum Information Processing}, 20(2):78, 2021.
\newblock arXiv:2004.10708 [quant-ph].

\bibitem[LKDW18]{leditzky2018approaches}
Felix Leditzky, Eneet Kaur, Nilanjana Datta, and Mark~M. Wilde.
\newblock Approaches for approximate additivity of the {H}olevo information of quantum channels.
\newblock {\em Physical Review A}, 97(1):012332, January 2018.
\newblock arXiv:1709.01111 [quant-ph].

\bibitem[LPP15]{li2015some}
Jin Li, Rajesh Pereira, and Sarah Plosker.
\newblock Some geometric interpretations of quantum fidelity.
\newblock {\em Linear Algebra and its Applications}, 487:158--171, 2015.

\bibitem[LT15]{lin2015investigating}
Simon~M. Lin and Marco Tomamichel.
\newblock Investigating properties of a family of quantum {R}{\'e}nyi divergences.
\newblock {\em Quantum Information Processing}, 14(4):1501--1512, 2015.
\newblock arXiv:1408.6897 [quant-ph].

\bibitem[LZ04]{LZ04}
Shunlong Luo and Qiang Zhang.
\newblock Informational distance on quantum-state space.
\newblock {\em Physical Review A}, 69(3):032106, March 2004.

\bibitem[Mat67]{Matusita1967notion}
Kameo Matusita.
\newblock On the notion of affinity of several distributions and some of its applications.
\newblock {\em Annals of the Institute of Statistical Mathematics}, 19:181--192, December 1967.

\bibitem[Mat10a]{matsumoto2010reverse_relent}
Keiji Matsumoto.
\newblock Reverse test and characterization of quantum relative entropy, 2010.
\newblock arXiv:1010.1030 [quant-ph].

\bibitem[Mat10b]{matsumoto2010reverse}
Keiji Matsumoto.
\newblock Reverse test and quantum analogue of classical fidelity and generalized fidelity, 2010.
\newblock arXiv:1006.0302 [quant-ph].

\bibitem[Mat14]{matsumoto2014quantum}
Keiji Matsumoto.
\newblock Quantum fidelities, their duals, and convex analysis, 2014.
\newblock arXiv:1408.3462 [quant-ph].

\bibitem[Mat18]{Matsumoto2018}
Keiji Matsumoto.
\newblock A new quantum version of $f$-divergence.
\newblock In {\em Reality and Measurement in Algebraic Quantum Theory}, volume 261, pages 229--273. Springer Singapore, 2018.

\bibitem[MBV22]{mosonyi2022geometric}
Mil{\'a}n Mosonyi, Gergely Bunth, and P{\'e}ter Vrana.
\newblock Geometric relative entropies and barycentric {R}{\'e}nyi divergences, 2022.
\newblock arXiv:2207.14282 [quant-ph].

\bibitem[MH22]{mosonyi2022test}
Mil{\'a}n Mosonyi and Fumio Hiai.
\newblock Test-measured {R}{\'e}nyi divergences.
\newblock {\em IEEE Transactions on Information Theory}, 69(2):1074--1092, 2022.
\newblock arXiv:2201.05477 [quant-ph].

\bibitem[MLDS{\etalchar{+}}13]{muller2013quantum}
Martin M{\"u}ller-Lennert, Fr{\'e}d{\'e}ric Dupuis, Oleg Szehr, Serge Fehr, and Marco Tomamichel.
\newblock On quantum {R\'e}nyi entropies: A new generalization and some properties.
\newblock {\em Journal of Mathematical Physics}, 54(12):122203, 2013.
\newblock arXiv:1306.3142 [quant-ph].

\bibitem[MNW23]{mnw2023}
Hemant~K. Mishra, Michael Nussbaum, and Mark~M. Wilde.
\newblock On the optimal error exponents for classical and quantum antidistinguishability, 2023.
\newblock arXiv:2309.03723v2 [quant-ph], accepted for publication in \textit{Letters in Mathematical Physics}.

\bibitem[MO17]{mosonyi2017strong}
Mil{\'a}n Mosonyi and Tomohiro Ogawa.
\newblock Strong converse exponent for classical-quantum channel coding.
\newblock {\em Communications in Mathematical Physics}, 355:373--426, 2017.
\newblock arXiv:1409.3562 [quant-ph].

\bibitem[NR18]{nasser2018polar}
Rajai Nasser and Joseph~M. Renes.
\newblock Polar codes for arbitrary classical-quantum channels and arbitrary cq-{MAC}s.
\newblock {\em IEEE Transactions on Information Theory}, 64(11):7424--7442, 2018.
\newblock arXiv:1701.03397 [cs.IT].

\bibitem[Nö14]{Notzel_2014}
Janis Nötzel.
\newblock Hypothesis testing on invariant subspaces of the symmetric group: {Part I. Quantum Sanov}'s theorem and arbitrarily varying sources.
\newblock {\em Journal of Physics A: Mathematical and Theoretical}, 47(23):235303, May 2014.
\newblock arXiv:1310.5553 [quant-ph].

\bibitem[Pet85]{P85}
{D\'enes} Petz.
\newblock Quasi-entropies for states of a von {N}eumann algebra.
\newblock {\em Publications of the Research Institute for Mathematical Sciences}, 21(4):787--800, 1985.

\bibitem[Pet86]{P86}
D\'enes Petz.
\newblock Quasi-entropies for finite quantum systems.
\newblock {\em Reports on Mathematical Physics}, 23:57--65, 1986.

\bibitem[Pia09]{P09}
Marco Piani.
\newblock Relative entropy of entanglement and restricted measurements.
\newblock {\em Physical Review Letters}, 103(16):160504, October 2009.
\newblock arXiv:0904.2705 [quant-ph].

\bibitem[PV08]{Lautum}
Daniel~P. Palomar and Sergio Verdu.
\newblock Lautum information.
\newblock {\em IEEE Transactions on Information Theory}, 54(3):964--975, 2008.

\bibitem[RASW23]{RASW23}
Soorya Rethinasamy, Rochisha Agarwal, Kunal Sharma, and Mark~M. Wilde.
\newblock Estimating distinguishability measures on quantum computers.
\newblock {\em Physical Review A}, 108(1):012409, July 2023.
\newblock arXiv:2108.08406 [quant-ph].

\bibitem[R{\'e}n61]{renyi1961measures}
Alfr{\'e}d R{\'e}nyi.
\newblock On measures of entropy and information.
\newblock In {\em Proceedings of the Fourth Berkeley Symposium on Mathematical Statistics and Probability, Volume 1: Contributions to the Theory of Statistics}, volume~4, pages 547--562. University of California Press, January 1961.

\bibitem[{\c{S}}TA09]{averageBattacharya}
Eren {\c{S}}a\c{s}o\u{g}lu, Emre Telatar, and Erdal Arikan.
\newblock Polarization for arbitrary discrete memoryless channels.
\newblock In {\em IEEE Information Theory Workshop}, pages 144--148, 2009.
\newblock arXiv:0908.0302 [cs.IT].

\bibitem[Str65]{stratonovich1965information}
R.~L. Stratonovich.
\newblock Information capacity of a quantum communications channel. {I}.
\newblock {\em Soviet Radiophysics}, 8(1):82--91, January 1965.

\bibitem[SW12]{SW12}
Naresh Sharma and Naqueeb~A. Warsi.
\newblock On the strong converses for the quantum channel capacity theorems, May 2012.
\newblock arXiv:1205.1712 [quant-ph].

\bibitem[Tou74]{toussaint1974some}
Godfried~T. Toussaint.
\newblock Some properties of {M}atusita’s measure of affinity of several distributions.
\newblock {\em Annals of the Institute of Statistical Mathematics}, 26(3):389--394, 1974.

\bibitem[Uhl76]{uhlmann1976transition}
Armin Uhlmann.
\newblock The ``transition probability'' in the state space of a $\ast$-algebra.
\newblock {\em Reports on Mathematical Physics}, 9(2):273--279, 1976.

\bibitem[Ume62]{umegaki1962conditional}
Hisaharu Umegaki.
\newblock Conditional expectation in an operator algebra, iv (entropy and information).
\newblock In {\em Kodai Mathematical Seminar Reports}, volume~14, pages 59--85. Department of Mathematics, Tokyo Institute of Technology, 1962.

\bibitem[Wat09]{watrous2012simpler}
John Watrous.
\newblock Semidefinite programs for completely bounded norms.
\newblock {\em Theory of Computing}, 5(11):217--238, 2009.
\newblock arXiv:0901.4709 [quant-ph].

\bibitem[WWY14]{wilde2014strong}
Mark~M. Wilde, Andreas Winter, and Dong Yang.
\newblock Strong converse for the classical capacity of entanglement-breaking and {H}adamard channels via a sandwiched {R\'e}nyi relative entropy.
\newblock {\em Communications in Mathematical Physics}, 331:593--622, 2014.
\newblock arXiv:1306.1586 [quant-ph].

\bibitem[Zha20]{Z20}
Haonan Zhang.
\newblock From {Wigner-Yanase-Dyson} conjecture to {Carlen-Frank-Lieb} conjecture.
\newblock {\em Advances in Mathematics}, 365:107053, 2020.
\newblock arXiv:1811.01205 [math.FA].

\end{thebibliography}

\appendix

 \section{\hspace{17mm}
(Proofs)}
\label{Sec:Proofs}

In the following appendices, we present technical proofs of the results that are stated in the main text of the paper. 

\subsection{Proof of Equation~(\ref{eq:log-euc-fid-limit}) (Log-Euclidean fidelity as limit of $z$-fidelity)}

\label{sec:proof-log-euclid-fid}

For $\varepsilon>0$, define $\rho(\varepsilon)\coloneqq
\rho+\varepsilon I$ and $\sigma(\varepsilon)\coloneqq
\sigma+\varepsilon I$.
Observe that 
\begin{equation}
    F_z(\rho,\sigma) = \lim_{\varepsilon \to 0^+} F_z(\rho(\varepsilon),\sigma(\varepsilon)) = \inf_{\varepsilon >0} F_z(\rho(\varepsilon),\sigma(\varepsilon)),
\end{equation}
which follows from the operator monotonicity of the function $x^p$ for $p\in(0,1]$ and monotonicity of $(\cdot) \to \operatorname{Tr}[(\cdot)^z]$ for all $z\geq 1/2$. Now recalling that the $z$-fidelities are antimonotone for $z\geq 1/2$ \cite[Proposition~6]{lin2015investigating} and continuous in $z$, we conclude that
\begin{align}
    \lim_{z\to \infty} F_z(\rho,\sigma) & = 
    \inf_{z \geq \frac{1}{2}} F_z(\rho,\sigma) \\
    & = \inf_{z \geq \frac{1}{2}} \inf_{\varepsilon >0} F_z(\rho(\varepsilon),\sigma(\varepsilon)) \\
    & = \inf_{\varepsilon >0} \inf_{z \geq \frac{1}{2}}  F_z(\rho(\varepsilon),\sigma(\varepsilon)) \\
    & = \inf_{\varepsilon >0} \lim_{z \to \infty}  F_z(\rho(\varepsilon),\sigma(\varepsilon)) \\
    & = \inf_{\varepsilon >0} \operatorname{Tr}\!\left[\exp\!\left(\frac{1}{2} (\ln \rho(\varepsilon) + \ln \sigma(\varepsilon)\right)\right].
\end{align}
The last equality follows from an application of the Lie--Trotter product formula. Similar observations were made in \cite[Section~4]{audenaert2015alpha} using the fact that for operators $A$ and $B$, $\lim_{z \to \infty} \left(\exp(A/z) \exp(B/z)\right)^z = \exp(A+B)$ with the choice $A= \frac{1}{2} \ln \rho(\varepsilon)$ and $B=\frac{1}{2} \ln \sigma(\varepsilon)$.

The last equality in~\eqref{eq:log-euc-fid-limit} follows from operator monotonicity of the logarithm function and monotonicity of $(\cdot) \to \operatorname{Tr}[\exp(\cdot)]$.

\subsection{Proof of Proposition~\ref{prop:unif-cont-Matusita} (Uniform continuity of Matusita fidelity)}

\label{sec:proof-unif-cont-Matusita}

Consider that
\begin{align}
&  \left\vert F_{r}(p_{1},\ldots,p_{r})-F_{r}(q_{1},\ldots,q_{r})\right\vert
\nonumber\\
&  =\left\vert \sum_{x}\left(  p_{1}(x)\cdots p_{r}(x)\right)  ^{\frac{1}{r}
}-\sum_{x}\left(  q_{1}(x)\cdots q_{r}(x)\right)  ^{\frac{1}{r}}\right\vert \\
&  =\left\vert \sum_{x}\left[  \left(  p_{1}(x)\cdots p_{r}(x)\right)
^{\frac{1}{r}}-\left(  q_{1}(x)\cdots q_{r}(x)\right)  ^{\frac{1}{r}}\right]
\right\vert \\
&  =\left\vert \sum_{i=1}^{r}\sum_{x}\left[  \left(  p_{1}(x)\cdots
p_{i-1}(x)\right)  ^{\frac{1}{r}}\left(  p_{i}(x)^{\frac{1}{r}}-q_{i}
(x)^{\frac{1}{r}}\right)  \left(  q_{i+1}(x)\cdots q_{r}(x)\right)  ^{\frac
{1}{r}}\right]  \right\vert \\
&  \leq\sum_{i=1}^{r}\left\vert \sum_{x}\left[  \left(  p_{1}(x)\cdots
p_{i-1}(x)\right)  ^{\frac{1}{r}}\left(  p_{i}(x)^{\frac{1}{r}}-q_{i}
(x)^{\frac{1}{r}}\right)  \left(  q_{i+1}(x)\cdots q_{r}(x)\right)  ^{\frac
{1}{r}}\right]  \right\vert \\
&  \leq\sum_{i=1}^{r}\left\Vert \left(  p_{1}\cdots p_{i-1}q_{i+1}\cdots
q_{r}\right)  ^{\frac{1}{r}}\right\Vert _{\frac{r}{r-1}}\left\Vert p_{i}
^{\frac{1}{r}}-q_{i}^{\frac{1}{r}}\right\Vert _{r}\\
&  =\sum_{i=1}^{r}F_{r-1}(p_{1},\ldots,p_{i-1},q_{i+1},\ldots,q_{r}
)^{\frac{r-1}{r}}\left\Vert p_{i}^{\frac{1}{r}}-q_{i}^{\frac{1}{r}
}\right\Vert _{r}\\
&  \leq\sum_{i=1}^{r}\left\Vert p_{i}^{\frac{1}{r}}-q_{i}^{\frac{1}{r}
}\right\Vert _{r}\\
&  =r\left(  \frac{1}{r}\sum_{i=1}^{r}\left(  \left\Vert p_{i}^{\frac{1}{r}
}-q_{i}^{\frac{1}{r}}\right\Vert _{r}^{r}\right)  ^{\frac{1}{r}}\right)  \\
&  \leq r\left(  \frac{1}{r}\sum_{i=1}^{r}\left\Vert p_{i}^{\frac{1}{r}}
-q_{i}^{\frac{1}{r}}\right\Vert _{r}^{r}\right)  ^{\frac{1}{r}}\\
&  \leq r\left(  \frac{1}{r}\sum_{i=1}^{r}\left\Vert \sqrt{p_{i}}-\sqrt{q_{i}}\right\Vert
_{2}^2\right)  ^{\frac{1}{r}}\\
&  \leq r \varepsilon  ^{\frac{1}{r}}.
\end{align}
The third equality follows from rewriting as a telescoping sum. The first
inequality follows from the triangle inequality. The second inequality follows
from H\"older's inequality with parameters $r,s\geq1$ satisfying $\frac{1}
{r}+\frac{1}{s}=1$, so that $s=\frac{r}{r-1}$. The fourth equality follows from
applying definitions, and the third inequality follows because $F_{r}\leq1$ for
all probability distributions and for every integer $r\geq2$.
The fourth inequality follows from the concavity of function $x^{1/r}$ for $r\geq 2$ for all $x \geq 0$. To see the penultimate
inequality, consider 
that the following inequality holds for all distributions $p$ and $q$ by using \cite[Theorem~5]{Matusita1967notion}:
\begin{equation}
\left\Vert p^{\frac{1}{r}}-q^{\frac{1}{r}}\right\Vert _{r}^{r}=\sum
_{x}\left\vert p(x)^{\frac{1}{r}}-q(x)^{\frac{1}{r}}\right\vert ^{r}\leq
\sum_{x}\left( \sqrt{p(x)}-\sqrt{q(x)}\right)^2 =\left\Vert \sqrt{p}-\sqrt{q}\right\Vert _{2}^2.
\end{equation}
Finally, the last inequality follows from the assumption $\frac{1}{r}\left(  \sum_{i=1}^{r} d_H (p_i, q_i)^2 \right)  \leq\varepsilon$, along with~\eqref{eq:Helinger_classical}.

\subsection{Proof of Proposition~\ref{prop:ineqs-avg-k-wise-classical} (Inequalities relating classical average $k$-wise fidelities)}

\label{sec:proof-avg-k-wise-order-classical}

It suffices to prove the following inequality:
\begin{equation}
F_{k-1,r}(\rho_{1},\ldots,\rho_{r})\geq F_{k,r}(\rho_{1},\ldots,\rho_{r}),
\end{equation}
for all $k\in\left\{  3,\ldots,r\right\}  $. So we prove this one, using a generalization of the approach used in the proof of \Cref{prop:r_root_to_commuting_fidelity}.
Let $S_{k-1}(i_{1},\ldots,i_{k})$ denote all size-($k-1$) subsets of
$i_{1},\ldots,i_{k}$, of which there are $k$ of them. For example, if $r=7$,
$k=4$, and $i_{1}=2$, $i_{2}=4$, $i_{3}=5$, $i_{4}=7$, then
\begin{equation}
S_{k-1}(i_{1},i_{2},i_{3},i_{4})=S_{3}(2,4,5,7)=\left\{  \left\{
2,4,5\right\}  ,\left\{  2,4,7\right\}  ,\left\{  2,5,7\right\}
,\{4,5,7\}\right\}  .
\end{equation}
Note that each symbol $i_{j}$ appears in exactly $k-1$ elements of
$S_{k-1}(i_{1},\ldots,i_{k})$. Then consider that
\begin{align}
F_{k}(\rho_{i_{1}},\ldots,\rho_{i_{k}}) &  =\sum_{x}\left(  \rho_{i_{1}}(x)\cdots
\rho_{i_{k}}(x)\right)  ^{\frac{1}{k}}\\
&  =\sum_{x}\left(  \prod\limits_{t\in S_{k-1}(i_{1},\ldots,i_{k})}\left[
\rho_{t(1)}(x)\cdots \rho_{t(k-1)}(x)\right]  ^{\frac{1}{k-1}}\right)  ^{\frac{1}
{k}}\label{eq:reorder-product-k-1} \\
&  \leq\sum_{x}\frac{1}{k}\sum_{t\in S_{k-1}(i_{1},\ldots,i_{k})}\left[
\rho_{t(1)}(x)\cdots \rho_{t(k-1)}(x)\right]  ^{\frac{1}{k-1}}\\
&  =\frac{1}{k}\sum_{t\in S_{k-1}(i_{1},\ldots,i_{k})}\sum_{x}\left[
\rho_{t(1)}(x)\cdots \rho_{t(k-1)}(x)\right]  ^{\frac{1}{k-1}}\\
&  =\frac{1}{k}\sum_{t\in S_{k-1}(i_{1},\ldots,i_{k})}F_{k-1}(\rho_{t(1)}
,\ldots,\rho_{t(k-1)}),
\end{align}
where the inequality follows as a consequence of the inequality of arithmetic
and geometric means. Now consider that
\begin{align}
F_{k,r}(\rho_1, \ldots, \rho_r) & =  \frac{1}{\binom{r}{k}}\sum_{i_{1}<\cdots<i_{k}}F_{k}(\rho_{i_{1}}
,\ldots,\rho_{i_{k}})\nonumber\\
&  \leq\frac{1}{\binom{r}{k}}\sum_{i_{1}<\cdots<i_{k}}\frac{1}{k}\sum_{t\in
S_{k-1}(i_{1},\ldots,i_{k})}F_{k-1}(\rho_{t(1)},\ldots,\rho_{t(k-1)}) \label{eq:avg-k-wise-final-steps}\\
&  =\frac{k-1!r-k!}{r!}\sum_{i_{1}<\cdots<i_{k}}\sum_{t\in S_{k-1}
(i_{1},\ldots,i_{k})}F_{k-1}(\rho_{t(1)},\ldots,\rho_{t(k-1)})\\
&  =\frac{k-1!r-k!}{r!}\left(  r-k+1\right)  \sum_{i_{1}<\cdots<i_{k-1}
}F_{k-1}(\rho_{i_{1}},\ldots,\rho_{i_{k-1}})\\
&  =\frac{1}{\binom{r}{k-1}}\sum_{i_{1}<\cdots<i_{k-1}}F_{k-1}(\rho_{i_{1}
},\ldots,\rho_{i_{k-1}})\\
& = F_{k-1,r}(\rho_1, \ldots, \rho_r).
\label{eq:avg-k-wise-final-steps-final}
\end{align}
The second equality follows because the sum in the third line includes each
distinct term $r-k+1$ times. Indeed, the sum in the third line arises by
picking size-$k$ subsets from $\left[  r\right]  $, and from each of these
size-$k$ subsets, picking size-$\left(  k-1\right)  $ subsets, thus leading to
a total number of subsets given by $\binom{r}{k}\binom{k}{k-1}$. However,
picking subsets in this way leads to $\binom{r}{k-1}$ distinct subsets, each
occurring with multiplicity $\frac{\binom{r}{k}\binom{k}{k-1}}{\binom{r}{k-1}
}=r-k+1$.

\subsection{Proof of Proposition~\ref{Prop:uniform_cont_average_pairwise} (Uniform continuity of average pairwise fidelities)}

\label{Proof:uniform_cont_average_pairwise}

Define the shorthand $\mathcal{S}_{\rho} \coloneqq (\rho_1,\ldots, \rho_r)$ and $\mathcal{S}_{\sigma} \coloneqq (\sigma_1,\ldots, \sigma_r)$.
Consider that
\begin{align}
F_{U}(\mathcal{S}_{\rho})-F_{U}(\mathcal{S}_{\sigma}) &  =\frac{2}{r\left(
r-1\right)  }\sum_{i<j}F(\rho_{i},\rho_{j})-\frac{2}{r\left(  r-1\right)
}\sum_{i<j}F(\sigma_{i},\sigma_{j})\\
&  =\frac{2}{r\left(  r-1\right)  }\sum_{i<j}\left[  F(\rho_{i},\rho
_{j})-F(\sigma_{i},\sigma_{j})\right]. \label{eq:Uhlmann-average-fidelity-difference-two-sets-quantum-states}
\end{align}
Then
\begin{align}
&  2\left[  F(\rho_{i},\rho_{j})-F(\sigma_{i},\sigma_{j})\right]
\nonumber\\
&  =2\left[  1-F(\sigma_{i},\sigma_{j})-\left(  1-F(\rho_{i},\rho
_{j})\right)  \right]  \\
&  =2\left[  \sqrt{1-F(\sigma_{i},\sigma_{j})}\sqrt{1-F(\sigma
_{i},\sigma_{j})}-\sqrt{1-F(\rho_{i},\rho_{j})}\sqrt{1-F(\rho_{i}
,\rho_{j})}\right]  \\
&  =d_{B}(\sigma_{i},\sigma_{j})d_{B}(\sigma_{i},\sigma_{j})-d_{B}(\rho
_{i},\rho_{j})d_{B}(\rho_{i},\rho_{j})\\
&  =d_{B}(\sigma_{i},\sigma_{j})d_{B}(\sigma_{i},\sigma_{j})-d_{B}(\sigma
_{i},\sigma_{j})d_{B}(\rho_{i},\rho_{j})\nonumber\\
&  \qquad+d_{B}(\sigma_{i},\sigma_{j})d_{B}(\rho_{i},\rho_{j})-d_{B}(\rho
_{i},\rho_{j})d_{B}(\rho_{i},\rho_{j})\\
&  =\left[  d_{B}(\sigma_{i},\sigma_{j})+d_{B}(\rho_{i},\rho_{j})\right]
\left[  d_{B}(\sigma_{i},\sigma_{j})-d_{B}(\rho_{i},\rho_{j})\right]  \\
&  \leq\left[  d_{B}(\sigma_{i},\sigma_{j})+d_{B}(\rho_{i},\rho_{j})\right]\times\notag\\
& \qquad \left[  d_{B}(\sigma_{i},\rho_{i})+d_{B}(\rho_{i},\rho_{j})+d_{B}(\rho
_{j},\sigma_{j})-d_{B}(\rho_{i},\rho_{j})\right]  \\
&  =\left[  d_{B}(\sigma_{i},\sigma_{j})+d_{B}(\rho_{i},\rho_{j})\right]
\left[  d_{B}(\sigma_{i},\rho_{i})+d_{B}(\rho_{j},\sigma_{j})\right]  \\
&  \leq 2\sqrt{2} \left[  d_{B}(\sigma_{i},\rho_{i})+d_{B}(\rho_{j},\sigma_{j})\right] \\
& = 2\sqrt{2} \left[  d_{B}(\rho_{i},\sigma_{i})+d_{B}(\rho_{j},\sigma_{j})\right],
\end{align}
where the first inequality follows from the triangular inequality of the Bures distance and the second inequality because $d_B(\omega,\tau) \leq \sqrt{2}$ for all quantum states $\omega$ and $\tau$.
This implies that
\begin{equation}
F(\rho_{i},\rho_{j})-F(\sigma_{i},\sigma_{j})\leq \sqrt{2} \left[  d_{B}(\rho_{i},\sigma_{i})+d_{B}(\rho_{j},\sigma_{j})\right]. \label{eq:fidelity-difference-two-sets-quantum-states}
\end{equation}
Substituting~\eqref{eq:fidelity-difference-two-sets-quantum-states} into~\eqref{eq:Uhlmann-average-fidelity-difference-two-sets-quantum-states} and using the fact that
\begin{equation}
    \sum_{i<j}  \left[  d_{B}(\rho_{i},\sigma_{i})+d_{B}(\rho_{j},\sigma_{j})\right] = (r-1) \sum_{i=1}^r d_{B}(\rho_{i},\sigma_{i})
\end{equation}
gives
\begin{align}
  F_{U}(\mathcal{S}_{\rho})-F_{U}(\mathcal{S}_{\sigma})  
 &  \leq \sqrt{2} \left(\frac{2}{r\left(  r-1\right)  }\right) \sum_{i<j}  \left[  d_{B}(\rho_{i},\sigma_{i})+d_{B}(\rho_{j},\sigma_{j})\right] \\
 & = 2\sqrt{2} \left(\frac{1}{r} \sum_{i=1}^r d_{B}(\rho_{i},\sigma_{i})\right) \\
 & \leq 2\sqrt{2} \varepsilon,
 \end{align}
 where the last inequality follows from the assumption that $\frac{1}{r} \sum_{i=1}^r d_{B}(\rho_{i},\sigma_{i})\leq \varepsilon$.
 This concludes the proof of the inequality in~\eqref{eq:Uhlmann-average-fidelity-difference-mod-two-sets-quantum-states-2}.

A similar proof works for the inequality in 
\eqref{eq:holevo-pairwise-unif-cont-2} by replacing all instances of Uhlmann fidelity and Bures distance  with Holevo fidelity and quantum Hellinger distance, respectively.

\subsection{Proof of Proposition~\ref{prop:dual_SDP_multi} (Dual of multivariate SDP fidelity)} 

\label{Sec:Proof_dual_SDP}
Recall that a standard primal SDP and its dual are given as follows (cf.~\cite[Definition~2.20]{khatri2020principles}): 
\begin{align}
\texttt{PRIMAL} \coloneqq &\inf_{Y\geq0}\left\{  \operatorname{Tr}[BY]:\Phi^{\dag}(Y)\geq A\right\}, \label{eq:sdp_primal} \\
\texttt{DUAL} \coloneqq &\sup_{X\geq0}\left\{  \operatorname{Tr}[AX]:\Phi(X)\leq B\right\},\label{eq:sdp_dual}
\end{align}
where $A$ and $B$ are Hermitian matrices and $\Phi$ is a Hermiticity-preserving superoperator. It is always true that $\texttt{PRIMAL} \geq \texttt{DUAL}$, and  strong duality is said to hold in the case of equality. A sufficient condition for strong duality to hold is that there exists a feasible point for the \texttt{DUAL} and a \textit{strictly feasible} point for the \texttt{PRIMAL}; i.e., the latter meaning that there exists a $Y$ such that $Y>0$ and $\Phi^{\dagger}(Y) > A$. This is known as Slater's condition.
In what follows, we identify the $A, B$, and $\Phi^{\dagger}$ from the multivariate SDP formulation~\eqref{eq:multivar-fid} to derive its dual, and we then argue that strong duality holds using Slater's theorem.

Comparing the primal SDP~\eqref{eq:sdp_primal} to the multivariate SDP fidelity~\eqref{eq:multivar-fid} without considering the normalization term $r(r-1)$, 
we make the following choices of $A, B$, and $\Phi^{\dagger}$: given $Y=\sum_{i,j=1}^r |i\rangle\!\langle j| \otimes Y_{ij}$, define
\begin{align}
    A & \coloneqq \sum_{i,j\in [r]: i \neq j} |i\rangle\!\langle j| \otimes I , \\
    B & \coloneqq \sum_{i=1}^r |i\rangle\!\langle i| \otimes \rho_i, \\
    \Phi^\dagger(Y) & \coloneqq \sum_{i=1}^r |i\rangle\!\langle i| \otimes Y_{ii}.
\end{align}
Thus, given any Hermitian matrix $X=\sum_{i,j=1}^r |i\rangle\!\langle j| \otimes X_{ij}$ and using the relation $\Tr\!\left[ Y \Phi(X)  \right] = \Tr\!\left[ \Phi^\dagger(Y) X   \right]$, we obtain
\begin{equation}
    \Phi(X)= \sum_{i=1}^r |i\rangle\!\langle i| \otimes X_{ii}.
\end{equation}
Also, we have 
\begin{equation} \label{eq:objective_function}
    \Tr[AX]= \sum_{i \neq j} \Tr[X_{ij}]= 2 \sum_{i <j}\mathfrak{R}\left[ \Tr[X_{ij}]\right].
\end{equation}
The \texttt{DUAL} condition $\Phi(X) \leq B$ is given by
\begin{align}
  & \sum_{i=1}^r |i\rangle\!\langle i| \otimes X_{ii}  \leq \sum_{i=1}^r |i\rangle\!\langle i| \otimes \rho_i,
\end{align}
which, by adding $\sum_{i \neq j} |i\rangle\!\langle i| \otimes X_{ij}$ to both sides,  evaluates to
\begin{align}
    0 \leq  X\leq \sum_{i=1}^r |i\rangle\!\langle i| \otimes \rho_i + \sum_{i \neq j} |i\rangle\!\langle i| \otimes X_{ij}. \label{eq:constraint}
\end{align}
Using the relations~\eqref{eq:objective_function} and~\eqref{eq:constraint} in the \texttt{DUAL} form and the fact that $X$ is Hermitian gives the desired dual formulation~\eqref{eq:multivariate_fid_sdp_dual}.

We use Slater's condition to show that strong duality holds, which requires finding a strictly feasible point in the SDP formulation~\eqref{eq:multivar-fid} and a feasible point in the dual formulation~\eqref{eq:multivariate_fid_sdp_dual}. The positive definite matrix $Y=\sum_{i=1}^r |i\rangle\!\langle i| \otimes r I$ is a strictly feasible point because $\Phi^{\dagger}(Y)-A>0$. To see this, consider that
\begin{align}
   \Phi^{\dagger}(Y)-A &= \sum_{i=1}^r |i\rangle\!\langle i| \otimes rI - \sum_{i \neq j} |i\rangle\!\langle j| \otimes I 
   \label{eq:multi-SDP-slater-choice} \\
   &= \sum_{i=1}^r (r+1) |i\rangle\!\langle i| \otimes I - \sum_{i, j=1}^r |i\rangle\!\langle j| \otimes I \\
   &=\big((r+1) I - uu^T \big) \otimes I,
\end{align}
where $u$ is the $r$-column vector of all ones. The matrix $(r+1) I - uu^T$ is positive definite because its (distinct) eigenvalues are $1$ and $r+1$. We thus conclude that $\Phi^{\dagger}(Y)-A>0$. A feasible point for~\eqref{eq:multivariate_fid_sdp_dual} consists of simply picking $X_{ij}= 0 $ for all~$i,j$. 

\subsection{Proof of Theorem~\ref{thm:uniform_cont_SDP_K} (Uniform continuity of $K^\star$-representation)}

\label{Proof:uniform_cont_SDP_F}

Recall by definition that 
\begin{equation}
F_{K^\star}(\rho_1, \ldots, \rho_r)=\frac{1}{r-1}\sup_{K\geq0}\left\{
\begin{array}
[c]{c}
\langle\psi^{\rho}|K\otimes I_d|\psi^{\rho}\rangle-1:\\
K=I_{r}\otimes I_{d}+\sum_{i\neq j}|i\rangle\!\langle j|\otimes K_{ij}\geq0,
\end{array}
\right\}  ,\label{eq:K-formulation}
\end{equation}
where $|\psi^{\rho}\rangle=\frac{1}{\sqrt{r}}\sum_{i=1}^{r}|i\rangle|\phi^{\rho_{i}
}\rangle$ and $|\phi^{\rho_{i}}\rangle$ is an arbitrary purification of $\rho_{i}$.
Fix
\begin{equation}
K=I_{r}\otimes I_{d}+\sum_{i\neq j}|i\rangle\!\langle j|\otimes K_{ij}
\geq0.    
\end{equation}
Also fix
\begin{align}
|\psi^{\rho}\rangle & =\frac{1}{\sqrt{r}}\sum_{i=1}^{r}|i\rangle|\phi
^{\rho_{i}}\rangle,\\
|\psi^{\sigma}\rangle & =\frac{1}{\sqrt{r}}\sum_{i=1}^{r}|i\rangle
|\phi^{\sigma_{i}}\rangle,
\end{align}
where these purifications are those that achieve the (root)\ fidelity, so that
$\langle\phi^{\sigma_{i}}|\phi^{\rho_{i}}\rangle=F(\rho_{i},\sigma_{i})$.
H\"older's inequality implies that
\begin{align}
& \left\vert \langle\psi^{\rho}|K\otimes I_d|\psi^{\rho}\rangle-\langle
\psi^{\sigma}|K\otimes I_d|\psi^{\sigma}\rangle\right\vert \nonumber\\
& =\left\vert \operatorname{Tr}[\left(  |\psi^{\rho}\rangle\!\langle\psi^{\rho
}|-|\psi^{\sigma}\rangle\!\langle\psi^{\sigma}|\right)  K\otimes I_d]\right\vert
\\
& \leq\left\Vert |\psi^{\rho}\rangle\!\langle\psi^{\rho}|-|\psi^{\sigma}
\rangle\!\langle\psi^{\sigma}|\right\Vert _{1}\left\Vert K\otimes I_d\right\Vert
_{\infty}\\
& =\left\Vert |\psi^{\rho}\rangle\!\langle\psi^{\rho}|-|\psi^{\sigma}
\rangle\!\langle\psi^{\sigma}|\right\Vert _{1}\left\Vert K\right\Vert _{\infty
}\\
& \leq\left\Vert |\psi^{\rho}\rangle\!\langle\psi^{\rho}|-|\psi^{\sigma}
\rangle\!\langle\psi^{\sigma}|\right\Vert _{1}\left\Vert r\left(  I_{r}\otimes
I_{d}\right)  \right\Vert _{\infty}\\
& =r\left\Vert |\psi^{\rho}\rangle\!\langle\psi^{\rho}|-|\psi^{\sigma}
\rangle\!\langle\psi^{\sigma}|\right\Vert _{1}\\
& =r\sqrt{1-\left\vert \langle\psi^{\sigma}|\psi^{\rho}\rangle\right\vert
^{2}}\\
& =r\sqrt{1-\left\vert \frac{1}{r}\sum_{i=1}^{r}\langle\phi^{\sigma_{i}}
|\phi^{\rho_{i}}\rangle\right\vert ^{2}}\\
& =r\sqrt{1-\left(  \frac{1}{r}\sum_{i=1}^{r}F(\rho_{i},\sigma_{i})\right)
^{2}}\\
& \leq r\sqrt{1-\left(  1-\varepsilon\right)  ^{2}}\\
& =r\sqrt{\varepsilon\left(  2-\varepsilon\right)  }.\label{eq:k-representation-difference-upper-bound}
\end{align}
The second to last inequality follows because 
\begin{align}
  I_{r}\otimes I_{d}&= \sum
_{i=1}^r\left(  |i\rangle\!\langle i|\otimes I_d\right)  K\left(  |i\rangle\!\langle
i|\otimes I_d\right)  \\
& =\frac{1}{r}\sum_{k=0}^{r-1}\left(  Z(k)\otimes I_d\right)  K\left(
Z(k)\otimes I_d\right)  ^{\dag}\\
& \geq\frac{1}{r}K,
\end{align}
where $Z(k)\coloneqq \sum_{\ell=0}^{r-1}e^{2\pi i\ell k/r}|\ell\rangle\!\langle\ell|$ is
the Heisenberg phase shift operator. The last inequality follows from
antimonotonicity of the function $\sqrt{1-x^{2}}$ on the interval $x\in\left[
0,1\right]  $ and the fact that $\frac{1}{r}\sum_{i=1}^{r}F(\rho_{i}
,\sigma_{i})\geq1-\varepsilon$ by assumption (see~\eqref{eq:cont-SDP-assumption}).
From~\eqref{eq:k-representation-difference-upper-bound} we thus get
\begin{align}
& \frac{1}{r-1}\left(  \langle\psi^{\rho}|K\otimes I_d|\psi^{\rho}\rangle
-1\right)    \notag \\
& \leq\frac{1}{r-1}\left(  \langle\psi^{\sigma}|K\otimes
I_d|\psi^{\sigma}\rangle-1\right)  +\frac{r}{r-1}\sqrt{\varepsilon\left(
2-\varepsilon\right)  }\\
& \leq F_{K^\star}(\sigma_1, \ldots, \sigma_r)+\frac{r}{r-1}\sqrt{\varepsilon\left(
2-\varepsilon\right)  },
\end{align}
where the last inequality follows from the expression in
\eqref{eq:K-formulation}. Since the inequality holds for all $K$ satisfying
the condition, we conclude that
\begin{equation}
F_{K^\star}(\rho_1, \ldots, \rho_r)\leq F_{K^\star}(\sigma_1, \ldots, \sigma_r)+\frac{r}{r-1}\sqrt
{\varepsilon\left(  2-\varepsilon\right)  },
\end{equation}
after applying~\eqref{eq:K-formulation} once again. 
The other inequality
follows from the same reasoning, but instead swapping $\rho_1, \ldots, \rho_r$
and $\sigma_1, \ldots, \sigma_r$.

\subsection{Monotonicity of multivariate SDP fidelity}

For every tuple of states $\rho_1, \ldots , \rho_r $ (not necessarily invertible) and for all $\varepsilon > 0$, all states of the following form are invertible:
\begin{align}\label{eq:epsilon_states}
    \rho_i^{(\varepsilon)} \coloneqq \dfrac{1}{1+\varepsilon} \left(\rho_i +  \frac{\varepsilon}{d} I \right).
\end{align}
As stated below, the multivariate SDP fidelity $F_{\operatorname{SDP}}(\rho_1^{(\varepsilon)}, \ldots, \rho_r^{(\varepsilon)})$ is monotonic as a function of $\varepsilon$, up to a multiplicative factor.
\begin{lemma}[Monotonicity]\label{prop:monotonicity}
    Let $\rho_1,\ldots, \rho_r$ be  quantum states, and fix $\varepsilon_1$ and $\varepsilon_2$ such that $0\leq \varepsilon_1 \leq  \varepsilon_2$. Then 
    \begin{equation}
(1+\varepsilon_1)F_{\operatorname{SDP}}\!\left(\rho_1^{(\varepsilon_1)}, \ldots, \rho_r^{(\varepsilon_1)}\right)  \leq (1+\varepsilon_2)F_{\operatorname{SDP}}\!\left(\rho_1^{(\varepsilon_2)}, \ldots, \rho_r^{(\varepsilon_2)}\right).\label{eq:monotonicity-sdp-fidelity}
    \end{equation}
\end{lemma}

\begin{proof}
     Let $(Y_i)_{i=1}^r$ be a candidate tuple for the SDP for $F_{\operatorname{SDP}}\!\left(\rho_1^{(\varepsilon_2)}, \ldots, \rho_r^{(\varepsilon_2)}\right)$.
Then
\begin{align}
    \frac{1}{r(r-1)} \sum_{i=1}^n \Tr\! \left[ Y_i \rho_i^{(\varepsilon_2)} \right] 
    &= \frac{1}{r(r-1)(1+\varepsilon_2)} \sum_{i=1}^n \Tr\! \left[ Y_i (\rho_i + \varepsilon_2 I/d) \right]  \\ 
    & \geq \frac{1}{r(r-1)(1+\varepsilon_2) } \sum_{i=1}^n \Tr\! \left[ Y_i (\rho_i + \varepsilon_1 I/d) \right] \\ 
    & =\frac{(1+ \varepsilon_1)}{r(r-1)(1+ \varepsilon_2)}  \sum_{i=1}^n \Tr\! \left[ Y_i \rho_i^{(\varepsilon_1)} \right] \\ 
    &\geq \frac{(1+ \varepsilon_1)}{(1+ \varepsilon_2)}  F\!\left(\rho_1^{(\varepsilon_1)}, \ldots, \rho_r^{(\varepsilon_1)}\right).
\end{align}
The first inequality holds because $\varepsilon_1 \leq  \varepsilon_2$, and the last inequality follows from  \cref{def:multivar-fid}. 
By optimizing the left-hand side of the above over all $(Y_i)_{i=1}^r$ satisfying the constraints in the SDP for $F_{\operatorname{SDP}}\!\left(\rho_1^{(\varepsilon_2)}, \ldots, \rho_r^{(\varepsilon_2)}\right)$, and then rearranging the terms, we arrive at the desired inequality~\eqref{eq:monotonicity-sdp-fidelity}.
\end{proof}

\subsection{Proof of Proposition~\ref{prop:another_formulation_multi_SDP_fid} (Alternative formulation of multivariate SDP fidelity)} 

\label{Sec:proof_another_formulation_multi}

The  proof of \Cref{prop:another_formulation_multi_SDP_fid} is structured in four steps as follows: 
In the first step we show an equivalent formulation for the $K^\star$-representation in~\eqref{eq:K_star_F}.
This helps us in the second step to prove that $F_{K^\star}(\rho_1, \ldots, \rho_r) \leq F_{\operatorname{SDP}}(\rho_1, \ldots, \rho_r)$ for general states. 
In the third step, we show that $F_{K^\star}(\rho_1, \ldots, \rho_r) = F_{\operatorname{SDP}}(\rho_1, \ldots, \rho_r)$ for invertible states. 
Finally, in the fourth step, with the assistance of \cref{prop:monotonicity} and \cref{thm:uniform_cont_SDP_K} among other things, we establish the equivalence~\eqref{eq:sdp_fidelity_K_star_equivalence} for the general setting. 

\medskip

\noindent\underline{\textbf{Step~1}: Equivalent formulation for $K^\star$-representation}
\medskip 

Define 
\begin{multline}\label{eq:multivariateFidelity_another_eq_int}
        F_{\operatorname{Int}}(\rho_1, \ldots, \rho_r)  
        \coloneqq \\ \frac{2}{r(r-1)} \sup_{K_{ji}=K_{ij}^\dagger } \left\{  \sum_{i<j} \mathfrak{R}\!\left[ \Tr\!\left[ \rho_i^{1/2}K_{ij} \rho_j^{1/2}\right]\right]: I_r \otimes I_d+ \sum_{i \neq j} |i\rangle\!\langle j| \otimes K_{ij} \geq 0 \right\} .
    \end{multline} 
We first show that $F_{\operatorname{Int}}(\rho_1, \ldots, \rho_r) = F_{K^\star}(\rho_1, \ldots, \rho_r)$. 
For all $i\in [r]$, define the canonical purification of $\rho_i$ as 
\begin{equation}
    | \phi^{\rho_i} \rangle =(\rho_i^{1/2} \otimes I_d) | \Gamma \rangle,
\end{equation}
where $|\Gamma\rangle \coloneqq  \sum_{i=1}^d |i\rangle\otimes |i\rangle$ is the unnormalized maximally entangled vector and $d$ denotes the dimension of the Hilbert space on which each $\rho_i$ is defined.
With this, we can write each term in the objective function as 
\begin{equation}
    \mathfrak{R}\!\left[ \Tr\!\left[ \rho_i^{1/2}K_{ij} \rho_j^{1/2}\right]\right]= \mathfrak{R}\!\left[ \langle\phi^{\rho_i}| K_{ij} \otimes I_d | \phi^{\rho_j} \rangle \right].
\end{equation}
Recalling the definition of $K$ from~\eqref{eq:multivariateFidelity_another_eq_int}
\begin{equation} \label{eq:K_matrix_recall}
    K \coloneqq I_r \otimes I_d+ \sum_{i \neq j} |i\rangle\!\langle j| \otimes K_{ij} ,
\end{equation}
and considering that $K\geq 0$,
we also obtain that
\begin{equation}
    \frac{1}{\sqrt{r}} \left( \sum_{k=1}^r \langle k| \langle \phi^{\rho_k}| \right)\left(   I_r \otimes I_d \otimes I_d+ \sum_{i \neq j=1}^r  |i \rangle\! \langle j| \otimes K_{ij} \otimes I_d\right) \frac{1}{\sqrt{r}} \left( \sum_{\ell=1}^r |\ell \rangle  |\phi^{\rho_\ell} \rangle \right) \geq 0
\end{equation}
and 
\begin{align}
  &\frac{1}{\sqrt{r}} \left( \sum_{k=1}^r \langle k| \langle \phi^{\rho_k}| \right)\left(   I_r \otimes I_d \otimes I_d+ \sum_{i \neq j=1}^r |i \rangle\! \langle j| \otimes K_{ij} \otimes I_d\right) \frac{1}{\sqrt{r}} \left( \sum_{\ell=1}^r |\ell \rangle  |\phi^{\rho_\ell} \rangle \right)  \notag \\
  &= 1 + \frac{1}{r} \left(\sum_{i \neq j=1}^r \langle\phi^{\rho_i}| K_{ij} \otimes I_d | \phi^{\rho_j} \rangle  \right) \\
  &= 1 + \frac{2}{r} \left(\sum_{i < j} \mathfrak{R}\left[\langle\phi^{\rho_i}| K_{ij} \otimes I_d | \phi^{\rho_j} \rangle\right]  \right).
\end{align}
The above evaluates to the following equality: 
\begin{equation}\label{eq:pairwise_relation_fidelity}
\sum_{i < j} \mathfrak{R}\left[\langle\phi^{\rho_i}| K_{ij} \otimes I_d | \phi^{\rho_j} \rangle\right] = \frac{r}{2} \left( \langle \psi | K \otimes I_d | \psi \rangle -1\right),
\end{equation}
where 
\begin{equation}
     | \psi\rangle = \frac{1}{\sqrt{r}} \sum_{i=1}^r |i \rangle | \phi^{\rho_i} \rangle.
\end{equation}
Then, for the canonical purification of each $\rho_i$, we deduce that the objective function of the optimizations in both $F_{\operatorname{Int}}(\rho_1, \ldots, \rho_r)$ and $F_{K^\star}(\rho_1, \ldots, \rho_r)$ are equivalent. 

For any other purification, recall that all purifications are related by an isometry. 
Let 
\begin{equation}
    |\Phi^{\rho_i}\rangle = ( U_i \otimes I_d )|\phi^{\rho_i}\rangle,
\end{equation}
where $U_i$ is an isometry for all $i \in [r]$.
Then, the objective function inside the SDP (i.e., without considering the constants) evaluates to 
\begin{equation}
\sum_{i < j} \mathfrak{R}\left[\langle\Phi^{\rho_i}| K_{ij} \otimes I_d | \Phi^{\rho_j} \rangle\right] = \sum_{i < j} \mathfrak{R}\left[\langle\phi^{\rho_i}| U_i^\dagger K_{ij} U_j\otimes I_d | \phi^{\rho_j} \rangle\right].
\end{equation}
Also note that the matrix formed by replacing $K_{ij}$ in~\eqref{eq:K_matrix_recall} with $U_i^\dagger K_{ij} U_j$ also satisfies the constraint due to: 
\begin{equation}
    \sum_{i=1}^r |i \rangle\!\langle i| \otimes U_i^\dagger  \left(I_r \otimes I_d+ \sum_{i \neq j} |i\rangle\!\langle j| \otimes K_{ij}  \right) \sum_{k=1}^r |k \rangle\!\langle k| \otimes U_k \geq 0,
\end{equation}
which follows from the fact that $U_i^\dagger U_i =I$.
Then, the previous arguments in the canonical purification follow for all purifications of the considered states. We then conclude that 
\begin{align}
 & F_{K^\star}(\rho_1, \ldots, \rho_r)=  F_{\operatorname{Int}}(\rho_1, \ldots, \rho_r)  \\
 &=
 \frac{2}{r(r-1)} \sup_{K_{ji}=K_{ij}^\dagger } \left\{  \sum_{i<j} \mathfrak{R}\!\left[ \Tr\!\left[ \rho_i^{1/2}K_{ij} \rho_j^{1/2}\right]\right]: I_r \otimes I_d+ \sum_{i \neq j} |i\rangle\!\langle j| \otimes K_{ij} \geq 0 \right\} \label{eq:multivariateFidelity_another_eq}.
\end{align}

\noindent\underline{\textbf{Step~2}: $F_{K^\star}(\rho_1,\ldots, \rho_r)\leq F_{\operatorname{SDP}}(\rho_1, \ldots, \rho_r)$ for general states} 
\medskip

Consider the constraint: $I_r \otimes I_d+ \sum_{i \neq j} |i\rangle\!\langle j| \otimes K_{ij} \geq 0 $ in~\eqref{eq:multivariateFidelity_another_eq}. 
By left and right multiplying the above inequality by 
$\sum_{i=1}^r |i\rangle\!\langle i| \otimes \rho_i^{1/2}$, we obtain 
\begin{equation}
    \sum_{i=1}^r |i\rangle\!\langle i| \otimes \rho_i + \sum_{i \neq j} |i\rangle\!\langle j| \otimes \rho_i ^{1/2} K_{ij}   \rho_j^{1/2} \geq 0.
\end{equation}
Then referring to the dual SDP of multivariate SDP fidelity in \cref{prop:dual_SDP_multi}, we see that a feasible point in the optimization~\eqref{eq:multivariate_fid_sdp_dual} is given by $X_{ij}= \rho_i ^{1/2} K_{ij}   \rho_j^{1/2} $ for $i\neq j$.
With that we have:
\begin{equation}
   \frac{2}{r(r-1)}   \sum_{i<j} \mathfrak{R}\!\left[ \Tr\!\left[ \rho_i^{1/2}K_{ij} \rho_j^{1/2}\right]\right] \leq F_{\operatorname{SDP}}(\rho_1,\ldots, \rho_r).
\end{equation}
Then, optimizing over all $K$ satisfying the required constraint, we arrive at the claim:
\begin{equation}\label{eq:inequality_K_SDP_general}
    F_{K^\star}(\rho_1,\ldots,\rho_r) \leq F_{\operatorname{SDP}}(\rho_1,\ldots, \rho_r).
\end{equation}

\bigskip
\noindent\underline{\textbf{Step~3}: $F_{K^\star}(\rho_1,\ldots,\rho_r) =F_{\operatorname{SDP}}(\rho_1,\ldots, \rho_r)$ for invertible states}
\medskip 

Consider an arbitrary feasible point in the dual SDP~\eqref{eq:multivariate_fid_sdp_dual} of multivariate SDP fidelity:
\begin{equation}
     \sum_{i=1}^r |i\rangle\!\langle i| \otimes \rho_i + \sum_{i \neq j} |i\rangle\!\langle j| \otimes X_{ij} \geq 0.
\end{equation}
Using the assumption that $\rho_i$  is invertible for all $i \in [ r]$, by left and right multiplying the above inequality by the positive definite matrix
$\sum_{i=1}^r |i\rangle\!\langle i| \otimes \rho_i^{-1/2}$, we obtain 
\begin{equation}
    \sum_{i=1}^r |i\rangle\!\langle i| \otimes I_d + \sum_{i \neq j} |i\rangle\!\langle j| \otimes  \rho_i^{-1/2} X_{ij} \rho_j^{-1/2}   \geq 0.
\end{equation}
This implies that the matrices $K_{ij}=\rho_i^{-1/2} X_{ij} \rho_j^{-1/2}$ for $i \neq j$ form a feasible point for the optimization in~\eqref{eq:multivariateFidelity_another_eq}.
With this choice of a feasible point, we get 
\begin{equation}
     \frac{2}{r(r-1)}   \sum_{i<j} \mathfrak{R}\!\left[ \Tr\!\left[ X_{ij}\right]\right] \leq F_{K^\star}(\rho_1,\ldots, \rho_r).
\end{equation}
Since the above inequality holds for an arbitrary feasible point in~\eqref{eq:multivariate_fid_sdp_dual}, it follows by the definition that
\begin{equation}\label{eq:sdp_fidelity_upper_bound_kstar_invertible}
      F_{\operatorname{SDP}}(\rho_1,\ldots, \rho_r) \leq   F_{K^\star}(\rho_1,\ldots,\rho_r).
\end{equation}
Combining~\eqref{eq:sdp_fidelity_upper_bound_kstar_invertible} with~\eqref{eq:inequality_K_SDP_general}, we get
\begin{equation}\label{eq:equality_K_SDP_invertible}
    F_{\operatorname{SDP}}(\rho_1,\ldots, \rho_r) =   F_{K^\star}(\rho_1,\ldots,\rho_r)
\end{equation}
for invertibles states.

\bigskip
\noindent\underline{\textbf{Step~4}: Extending to non-invertible states}
\medskip 

For all $i\in [r]$ and $\varepsilon > 0 $,
\begin{equation}
     \rho_i^{(\varepsilon)} \coloneqq \dfrac{1}{1+\varepsilon} \left(\rho_i +  \frac{\varepsilon}{d} I \right)
\end{equation}
is an invertible state. Then, we have that 
\begin{equation}
  \frac{1}{2} \left  \| \rho_i - \rho_i^{(\varepsilon)}  \right\|_1 = \frac{\varepsilon}{2(1+\varepsilon)} \left\| \rho_i - \frac{I}{d} \right\|_1 \leq \frac{ \varepsilon}{1+\varepsilon} \eqqcolon \varepsilon',
\end{equation}
leading to $ F\!\left(\rho_i, \rho_i^{(\varepsilon)}\right) \geq 1- \varepsilon' $. 
Then by \cref{thm:uniform_cont_SDP_K}, we get
\begin{equation} \label{eq:continuity_k_represent}
    \left| F_{K^\star}(\rho_1, \ldots, \rho_r) -F_{K^\star}\!\left( \rho_1^{(\varepsilon)}, \ldots, \rho_r^{(\varepsilon)} \right)   \right| \leq \frac{r}{r-1} \sqrt{\varepsilon'(2-\varepsilon')}.
\end{equation}

Using \cref{prop:monotonicity} with the choice $\varepsilon_1=0$ and $\varepsilon_2=\varepsilon >0$, we have that 
\begin{align}
F_{\operatorname{SDP}}\!\left(\rho_1, \ldots, \rho_r \right)  &\leq (1+\varepsilon)F_{\operatorname{SDP}}\!\left(\rho_1^{(\varepsilon)}, \ldots, \rho_r^{(\varepsilon)}\right)\\ 
& \leq F_{\operatorname{SDP}}\!\left(\rho_1^{(\varepsilon)}, \ldots, \rho_r^{(\varepsilon)}\right) + \varepsilon, \label{eq:otherside_bound_SDP_U}
\end{align}
where the last inequality holds because $F_{\operatorname{SDP}}\!\left(\rho_1^{(\varepsilon)}, \ldots, \rho_r^{(\varepsilon)}\right) \leq 1$.
Then consider that
\begin{align}
F_{\operatorname{SDP}}\!\left(\rho_1^{(\varepsilon)}, \ldots, \rho_r^{(\varepsilon)}\right)&= F_{K^\star}\!\left(\rho_1^{(\varepsilon)}, \ldots, \rho_r^{(\varepsilon)}\right)\\ 
& \leq F_{K^\star}(\rho_1, \ldots, \rho_r) +\frac{r}{r-1} \sqrt{\varepsilon'(2-\varepsilon')} \\ 
& \leq F_{\operatorname{SDP}}\!\left(\rho_1, \ldots, \rho_r \right) + \frac{r}{r-1} \sqrt{\varepsilon'(2-\varepsilon')}, \label{eq:onesided_bound_SDP_u}
\end{align}
where the first equality follows from~\eqref{eq:equality_K_SDP_invertible}; the second inequality from~\eqref{eq:continuity_k_represent}; and finally the third inequality from~\eqref{eq:inequality_K_SDP_general} for general states. 

Combining~\eqref{eq:otherside_bound_SDP_U} and~\eqref{eq:onesided_bound_SDP_u}, we arrive at 
\begin{align}
    \left| F_{\operatorname{SDP}}(\rho_1, \ldots, \rho_r) -F_{\operatorname{SDP}}\!\left( \rho_1^{(\varepsilon)}, \ldots, \rho_r^{(\varepsilon)} \right)   \right| &\leq \max\!\left\{\varepsilon, \frac{r}{r-1} \sqrt{\varepsilon'(2-\varepsilon')} \right\}\\
    & \leq 2 \sqrt{\varepsilon'(2-\varepsilon')} \\
    & = 2 \sqrt{\frac{ \varepsilon}{1+\varepsilon}\left(2-\frac{ \varepsilon}{1+\varepsilon}\right)}.
\label{eq:continuity_SDP}
\end{align}

With the established machinery, consider 
\begin{align}
    &\left| F_{\operatorname{SDP}}(\rho_1, \ldots, \rho_r) -  F_{K^\star}(\rho_1, \ldots, \rho_r) \right|  \cr
    &\leq  \left| F_{\operatorname{SDP}}(\rho_1, \ldots, \rho_r) -F_{\operatorname{SDP}}\!\left( \rho_1^{(\varepsilon)}, \ldots, \rho_r^{(\varepsilon)} \right)   \right|     \notag \\
    & \quad \quad + \left|F_{\operatorname{SDP}}\!\left( \rho_1^{(\varepsilon)}, \ldots, \rho_r^{(\varepsilon)} \right)  -F_{K^\star}\!\left( \rho_1^{(\varepsilon)}, \ldots, \rho_r^{(\varepsilon)} \right)   \right| \notag \\ & \quad \quad  + \left| F_{K^\star}(\rho_1, \ldots, \rho_r) -F_{K^\star}\!\left( \rho_1^{(\varepsilon)}, \ldots, \rho_r^{(\varepsilon)} \right)   \right|\\
    &\leq  2 \sqrt{\frac{ \varepsilon}{1+\varepsilon}\left(2-\frac{ \varepsilon}{1+\varepsilon}\right)} + 0 + \frac{r}{r-1} \sqrt{\frac{ \varepsilon}{1+\varepsilon}\left(2-\frac{ \varepsilon}{1+\varepsilon}\right)},
\end{align}
where the first inequality follows from triangular inequality; the second inequality from using~\eqref{eq:continuity_k_represent} and~\eqref{eq:continuity_SDP} along with the fact that $F_{\operatorname{SDP}}\!\left( \rho_1^{(\varepsilon)}, \ldots, \rho_r^{(\varepsilon)} \right) =F_{K^\star}\!\left( \rho_1^{(\varepsilon)}, \ldots, \rho_r^{(\varepsilon)} \right)$ due to~\eqref{eq:equality_K_SDP_invertible}.

Finally, by taking limit ${\varepsilon} \to 0$ in the above inequality, we conclude that 
\begin{equation}
    \left| F_{\operatorname{SDP}}(\rho_1, \ldots, \rho_r) -  F_{K^\star}(\rho_1, \ldots, \rho_r) \right| =0,
\end{equation}
which completes the proof for the general case.

\subsection{Proof of Theorem~\ref{thm:properties_SDP_fidelity} (Properties of multivariate SDP fidelity)} 

\label{Sec:proof_properties}

Note that multivariate SDP fidelity satisfies reduction to classical average pairwise fidelity, faithfulness, and orthogonality by \cref{thm:SDP_fidelity_pairwise_upper_and_lower} and the fact that average pairwise fidelities satisfy these properties as stated in \cref{Prop:properties_average_pairwise_z}.

\medskip

\noindent\underline{Data processing:}
Let $\cN: \mathscr{L}(\cH_A) \to \mathscr{L}(\cH_{B})$ be a channel, where the dimension of $\cH_A$ and $\cH_B$ be $d_A$ and $d_B$, respectively. 
    For $i \in [r]$, let $Y_i$  be such that 
    $\sum_{i=1}^r |i\rangle\!\langle i| \otimes Y_i \geq \sum_{i \neq j} |i\rangle\!\langle j| \otimes I_{d_B} $. 
    Consider that
    \begin{equation}
        \frac{1}{r(r-1)} \sum_{i=1}^r \Tr\!\left[ Y_i \cN(\rho_i) \right] = \frac{1}{r(r-1)} \sum_{i=1}^r \Tr\!\left[ \cN^\dagger(Y_i) \rho_i \right].
    \end{equation}
Then, due to $\cN$ being completely positive, along with $\cN$ being trace-preserving (so that $\cN^\dagger(I_{d_B}) =I_{d_A}$),  we have that
\begin{equation}
    0 \leq \sum_{i=1}^r |i\rangle\!\langle i| \otimes \cN^\dagger(Y_i) - \sum_{i \neq j} |i\rangle\!\langle j| \otimes \cN^\dagger(I_{d_B}) = \sum_{i=1}^r |i\rangle\!\langle i| \otimes \cN^\dagger(Y_i) - \sum_{i \neq j} |i\rangle\!\langle j| \otimes I_{d_A}   .
\end{equation}
With that, the tuple $(\cN^\dagger(Y_i) )_{i=1}^r$ is a candidate for $F_{\operatorname{SDP}}(\rho_1, \ldots, \rho_r)$. Thus, 
\begin{equation}
    F_{\operatorname{SDP}}(\rho_1, \ldots, \rho_r) \leq  \frac{1}{r(r-1)} \sum_{i=1}^r \Tr\!\left[ Y_i \cN(\rho_i) \right]. 
\end{equation}
This holds for all $(Y_i)_{i=1}^r$ satisfying $\sum_{i=1}^r |i\rangle\!\langle i| \otimes Y_i \geq \sum_{i \neq j} |i\rangle\!\langle j| \otimes I_{d_A} $. Finally, taking the infimum over all such $(Y_i)_{i=1}^r$ concludes the proof.

\begin{remark}
    We note that the proof for data processing given above holds not just for completely positive, trace-preserving maps, but more generally for $r$-positive, trace-preserving maps.
\end{remark}

\medskip 

\noindent\underline{Symmetry:} This follows directly from the definition of the SDP fidelity.

\medskip

\noindent\underline{Direct-sum property:}
    Let $( Y_i^x)_{i=1}^r$ for all $x \in \cX$ satisfy 
    $ \sum_{i=1}^r |i\rangle\!\langle i| \otimes Y^x_i - \sum_{i \neq j} |i\rangle\!\langle j| \otimes I \geq 0$.
    Consider that
    \begin{align}
       & \frac{1}{r (r-1)} \sum_{x \in \cX} p(x)  \sum_{i=1}^r \Tr[ Y_i^x \rho_i^x] \notag \\ 
        & =  \frac{1}{r(r-1)} \sum_{i=1}^r  \Tr\! \left[\sum_{x \in \cX} p(x)  Y_i^x \rho_i^x \right]  \\ 
        &= \frac{1}{r(r-1)} \sum_{i=1}^r  \Tr\! \left[\left(\sum_{x' \in \cX} |x' \rangle\! \langle x'| \otimes Y_i^{x'} \right)  \sum_{x \in \cX} p(x)   |x \rangle\!\langle x | \otimes \rho_i^x \right].\label{eq:classical-quantum-state-trace-relation}
    \end{align}
    From the conditions satisfied by $\left(Y_i^x\right)_{i=1}^r$, we also get 
    \begin{align}
       \sum_{i=1}^r |i\rangle\!\langle i|  \otimes |x \rangle\!\langle x| \otimes Y^x_i  - \sum_{i \neq j} |i\rangle\!\langle j| \otimes |x \rangle\!\langle x|\otimes I_d  \geq 0. 
    \end{align}
    By summing over $x \in \cX$, we get
     \begin{align}
       \sum_{i=1}^r |i\rangle\!\langle i| \otimes \sum_{x \in \cX } |x \rangle\!\langle x| \otimes Y^x_i   - \sum_{i \neq j} |i\rangle\!\langle j| \otimes  I_\mathcal{|X|}  \otimes I_d \geq 0. 
    \end{align}
    Then, $\left ( \sum_{x \in \cX } |x \rangle\!\langle x| \otimes Y^x_i\right )_{i=1}^r$ forms a feasible point for the optimization in \sloppy $F_{\operatorname{SDP}}\! \left( \sum_{x \in \cX} p(x) |x\rangle\!\langle x| \otimes \rho_1^x, \ldots, \sum_{x \in \cX} p(x)  |x\rangle\!\langle x| \otimes \rho_r^x\right)$.
    Using the relation~\eqref{eq:classical-quantum-state-trace-relation}, this leads to
\begin{multline}
F_{\operatorname{SDP}}\! \left( \sum_{x \in \cX} p(x) |x\rangle\!\langle x| \otimes \rho_1^x, \ldots, \sum_{x \in \cX} p(x)  |x\rangle\!\langle x| \otimes \rho_r^x\right) \\
     \leq \frac{1}{r (r-1)} \sum_{x \in \cX} p(x)  \sum_{i=1}^r \Tr[ Y_i^x \rho_i^x]. 
\end{multline}
Since the last inequality holds for all $\left(Y_i^x\right)_{i=1}^r$ and for all $x \in \cX$, this also holds for  the optimal ones achieving $ F(\rho_1^x, \ldots, \rho_r^x)$.
We thus get 
 \begin{equation}\label{eq:forward}
         F_{\operatorname{SDP}}\! \left( \sum_{x \in \cX} p(x) |x\rangle\!\langle x| \otimes \rho_1^x, \ldots, \sum_{x \in \cX} p(x)  |x\rangle\!\langle x| \otimes \rho_r^x\right) \leq  \sum_{x \in \cX} p(x)  F_{\operatorname{SDP}}(\rho_1^x, \ldots, \rho_r^x).
    \end{equation}

To prove the reverse direction, 
let $\left( Z_i\right)_{i=1}^r$ satisfy 
    $ \sum_{i=1}^r |i\rangle\!\langle i| \otimes Z_i - \sum_{i \neq j} |i\rangle\!\langle j|  \otimes I_\mathcal{|X|} \otimes I_d \geq 0$. Then
   \begin{multline} \label{eq:different_forms_with_tensor}
       \frac{1}{r(r-1)} \sum_{i=1}^r  \Tr\! \left[ Z_i  \sum_{x \in \cX} p(x)   |x \rangle\!\langle x | \otimes \rho_i^x \right] \\ = \sum_{x \in \cX} p(x)  \frac{1}{r(r-1)} \sum_{i=1}^r  \Tr\! \left[ (\langle x| \otimes I_d  ) Z_i ( | x\rangle \otimes I_d ) \rho_i^x \right].
   \end{multline}
From the conditions satisfied by  $\left( Z_i\right)_{i=1}^r$, we also get for all $x \in \cX$ that
    \begin{multline}
        \sum_{i=1}^r |i\rangle\!\langle i|  \otimes (\langle x| \otimes I_d  ) Z_i ( | x\rangle \otimes I_d )   - \sum_{i \neq j} |i\rangle\!\langle j| \otimes  \langle x|x\rangle \otimes I_d  \geq 0 \\ \Longrightarrow \qquad \sum_{i=1}^r |i\rangle\!\langle i|  \otimes (\langle x| \otimes I_d  ) Z_i ( | x\rangle \otimes I_d )   - \sum_{i \neq j} |i\rangle\!\langle j| \otimes   I_d  \geq 0.
    \end{multline}
Then $ \left( (\langle x| \otimes I_d  ) Z_i ( | x\rangle \otimes I_d )\right)_{i=1}^r $ is a candidate for $ F_{\operatorname{SDP}}(\rho_1^x, \ldots, \rho_r^x)$, leading to 
\begin{equation}
    F_{\operatorname{SDP}}(\rho_1^x, \ldots, \rho_r^x) \leq  \frac{1}{r(r-1)} \sum_{i=1}^r  \Tr\! \left[  (\langle x| \otimes I_d  ) Z_i ( | x\rangle \otimes I_d )\rho_i^x \right].
\end{equation}
Summing over all $x \in \cX$ together with~\eqref{eq:different_forms_with_tensor} leads to
\begin{equation}
    \sum_{x \in \cX} p(x) F_{\operatorname{SDP}}(\rho_1^x, \ldots, \rho_r^x) \leq  \frac{1}{r(r-1)} \sum_{i=1}^r  \Tr\! \left[ Z_i  \sum_{x \in \cX} p(x)   |x \rangle\!\langle x | \otimes \rho_i^x \right].
\end{equation}
Since the above inequality also holds for an optimal choice of $(Z_i)_{i=1}^r$, we arrive at
 \begin{equation}\label{eq:reverse}
         \sum_{x \in \cX} p(x)  F_{\operatorname{SDP}}(\rho_1^x, \ldots, \rho_r^x) \leq  F_{\operatorname{SDP}}\! \left( \sum_{x \in \cX} p(x) |x\rangle\!\langle x| \otimes \rho_1^x, \ldots, \sum_{x \in \cX} p(x)  |x\rangle\!\langle x| \otimes \rho_r^x\right)  .
    \end{equation}
  Lastly, by combining~\eqref{eq:forward}  and~\eqref{eq:reverse}, we conclude the proof of~\eqref{eq:CQequality}. 

\medskip 
\noindent\underline{Joint concavity:}
Using the direct-sum property (property 6 in \cref{thm:properties_SDP_fidelity}), we have
    \begin{equation}
         F_{\operatorname{SDP}}\! \left( \sum_{x \in \cX} p(x) |x\rangle\!\langle x| \otimes \rho_1^x, \ldots, \sum_{x \in \cX} p(x)  |x\rangle\!\langle x| \otimes \rho_r^x\right) = \sum_{x \in \cX} p(x)  F_{\operatorname{SDP}}(\rho_1^x, \ldots, \rho_r^x).
    \end{equation}
    Then, using the above and applying data processing  with respect to the partial trace channel (i.e., taking a partial trace over the classical system) to the left-hand side concludes the proof.

\subsection{Proof of Proposition~\ref{prop:coarse_graining} (Coarse-graining property of SDP fidelity)} \label{Proof:coarse_graining}
    Let $(Y_i)_{i=1}^{r+m}$ be such that 
    \begin{equation}
        \sum_{i=1}^{r +m}|i\rangle\!\langle i| \otimes Y_i \geq  \sum_{i \neq j} |i\rangle\!\langle j| \otimes I.
        \label{eq:full-r-and-m-constraint}
    \end{equation}
    Then that choice is a possible candidates for the SDP of $ F_{\operatorname{SDP}}(\rho_1,\ldots,\rho_r, \dots, \rho_{r+m})$ with the extension of~\eqref{eq:multivar-fid} to $r+m$ states. Consider 
    \begin{align}
          \sum_{i=1}^{r+m}\Tr[Y_i \rho_i] & \geq \sum_{i=1}^{r}\Tr[Y_i \rho_i]  \\
          &  \geq r(r-1)  F_{\operatorname{SDP}}(\rho_1, \dots, \rho_r)
    \end{align}
    where the first inequality follows from $Y_i \geq 0$ for all $i \in [r+m]$, and the last inequality from the fact that $(Y_i)_{i=1}^r$ ( $r$ elements from the tuple $(Y_i)_{i=1}^{r+m}$) forms a possible candidate for the SDP of $ F_{\operatorname{SDP}}(\rho_1,\ldots, \rho_r)$ in~\eqref{eq:multivar-fid} due to~\eqref{eq:full-r-and-m-constraint} implying the following operator inequality holds:
    \begin{equation}
        \sum_{i=1}^r |i\rangle\!\langle i| \otimes Y_i \geq  \sum_{i \neq j} |i\rangle\!\langle j| \otimes I.
    \end{equation} We conclude the proof by infimizing the LHS above over all possible candidates $(Y_i)_{i=1}^{r+m}$ satisfying the constraints, which renders 
    \begin{equation}
        (r+m) (r+m-1)  F_{\operatorname{SDP}}(\rho_1,\ldots,\rho_r, \dots, \rho_{r+m}) \geq r(r-1)  F_{\operatorname{SDP}}(\rho_1, \dots, \rho_r).
    \end{equation}

\subsection{Proof of Proposition~\ref{prop:SDP-secrecy-SDP} (SDPs for the secrecy measure~$S_{\operatorname{SDP}}$)}

\label{app:SDP-secrecy-SDP}

Starting from the $K^{\star}$ representation in \cref{def:K-star-rep}, observe that the
constraint there is equivalent to $\left(  \Delta_{r}\otimes\operatorname{id}
_{d}\right)  \left(  K\right)  =I_{r} \otimes I_{d}$. Also, observe that
\begin{equation}
S_{\operatorname{SDP}}(\rho_{1},\ldots,\rho_{r})^{2}=\frac{1}{r}\left[  \left(
r-1\right)  F_{\operatorname{SDP}}(\rho_{1},\ldots,\rho_{r})+1\right]  .
\end{equation}
The equality in \eqref{eq:SDP-secrecy-measure} then follows after some simple
algebra and recalling that $|\varphi\rangle$ above and $|\psi\rangle$ in \eqref{eq:K_star_F}
are related as $|\psi\rangle=\frac{1}{\sqrt{r}}|\varphi\rangle$. The equality
in \eqref{eq:SDP-secrecy-measure-SDP-dual} follows from SDP duality theory. Recall $\omega\coloneqq \operatorname{Tr}_{S}[|\varphi\rangle\!\langle\varphi|_{XRS}]$ and 
consider that
\begin{align}
& \sup_{K\geq0}\left\{  \langle\varphi|K\otimes I_{d}|\varphi\rangle:\left(
\Delta_{r}\otimes\operatorname{id}_{d}\right)  \left(  K\right)
=I_{r} \otimes I_d \right\}  \nonumber\\
& =\sup_{K\geq0}\left\{  \operatorname{Tr}[K\omega]:\left(  \Delta_{r}
\otimes\operatorname{id}_{d}\right)  \left(  K\right)  =I_{r} \otimes I_{d} \right\}  \\
& =\sup_{K\geq0}\left\{  \operatorname{Tr}[K\omega]+\inf_{Y\in\operatorname{Herm}
}\operatorname{Tr}[Y\left(  I_r \otimes I_d-\left(  \Delta_{r}\otimes\operatorname{id}
_{d}\right)  \left(  K\right)  \right)  ]\right\}  \\
& =\sup_{K\geq0}\inf_{Y\in\operatorname{Herm}}\left\{  \operatorname{Tr}
[K\omega]+\operatorname{Tr}[Y\left(  I_r \otimes I_d-\left(  \Delta_{r}\otimes
\operatorname{id}_{d}\right)  \left(  K\right)  \right)  ]\right\}  \\
& \leq\inf_{Y\in\operatorname{Herm}}\sup_{K\geq0}\left\{  \operatorname{Tr}
[K\omega]+\operatorname{Tr}[Y\left(  I_r \otimes I_d-\left(  \Delta_{r}\otimes
\operatorname{id}_{d}\right)  \left(  K\right)  \right)  ]\right\}  \\
& =\inf_{Y\in\operatorname{Herm}}\sup_{K\geq0}\left\{  \operatorname{Tr}
[Y]+\operatorname{Tr}[K\omega]-\operatorname{Tr}[\left(  \Delta_{r}
\otimes\operatorname{id}_{d}\right)  (Y)K]\right\}  \\
& =\inf_{Y\in\operatorname{Herm}}\left\{  \operatorname{Tr}[Y]+\sup_{K\geq
0}\operatorname{Tr}[K\left(  \omega-\left(  \Delta_{r}\otimes\operatorname{id}
_{d}\right)  (Y)\right)  ]\right\}  \\
& =\inf_{Y\in\operatorname{Herm}}\left\{  \operatorname{Tr}[Y]:\left(  \Delta
_{r}\otimes\operatorname{id}_{d}\right)  (Y)\geq\omega\right\}  .
\end{align}
Strong duality holds by picking $K=I_{r} \otimes I_d$ in the primal (feasible)\ and
$Y=2\lambda_{\max}(\omega)I$ in the dual (strictly feasible).

\subsection{Proof of Proposition~\ref{prop:Multiplicativity-S-SDP} (Multiplicativity of the secrecy measure $S_{\operatorname{SDP}}$)}

\label{app:Multiplicativity-S-SDP}

This follows by using SDP duality. Let $d_{1}$ denote the dimension of the
states $\rho_{1},\ldots,\rho_{r_{1}}$ and $d_{2}$ the dimension of the states
$\sigma_{1},\ldots,\sigma_{r_{2}}$. Define
\begin{align}
|\varphi_{1}\rangle & \coloneqq \sum_{i=1}^{r_{1}}|i\rangle|\psi^{\rho_{i}}\rangle,\\
|\varphi_{2}\rangle & \coloneqq \sum_{j=1}^{r_{2}}|j\rangle|\psi^{\sigma_{j}}\rangle,
\end{align}
where $|\psi^{\rho_{i}}\rangle$ purifies $\rho_{i}$ and $|\psi^{\sigma_{j}
}\rangle$ purifies $\sigma_{j}$ for all $i \in {[r_1]}$ and $j \in {[r_2]}$. Then
\begin{equation}
|\varphi_{1}\rangle|\varphi_{2}\rangle\simeq|\varphi\rangle\coloneqq \sum_{i=1}
^{r_{1}}\sum_{j=1}^{r_{2}}|i\rangle|j\rangle|\psi^{\rho_{i}}\rangle
|\psi^{\sigma_{j}}\rangle
\end{equation}
is a choice of $|\varphi\rangle$ for the tuple $\left(  \rho_{i}\otimes
\sigma_{j}\right)  _{ij}$ (where we note that $\simeq$ indicates we have
permuted some systems). Let $K_{i}$ satisfy $K_{i}\geq0$ and $\left(
\Delta_{r_{i}}\otimes\operatorname{id}_{d_i}\right)  \left(  K_{i}\right)
=I_{r_{i}} \otimes I_{d_i}$, for $i\in\left\{  1,2\right\}  $. Then it follows that
$K_{1}\otimes K_{2}$ is feasible for $S_{\operatorname{SDP}}(\left(  \rho_{i}
\otimes\sigma_{j}\right)  _{ij})$ because
\begin{align}
K_{1}\otimes K_{2}  & \geq0,\\
\left(  \Delta_{r_{1}}\otimes\operatorname{id}_{d_{1}}\otimes\Delta_{r_{2}
}\otimes\operatorname{id}_{d_{2}}\right)  \left(  K_{1}\otimes K_{2}\right)
& =I_{r_{1}}\otimes I_{d_{1}} \otimes I_{r_{2}}\otimes I_{d_{2}},
\end{align}
and we can think of $\Delta_{r_{1}}\otimes\Delta_{r_{2}}\simeq\Delta
_{r_{1}r_{2}}$ and $\operatorname{id}_{d_{1}}\otimes\operatorname{id}_{d_{2}
}\simeq\operatorname{id}_{d_{1}d_{2}}$. Furthermore,
\begin{align}
& \frac{1}{r_{1}^{2}}\langle\varphi_{1}|K_{1}\otimes I_{d_1}|\varphi_{1}
\rangle\cdot\frac{1}{r_{2}^{2}}\langle\varphi_{2}|K_{2}\otimes I_{d_2}
|\varphi_{2}\rangle\nonumber\\
& =\frac{1}{\left(  r_{1}r_{2}\right)  ^{2}}\langle\varphi|K_{1}\otimes
K_{2}\otimes I_{d_1}\otimes I_{d_2}|\varphi\rangle\\
& \leq S_{\operatorname{SDP}}(\left(  \rho_{i}\otimes\sigma_{j}\right)  _{ij})^{2}.
\end{align}
Since the inequality holds for all feasible $K_{1}$ and $K_{2}$, we conclude
that
\begin{equation}
S_{\operatorname{SDP}}(\rho_{1},\ldots,\rho_{r_{1}})^{2}\cdot S_{\operatorname{SDP}}
(\sigma_{1},\ldots,\sigma_{r_{2}})^{2}\leq S_{\operatorname{SDP}}(\left(  \rho
_{i}\otimes\sigma_{j}\right)  _{ij})^{2}
,\label{eq:primal-multiplica-SDP-secrecy}
\end{equation}
which constitutes one direction of the proof.

For the other direction, we use the dual SDP\ in
\eqref{eq:SDP-secrecy-measure-SDP-dual}. Let $Y_{i}$ be Hermitian and satisfy
\begin{equation}
\left(  \Delta_{r_{i}}\otimes\operatorname{id}_{d_{i}}\right)  (Y_{i}
)\geq\omega_{i},
\end{equation}
where $\omega_{i}\coloneqq \operatorname{Tr}_{S_{i}}[|\varphi_{i}\rangle\!\langle
\varphi_{i}|_{X_{i}R_{i}S_{i}}]$, for $i\in\left\{  1,2\right\}  $. So $Y_{1}$
is feasible for $S_{\operatorname{SDP}}(\rho_{1},\ldots,\rho_{r_{1}})$ and $Y_{2}$ for
$S_{\operatorname{SDP}}(\sigma_{1},\ldots,\sigma_{r_{2}})$. Then, since $\omega
_{i}\geq0$, it follows that $\left(  \Delta_{r_{i}}\otimes\operatorname{id}
_{d_{i}}\right)  (Y_{i})\geq0$, for $i\in\left\{  1,2\right\}  $. Thus,
$Y_{1}\otimes Y_{2}$ is Hermitian and
\begin{align}
\omega_{1}\otimes\omega_{2}  & \leq\left(  \Delta_{r_{1}}\otimes
\operatorname{id}_{d_{1}}\right)  (Y_{1})\otimes\omega_{2}\\
& \leq\left(  \Delta_{r_{1}}\otimes\operatorname{id}_{d_{1}}\right)
(Y_{1})\otimes\left(  \Delta_{r_{1}}\otimes\operatorname{id}_{d_{1}}\right)
(Y_{2})\\
& =\left(  \Delta_{r_{1}}\otimes\operatorname{id}_{d_{1}}\otimes\Delta_{r_{1}
}\otimes\operatorname{id}_{d_{1}}\right)  (Y_{1}\otimes Y_{2}).
\end{align}
Thus, $Y_{1}\otimes Y_{2}$ is feasible for the dual SDP\ of $S_{\text{SDP}
}(\left(  \rho_{i}\otimes\sigma_{j}\right)  _{ij})$ in
\eqref{eq:SDP-secrecy-measure-SDP-dual}. As such, we conclude that
\begin{align}
\frac{1}{r_{1}^{2}}\operatorname{Tr}[Y_{1}]\cdot\frac{1}{r_{2}^{2}
}\operatorname{Tr}[Y_{2}]  & =\frac{1}{\left(  r_{1}r_{2}\right)  ^{2}
}\operatorname{Tr}[Y_{1}\otimes Y_{2}]\\
& \geq S_{\operatorname{SDP}}(\left(  \rho_{i}\otimes\sigma_{j}\right)  _{ij})^{2}.
\end{align}
Since the inequality holds for all feasible $Y_{1}$ and $Y_{2}$, we conclude
that
\begin{equation}
S_{\operatorname{SDP}}(\rho_{1},\ldots,\rho_{r_{1}})^{2}\cdot S_{\operatorname{SDP}}
(\sigma_{1},\ldots,\sigma_{r_{2}})^{2}\geq S_{\operatorname{SDP}}(\left(  \rho
_{i}\otimes\sigma_{j}\right)  _{ij})^{2}.
\end{equation}
Combining with \eqref{eq:primal-multiplica-SDP-secrecy}, this concludes the proof.

\subsection{Proof of Proposition~\ref{Prop:pure_states_SDP_F} (Multivariate SDP fidelity for pure states)}

\label{Proof:pure_states_SDP_F}

Consider the general formula for arbitrary states $\left(  \rho_{1}
,\ldots,\rho_{r}\right)  $ from \cref{prop:another_formulation_multi_SDP_fid}:
\begin{multline}
F_{\operatorname{SDP}}(\rho_{1},\ldots,\rho_{r})=\label{eq:K-form-SDP-fid}\\
\frac{1}{r-1}\sup_{K\geq0}\left\{  \langle\psi|K\otimes I_{d}|\psi
\rangle-1:K=\sum_{i,j=1}^{r}|i\rangle\!\langle j|\otimes K_{ij},\ K_{ii}
=I_{d}\ \forall i\in\left[  r\right]  \right\}  ,
\end{multline}
where $|\psi\rangle=\frac{1}{\sqrt{r}}\sum_{i=1}^{r}|i\rangle|\psi^{\rho_{i}
}\rangle$ and $|\psi^{\rho_{i}}\rangle$ is a purification of $\rho_{i}$. In
the case that each $\rho_{i}$ is pure, the purification system for each state
can be taken to be a trivial one-dimensional system. So then each matrix
$K_{ij}$ reduces to a scalar $k_{ij} \in \mathbb{C}$, and the constraints in
\eqref{eq:K-form-SDP-fid} reduce to those in~\eqref{eq:K-form-SDP-fid-pure}.
Then
\begin{align}
& \langle\psi|K\otimes I_{d}|\psi\rangle\nonumber\\
& =\left(  \frac{1}{\sqrt{r}}\sum_{i^{\prime\prime}=1}^{r}\langle
i^{\prime\prime}|\langle\psi_{i^{\prime\prime}}|\right)  \left(  \sum
_{i,j=1}^{r}k_{ij}|i\rangle\!\langle j|\otimes I_{d}\right)  \left(  \frac
{1}{\sqrt{r}}\sum_{i^{\prime}=1}^{r}|i^{\prime}\rangle|\psi_{i^{\prime}
}\rangle\right)  \\
& =\frac{1}{r}\sum_{i^{\prime\prime},i,j,i^{\prime}=1}^{r}k_{ij}\langle
i^{\prime\prime}|i\rangle\!\langle j|i^{\prime}\rangle\ \langle\psi
_{i^{\prime\prime}}|\psi_{i^{\prime}}\rangle\\
& =\frac{1}{r}\sum_{i,j=1}^{r}k_{ij}\langle\psi_{i}|\psi_{j}\rangle\\
& =\frac{1}{r}\sum_{i=1}^{r}k_{ii}\langle\psi_{i}|\psi_{i}\rangle+\frac{1}
{r}\sum_{i,j=1:i\neq j}^{r}k_{ij}\langle\psi_{i}|\psi_{j}\rangle\\
& =1+\frac{2}{r}\sum_{i<j}\mathfrak{R}[k_{ij}\langle\psi_{i}|\psi
_{j}\rangle].\label{eq:K_star_calculation_pure_states}
\end{align}
Substituting~\eqref{eq:K_star_calculation_pure_states} into the objective function in~\eqref{eq:K-form-SDP-fid} and
simplifying then gives~\eqref{eq:K-form-SDP-fid-pure}.

\subsection{Proof of Equation \eqref{eq:dual-alt-sdp-secrecy} (Dual SDP for secrecy measure~$S$)}

\label{app:dual-alt-sdp-secrecy}

We derive the dual SDP\ by means of the Lagrange multiplier method. Define the
shorthand
\begin{equation}\label{eq:S_i_Lagrange}
S_{i}\equiv
\begin{bmatrix}
Y_{i} & W_{i}\\
W_{i}^{\dag} & Z_{i}
\end{bmatrix}
,
\end{equation}
and consider that
\begin{align}
& \sup_{\substack{X_{1},\ldots,X_{r}\in\mathcal{L},\\\sigma\geq0}}\left\{
\sum_{i=1}^{r}\operatorname{Tr}[X_i]+\operatorname{Tr}[X_i^{\dag}
]:\operatorname{Tr}[\sigma]=1,\
\begin{bmatrix}
\rho_{i} & X_{i}\\
X_{i}^{\dag} & \sigma
\end{bmatrix}
\geq0\ \forall i\in\left[  r\right]  \right\}  \nonumber\\
& =\sup_{\substack{X_{1},\ldots,X_{r}\in\mathcal{L},\\\sigma\geq0}}\left\{
\begin{array}
[c]{c}
\sum_{i=1}^{r}2\,\mathfrak{R}[\operatorname{Tr}[X_i]]\\
+\inf_{\substack{W_{i},Y_{i},Z_{i}\in\mathcal{L}\ \forall i,\\S_{i}
\geq0 \ \forall i, \ \lambda\in\mathbb{R}}}\left\{
\begin{array}
[c]{c}
\lambda\left(  1-\operatorname{Tr}[\sigma]\right)  \\
+\sum_{i=1}^{r}\operatorname{Tr}\!\left[
\begin{bmatrix}
Y_{i} & W_{i}\\
W_{i}^{\dag} & Z_{i}
\end{bmatrix}
\begin{bmatrix}
\rho_{i} & X_{i}\\
X_{i}^{\dag} & \sigma
\end{bmatrix}
\right]
\end{array}
\right\}
\end{array}
\right\}  \\
& =\sup_{\substack{X_{1},\ldots,X_{r}\in\mathcal{L},\\\sigma\geq0}
}\inf_{\substack{W_{i},Y_{i},Z_{i}\in\mathcal{L}\ \forall i,\\S_{i} 
\geq0 \ \forall i, \ ,\lambda\in\mathbb{R}}}\left\{
\begin{array}
[c]{c}
\sum_{i=1}^{r}2\,\mathfrak{R}[\operatorname{Tr}[X_i]]+\lambda\left(
1-\operatorname{Tr}[\sigma]\right)  \\
+\sum_{i=1}^{r}\operatorname{Tr}\!\left[
\begin{bmatrix}
Y_{i} & W_{i}\\
W_{i}^{\dag} & Z_{i}
\end{bmatrix}
\begin{bmatrix}
\rho_{i} & X_{i}\\
X_{i}^{\dag} & \sigma
\end{bmatrix}
\right]
\end{array}
\right\}  \\
& =\sup_{\substack{X_{1},\ldots,X_{r}\in\mathcal{L},\\\sigma\geq0}
}\inf_{\substack{W_{i},Y_{i},Z_{i}\in\mathcal{L}\ \forall i,\\S_{i}
\geq0 \ \forall i, \ ,\lambda\in\mathbb{R}}}\left\{
\begin{array}
[c]{c}
\sum_{i=1}^{r}2\,\mathfrak{R}[\operatorname{Tr}[X_i]]+\lambda\left(
1-\operatorname{Tr}[\sigma]\right)  \\
+\sum_{i=1}^{r}\operatorname{Tr}[Y_{i}\rho_{i}]+2\,\mathfrak{R}
[\operatorname{Tr}[W_{i}X_{i}^{\dag}]]+\operatorname{Tr}[Z_{i}\sigma]
\end{array}
\right\}  \\
& =\sup_{\substack{X_{1},\ldots,X_{r}\in\mathcal{L},\\\sigma\geq0}
}\inf_{\substack{W_{i},Y_{i},Z_{i}\in\mathcal{L}\ \forall i,\\S_{i}
\geq0 \ \forall i, \ ,\lambda\in\mathbb{R}}}\left\{
\begin{array}
[c]{c}
\sum_{i=1}^{r}2\,\mathfrak{R}[\operatorname{Tr}[X_i]]+\lambda\left(
1-\operatorname{Tr}[\sigma]\right)  \\
+\sum_{i=1}^{r}\operatorname{Tr}[Y_{i}\rho_{i}]+2\,\mathfrak{R}
[\operatorname{Tr}[W_{i}X_{i}^{\dag}]]+\operatorname{Tr}[Z_{i}\sigma]
\end{array}
\right\}  \\
& \leq\inf_{\substack{W_{i},Y_{i},Z_{i}\in\mathcal{L}\ \forall i,\\S_{i}
\geq0 \ \forall i, \ ,\lambda\in\mathbb{R}}}\sup_{\substack{X_{1},\ldots,X_{r}\in
\mathcal{L},\\\sigma\geq0}}\left\{
\begin{array}
[c]{c}
\sum_{i=1}^{r}2\,\mathfrak{R}[\operatorname{Tr}[X_i]]+\lambda\left(
1-\operatorname{Tr}[\sigma]\right)  \\
+\sum_{i=1}^{r}\operatorname{Tr}[Y_{i}\rho_{i}]+2\,\mathfrak{R}
[\operatorname{Tr}[W_{i}X_{i}^{\dag}]]+\operatorname{Tr}[Z_{i}\sigma]
\end{array}
\right\}  \\
& =\inf_{\substack{W_{i},Y_{i},Z_{i}\in\mathcal{L}\ \forall i,\\S_{i}
\geq0,\lambda\in\mathbb{R}}}\left\{
\begin{array}
[c]{c}
\sum_{i=1}^{r}\operatorname{Tr}[Y_{i}\rho_{i}]+\lambda\\
+\sup_{\substack{X_{1},\ldots,X_{r}\in\mathcal{L},\\\sigma\geq0}}\left\{
\begin{array}
[c]{c}
\sum_{i=1}^{r}2\,\mathfrak{R}[\operatorname{Tr}[\left(  W_{i}+I\right)
X_{i}^{\dag}]]\\
+\operatorname{Tr}\!\left[  \left(  \sum_{i=1}^{r}Z_{i}-\lambda I\right)
\sigma\right]
\end{array}
\right\}
\end{array}
\right\}  \\
& =\inf_{\substack{Y_{i},Z_{i}\in\mathcal{L}\ \forall i,\\\lambda\in
\mathbb{R}}}\left\{  \sum_{i=1}^{r}\operatorname{Tr}[Y_{i}\rho_{i}
]+\lambda:\sum_{i=1}^{r}Z_{i}\leq\lambda I,\
\begin{bmatrix}
Y_{i} & -I\\
-I & Z_{i}
\end{bmatrix}
\geq0\ \forall i\in\left[  r\right]  \right\}  \\
& =\inf_{\substack{Y_{i},Z_{i}\geq0\ \forall i,\\\lambda\geq0}}\left\{
\sum_{i=1}^{r}\operatorname{Tr}[Y_{i}\rho_{i}]+\lambda:\sum_{i=1}^{r}Z_{i}
\leq\lambda I,\
\begin{bmatrix}
Y_{i} & -I\\
-I & Z_{i}
\end{bmatrix}
\geq0\ \forall i\in\left[  r\right]  \right\}  .
\end{align}
The penultimate equality follows because
\begin{equation}
    \sup_{\substack{X_{1},\ldots,X_{r}\in\mathcal{L}}}
\sum_{i=1}^{r}2\,\mathfrak{R}[\operatorname{Tr}[\left(  W_{i}+I\right)
X_{i}^{\dag}]] = 
\begin{cases}
    0 & \text{ if } W_i = -I \ \forall i \in [r] \\
    +\infty & \text{ else}
\end{cases},
\end{equation}
and
\begin{equation}
    \sup_{\sigma \geq 0}  \operatorname{Tr}\!\left[  \left(  \sum_{i=1}^{r}Z_{i}-\lambda I\right)
\sigma\right] = 
\begin{cases}
    0 & \text{ if } \sum_{i=1}^{r}Z_{i}
\leq\lambda I \\
+\infty & \text{ else }
\end{cases}.
\end{equation}
Strong duality follows by applying Slater's theorem while picking $X_{i}=0$
and $\sigma=I/d$ in the primal (feasible choices), and $Y_{i}=Z_{i}=2I$ and
$\lambda=3r$ in the dual (strictly feasible choices).

\subsection{Proof of Proposition~\ref{prop:another_rep_secrecy_based} (Another representation of secrecy-based multivariate fidelity) }
\label{proof_another_rep_secrecy_based}

The proof follows ideas presented in  \cite[Eqs.~(A70)--(A76)]{RASW23} but makes some further observations beyond those presented there.
Our goal is to rewrite the following secrecy measure:
\begin{equation}
{S(\rho_1, \ldots, \rho_r)^2} = \sup_{\sigma\in\mathcal{D}} {F(\rho_{XA},\rho_{X}\otimes\sigma)^2},
\end{equation}
where $\rho_{XA}$ is defined in \eqref{eq:cq-state-secrecy-meas}.
For all $i\in\left[  r\right]  $, let $|\phi^{\rho_{i}}\rangle_{RA}$ be a
purification of $\rho_{i}$. Let $|\psi^{\sigma}\rangle_{RA}$ be a purification
of $\sigma$. Then it follows that
\begin{equation}
\frac{1}{\sqrt{r}}\sum_{i=1}^{r}|i\rangle_{X}|i\rangle_{R^{\prime}}|\phi
^{\rho_{i}}\rangle_{RA}
\end{equation}
is a purification of $\rho_{XA}$ and
\begin{equation}
|\Phi\rangle_{XR^{\prime}}|\psi^{\sigma}\rangle_{RA}
\end{equation}
is a purification of $\rho_{X}\otimes\sigma$, where $|\Phi\rangle_{XR^{\prime
}}\coloneqq \frac{1}{\sqrt{r}}\sum_{i=1}^{r}|i\rangle_{X}|i\rangle_{R^{\prime}}$. It
follows that
\begin{align}
& \sup_{\sigma\in\mathcal{D}} {F(\rho_{XA},\rho_{X}\otimes\sigma
)^2}\nonumber\\
& =\sup_{\sigma\in\mathcal{D},U\in\mathscr{U}}\left\vert \langle
\Phi|_{XR^{\prime}}\langle\psi^{\sigma}|_{RA}U_{R^{\prime}R}\frac{1}{\sqrt{r}
}\sum_{i=1}^{r}|i\rangle_{X}|i\rangle_{R^{\prime}}|\phi^{\rho_{i}}\rangle
_{RA}\right\vert ^{2} \label{eq:secrecy-fidelity-uhlmann}\\
& =\sup_{|\psi^{\sigma}\rangle,U\in\mathscr{U}}\left\vert \langle
\Phi|_{XR^{\prime}}\langle\psi^{\sigma}|_{RA}U_{R^{\prime}R}\frac{1}{\sqrt{r}
}\sum_{i=1}^{r}|i\rangle_{X}|i\rangle_{R^{\prime}}|\phi^{\rho_{i}}\rangle
_{RA}\right\vert ^{2}\\
& =\sup_{U\in\mathscr{U}}\left\Vert \langle\Phi|_{XR^{\prime}}U_{R^{\prime}
R}\frac{1}{\sqrt{r}}\sum_{i=1}^{r}|i\rangle_{X}|i\rangle_{R^{\prime}}
|\phi^{\rho_{i}}\rangle_{RA}\right\Vert _{2}^{2}
\end{align}
The first equality follows from Uhlmann's theorem~\cite{uhlmann1976transition}. The second equality follows
because the optimization over every state $\sigma$ and every unitary
$U_{R^{\prime}R}$ is equivalent to optimizing over every state vector
$|\psi^{\sigma}\rangle_{RA}$ and unitary $U_{R^{\prime}R}$. The third equality
follows from the variational characterization of the Euclidean norm (i.e., $\left\| |\phi \rangle \right\|_2^2 = \sup_{|\psi\rangle} | \langle\psi|\phi\rangle|^2$, where $|\psi\rangle$ is a state vector). Now
defining
\begin{equation}
P_{R}^{i}\coloneqq \langle i|_{R^{\prime}}U_{R^{\prime}R}|i\rangle_{R^{\prime}},
\end{equation}
consider that
\begin{align}
& \sup_{U\in\mathscr{U}}\left\Vert \langle\Phi|_{XR^{\prime}}U_{R^{\prime}
R}\frac{1}{\sqrt{r}}\sum_{i=1}^{r}|i\rangle_{X}|i\rangle_{R^{\prime}}
|\phi^{\rho_{i}}\rangle_{RA}\right\Vert _{2}^{2}\nonumber\\
& =\sup_{U\in\mathscr{U}}\left\Vert \frac{1}{\sqrt{r}}\sum_{i^{\prime}=1}
^{r}\langle i^{\prime}|_{X}\langle i^{\prime}|_{R^{\prime}}U_{R^{\prime}
R}\frac{1}{\sqrt{r}}\sum_{i=1}^{r}|i\rangle_{X}|i\rangle_{R^{\prime}}
|\phi^{\rho_{i}}\rangle_{RA}\right\Vert _{2}^{2}\\
& =\frac{1}{r^{2}}\sup_{U\in\mathscr{U}}\left\Vert \sum_{i^{\prime},i=1}
^{r}\langle i^{\prime}|i\rangle_{X}\langle i^{\prime}|_{R^{\prime}
}U_{R^{\prime}R}|i\rangle_{R^{\prime}}|\phi^{\rho_{i}}\rangle_{RA}\right\Vert
_{2}^{2}\\
& =\frac{1}{r^{2}}\sup_{U\in\mathscr{U}}\left\Vert \sum_{i=1}^{r}\langle
i|_{R^{\prime}}U_{R^{\prime}R}|i\rangle_{R^{\prime}}|\phi^{\rho_{i}}
\rangle_{RA}\right\Vert _{2}^{2}\\
& =\frac{1}{r^{2}}\sup_{U\in\mathscr{U}}\left\Vert \sum_{i=1}^{r}P_{R}
^{i}|\phi^{\rho_{i}}\rangle_{RA}\right\Vert _{2}^{2}\\
& \leq\frac{1}{r^{2}}\sup_{P_{R}^{i}:\left\Vert P_{R}^{i}\right\Vert _{\infty
}\leq1}\sum_{i,j=1}^{r}\langle\phi^{\rho_{j}}|_{RA}\left(  P_{R}^{j}\right)
^{\dag}P_{R}^{i}|\phi^{\rho_{i}}\rangle_{RA}.
\end{align}
Now let us apply~\cref{lem:assist_uniform_mix_unitaries} stated below,  
which says that every contraction can be
written as a uniform convex combination of two unitaries, so that $P_{R}
^{i}=\frac{1}{2}\left(  V_{R}^{i,0}+V_{R}^{i,1}\right)  $, where $V_{R}^{i,0}$
and $V_{R}^{i,1}$ are unitaries for all $i\in\left[  r\right]  $. Then
\begin{align}
 & \sum_{i,j=1}^{r}\langle\phi^{\rho_{j}}|_{RA}\left(  P_{R}^{j}\right)
^{\dag}P_{R}^{i}|\phi^{\rho_{i}}\rangle_{RA}\notag \\
& =\frac{1}{4}\sum_{i,j=1}^{r}\langle\phi^{\rho_{j}}|_{RA}\left(  \sum
_{k\in\left\{  0,1\right\}  }V_{R}^{j,k}\right)  ^{\dag}\sum_{\ell\in\left\{
0,1\right\}  }V_{R}^{i,\ell}|\phi^{\rho_{i}}\rangle_{RA}\\
& =\frac{1}{4}\sum_{i,j=1}^{r}\sum_{k,\ell\in\left\{  0,1\right\}  }
\langle\phi^{\rho_{j}}|_{RA}\left(  V_{R}^{j,k}\right)  ^{\dag}V_{R}^{i,\ell
}|\phi^{\rho_{i}}\rangle_{RA}\\
& =\frac{1}{4}\sum_{k,\ell\in\left\{  0,1\right\}  }\sum_{i,j=1}^{r}
\langle\phi^{\rho_{j}}|_{RA}\left(  V_{R}^{j,k}\right)  ^{\dag}V_{R}^{i,\ell
}|\phi^{\rho_{i}}\rangle_{RA}.
\end{align}
Now define
\begin{align}
|\psi\rangle_{XRA}  & \coloneqq \sum_{i=1}^{r}|i\rangle_X|\phi^{\rho_{i}}\rangle_{RA},\\
W^{k}  & \coloneqq \sum_{j=1}^{r}\langle j|_{X}\otimes V_{R}^{j,k},
\end{align}
and observe that
\begin{align}
& \frac{1}{4}\sum_{k,\ell\in\left\{  0,1\right\}  }\sum_{i,j=1}^{r}\langle
\phi^{\rho_{j}}|_{RA}\left(  V_{R}^{j,k}\right)  ^{\dag}V_{R}^{i,\ell}
|\phi^{\rho_{i}}\rangle_{RA}  \notag \\
& =\frac{1}{4}\sum_{k,\ell\in\left\{  0,1\right\}  }\langle\psi|_{XRA}
\left(W^{k
}\right)^\dag W^{\ell}|\psi\rangle_{XRA}\\
& \leq\frac{1}{2}\sum_{k\in\left\{  0,1\right\}  }\langle\psi|_{XRA}\left(W^{k
}\right)^\dag W^{k}|\psi\rangle_{XRA}\\
& \leq\sup_{k\in\left\{  0,1\right\}  }\langle\psi|_{XRA}\left(W^{k
}\right)^\dag W^{k}
|\psi\rangle_{XRA}\\
& =\sup_{k\in\left\{  0,1\right\}  }\sum_{i,j=1}^{r}\langle\phi^{\rho_{j}
}|_{RA}\left(  V_{R}^{j,k}\right)  ^{\dag}V_{R}^{i,k}|\phi^{\rho_{i}}
\rangle_{RA}\\
& \leq\sup_{\left(  V^{i}\right)  _{i}}\sum_{i,j=1}^{r}\langle\phi^{\rho_{j}
}|_{RA}\left(  V_{R}^{j}\right)  ^{\dag}V_{R}^{i}|\phi^{\rho_{i}}\rangle_{RA}.
\end{align}
The first inequality follows because
\begin{multline}
\left(  W^{0}-W^{1}\right)  ^{\dag}\left(  W^{0}-W^{1}\right)    \geq 0   \\ \quad \Leftrightarrow\quad  \left(  W^{0}\right)  ^{\dag}W^{0}+\left(  W^{1}\right)
^{\dag}W^{1}   \geq\left(  W^{0}\right)  ^{\dag}W^{1}+\left(  W^{1}\right)
^{\dag}W^{0},
\end{multline}
which implies the operator inequality
\begin{equation}
\sum_{k,\ell\in\left\{  0,1\right\}  }\left(W^{k }\right)^\dag W^{\ell}\leq2\sum_{k\in\left\{
0,1\right\}  }\left(W^{k }\right)^\dag W^{k}.
\end{equation}
The second inequality follows because the average does not exceed the maximum. The last inequality follows by optimizing over every tuple $(V^i_R)_{i=1}^r$ of unitaries.
Thus, 
\begin{equation}
	\sup_{\sigma\in\mathcal{D}}{F(\rho_{XA},\rho_{X}\otimes\sigma)^{2}}\leq \frac{1}{r^2}
	\sup_{\left(  V^{i}\right)  _{i}}\sum_{i,j=1}^{r}\langle\phi^{\rho_{j}}
	|_{RA}\left(  V_{R}^{j}\right)  ^{\dag}V_{R}^{i}|\phi^{\rho_{i}}\rangle_{RA}, \label{eq:secrecy-sup-upper-bound}
\end{equation}
and equality follows by picking $U_{R^{\prime}R}$ in \eqref{eq:secrecy-fidelity-uhlmann} to be a controlled unitary of
the form
\begin{equation}
	U_{R^{\prime}R}=\sum_{i=1}^{r}|i\rangle\!\langle i|_{R^{\prime}}\otimes
	V_{R}^{i},
\end{equation}
with the tuple of unitaries $(V^i_R)_{i=1}^r$ achieving the supremum on the right-hand side of \eqref{eq:secrecy-sup-upper-bound},
implying the statement made in \cref{rem:prover-opt-unitary}.

In summary, we showed that
\begin{equation}
S(\rho_1, \ldots, \rho_r)^2 = \sup_{\sigma\in\mathcal{D}}{F(\rho_{XA},\rho_{X}\otimes\sigma)^{2}}= \frac{1}{r^2}
\sup_{\left(  V^{i}\right)  _{i}}\sum_{i,j=1}^{r}\langle\phi^{\rho_{j}}
|_{RA}\left(  V_{R}^{j}\right)  ^{\dag}V_{R}^{i}|\phi^{\rho_{i}}\rangle_{RA},
\end{equation}
which gives~\eqref{eq:secrecy-as-pairwise} after taking a square root. 

To prove \eqref{eq:secrecy-fid-as-pairwise}, observe that
\begin{align}
& S(\rho_1, \ldots, \rho_r)^2 \notag \\
    & = \frac{1}{r^2}
\sup_{\left(  V^{i}\right)  _{i}}
\sum_{i,j=1}^{r}\langle\phi^{\rho_{j}}
|_{RA}\left(  V_{R}^{j}\right)  ^{\dag}V_{R}^{i}|\phi^{\rho_{i}}\rangle_{RA} \notag \\
& = \frac{1}{r^2}
\sup_{\left(  V^{i}\right)  _{i}}
\left(
\sum_{i=1}^{r}\langle\phi^{\rho_{i}}
|_{RA}\left(  V_{R}^{i}\right)  ^{\dag}V_{R}^{i}|\phi^{\rho_{i}}\rangle_{RA}
+2
\sum_{i<j}\mathfrak{R}\!\left[\langle\phi^{\rho_{j}}
|_{RA}\left(  V_{R}^{j}\right)  ^{\dag}V_{R}^{i}|\phi^{\rho_{i}}\rangle_{RA}
\right]\right) \\
& = \frac{1}{r^2}
\sup_{\left(  V^{i}\right)  _{i}}
\left(
r
+2
\sum_{i<j}\mathfrak{R}\!\left[\langle\phi^{\rho_{j}}
|_{RA}\left(  V_{R}^{j}\right)  ^{\dag}V_{R}^{i}|\phi^{\rho_{i}}\rangle_{RA}
\right]\right) \\
& = 
\frac{1}{r}
+\frac{2}{r^2}
\sup_{\left(  V^{i}\right)  _{i}}
\sum_{i<j}\mathfrak{R}\!\left[\langle\phi^{\rho_{j}}
|_{RA}\left(  V_{R}^{j}\right)  ^{\dag}V_{R}^{i}|\phi^{\rho_{i}}\rangle_{RA}
\right].
\end{align}
Combining this with the definition of $F_S$ in \cref{def:secrecy_based_multi} and performing some algebra then concludes the proof.

\begin{lemma} \label{lem:assist_uniform_mix_unitaries}
Suppose that the square matrix $A$ is a contraction, that is, $\left \|A\right \|_{\infty} \leq 1$. Then it can be written as a uniform mixture of two unitaries: $A = (U_1 + U_2)/2$, where $U_1$ and $U_2$ are unitaries.
\end{lemma}

\begin{proof} The basic construction was given in the proof of~\cite[Theorem~1]{unitaries_convex_combinations85}. We provide a proof for completeness. Under the assumptions stated,
$A$ has a polar decomposition as $A=UM$, where $U$ is a unitary and $M$ is a PSD matrix with eigenvalues between zero and one. Pick the following unitaries:
$
    V_1 \coloneqq M + i( I-M^2)^{1/2} $ and  $V_2\coloneqq  V_1^\dagger$.
With that choice we have that $M= (V_1 +V_2)/2$, and so 
$
    A=\frac{1}{2}\left( U V_1 + U V_2\right),
$
concluding the proof that $A$ can be written as a uniform mixture of two unitaries.
\end{proof}

\subsection{Proof of Proposition~\ref{prop:uniform_cont_secrecy_multi_F} (Uniform continuity of secrecy-based multivariate fidelity)} \label{proof:uniform_cont_secrecy_multi_F}
    
    By \cref{def:secrecy_based_multi},
    \begin{align}
       &  \left | F_S(\rho_1, \ldots, \rho_r) - F_S(\sigma_1, \ldots, \sigma_2) \right | \notag \\
       &= \frac{r}{r-1} \left| S(\rho_1, \ldots, \rho_r )^2 - S(\sigma_1, \ldots, \sigma_r )^2  \right|
       \label{eq:F_S-bound-1} \\
       &=  \frac{r}{r-1} \left| S(\rho_1, \ldots, \rho_r ) + S(\rho_1, \ldots, \rho_r )\right| \left| S(\rho_1, \ldots, \rho_r )- S(\rho_1, \ldots, \rho_r ) \right|\\
       & \leq  \frac{2r}{r-1} \left| S(\rho_1, \ldots, \rho_r )- S(\rho_1, \ldots, \rho_r ) \right|,
       \label{eq:F_S-bound-last} 
    \end{align}
    where the last inequality follows because $S(\rho_1, \ldots,\rho_r) \leq 1$ for every tuple $(\rho_i)_{i=1}^r$ of states.

    By using~\eqref{eq:secrecy_measure}, we have that
    \begin{align}
        \left| S(\rho_1, \ldots, \rho_r )- S(\sigma_1, \ldots, \sigma_r ) \right| & \leq \sup_{\omega \in \mathscr{D}} \left|  \frac{1}{r} \sum_{i=1}^r \left( F(\rho_i, \omega) - F(\sigma_i, \omega) \right)\right| \\
         & \leq \sup_{\omega \in \mathscr{D}} \frac{1}{r} \sum_{i=1}^r  \left|  F(\rho_i, \omega) - F(\sigma_i, \omega)\right| 
    \end{align}
For an arbitrary $\omega \in \mathscr{D}$, consider that
\begin{align}
     |F(\rho_i,\omega) - F(\sigma_i,\omega) |
    & =\left | \left[ 1-F(\sigma_i,\omega)\right] - \left[ 1-F(\rho_i,\omega)\right] \right | \\
    &=\frac{1}{2} \left| d_B(\sigma_i, \omega)^2 - d_B(\rho_i, \omega)^2 \right| \\ 
    &=\frac{1}{2} \left[ d_B(\sigma_i, \omega)+ d_B(\rho_i, \omega) \right] \left|d_B(\sigma_i, \omega) -d_B(\rho_i, \omega) \right| \\ 
    & \leq \frac{1}{2}\left[ d_B(\sigma_i, \omega)+ d_B(\rho_i, \omega) \right] d_B(\sigma_i, \rho_i) \\
    & \leq \sqrt{2} d_B(\sigma_i, \rho_i),
\end{align}
  where the penultimate inequality follows from triangular inequality of the Bures distance and the last inequality because the Bures distance is bounded by $\sqrt{2}$.
Together with that, we have 
\begin{equation}
     \left| S(\rho_1, \ldots, \rho_r )- S(\sigma_1, \ldots, \sigma_r ) \right| \leq \frac{\sqrt{2}}{r} \sum_{i=1}^r d_B(\sigma_i, \rho_i) \leq \sqrt{2} \varepsilon,
     \label{eq:final-S-cont-bound}
\end{equation}
  where the last inequality follows from the assumption $\frac{1}{r}\sum_{i=1}^r d_B(\rho_i,\sigma_i) \leq  \varepsilon$. Finally combining~\eqref{eq:F_S-bound-1}--\eqref{eq:F_S-bound-last} and~\eqref{eq:final-S-cont-bound}, we conclude that 
  \begin{equation}
       |F_S(\rho_1, \ldots, \rho_r) - F_S(\sigma_1, \ldots, \sigma_2)|  \leq \frac{r}{r-1} 2 \sqrt{2} \varepsilon,
  \end{equation}
  completing the proof.

\subsection{Proof of Theorem~\ref{thm:properties_secrecy_multi_F} (Properties of secrecy-based multivariate fidelity)}
\label{proof:Properties_secrecy_fidelity}

Note that the secrecy-based multivariate fidelity satisfies reduction to classical average pairwise fidelity, faithfulness, and orthogonality by \cref{thm:SDP_fidelity_pairwise_upper_and_lower} and the fact that the average pairwise Holevo and Uhlmann fidelities satisfy these properties, as stated in \cref{Prop:properties_average_pairwise_z}.

\medskip
\noindent \underline{Symmetry:} This follows by the definition of the secrecy measure in~\eqref{eq:secrecy_measure}. 

\bigskip

\noindent \underline{Data processing:}
    Consider that
    \begin{align}
      \sup_{\sigma  \in \mathscr{D}}\frac{1}{r}\sum_{i=1}^{r}F(  \rho_{i}, \sigma) & \leq   \sup_{\sigma  \in \mathscr{D}}\frac{1}{r}\sum_{i=1}^{r}F\!\left( \cN( \rho_{i}), \cN(\sigma) \right) \\
      & \leq  \sup_{\sigma'  \in \mathscr{D}}\frac{1}{r}\sum_{i=1}^{r}F\!\left( \cN( \rho_{i}), \sigma' \right),
    \end{align}
    where the first inequality follows from data processing of the Uhlmann fidelity (\cref{prop:properties_bivariate_F}, Property 2); and the last inequality by supremizing over a larger set. 
    With that we arrive at 
     \begin{equation}
            S(\rho_1, \ldots, \rho_r) \leq S\!\left( \cN(\rho_1), \ldots, \cN(\rho_r) \right).
        \end{equation}
By recalling \cref{def:secrecy_based_multi}, we conclude that 
 \begin{equation}
            F_S(\rho_1, \ldots, \rho_r) \leq F_S\!\left( \cN(\rho_1), \ldots, \cN(\rho_r) \right).
        \end{equation}

\bigskip 
\noindent \underline{Direct-sum property:}
 A purification of the state $\sum_{x \in \cX} p(x) |x\rangle\!\langle x| \otimes \rho_i^x$ is 
\begin{equation}\label{eq:purification_alltogether_i}
      |\Psi^i\rangle_{X\bar{X}RA} \coloneqq  \sum_x \sqrt{p(x)} |x\rangle_X\otimes |x\rangle_{\bar{X}} \otimes | \Phi^{\rho_i^x}\rangle_{RA},
    \end{equation}
    where $| \Phi^{\rho_i^x}\rangle$ is a purification of the state $\rho_i^x$. 
For all $x \in \cX$, let $(V^{x,i}_{R})_{i=1}^r$ be a tuple of unitaries  achieving the optimum value of $S\!\left(\rho_1^x, \ldots, \rho_r^x \right)^2$ related to the tuple of states $\left( \rho_{i}^x \right)_{i}$ in~\eqref{eq:secrecy-as-pairwise}.
That is, 
\begin{equation}\label{eq:optimum_value_S_2}
    S\!\left(\rho_1^x, \ldots, \rho_r^x \right)^2= \frac{1}{r^2}  \sum_{i,j=1}^r  \left\langle \Phi^{\rho_i^x} \right|_{RA}\left(V^{x,i}_{R}\right)^\dagger  V^{x,j}_{R} \left| \Phi^{\rho_j^x} \right \rangle_{RA},
\end{equation}
Then, choose the controlled unitary 
\begin{equation}\label{eq:controlled_unitary_choice}
    V^i_{\bar{X}R } \coloneqq  \sum_x  |x\rangle\!\langle x|_{\bar{X}} \otimes  V^{x,i}_{R } .
\end{equation}
With the operators defined, consider that 
\begin{align}
    &\frac{1}{r^2} \sum_{i,j=1}^r \left\langle \Psi^{i} \right|_{X\bar{X}RA}\left(V^i_{\bar{X}R}\right)^\dagger V^j_{\bar{X}R} \left| \Psi^{j} \right \rangle_{X\bar{X}RA} \notag \\
    &= \frac{1}{r^2} \sum_{i,j=1}^r \left(\sum_x \sqrt{p(x)} \langle x|_X \otimes\langle x|_{\bar{X}} \otimes \langle \Phi^{\rho_i^x}|_{RA}\right) \left( I_X\otimes \sum_{x'}  |x'\rangle\!\langle x'|_{\bar{X}} \otimes  \left(V^{x',i}_{R}  \right)^\dagger \right)  \times \notag \\
    & \qquad\qquad\left( I_X \otimes \sum_{x''}  |x''\rangle\!\langle x''|_{\bar{X}} \otimes   V^{x'',j}_{R}\right) \left(\sum_{y} \sqrt{p(y)} |y\rangle_X \otimes |y\rangle_{\bar{X}} \otimes | \Phi^{\rho_i^y}\rangle_{RA} \right)   \\ 
    &= \frac{1}{r^2} \sum_{i,j=1}^r \sum_x p(x) \left\langle \Phi^{\rho_i^x} \right|_{RA}\left(V^{x,i}_{R}\right)^\dagger  V^{x,j}_{R} \left| \Phi^{\rho_j^x} \right \rangle_{RA}\\
    &=\sum_x p(x)\frac{1}{r^2} \sum_{i,j=1}^r  \left\langle \Phi^{\rho_i^x} \right|_{RA}\left(V^{x,i}_{R}\right)^\dagger  V^{x,j}_{R} \left| \Phi^{\rho_j^x} \right \rangle_{RA} \\
    &= \sum_x p(x) S(\rho_1^x, \ldots, \rho_r^x)^2,
\end{align}
where the first equality follows by substituting the definitions in~\eqref{eq:purification_alltogether_i} and~\eqref{eq:controlled_unitary_choice}, and the last equality from~\eqref{eq:optimum_value_S_2}.

This shows that the controlled unitary in~\eqref{eq:controlled_unitary_choice} is a possible candidate for the unitary in the optimization of $ S\!\left(\sum_{x \in \cX} p(x) |x\rangle\!\langle x| \otimes \rho_1^x, \ldots, \sum_{x \in \cX} p(x)  |x\rangle\!\langle x| \otimes \rho_r^x\right)^2$. 
With that, we conclude the following inequality:
\begin{equation}\label{eq:one_sided_concavity}
      S\!\left(\sum\nolimits_{x \in \cX} p(x) |x\rangle\!\langle x| \otimes \rho_1^x, \ldots, \sum\nolimits_{x \in \cX} p(x)  |x\rangle\!\langle x| \otimes \rho_r^x\right)^2 \\
         \geq  \sum_{x \in \cX} p(x) S(\rho_1^x, \ldots, \rho_r^x)^2.  
\end{equation}

Now we prove the opposite inequality. Let $\Delta(\cdot)\coloneqq \sum_{x}
|x\rangle\!\langle x|(\cdot)|x\rangle\!\langle x|$ denote the dephasing channel.
Recall that
\begin{multline}
S\!\left(  \sum\nolimits_{x}p(x)|x\rangle\!\langle x|\otimes\rho_{1}^{x},\ldots
,\sum\nolimits_{x}p(x)|x\rangle\!\langle x|\otimes\rho_{r}^{x}\right)^2  \\
=\left[  \sup_{\sigma_{XA}}\frac{1}{r}\sum_{i=1}^{r}F\!\left(  \sum
_{x}p(x)|x\rangle\!\langle x|\otimes\rho_{i}^{x},\sigma_{XA}\right)  \right]
^{2}.
\end{multline}
Let $\sigma_{XA}$ be an arbitrary state, and define the probability
distribution $q$ and the tuple $\left(  \sigma^{x}\right)  _{x}$ of states in terms of
\begin{equation}
\left(  \Delta\otimes\operatorname{id}\right)  \left(  \sigma_{XA}\right)
=\sum_{x}q(x)|x\rangle\!\langle x|\otimes\sigma^{x},
\end{equation}
where $\operatorname{id}$ is the identity channel. 
Then
\begin{align}
& \frac{1}{r}\sum_{i=1}^{r}F\!\left(  \sum_{x}p(x)|x\rangle\!\langle x|\otimes
\rho_{i}^{x},\sigma_{XA}\right)  \nonumber\\
& \leq\frac{1}{r}\sum_{i=1}^{r}F\!\left(  \left(  \Delta\otimes\operatorname{id}
\right)  \left(  \sum_{x}p(x)|x\rangle\!\langle x|\otimes\rho_{i}^{x}\right)
,\left(  \Delta\otimes\operatorname{id}\right)  \left(  \sigma_{XA}\right)
\right)  \\
& =\frac{1}{r}\sum_{i=1}^{r}F\!\left(  \sum_{x}p(x)|x\rangle\!\langle
x|\otimes\rho_{i}^{x},\sum_{x}q(x)|x\rangle\!\langle x|\otimes\sigma^{x}\right)
\\
& =\frac{1}{r}\sum_{i=1}^{r}\sum_{x}\sqrt{p(x)q(x)}F\!\left(  \rho_{i}
^{x},\sigma^{x}\right)  \\
& =\sum_{x}\sqrt{p(x)q(x)}\frac{1}{r}\sum_{i=1}^{r}F\!\left(  \rho_{i}
^{x},\sigma^{x}\right)  .
\end{align}
Now applying the Cauchy--Schwarz inequality, we conclude that
\begin{align}
& \sum_{x}\sqrt{p(x)q(x)}\frac{1}{r}\sum_{i=1}^{r}F\!\left(  \rho_{i}^{x}
,\sigma^{x}\right) \notag  \\
& \leq\sqrt{\sum_{x}q(x)}\sqrt{\sum_{x}p(x)\left[  \frac{1}{r}\sum_{i=1}
^{r}F\!\left(  \rho_{i}^{x},\sigma^{x}\right)  \right]  ^{2}}\\
& =\sqrt{\sum_{x}p(x)\left[  \frac{1}{r}\sum_{i=1}^{r}F\!\left(  \rho_{i}
^{x},\sigma^{x}\right)  \right]  ^{2}}\\
& \leq\sqrt{\sum_{x}p(x)S(\rho_{1}^{x},\ldots,\rho_{r}^{x})^{2}}.
\end{align}
We have thus shown that the following inequality holds for every state $\sigma_{XA}$:
\begin{equation}
\left[  \frac{1}{r}\sum_{i=1}^{r}F\!\left(  \sum_{x}p(x)|x\rangle\!\langle
x|\otimes\rho_{i}^{x},\sigma_{XA}\right)  \right]  ^{2}\leq\sum_{x}
p(x)S(\rho_{1}^{x},\ldots,\rho_{r}^{x})^{2}.
\end{equation}
By taking the optimization over every state $\sigma_{XA}$, we conclude that
\begin{equation} \label{eq:two_sided_con}
S\!\left(  \sum\nolimits_{x}p(x)|x\rangle\!\langle x|\otimes\rho_{1}^{x},\ldots
,\sum\nolimits_{x}p(x)|x\rangle\!\langle x|\otimes\rho_{r}^{x}\right)^{2}  \leq\sum
_{x}p(x)S(\rho_{1}^{x},\ldots,\rho_{r}^{x})^{2}.
\end{equation}
Note that equality is achieved by picking $\sigma_{XA} = \sum_x p(x) |x\rangle\!\langle x| \otimes \sigma^x$, where $\sigma^x$ is an optimal choice for $S(\rho_{1}^{x},\ldots,\rho_{r}^{x})^{2}$.

Combining~\eqref{eq:one_sided_concavity} and~\eqref{eq:two_sided_con}, it follows that $S(\cdot)^2$ satisfies the direct-sum property. 
Then, applying the direct-sum property and the fact that $p(x)$ is a probability distribution, we find that
\begin{align}
& rS\!\left(  \sum\nolimits_{x}p(x)|x\rangle\!\langle x|\otimes\rho_{1}^{x},\ldots
,\sum\nolimits_{x}p(x)|x\rangle\!\langle x|\otimes\rho_{r}^{x}\right)^{2} -1 \notag \\
& = r\left(\sum
_{x}p(x)S(\rho_{1}^{x},\ldots,\rho_{r}^{x})^{2}\right)-1\\
& = \sum
_{x}p(x)\left(rS(\rho_{1}^{x},\ldots,\rho_{r}^{x})^{2}-1\right).
\end{align}
Finally, by recalling~\cref{def:secrecy_based_multi} and dividing the above by $(r-1)$, we conclude the proof of the direct-sum property of $F_S$.

\subsection{Proof of Proposition~\ref{prop:coarse_graining_secrecy_F} (Coarse-graining property of secrecy-based fidelity)} \label{proof:coarse_graining_secrecy_F}

Let $(V^i_R)_{i=1}^r$ be a tuple of unitaries, and let $|\phi^{\rho_i} \rangle_{RA}$ be a purification of $\rho_i$ for all $i \in [r+m]$. Choose $V^{r+1}_R$ to be a unitary such that 
\begin{multline}
 \sum_{i=1}^{r} \mathfrak{R}\! \left[\left\langle \phi^{\rho_{r+1}} \right|_{RA}\left(V^{r+1}_{R }\right)^\dagger V^i_{R} \left| \phi^{\rho_i} \right \rangle_{RA}\right]  \\ =  \mathfrak{R}\! \left[\left\langle \phi^{\rho_{r+1}} \right|_{RA}\left(V^{r+1}_{R }\right)^\dagger \left(  \sum_{i=1}^{r} V^i_{R} \left| \phi^{\rho_i} \right \rangle_{RA}\right)\right]  \geq 0,
\end{multline}
which is possible by adjusting the unitary $V^{r+1}_R$ with a global phase.
With this choice, we find that 
\begin{align}
    & \sum_{1\leq i < j\leq r}\mathfrak{R}\! \left[\left\langle \phi^{\rho_j} \right|_{RA}\left(V^j_{R }\right)^\dagger V^i_{R} \left| \phi^{\rho_i} \right \rangle_{RA}\right]  \notag \\
    & \leq \sum_{1\leq i < j\leq r}\mathfrak{R}\! \left[\left\langle \phi^{\rho_j} \right|_{RA}\left(V^j_{R }\right)^\dagger V^i_{R} \left| \phi^{\rho_i} \right \rangle_{RA}\right] + \sum_{i=1}^{r} \mathfrak{R}\! \left[\left\langle \phi^{\rho_{r+1}} \right|_{RA}\left(V^{r+1}_{R }\right)^\dagger V^i_{R} \left| \phi^{\rho_i} \right \rangle_{RA}\right] \\ 
    &= \sum_{1\leq i < j\leq r+1} \mathfrak{R}\! \left[\left\langle \phi^{\rho_j} \right|_{RA}\left(V^j_{R }\right)^\dagger V^i_{R} \left| \phi^{\rho_i} \right \rangle_{RA}\right].
\end{align}
Following the same approach, we can successively choose unitaries in the tuple $(V^k_R)_{k=r+2}^{r+m}$  such that the following inequality holds for all $k \geq r+2$:
\begin{equation}
 \sum_{i=1}^{k-1} \mathfrak{R}\! \left[\left\langle \phi^{\rho_k} \right|_{RA}\left(V^k_{R }\right)^\dagger V^i_{R} \left| \phi^{\rho_i} \right \rangle_{RA}\right] = \mathfrak{R}\! \left[\left\langle \phi^{\rho_k} \right|_{RA}\left(V^k_{R }\right)^\dagger\left(\sum_{i=1}^{k-1}  V^i_{R} \left| \phi^{\rho_i} \right \rangle_{RA}\right)\right] \geq 0.
\end{equation}
This selection process leads to the following inequality holding:
\begin{equation}
    \sum_{1\leq i < j\leq r}\mathfrak{R}\! \left[\left\langle \phi^{\rho_j} \right|_{RA}\left(V^j_{R }\right)^\dagger V^i_{R} \left| \phi^{\rho_i} \right \rangle_{RA}\right] \leq \sum_{1\leq i < j\leq r + m} \mathfrak{R}\! \left[\left\langle \phi^{\rho_j} \right|_{RA}\left(V^j_{R }\right)^\dagger V^i_{R} \left| \phi^{\rho_i} \right \rangle_{RA}\right].
\end{equation}
Then  bounding the right-hand side from above by taking the supremum over tuples $(V^i_R)_{i=1}^{r+m}$ of unitaries and using~\eqref{eq:secrecy-fid-as-pairwise}, we have that
\begin{multline}
     \sum_{1\leq i < j\leq r}\mathfrak{R}\! \left[\left\langle \phi^{\rho_j} \right|_{RA}\left(V^j_{R }\right)^\dagger V^i_{R} \left| \phi^{\rho_i} \right \rangle_{RA}\right] \\  \leq \frac{(r+m)(r+m-1)}{2} F_S(\rho_1,\ldots, \rho_r, \ldots, \rho_{r+m}).
\end{multline}
Taking the supremum of the left-hand side of this inequality over tuples $(V^k_R)_{k=1}^r$ of unitaries and using~\eqref{eq:secrecy-fid-as-pairwise}  again, we arrive at 
\begin{equation}
\frac{r(r-1)}{2} F_S(\rho_1,\ldots, \rho_r)  \leq  \frac{(r+m)(r+m-1)}{2} F_S(\rho_1,\ldots, \rho_r, \ldots, \rho_{r+m}),
\end{equation}
from which the claim of the proposition follows.

\subsection{Proof of Theorem~\ref{thm:SDP_fidelity_pairwise_upper_and_lower} (Relations between multivariate fidelities)} 

\label{proof:SDP_fidelity_bounds_with_pairwise}

We prove the chain of inequalities in the theorem statement from left to right. 
\bigskip

\noindent \underline{Proof of $F_{H}(\rho_1, \ldots,\rho_r) \leq F_S(\rho_1, \ldots, \rho_r)$:}
\medskip

Recall the Petz~\cite{P86} and sandwiched~\cite{muller2013quantum,wilde2014strong}
R\'enyi relative entropies, defined for $\alpha \in (0,1)\cup(1,\infty)$ respectively as follows:
\begin{align}
D_{\alpha}(\rho\Vert\sigma)  & \coloneqq \frac{1}{\alpha-1}\ln\operatorname{Tr}
\left[\rho^{\alpha}\sigma^{1-\alpha}\right],\\
\widetilde{D}_{\alpha}(\rho\Vert\sigma)  & \coloneqq \frac{1}{\alpha-1}\ln
\operatorname{Tr}\!\left[\left(  \sigma^{\left(  1-\alpha\right)  /2\alpha}\rho
\sigma^{\left(  1-\alpha\right)  /2\alpha}\right)  ^{\alpha}\right].
\end{align}
Also note that $\widetilde{D}_{\frac{1}{2}}(\rho\Vert\sigma) = -2 \ln F(\rho,\sigma)$ and ${D}_{\frac{1}{2}}(\rho\Vert\sigma) = -2 \ln F_H(\rho,\sigma)$.
The following inequality is known for $\alpha\in\left(  0,1\right)  $ from \cite[Lemma~3]{datta2014limit}:
\begin{equation}
\widetilde{D}_{\alpha}(\rho\Vert
\sigma)\leq D_{\alpha}(\rho\Vert\sigma).\label{eq:petz-sandwiched-ineqs}
\end{equation}
Let
\begin{equation}
\rho_{XA}\coloneqq \frac{1}{r}\sum_{i=1}^{r}|i\rangle\!\langle i|_{X}\otimes\rho_{i}.
\end{equation}
Plugging into~\eqref{eq:petz-sandwiched-ineqs}, taking  infima with respect to a state
$\sigma_A$, and setting $\alpha=1/2$ then gives
\begin{equation}
\inf_{\sigma_{A}}\widetilde{D}_{\frac{1}{2}}(\rho
_{XA}\Vert\rho_{X}\otimes\sigma_{A})\leq\inf_{\sigma_{A}}D_{\frac{1}{2}}
(\rho_{XA}\Vert\rho_{X}\otimes\sigma_{A}).\label{eq:fidelity-holevo-infos}
\end{equation}
Consider that
\begin{align}
\label{eq:sandwiched-to-secrecy-1}
\inf_{\sigma_{A}}\widetilde{D}_{\frac{1}{2}}(\rho_{XA}\Vert\rho_{X}
\otimes\sigma_{A})  & =\inf_{\sigma_{A}}\left[  -2\ln F\!\left(  \rho_{XA}
,\rho_{X}\otimes\sigma_{A}\right)  \right]  \\
& =-2\ln\sup_{\sigma_{A}}F\!\left(  \rho_{XA},\rho_{X}\otimes\sigma_{A}\right)
\\
& =-2\ln\sup_{\sigma}\frac{1}{r}\sum_{i=1}^{r}F(  \rho_{i},\sigma) \\
& = -2 \ln S(\rho_1, \ldots, \rho_r)
\label{eq:sandwiched-to-secrecy-last}
.
\end{align}
For the penultimate equality, we used the direct-sum property of the fidelity in~\eqref{eq:CQequality}.

Now applying  \cite[Corollary~8]{gupta2015multiplicativity} (which holds for $\alpha
\in(0,1)$), consider that
\begin{align}
& \inf_{\sigma_{A}}D_{\frac{1}{2}}(\rho_{XA}\Vert\rho_{X}\otimes\sigma
_{A})\nonumber\\
& =-\ln\operatorname{Tr}\!\left[  \left(  \operatorname{Tr}_{X}[\rho_{X}
^{1/2}\rho_{XA}^{1/2}]\right)  ^{2}\right] \label{eq:ln-SH-explicit-form} \\
& =-\ln\operatorname{Tr}\!\left[  \left(  \operatorname{Tr}_{X}\!\left[  \left(
\sum_{i}\frac{1}{r}|i\rangle\!\langle i|\otimes I_d\right)  ^{1/2}\left(  \sum
_{j}\frac{1}{r}|j\rangle\!\langle j|\otimes\rho_{j}\right)  ^{1/2}\right]
\right)  ^{2}\right]  \\
& =-\ln\operatorname{Tr}\!\left[  \left(  \operatorname{Tr}_{X}\!\left[  \sum
_{i}\frac{1}{r}|i\rangle\!\langle i|\otimes\rho_{i}^{1/2}\right]  \right)
^{2}\right]  \\
& =-\ln\operatorname{Tr}\!\left[  \left(  \sum_{i}\frac{1}{r}\rho_{i}
^{1/2}\right)  ^{2}\right]  \\
& =-\ln\operatorname{Tr}\!\left[  \left(  \sum_{i}\frac{1}{r}\rho_{i}
^{1/2}\right)  \left(  \sum_{j}\frac{1}{r}\rho_{j}^{1/2}\right)  \right]  \\
& =-\ln\operatorname{Tr}\!\left[  \frac{1}{r^{2}}\sum_{i}\rho_{i}+\frac{1}
{r^{2}}\sum_{i\neq j}\rho_{i}^{1/2}\rho_{j}^{1/2}\right]  \\
& =-\ln\!\left(  \frac{1}{r}+\frac{1}{r^{2}}\sum_{i\neq j}\operatorname{Tr}
\!\left[  \rho_{i}^{1/2}\rho_{j}^{1/2}\right]  \right)  \\
& =-\ln\!\left(  \frac{1}{r}+\frac{2}{r^{2}}\sum_{i<j}F_{H}(\rho_{i},\rho
_{j})\right)  \\
& =-\ln\!\left(  \frac{1}{r}+\frac{r-1}{r}\left[  \frac{2}{r\left(  r-1\right)
}\sum_{i<j}F_{H}(\rho_{i},\rho_{j})\right]  \right)  \\
& =-\ln\!\left(  \frac{1}{r}+\frac{r-1}{r}F_{H}\!\left(\rho_1, \ldots, \rho_r \right)\right)  . \label{eq:relation_to_fidelity_d_1/2}
\end{align}
Combining~\eqref{eq:fidelity-holevo-infos},~\eqref{eq:sandwiched-to-secrecy-last}, and~\eqref{eq:relation_to_fidelity_d_1/2}, we get
\begin{align}
-2\ln S(\rho_1, \ldots, \rho_r)  
& \leq-\ln\!\left(  \frac{1}{r}+\frac{r-1}{r}F_{H}\!\left(\rho_1, \ldots, \rho_r\right)\right)  ,
\end{align}
which is the same as
\begin{align}
S(\rho_1, \ldots, \rho_r)^{2} 
 \geq\frac{1}{r}+\frac{r-1}{r}F_{H}\!\left(\rho_1, \ldots, \rho_r\right).
 \end{align}
 Then, by rearranging the terms in the above inequality, we arrive at the required bound.

 We also provide an alternative proof for this with the use of~\eqref{eq:secrecy-fid-as-pairwise}. Choose $V^{i}_R =I$  for all $i \in [r]$ and purifications to be canonical purifications where $|\phi^{\rho_i}\rangle = \left(\rho_i^{1/2} \otimes I \right) | \Gamma \rangle $, where 
 $| \Gamma \rangle \coloneqq \sum_{i} |i \rangle |i\rangle$ for all $i \in [r]$. 
 With that choice
 \begin{align}
    \mathfrak{R}\! \left[\left\langle \phi^{\rho_j} \right|_{RA}\left(V^j_{R }\right)^\dagger V^i_{R} \left| \phi^{\rho_i} \right \rangle_{RA}\right] \notag 
   & =\left\langle \Gamma \middle |   \left(\rho_i^{1/2} \otimes I \right) \left(\rho_j^{1/2} \otimes I \right) \middle| \Gamma \right\rangle \\
   & =\Tr\!\left[\rho_i^{1/2}\rho_j^{1/2} \right].
 \end{align}
 Since this choice is a feasible candidate for the optimization in~\eqref{eq:secrecy-fid-as-pairwise}, 
we arrive at
 \begin{equation}
     \frac{2}{r(r-1)}\sum_{i <j} \Tr\!\left[\rho_i^{1/2}\rho_j^{1/2} \right] \leq F_S(\rho_1, \ldots, \rho_r),
 \end{equation}
 which is the desired inequality.

\medskip
\noindent \underline{Proof of $F_{S}(\rho_1, \ldots,\rho_r) \leq F_{\operatorname{SDP}}(\rho_1, \ldots, \rho_r)$:}

\medskip

For a tuple of unitaries $\left(V^{i}_R\right)_{i=1}^r$, define
\begin{equation}
    M \coloneqq \sum_{i} \langle i| \otimes  V^i_{R}.
\end{equation}
Then
\begin{equation}
    \sum_{i,j} |i\rangle\!\langle j| \otimes \left(V^i_{R } \right)^\dag V^j_{R} = M^\dag M\geq0.
\end{equation}
Using the fact that $V^i_R$ is a unitary, leading to $\left(V^i_{R } \right)^\dag V^i_{R}=I_d $ for all $i \in [r]$, we see that
\begin{equation}
    I_r \otimes I_d+ \sum_{i \neq j} |i\rangle\!\langle j| \otimes \left(V^i_{R} \right)^\dag V^j_{R } = \sum_{i,j} |i\rangle\!\langle j| \otimes \left(V^i_{R } \right)^\dag V^j_{R} \geq 0 .
\end{equation}
Then the choice $K_{ij}= \left(V^i_{R} \right)^\dag V^j_{R }$ for $i\neq j \in [r]$ and $K_{ii}=I_d$ for $i \in [r]$ is a possible candidate for the optimization for $F_{\operatorname{SDP}}(\rho_1,\ldots, \rho_r)$ in~\eqref{eq:SDP_fidelity_with_pairwise_structure}. To this end, we have 
\begin{equation}
    \frac{2}{r(r-1)}\sum_{i<j} \mathfrak{R}\! \left[\left\langle \phi^{\rho_j} \right|_{RA}\left(V^j_{R }\right)^\dagger V^i_{R} \left| \phi^{\rho_i} \right \rangle_{RA}\right] \leq F_{\operatorname{SDP}}(\rho_1,\ldots, \rho_r).
\end{equation}
Supremizing over $\left(V^{i}_R\right)_{i=1}^r$ and 
recalling~\eqref{eq:secrecy-fid-as-pairwise}, the desired inequality follows:
\begin{equation}
     F_S(\rho_1, \ldots, \rho_r)  \leq F_{\operatorname{SDP}}(\rho_1,\ldots, \rho_r).
\end{equation}

\medskip

\noindent\underline{Proof of $F_{\operatorname{SDP}}(\rho_1, \ldots,\rho_r) \leq F_U(\rho_1, \ldots, \rho_r)$}: 

\medskip

  Suppose that $X_{ij}$ satisfies $X_{ji}=X_{ij}^\dagger$, along with the following inequality:
    \begin{equation}
        \sum_{i=1}^r |i\rangle\!\langle i| \otimes \rho_i + \sum_{i \neq j} |i\rangle\!\langle j| \otimes X_{ij} \geq 0.
    \end{equation}
 Then for each pair of states $\rho_i,\rho_j$ such that $i < j$, we also have 
 \begin{equation}
     \begin{bmatrix}
\rho_i & X_{ij} \\
X_{ij}^\dagger & \rho_j
\end{bmatrix} \geq 0.
 \end{equation}
 Using this together with the dual SDP of Uhlmann fidelity given in~\eqref{eq:bivariate_fid_dual}, we get 
 \begin{equation}
     2 \sum_{i <j} \mathfrak{R}\left[ \Tr[X_{ij}]\right] \leq 2 \sum_{i <j}  F(\rho_i, \rho_j).
 \end{equation}
 Taking the supremum over all $(X_{ij})_{ij}$ satisfying the inequality
 $\sum_{i=1}^r |i\rangle\!\langle i| \otimes \rho_i + \sum_{i \neq j} |i\rangle\!\langle j| \otimes X_{ij} \geq 0$, we obtain 
 \begin{equation}
     r(r-1) F_{\operatorname{SDP}}(\rho_1, \ldots, \rho_r) \leq  2 \sum_{i <j}^r  F(\rho_i, \rho_j),
 \end{equation}
 which implies the desired inequality after a rearrangement of terms.

 \medskip 
 \noindent\underline{Proof of $F_U(\rho_1,\ldots, \rho_r) \leq \sqrt{F_H(\rho_1,\ldots,\rho_r)}$:}

\medskip
 
Notice that for states $\rho$ and $\sigma$, we have
\begin{equation}F(\rho,\sigma) \leq \sqrt{F_H(\rho,\sigma)},
\end{equation}
which follows by adapting \cite[Theorem~6]{audenaert2008asymptotic} with the choice $s=1/2$, $A=\rho$, and $B=\sigma$. 
This gives 
\begin{align}
F_U(\rho_1,\ldots, \rho_r) 
    &= \frac{2}{r\left(  r-1\right)  }\sum_{i<j}F(\rho_{i},\rho_{j}) \\
    &\leq \frac{2}{r\left(  r-1\right)  }\sum_{i<j}\sqrt{F_H(\rho_{i},\rho_{j})}  \\
    &\leq \sqrt{ \frac{2}{r\left(  r-1\right)  }\sum_{i<j}F_H(\rho_{i},\rho_{j}) } \\
    &= \sqrt{F_H(\rho_1,\ldots,\rho_r)},
\end{align}
where the last inequality holds due to concavity of the square root function.

\subsection{Proof of Proposition~\ref{prop:lower_bound_multi_SDP} (Lower bound on multivariate SDP fidelity)} \label{Sec:Proof_lower_bound_multi_SDP}
 Fix an arbitrary permutation $\pi \in S_r$. For all $1 \leq i \leq \lfloor r/2 \rfloor$, we have from~\eqref{eq:multivariate_fid_sdp_dual} that
    \begin{align}
        F(\rho_{\pi(i)}, \rho_{\pi(i+\lfloor r/2 \rfloor)}) &= \sup_{Z_i} \left \{   \mathfrak{R}\!\left[ \Tr[Z_i]\right]: \begin{bmatrix}
        \rho_{\pi(i)} & Z_i \\
        Z_i^\dagger & \rho_{\pi(i+\lfloor r/2 \rfloor)}
        \end{bmatrix} \geq 0 \right \}. \label{eq:sdp-dual-bivariate-fidelity}
    \end{align}
    Consider arbitrary operators $Z_i$ satisfying the constraint of~\eqref{eq:sdp-dual-bivariate-fidelity}:
    \begin{equation}\label{eq:psd-rho-pi-i}
        \begin{bmatrix}
        \rho_{\pi(i)} & Z_i \\
        Z_i^\dagger & \rho_{\pi(i+\lfloor r/2 \rfloor)}
        \end{bmatrix} \geq 0, \qquad 1 \leq i \leq \lfloor r/2 \rfloor.
    \end{equation}
    Define operators $X_{ij}$ for $1 \leq i \neq j \leq r$ as
    \begin{align}
        X_{ij} \coloneqq 
        \begin{cases}
            Z_i & \text{if } 1\leq i \leq \lfloor r/2 \rfloor \text{ and } j=i+\lfloor r/2 \rfloor, \\
            Z_j^{\dagger} & \text{else if } i=j+\lfloor r/2 \rfloor \text{ and } 1\leq j \leq \lfloor r/2 \rfloor, \\
            0 & \text{otherwise},
        \end{cases}
    \end{align}
    where $0$ denotes the zero operator.
    The positive semi-definiteness of the matrices in~\eqref{eq:psd-rho-pi-i} implies that 
    \begin{align}
        \sum_{i=1}^r |i\rangle\!\langle i| \otimes \rho_{\pi(i)} + \sum_{i \neq j} |i\rangle\!\langle j| \otimes X_{ij} \geq 0.
    \end{align}
    Hence, the SDP dual~\eqref{eq:multivariate_fid_sdp_dual} and the permutation symmetry of the multivariable SDP fidelity yield
    \begin{align}
        F_{\operatorname{SDP}}(\rho_1, \ldots, \rho_r) 
            &\geq \frac{2}{ r(r-1)}\sum_{1 \leq i <j \leq r} \mathfrak{R}\!\left[ \Tr[X_{ij}]\right] \\
            &=\frac{2}{ r(r-1)} \sum_{i=1}^{\lfloor \frac{r}{2} \rfloor} \mathfrak{R}\!\left[\Tr[Z_i] \right].\label{eq:mult-fid-lower-bd-bivariate-objective-function}
    \end{align}
    Combining~\eqref{eq:sdp-dual-bivariate-fidelity} and~\eqref{eq:mult-fid-lower-bd-bivariate-objective-function}, we get 
\begin{equation}
    F_{\operatorname{SDP}}(\rho_1, \ldots, \rho_r) \geq \frac{2}{ r(r-1)} \sum_{i=1}^{\lfloor \frac{r}{2} \rfloor} F(\rho_{\pi(i)}, \rho_{\pi(i)+\lfloor \frac{r}{2} \rfloor}).
\end{equation}
Since the above inequality holds for arbitrary $\pi \in S_r$, the desired inequality~\eqref{eq:sdp-fid-lower-bound} follows.

\subsection{Proof of Proposition~\ref{prop:properties_measured_multi_G} (Properties of measured multivariate fidelity)}\label{proof:properties_measured_multivariate_G}
\noindent\underline{Symmetry:} 
Symmetry of the measured multivariate fidelity follows easily from the symmetry satisfied by the underlying classical multivariate fidelity. 

\medskip
\noindent\underline{Data processing:} 
Let $\Lambda'$ be a measurement channel, and let $\cN$ be a channel. Then
consider that 
\begin{equation}
    \mathbf{F}\!\left( \Lambda'\big(\cN(\rho_1)\big), \ldots, \Lambda'\big(\cN(\rho_r)\big)\right) \geq \inf_{( \Lambda_x)_x}  \mathbf{F}\!\left(\Lambda(\rho_1), \ldots, \Lambda(\rho_r)\right)
\end{equation}
due to the composed measurement channel $\Lambda' \circ \cN$ being a possible choice for a measurement channel acting on $\rho_1, \ldots, \rho_r$.
Infimizing over all measurement channels  $\Lambda'$, we arrive at 
\begin{equation}
 \inf_{( \Lambda'_x)_x}  \mathbf{F}\!\left( \Lambda'\big(\cN(\rho_1)\big), \ldots, \Lambda'\big(\cN(\rho_r)\big)\right) \geq \inf_{( \Lambda_x)_x}  \mathbf{F}\!\left(\Lambda(\rho_1), \ldots, \Lambda(\rho_r)\right),
\end{equation}
proving the data processing property. 

\medskip
\noindent\underline{Faithfulness:}
When $\rho_1=\cdots = \rho_r$, it follows that $\Lambda(\rho_1)=\cdots = \Lambda(\rho_r)$ for every measurement channel $\Lambda$. By the faithfulness of the underlying classical multivariate fidelity, we conclude that $\mathbf{F}(\Lambda(\rho_1), \ldots, \Lambda(\rho_r))=1$ for every measurement channel $\Lambda$, implying that
\begin{equation}
    \overline{\mathbf{F}}(\rho_1, \ldots, \rho_r)=1.
\end{equation}

To prove the reverse implication, suppose that $ \overline{\mathbf{F}}(\rho_1, \ldots, \rho_r) =1$. Due to the fact that the underlying classical multivariate fidelity is less than or equal to one, under all measurement channels, we have that 
$ \mathbf{F}\!\left(\Lambda(\rho_1), \ldots, \Lambda(\rho_r)\right)=1$. This includes tomographically complete measurements, which lead to distributions that are in one-to-one correspondence with the underlying density operators (i.e., there is a linear invertible map relating the resulting distributions and the original density operators).
From the faithfulness of the underlying classical multivariate fidelities, the distributions are the same, and due to the aforementioned bijection, the states are also the same. That is, if $\Tr[M_x(\rho_i-\rho_j)]=0$ for every measurement operator $M_x$ in a tomographically complete measurement, then $\rho_i=\rho_j$.

\medskip 
\noindent\underline{Direct-sum property:}
For all $i\in\left[  r\right]  $, define
\begin{equation}
\rho_{XA}^{i}\coloneqq \sum_{x}p(x)|x\rangle\!\langle x|\otimes\rho_{i}^{x},
\end{equation}
where $\left(  p(x)\right)  _{x}$ is a probability distribution, $\left\{
|x\rangle\right\}  _{x}$ is an orthonormal basis, and each $\rho_{i}^{x}$ is a
state.

Defining the measurement channels
\begin{align}
\Lambda(\omega_{XA})  & \coloneqq \sum_{z}\operatorname{Tr}[\Lambda_{z}\omega
_{XA}]|z\rangle\!\langle z|,\\
\Lambda^{x}(\tau_{A})  & \coloneqq \sum_{y}\operatorname{Tr}[\Lambda_{y}^{x}\tau
_{A}]|y\rangle\!\langle y|,\\
\widetilde{\Lambda}(\omega_{XA})  & \coloneqq \sum_{x,y}\operatorname{Tr}\!\left[
\left(  |x\rangle\!\langle x|\otimes\Lambda_{y}^{x}\right)  \omega_{XA}\right]
|x\rangle\!\langle x|\otimes|y\rangle\!\langle y|,
\end{align}
consider that
\begin{align}
&  \overline{\mathbf{F}}(\rho_{XA}^{1},\ldots,\rho_{XA}^{r})\nonumber\\
& =\inf_{\Lambda}\mathbf{F}(\Lambda(\rho_{XA}^{1}),\ldots,\Lambda(\rho_{XA}^{r}))\\
& \leq\inf_{\widetilde{\Lambda}}\mathbf{F}(\widetilde{\Lambda}(\rho_{XA}^{1}
),\ldots,\widetilde{\Lambda}(\rho_{XA}^{r}))\\
& =\inf_{\widetilde{\Lambda}}\mathbf{F}\!\left(  \sum_{x}p(x)|x\rangle\!\langle
x|\otimes\Lambda^{x}(\rho_{1}^{x}),\ldots,\sum_{x}p(x)|x\rangle\!\langle
x|\otimes\Lambda^{x}(\rho_{r}^{x})\right)  \\
& =\inf_{\widetilde{\Lambda}}\sum_{x}p(x)\mathbf{F}\!\left(  \Lambda^{x}(\rho_{1}
^{x}),\ldots,\Lambda^{x}(\rho_{r}^{x})\right)  \\
& =\sum_{x}p(x)\inf_{\Lambda^{x}}\mathbf{F}\!\left(  \Lambda^{x}(\rho_{1}^{x}
),\ldots,\Lambda^{x}(\rho_{r}^{x})\right)  \\
& =\sum_{x}p(x)\overline{\mathbf{F}}(\rho_{1}^{x},\ldots,\rho_{r}^{x}) \label{eq:forward_DS_M}.
\end{align}
The first inequality follows because $\widetilde{\Lambda}$ is a special kind
of measurement. The third equality follows from the direct-sum property of the
underlying classical multivariate fidelity.

For the opposite inequality, consider that
\begin{align}
& \overline{\mathbf{F}}(\rho_{XA}^{1},\ldots,\rho_{XA}^{r})\nonumber\\
& =\inf_{\Lambda}\mathbf{F}(\Lambda(\rho_{XA}^{1}),\ldots,\Lambda(\rho_{XA}^{r}))\\
& =\inf_{\Lambda}\mathbf{F}\!\left(  \sum_{x}p(x)\Lambda\left(  |x\rangle\!\langle
x|\otimes\rho_{1}^{x}\right)  ,\ldots,\sum_{x}p(x)\Lambda\left(
|x\rangle\!\langle x|\otimes\rho_{r}^{x}\right)  \right)  \\
& \geq\inf_{\Lambda}\sum_{x}p(x)\mathbf{F}\!\left(  \Lambda\left(  |x\rangle\!\langle
x|\otimes\rho_{1}^{x}\right)  ,\ldots,\Lambda\left(  |x\rangle\!\langle
x|\otimes\rho_{r}^{x}\right)  \right)  \\
& \geq\sum_{x}p(x) \overline{\mathbf{F}}(\rho_{1}^{x},\ldots,\rho_{r}^{x}). \label{eq:backward_DS_M}
\end{align}
The second equality follows from linearity of the measurement channel $\Lambda
$. The first inequality follows because the underlying classical fidelities
are jointly concave (as a consequence of the data-processing inequality and the direct-sum property). The final inequality follows because tensoring in the
state $|x\rangle\!\langle x|$ and performing the measurement $\Lambda$ is a
particular kind of measurement channel for the tuple $\left(  \rho_{1}
^{x},\ldots,\rho_{r}^{x}\right)  $, and so the resulting fidelity cannot be
smaller than $ \overline{\mathbf{F}}(\rho_{1}^{x},\ldots,\rho_{r}^{x})$, which is defined as
the infimum over all measurement channels.

By combining~\eqref{eq:forward_DS_M} and~\eqref{eq:backward_DS_M}, we conclude the proof.

\medskip 
\noindent\underline{Orthogonality:}
Suppose that $\rho_i \rho_j =0$ for all $i\neq j \in [r]$.  Then there exists a measurement to distinguish these states perfectly, leading to orthogonal distributions. Then all of the three underlying multivariate classical fidelities are equal to zero, thus proving weak orthogonality of the maximal extension with respect to all three underlying multivariate classical fidelities.

To prove the reverse implication, suppose that the underlying classical multivariate fidelity is the average pairwise fidelity and that $\overline{\mathbf{F}}(\rho_1, \ldots, \rho_r) =0$. Then applying \cref{prop:measured_unmeasured_ineq} gives that
\begin{equation}
    0= \overline{\mathbf{F}}(\rho_1, \ldots, \rho_r) \geq F_{U}(\rho_1, \ldots, \rho_r).
\end{equation}
By orthogonality of the average pairwise Uhlmann fidelity (\cref{Prop:properties_average_pairwise_z}), it follows  that $\rho_i \rho_j =0$ for all $i\neq j \in [r]$, concluding the proof of orthogonality when the underlying classical fidelity is the average pairwise fidelity.

\subsection{Proof of Proposition~\ref{Prop:properties_minimal_extension} (Properties of minimal extension of multivariate fidelity)}

\label{proof:properties_minimal}

Recall that
\begin{equation}
        \underline{ \mathbf{F}}(\rho_1, \ldots, \rho_r) \coloneqq \sup_{\substack{\cP \in \operatorname{CPTP}, \, \omega_i \in \mathscr{D}_c \forall i \in [r]}} \left\{\mathbf{F}\!\left( \omega_1, \ldots, \omega_r \right) : \cP(\omega_i) =\rho_i \,\forall i \in [r]\right \},
    \end{equation}
where $\cP$ is a classical-quantum channel with $\cP(\omega_i) =\rho_i$.

Symmetry follows by the symmetry of the underlying classical multivariate fidelity.

\medskip
\noindent\underline{Data processing:}
Let $\cN$ be a quantum channel. 
Let $\omega_i \in \mathscr{D}_c$ and let $\mathcal{P}$ be a preparation channel such that $\cP(\omega_i)= \rho_i$. This also leads to 
$\cN\big(\cP(\omega_i)\big)= \cN(\rho_i)$, thus providing a feasible preparation channel for $\cN(\rho_i)$, resulting from $\cN \circ \cP$.
So we have 
\begin{equation}
   \mathbf{F}\!\left( \omega_i, \ldots, \omega_r \right)  \leq \underline{ \mathbf{F}}\!\left(\cN(\rho_1), \ldots, \cN(\rho_r) \right).
\end{equation}
Then, supremizing over all $\cP$ and $\omega_i$, we conclude that 
\begin{equation}
  \underline{\mathbf{F}}\!\left(\rho_1, \ldots, \rho_r \right) \leq   \underline{\mathbf{F}}\!\left(\cN(\rho_1), \ldots, \cN(\rho_r) \right).
\end{equation}

\medskip 
\noindent\underline{Faithfulness:}
Suppose that $ \underline{\mathbf{F}}\!\left(\rho_1, \ldots, \rho_r \right) =1$. Then, we also have that ${\overline{\mathbf{F}}}\!\left(\rho_1, \ldots, \rho_r \right) =1$ from~\cref{prop:ineq_maximal_minimal_multi}. With the faithfulness property of maximal extension of multivariate fidelity in~\cref{prop:properties_measured_multi_G}, we find that $\rho_i=\rho_j$ for all $i,j \in [r]$.

To prove the reverse implication, suppose that $\rho_i=\rho_j =\rho$ for all $i,j \in [r]$.  Then we can set $\omega_i \coloneqq |0 \rangle\!\langle 0|$ for all $i\in[r]$, and a simple preparation channel consists of tracing out the input and replacing with $\rho$. The fidelity of the common state $\omega_i = |0 \rangle\!\langle 0|$ is equal to one, due to the faithfulness of the classical multivariate fidelity. Since the preparation fidelity involves a supremum over all preparations and the underlying fidelity cannot exceed one, this proves that $ \underline{\mathbf{F}}\!\left(\rho_1, \ldots, \rho_r \right) =1$.

\medskip 
\noindent\underline{Direct-sum property:}
For all $i\in\left[  r\right]  $, define
\begin{equation}
\rho_{XA}^{i}\coloneqq\sum_{x}p(x)|x\rangle\!\langle x|\otimes\rho_{i}^{x}, \label{eq:rho_XA_def}
\end{equation}
where $\left(  p(x)\right)  _{x}$ is a probability distribution, $\left\{
|x\rangle\right\}  _{x}$ is an orthonormal basis, and each $\rho_{i}^{x}$ is a
state.

Define the preparation channel $\mathcal{P}^{x}$ and the state $\omega_{i}
^{x}\in\mathscr{D}_{c}$ on a system~$S$, such that
\begin{equation}
\mathcal{P}^{x}(\omega_{i}^{x})=\rho_{i}^{x}\ \forall i\in\left[  r\right]  .
\end{equation}
Defining
\begin{align}
\mathcal{P}^{\prime}(\tau_{XS}) &  \coloneqq\sum_{x}|x\rangle\!\langle
x|\otimes\mathcal{P}^{x}(\langle x|_{X}\tau_{XS}|x\rangle_{X}),\\
\omega_{i}^{\prime} &  \coloneqq\sum_{x}p(x)|x\rangle\!\langle x|\otimes
\omega_{i}^{x},
\end{align}
we then find that, for all $i\in\left[  r\right]  $,
\begin{equation}
\mathcal{P}^{\prime}\left(  \omega_{i}^{\prime}\right)  =\sum_{x}
p(x)|x\rangle\!\langle x|\otimes\mathcal{P}^{x}(\omega_{i}^{x})=\rho_{XA}^{i}.
\end{equation}
Then consider that
\begin{align}
&  \underline{\mathbf{F}}(\rho_{XA}^{1},\ldots,\rho_{XA}^{r})\nonumber\\
&  =\sup_{\substack{\mathcal{P},\\\omega_{1},\ldots,\omega_{r}\in
\mathscr{D}_{c}}}\left\{  \mathbf{F}(\omega_{1},\ldots,\omega_{r}
):\mathcal{P}(\omega_{i})=\rho_{XA}^{i}\ \forall i\in\left[  r\right]
\right\}  \\
&  \geq\sup_{\substack{\mathcal{P}^{\prime},\\\omega_{1}^{\prime}
,\ldots,\omega_{r}^{\prime}\in\mathscr{D}_{c}}}\left\{  \mathbf{F}(\omega
_{1}^{\prime},\ldots,\omega_{r}^{\prime}):\mathcal{P}^{\prime}(\omega
_{i}^{\prime})=\rho_{XA}^{i}\ \forall i\in\left[  r\right]  \right\}  \\
&  =\sup_{\substack{\mathcal{P}^{\prime},\\\omega_{1}^{\prime},\ldots
,\omega_{r}^{\prime}\in\mathscr{D}_{c}}}\left\{
\begin{array}
[c]{c}
\mathbf{F\!}\left(  \sum_{x}p(x)|x\rangle\!\langle x|\otimes\omega_{1}
^{x},\ldots,\sum_{x}p(x)|x\rangle\!\langle x|\otimes\omega_{r}^{x}\right)  :\\
\mathcal{P}^{\prime}(\omega_{i}^{\prime})=\rho_{XA}^{i}\ \forall i\in\left[
r\right]
\end{array}
\right\}  \\
&  =\sup_{\substack{\mathcal{P}^{\prime},\\\omega_{1}^{\prime},\ldots
,\omega_{r}^{\prime}\in\mathscr{D}_{c}}}\left\{  \sum_{x}p(x)\mathbf{F}\!\left(
\omega_{1}^{x},\ldots,\omega_{r}^{x}\right)  :\mathcal{P}^{\prime}(\omega
_{i}^{\prime})=\rho_{XA}^{i}\ \forall i\in\left[  r\right]  \right\}  \\
&  =\sum_{x}p(x)\sup_{\substack{\mathcal{P}^{x},\\\omega_{1}^{x},\ldots
,\omega_{r}^{x}\in\mathscr{D}_{c}}}\left\{  \mathbf{F}\!\left(  \omega_{1}
^{x},\ldots,\omega_{r}^{x}\right)  :\mathcal{P}^{x}(\omega_{i}^{x})=\rho
_{i}^{x}\ \forall i\in\left[  r\right]  \right\}  \\
&  =\sum_{x}p(x)\underline{\mathbf{F}}(\rho_{1}^{x},\ldots,\rho_{r}^{x}),
\end{align}
where the third equality follows from the direct-sum property of the underlying classical multivariate fidelity.

Now we prove the following inequality
\begin{equation}
\underline{\mathbf{F}}(\rho_{XA}^{1},\ldots,\rho_{XA}^{r})\leq\sum
_{x}p(x)\underline{\mathbf{F}}(\rho_{1}^{x},\ldots,\rho_{r}^{x}
).\label{eq:other-prep-ineq}
\end{equation}
For a probability distribution $\left(  q_{i}(y)\right)  _{y}$, define the
state
\begin{equation}
\omega_{i}\coloneqq\sum_{y}q_{i}(y)|y\rangle\!\langle y|,
\end{equation}
for all $i\in\left[  r\right]  $. Then consider that the action of a general
preparation channel $\mathcal{P}$ is as follows
\begin{equation}
\mathcal{P}(\omega_{i})=\rho_{XA}^{i}\ \ \forall i\in\left[  r\right]  .
\end{equation}
Now define
\begin{equation}
\sigma_{XA}^{y} \coloneqq \mathcal{P}(|y\rangle\!\langle y|) ,
\end{equation}
so that
\begin{equation}
\sum_{y}q_{i}(y)\sigma_{XA}^{y}=\rho_{XA}^{i}\ \ \forall i\in\left[  r\right]
.
\end{equation}
Then from the action of the completely dephasing channel $\Delta$, define the
probability distribution $\left(  t(x|y)\right)  _{x}$ and tuple $\left(
\sigma_{A}^{x,y}\right)  _{x}$ of states as
\begin{equation}
\left(  \Delta\otimes\operatorname{id}\right)  (\sigma_{XA}^{y})=\sum
_{x}t(x|y)|x\rangle\!\langle x|\otimes\sigma_{A}^{x,y},
\end{equation}
and consider that
\begin{align}
\rho_{XA}^{i} &  =\left(  \Delta\otimes\operatorname{id}\right)  \left(
\rho_{XA}^{i}\right)  \\
&  =\sum_{y}q_{i}(y)\left(  \Delta\otimes\operatorname{id}\right)
(\sigma_{XA}^{y})\\
&  =\sum_{x,y}q_{i}(y)t(x|y)|x\rangle\!\langle x|\otimes\sigma_{A}^{x,y} . \label{eq:rho_XA without conditional}
\end{align}
Now observe that, for all $i \in [r]$,
\begin{equation}
\sum_{x}p(x)|x\rangle\!\langle x|=\operatorname{Tr}_{A}[\rho_{XA}^{i}
]=\sum_{x}\left(  \sum_{y}q_{i}(y)t(x|y)\right)  |x\rangle\!\langle x|,
\end{equation}
proving that the marginal over $y$ of the joint distribution $q_{i}(y)t(x|y)$
is $p(x)$. Now let us define the conditional distribution $s_{i}(y|x)$ by
\begin{equation}
s_{i}(y|x)p(x)=q_{i}(y)t(x|y),\label{eq:bayes-id}
\end{equation}
and by considering~\eqref{eq:rho_XA without conditional} and the above equality, we note that
\begin{align}
\rho_{XA}^{i} &  =\sum_{x,y}s_{i}(y|x)p(x)|x\rangle\!\langle x|\otimes
\sigma_{A}^{x,y}\\
&  =\sum_{x}p(x)|x\rangle\!\langle x|\otimes\sum_{y}s_{i}(y|x)\sigma_{A}
^{x,y}\\
&  =\sum_{x}p(x)|x\rangle\!\langle x|\otimes\rho_{i}^{x},
\end{align}
where the last equality follows by recalling~\eqref{eq:rho_XA_def}.
This implies that 
\begin{equation}
    \rho_i^x =\sum_{y}s_{i}(y|x)\sigma_{A}
^{x,y}.
\end{equation}
Thus, conditioned on $x$, generating $y$ at random according to $s_{i}(y|x)$
and preparing the state $\sigma_{A}^{x,y}$ is a particular way of preparing
$\rho_{i}^{x}$. That is, we can define the preparation channel $\mathcal{P}
^{x}$ as
\begin{equation}
\mathcal{P}^{x}(|y\rangle\!\langle y|)=\sigma_{A}^{x,y},
\end{equation}
so that, for all $i\in\left[  r\right]  $,
\begin{equation}
\mathcal{P}^{x}\left(  \sum_{y}s_{i}(y|x)|y\rangle\!\langle y|\right)
=\sum_{y}s_{i}(y|x)\mathcal{P}^{x}\left(  |y\rangle\!\langle y|\right)
=\sum_{y}s_{i}(y|x)\sigma_{A}^{x,y}=\rho_{i}^{x}.\label{eq:prepare-each-x}
\end{equation}
Putting everything together, let $\mathcal{P}$ and $\omega_{1},\ldots
,\omega_{r}$ be a particular way of preparing $\rho_{XA}^{1},\ldots,\rho
_{XA}^{r}$, respectively. Then
\begin{align}
&  \mathbf{F}(\omega_{1},\ldots,\omega_{r})\nonumber\\
&  =\mathbf{F}\!\left(  \sum_{y}q_{1}(y)|y\rangle\!\langle y|,\ldots,\sum
_{y}q_{r}(y)|y\rangle\!\langle y|\right)  \\
&  \leq\mathbf{F}\!\left(  \sum_{x,y}q_{1}(y)t(x|y)|x\rangle\!\langle
x|\otimes|y\rangle\!\langle y|,\ldots,\sum_{x,y}q_{r}(y)t(x|y)|x\rangle
\!\langle x|\otimes|y\rangle\!\langle y|\right)  \\
&  =\mathbf{F}\!\left(  \sum_{x,y}s_{1}(y|x)p(x)|x\rangle\!\langle
x|\otimes|y\rangle\!\langle y|,\ldots,\sum_{x,y}s_{r}(y|x)p(x)|x\rangle
\!\langle x|\otimes|y\rangle\!\langle y|\right)  \\
&  =\sum_{x}p(x)\mathbf{F}\!\left(  \sum_{y}s_{1}(y|x)|y\rangle\!\langle
y|,\ldots,\sum_{x,y}s_{r}(y|x)|y\rangle\!\langle y|\right)  \\
&  \leq\sum_{x}p(x)\sup_{\substack{\mathcal{P}^{x},\\\omega_{1}^{x}
,\ldots,\omega_{r}^{x}\in\mathscr{D}_{c}}}\left\{  \mathbf{F}\!\left(
\omega_{1}^{x},\ldots,\omega_{r}^{x}\right)  :\mathcal{P}^{x}(\omega_{i}
^{x})=\rho_{i}^{x}\ \forall i\in\left[  r\right]  \right\}  \\
&  =\sum_{x}p(x)\underline{\mathbf{F}}(\rho_{1}^{x},\ldots,\rho_{r}^{x}).
\end{align}
The first inequality follows from data processing. The second equality follows
from the identity in~\eqref{eq:bayes-id}. The third equality follows from the
direct-sum property for the underlying classical multivariate fidelity. The
last inequality follows because, from $\mathcal{P}$ and $\omega_{1}
,\ldots,\omega_{r}$, we have constructed a particular method of preparing
$\rho_{1}^{x},\ldots,\rho_{r}^{x}$ from the channel $\mathcal{P}^{x}$ and the
commuting states $\sum_{y}s_{1}(y|x)|y\rangle\!\langle y|,\ldots,\sum_{y}
s_{r}(y|x)|y\rangle\!\langle y|$, as given in~\eqref{eq:prepare-each-x}. The
proof of the inequality in~\eqref{eq:other-prep-ineq} is complete after
noticing that we have proven the inequality $\mathbf{F}(\omega_{1}
,\ldots,\omega_{r})\leq\sum_{x}p(x)\underline{\mathbf{F}}(\rho_{1}^{x}
,\ldots,\rho_{r}^{x})$ for all possible preparations of the states $\rho
_{XA}^{1},\ldots,\rho_{XA}^{r}$, and so we can take a supremum over all such preparations.

\medskip 
 \noindent\underline{Weak orthogonality:} Suppose that $\rho_i\rho_j=0$ for all $i\neq j \in [r]$. 
Then by using~\cref{prop:ineq_maximal_minimal_multi} and the weak orthogonality of maximal extension in~\cref{prop:properties_measured_multi_G}, we have that 
\begin{equation}
    \underline{ \mathbf{F}}(\rho_1, \ldots, \rho_r) \leq {\overline{ \mathbf{F}}}(\rho_1, \ldots, \rho_r) =0, 
\end{equation}
which implies that $  \underline{ \mathbf{F}}(\rho_1, \ldots, \rho_r)=0$ concluding the proof.

\subsection{Proof of Proposition~\ref{prop:additivity_log_Euclidean_divergence} (Additivity of multivariate log-Euclidean divergence)} 
\label{proof:additivity_log_eucli_divergence}

For $\varepsilon>0$, define $\rho_{i}(\varepsilon)\coloneqq \rho_{i}+\varepsilon I$
and $\sigma_{i}(\varepsilon)\coloneqq \sigma_{i}+\varepsilon I$, and consider that
\begin{align}
& H_{s}\!\left(  \rho_{1}(\varepsilon),\ldots,\rho_{r}(\varepsilon)\right)
+H_{s}\!\left(  \sigma_{1}(\varepsilon),\ldots,\sigma_{r}(\varepsilon)\right)
\nonumber\\
& =-\ln\operatorname{Tr}\!\left[  \exp\!\left(  \sum_{i=1}^{r}s_{i}\ln\rho
_{i}(\varepsilon)\right)  \right]  -\ln\operatorname{Tr}\!\left[  \exp\!\left(
\sum_{i=1}^{r}s_{i}\ln\sigma_{i}(\varepsilon)\right)  \right]  \\
& =-\ln\left(  \operatorname{Tr}\!\left[  \exp\!\left(  \sum_{i=1}^{r}s_{i}\ln
\rho_{i}(\varepsilon)\right)  \right]  \operatorname{Tr}\!\left[  \exp\!\left(
\sum_{i=1}^{r}s_{i}\ln\sigma_{i}(\varepsilon)\right)  \right]  \right)  \\
& =-\ln\operatorname{Tr}\!\left[  \exp\!\left(  \sum_{i=1}^{r}s_{i}\ln\rho
_{i}(\varepsilon)\right)  \otimes\exp\!\left(  \sum_{i=1}^{r}s_{i}\ln\sigma
_{i}(\varepsilon)\right)  \right]  \\
& =-\ln\operatorname{Tr}\!\left[  \left(  \exp\!\left(  \sum_{i=1}^{r}s_{i}\ln
\rho_{i}(\varepsilon)\right)  \otimes I\right)  \left(  I\otimes\exp\!\left(
\sum_{i=1}^{r}s_{i}\ln\sigma_{i}(\varepsilon)\right)  \right)  \right]  \\
& =-\ln\operatorname{Tr}\!\left[  \exp\!\left(  \sum_{i=1}^{r}s_{i}\ln\left(
\rho_{i}(\varepsilon)\otimes I\right)  \right)  \exp\!\left(  \sum_{i=1}
^{r}s_{i}\ln\left(  I\otimes\sigma_{i}(\varepsilon)\right)  \right)  \right]
\\
& =-\ln\operatorname{Tr}\!\left[  \exp\!\left(  \sum_{i=1}^{r}s_{i}\ln\left(
\rho_{i}(\varepsilon)\otimes I\right)  +\sum_{i=1}^{r}s_{i}\ln\left(
I\otimes\sigma_{i}(\varepsilon)\right)  \right)  \right]  \\
& =-\ln\operatorname{Tr}\!\left[  \exp\!\left(  \sum_{i=1}^{r}s_{i}\ln\left(
\rho_{i}(\varepsilon)\otimes\sigma_{i}(\varepsilon)\right)  \right)  \right]
\\
& =H_{s}\!\left(  \rho_{1}(\varepsilon)\otimes\sigma_{1}(\varepsilon
),\ldots,\rho_{r}(\varepsilon)\otimes\sigma_{r}(\varepsilon)\right)  .
\end{align}
The fifth equality follows because
\begin{align}
\exp\!\left(  \sum_{i=1}^{r}s_{i}\ln\rho_{i}(\varepsilon)\right)  \otimes I  &
=\exp\!\left(  \left(  \sum_{i=1}^{r}s_{i}\ln\rho_{i}(\varepsilon)\right)
\otimes I\right)  \\
& =\exp\!\left(  \sum_{i=1}^{r}s_{i}\left(  \left[  \ln\rho_{i}(\varepsilon
)\right]  \otimes I\right)  \right)  \\
& =\exp\!\left(  \sum_{i=1}^{r}s_{i}\ln\left(  \rho_{i}(\varepsilon)\otimes
I\right)  \right)  ,
\end{align}
where we used the facts that $\exp\!\left(  A\right)  \otimes I=\exp\!\left(
A\otimes I\right)  $ for Hermitian $A$, that $s$ is a probability
distribution, and that $\ln\left(  B\right)  \otimes I=\ln\left(  B\otimes
I\right)  $ for positive definite $B$. The sixth equality follows because
$\exp\!\left(  C+D\right)  =\exp\!\left(  C\right)  \exp\!\left(  D\right)  $ when
$C$ and $D$ commute. The penultimate equality follows from $\ln( A \otimes I) + \ln(I \otimes B) = \ln(A \otimes B)$  for positive definite $A$ and $B$. By taking the limit $\varepsilon\rightarrow0^{+}$ on both
sides of the last equality, we conclude the proof.

\subsection{Proof of Theorem~\ref{thm:properties_quantum_Matusita_multi} (Properties of multivariate log-Euclidean fidelity)} 

\label{proof:properties_quantum_Matusita_multi}

  Symmetry follows by the definition of oveloh information in~\eqref{eq:oveloh-info} or simply by the fact that the sum over $i$ in \eqref{eq:mult_log-Euc_fid_definition} can be performed in any order. 

\medskip
\noindent\underline{Reduction to classical Matusita multivariate fidelity:}
By the definition of log-Euclidean fidelity in~\eqref{eq:mult_log-Euc_fid_definition}, and applying the fact that the states are commuting, we see that 
\begin{equation}\label{eq:commuting_log_Euclidean}
    F_r^\flat(\rho_1,\ldots, \rho_r)  =\operatorname{Tr}\!\left[  \prod\limits_{i=1}^{r}\rho_{i}^{\frac{1}{r}}\right].
\end{equation}
Observing that the last expression is precisely the Matusita multivariate fidelity for commuting states (recall~\eqref{eq:Matusita_fid_commuting}) then establishes the claim.

  \medskip 
  \noindent\underline{Data processing:}
  Due to the data-processing inequality for the quantum relative entropy, we have 
  \begin{align}
\inf_{\sigma\in\mathscr{D}}
\frac{1}{r}\sum_{i=1}^{r}D(\sigma\Vert\rho
_{i})  & \geq \inf_{\sigma\in\mathscr{D}}
\frac{1}{r}\sum_{i=1}^{r}D\!\left(\cN(\sigma) \Vert \cN(\rho
_{i})\right) \notag \\ 
& \geq 
\inf_{\sigma' \in\mathscr{D}}
\frac{1}{r}\sum_{i=1}^{r}D\!\left(\sigma' \Vert \cN(\rho
_{i})\right),
  \end{align}
where the last inequality follows from performing an optimization over a larger set of states.
This then leads to 
\begin{equation}
    \exp\!\left( - \inf_{\sigma\in\mathscr{D}}\frac{1}{r}\sum_{i=1}^{r}D(\sigma\Vert\rho
_{i}) \right) \leq  \exp\!\left(-  \inf_{\sigma' \in\mathscr{D}}\frac{1}{r}\sum_{i=1}^{r}D\!\left(\sigma' \Vert \cN(\rho
_{i})\right) \right),
\end{equation}
due to the monotonicity of the exponential function. Recalling \eqref{eq:quantum_Matusita_oveloh}, we conclude the proof of data processing.

\medskip 
\noindent\underline{Faithfulness:} 
If $\rho_i=\rho$ for all $i \in [r]$, the states commute, so that by~\eqref{eq:commuting_log_Euclidean},
\begin{equation}
    F_r^\flat(\rho_1,\ldots, \rho_r)  =\operatorname{Tr}\!\left[  \prod\limits_{i=1}^{r}\rho^{\frac{1}{r}}\right] = \Tr[\rho].
\end{equation}
Since $\rho$ is a state, we conclude that $ F_r^\flat(\rho_1,\ldots, \rho_r)=1$.

To prove the reverse direction, suppose that $  F_r^\flat(\rho_1,\ldots, \rho_r) =1$, so that by~\eqref{eq:quantum_Matusita_oveloh} we have
\begin{equation}
     \inf_{\sigma\in\mathscr{D}}\frac{1}{r}\sum_{i=1}^{r}D(\sigma\Vert\rho
_{i}) =0.
\end{equation}
Since the quantum relative entropy is non-negative when evaluated on states, there exists $\sigma \in \mathscr{D}$ such that $D(\sigma \Vert \rho_i) =0 $ for all $i \in [r]$. Then,  faithfulness of quantum relative entropy \cite[Proposition~7.3]{khatri2020principles} implies that $\sigma= \rho_i$ for all $i$, which is the desired implication. 

\medskip 
\noindent\underline{Direct-sum property:}
Let 
\begin{equation}
    \rho_i \coloneqq \sum_{x} p(x) |x\rangle\!\langle x|_X \otimes \rho_i^x
\end{equation}
for all $i \in [r]$.
Then, by the functional calculus, we have 
\begin{equation}
    \ln \rho_i = \sum_{x}  |x\rangle\!\langle x|_X \otimes \ln\! \left(p(x) \rho_i^x \right) = \sum_{x}  |x\rangle\!\langle x|_X \otimes \ln\! \left(p(x) \right) I_d + \sum_{x}  |x\rangle\!\langle x|_X \otimes \ln\! \left( \rho_i^x \right).
\end{equation}
Summing over all $i \in [r]$ and averaging, we get 
\begin{equation}
    \sum_{i=1}^r \frac{1}{r} \ln \rho_i = \sum_{x}  |x\rangle\!\langle x|_X \otimes \ln\! \left(p(x) \right) I_d +  \sum_{x}  |x\rangle\!\langle x|_X \otimes \sum_{i=1}^r \frac{1}{r} \ln\! \left( \rho_i^x \right).
\end{equation}
Using the fact that the two summands on the right above commute, that $\Tr\!\left[  \exp(A+B)\right]=\Tr\!\left[  \exp(A) \exp(B)\right]$ for commuting operators $A$ and $B$, and again applying the functional calculus,
we arrive at 
\begin{align}
& \operatorname{Tr}\!\left[  \exp\!\left(  \sum_{i=1}^{r}\frac{1}{r}
\ln\rho_{i}\right)  \right] \notag \\ 
&= \Tr\!\left[ \exp\!\left(\sum_{x}  |x\rangle\!\langle x|_X \otimes \ln(p(x)) I_d \right) \exp\!\left(\sum_{x'}  |x'\rangle\!\langle x'|_X \otimes  \sum_{i=1}^r \frac{1}{r} \ln\!\left( \rho_i^{x'} \right)\right)  \right] \\ 
&= \Tr\!\left[ \left(\sum_{x}  |x\rangle\!\langle x|_X \otimes p(x) I_d \right) \left(\sum_{x'}  |x'\rangle\!\langle x'|_X \otimes \exp\!\left( \sum_{i=1}^r \frac{1}{r} \ln\!\left( \rho_i^{x'}\right) \right)\right)  \right] \\ 
& = \Tr\! \left[ \sum_{x} |x\rangle\!\langle x|_X \otimes p(x) \exp\!\left( \sum_{i=1}^r \frac{1}{r} \ln\!\left( \rho_i^x\right) \right)\right] \\ 
&= \sum_x p(x) \Tr\! \left[\exp\!\left( \sum_{i=1}^r \frac{1}{r} \ln\!\left( \rho_i^x\right) \right)\right].
\end{align}
Recalling~\eqref{eq:quantum_Matusita_oveloh} together with \cref{def:multi_log_eucli_fidelity}, we conclude the proof. 

\medskip 
\noindent\underline{Weak orthogonality:} Let $\rho_i\rho_j=0$ for all $i\neq j \in [r]$. Then, the channel mapping $\rho_i \to |i \rangle\!\langle i|$ for all $i \in [r]$ is a reversible mapping. Then, due to the data processing property of the log-Euclidean multivariate fidelity, we have 
\begin{equation}
    F^{\flat}_r (\rho_1,\ldots, \rho_r) =F^{\flat}_r( |1\rangle\!\langle 1|, \ldots, |r\rangle\!\langle r|).
\end{equation}
For $ \inf_{\sigma} \frac{1}{r} \sum_{i=1}^r D(\sigma \Vert |i\rangle\!\langle i| ) < \infty$, there should be $\sigma \in \mathscr{D}$ such that $D(\sigma \Vert |i\rangle\!\langle i| ) < \infty$ for all $i \in [r]$. This indeed happens only if $\operatorname{supp}(\sigma) \subseteq \operatorname{supp}(|i\rangle\!\langle i|)$ for all $i$. Since, however,  $\{|i\rangle\!\langle i|\}$ is a collection of orthonormal projections, we cannot find $\sigma$ such that it satisfies the finiteness condition for all $i \in [r]$. This implies that 
\begin{equation}
    \inf_{\sigma} \frac{1}{r} \sum_{i=1}^r D(\sigma \Vert |i\rangle\!\langle i| )= \infty.
\end{equation}
Then, by recalling \cref{eq:quantum_Matusita_oveloh}, we conclude that 
$F^{\flat}_r(\rho_1, \ldots, \rho_r) =0$.

\medskip 
\noindent\underline{Continuity:} 
The multivariate log-Euclidean fidelity $F_\flat$ given in Definition~\ref{def:multi-var-log-euclid} is continuous at tuples $(\rho_1,\ldots, \rho_r)$ of invertible states.
This follows from the fact that the matrix exponential and logarithmic functions are continuous at positive definite operators, which implies that the multivariate log-Euclidean fidelity is a composition of continuous functions.

\subsection{Proof of Proposition~\ref{prop:avg-k-wise-log-Euc-ordered} (Inequalities relating average $k$-wise log-Euclidean fidelities)}

\label{sec:avg-k-wise-log-Euc-ordered}

Let $S_{k-1}(i_{1},\ldots,i_{k})$ denote all size-($k-1$) subsets of
$i_{1},\ldots,i_{k}$, of which there are $k$ of them. Note that each symbol
$i_{j}$ appears in exactly $k-1$ elements of $S_{k-1}(i_{1},\ldots,i_{k})$.
Then consider that
\begin{align}
&  \frac{1}{k}\sum_{t\in S_{k-1}(i_{1},\ldots,i_{k})}F_{k-1}^{\flat}
(\rho_{t(1)},\ldots,\rho_{t(k-1)}) \notag \\
&  =\frac{1}{k}\sum_{t\in S_{k-1}(i_{1},\ldots,i_{k})}\exp\!\left(
-\inf_{\omega\in\mathscr{D}}\frac{1}{k-1}\sum_{j\in\left[  k-1\right]
}D(\omega\Vert\rho_{t(j)})\right)  \\
&  \geq\exp\!\left(  -\frac{1}{k}\sum_{t\in S_{k-1}(i_{1},\ldots,i_{k})}
\inf_{\omega\in\mathscr{D}}\frac{1}{k-1}\sum_{j\in\left[  k-1\right]
}D(\omega\Vert\rho_{t(j)})\right)  \\
&  \geq\exp\!\left(  -\inf_{\omega\in\mathscr{D}}\frac{1}{k\left(  k-1\right)
}\sum_{t\in S_{k-1}(i_{1},\ldots,i_{k})}\sum_{j\in\left[  k-1\right]
}D(\omega\Vert\rho_{t(j)})\right)  \\
&  =\exp\!\left(  -\inf_{\omega\in\mathscr{D}}\frac{1}{k}\sum_{\ell\in
\{i_{1},\ldots,i_{k}\}}D(\omega\Vert\rho_{\ell})\right)  \\
&  =F_{k}^{\flat}(\rho_{i_{1}},\ldots,\rho_{i_{k}}),
\end{align}
where the first inequality follows from the convexity of the exponential
function and the second from the fact that
\begin{multline}
\frac{1}{k}\sum_{t\in S_{k-1}(i_{1},\ldots,i_{k})}\inf_{\omega\in\mathscr{D}
}\frac{1}{k-1}\sum_{j\in\left[  k-1\right]  }D(\omega\Vert\rho_{t(j)}) \\ \leq
\inf_{\omega\in\mathscr{D}}\frac{1}{k\left(  k-1\right)  }\sum_{t\in
S_{k-1}(i_{1},\ldots,i_{k})}\sum_{j\in\left[  k-1\right]  }D(\omega\Vert
\rho_{t(j)}).
\end{multline}
The penultimate equality follows because
\begin{equation}
    \frac{1}{k\left(  k-1\right)
}\sum_{t\in S_{k-1}(i_{1},\ldots,i_{k})}\sum_{j\in\left[  k-1\right]
}D(\omega\Vert\rho_{t(j)})  
  =\frac{1}{k}\sum_{\ell\in
\{i_{1},\ldots,i_{k}\}}D(\omega\Vert\rho_{\ell}),
\end{equation}
as a generalization of the justification behind~\eqref{eq:reorder-product-k-1}.
Now consider that
\begin{align}
&  \frac{1}{\binom{r}{k}}\sum_{i_{1}<\cdots<i_{k}}F_{k}^{\flat}(\rho_{i_{1}
},\ldots,\rho_{i_{k}})\nonumber\\
&  \leq\frac{1}{\binom{r}{k}}\sum_{i_{1}<\cdots<i_{k}}\frac{1}{k}\sum_{t\in
S_{k-1}(i_{1},\ldots,i_{k})}F_{k-1}^{\flat}(\rho_{t(1)},\ldots,\rho
_{t(k-1)})\\
&  =\frac{k-1!r-k!}{r!}\sum_{i_{1}<\cdots<i_{k}}\sum_{t\in S_{k-1}
(i_{1},\ldots,i_{k})}F_{k-1}^{\flat}(\rho_{t(1)},\ldots,\rho_{t(k-1)})\\
&  =\frac{k-1!r-k!}{r!}\left(  r-k+1\right)  \sum_{i_{1}<\cdots<i_{k-1}
}F_{k-1}^{\flat}(\rho_{i_{1}},\ldots,\rho_{i_{k-1}})\\
&  =\frac{1}{\binom{r}{k-1}}\sum_{i_{1}<\cdots<i_{k-1}}F_{k-1}^{\flat}
(\rho_{i_{1}},\ldots,\rho_{i_{k-1}}).
\end{align}
The reasoning for these steps is precisely the same as that given for~\eqref{eq:avg-k-wise-final-steps}--\eqref{eq:avg-k-wise-final-steps-final}.

\subsection{Proof of Proposition~\ref{prop:dual_SDP_multi_geo} (Dual of multivariate SDP geometric fidelity)}

\label{Sec:Proof_dual_SDP_geo}

We follow the Lagrange multiplier method. Consider that
\begin{align}
& \sup_{X_{ij}=X_{ji}\in\text{Herm}}\left\{  \sum_{i<j}\operatorname{Tr}
[X_{ij}]:\sum_{i}|i\rangle\!\langle i|\otimes\rho_{i}+\sum_{i\neq j}
|i\rangle\!\langle j|\otimes X_{ij}\geq0\right\}  \nonumber\\
& =\sup_{X_{ij}=X_{ji}\in\text{Herm}}\Bigg\{  \sum_{i<j}\operatorname{Tr}
[X_{ij}]+\notag \\
& \qquad \inf_{\sum_{i,j}|i\rangle\!\langle j|\otimes Y_{ij}\geq0}
\operatorname{Tr}\!\left[  \left(  \sum_{i}|i\rangle\!\langle i|\otimes\rho
_{i}+\sum_{i\neq j}|i\rangle\!\langle j|\otimes X_{ij}\right)  \sum_{i^{\prime
},j^{\prime}}|i^{\prime}\rangle\!\langle j^{\prime}|\otimes Y_{i^{\prime
}j^{\prime}}\right]  \Bigg\}  \\
& =\sup_{X_{ij}=X_{ji}\in\text{Herm}}\inf_{\sum_{i,j}|i\rangle\!\langle
j|\otimes Y_{ij}\geq0}\Bigg\{  \sum_{i<j}\operatorname{Tr}[X_{ij}
]+\notag \\
& \qquad\qquad\qquad \operatorname{Tr}\!\left[  \left(  \sum_{i}|i\rangle\!\langle i|\otimes\rho
_{i}+\sum_{i\neq j}|i\rangle\!\langle j|\otimes X_{ij}\right)  \sum_{i^{\prime
},j^{\prime}}|i^{\prime}\rangle\!\langle j^{\prime}|\otimes Y_{i^{\prime
}j^{\prime}}\right]  \Bigg\}  \\
& =\sup_{X_{ij}=X_{ji}\in\text{Herm}}\inf_{\sum_{i,j}|i\rangle\!\langle
j|\otimes Y_{ij}\geq0}\left\{  \sum_{i<j}\operatorname{Tr}[X_{ij}]+\sum
_{i}\operatorname{Tr}[Y_{ii}\rho_{i}]+\sum_{i\neq j}\operatorname{Tr}
[X_{ij}Y_{ji}]\right\}  \\
& =\sup_{X_{ij}=X_{ji}\in\text{Herm}}\inf_{\sum_{i,j}|i\rangle\!\langle
j|\otimes Y_{ij}\geq0}\left\{  \sum_{i<j}\operatorname{Tr}[X_{ij}]+\sum
_{i}\operatorname{Tr}[Y_{ii}\rho_{i}]+\sum_{i<j}\operatorname{Tr}
[X_{ij}(Y_{ij}+Y_{ji}^{\dag})]\right\}  \\
& =\sup_{X_{ij}=X_{ji}\in\text{Herm}}\inf_{\sum_{i,j}|i\rangle\!\langle
j|\otimes Y_{ij}\geq0}\left\{  \sum_{i}\operatorname{Tr}[Y_{ii}\rho_{i}
]+\sum_{i<j}\operatorname{Tr}[X_{ij}(I_d+Y_{ij}+Y_{ji}^{\dag})]\right\}  \\
& \leq\inf_{\sum_{i,j}|i\rangle\!\langle j|\otimes Y_{ij}\geq0}\sup
_{X_{ij}=X_{ji}\in\text{Herm}}\left\{  \sum_{i}\operatorname{Tr}[Y_{ii}
\rho_{i}]+\sum_{i<j}\operatorname{Tr}[X_{ij}(I_d+Y_{ij}+Y_{ji}^{\dag})]\right\}
\\
& =\inf_{\sum_{i,j}|i\rangle\!\langle j|\otimes Y_{ij}\geq0}\left\{  \sum
_{i}\operatorname{Tr}[Y_{ii}\rho_{i}]+\sup_{X_{ij}=X_{ji}\in\text{Herm}}
\sum_{i<j}\operatorname{Tr}[X_{ij}(I_d+Y_{ij}+Y_{ji}^{\dag})]\right\}  \\
& =\inf_{Y_{ij}}\left\{  \sum_{i}\operatorname{Tr}[Y_{ii}\rho_{i}]:\sum
_{i,j}|i\rangle\!\langle j|\otimes Y_{ij}\geq0,\ Y_{ij}+Y_{ji}^{\dag
}=-I_d\ \forall i\neq j\right\}
\\
& =\frac{1}{2}\inf_{Y_{ij}}\left\{  \sum_{i}\operatorname{Tr}[Y_{ii}\rho_{i}]:\sum
_{i,j}|i\rangle\!\langle j|\otimes Y_{ij}\geq0,\ \frac{1}{2}(Y_{ij}+Y_{ji}^{\dag
})=-I_d\ \forall i\neq j\right\}.
\end{align}
Strong duality holds by making similar choices discussed around \eqref{eq:multi-SDP-slater-choice}. To be clear, we pick $Y_{ii} = r I$ and $Y_{ij} = Y_{ji} = -I$ for $i\neq j$, and the analysis starting at \eqref{eq:multi-SDP-slater-choice} establishes that $\sum
_{i,j}|i\rangle\!\langle j|\otimes Y_{ij} > 0$. For the other SDP, we pick $X_{ij} = 0$ for $i \neq j$.

\subsection{Proof of \cref{prop:SDP-geo-props} (Properties of multivariate SDP geometric fidelity)}

\label{app:SDP-geo-props}

    Proofs for data processing, symmetry, and the direct-sum property follow by proceeding similarly to our proof approaches for \cref{thm:properties_SDP_fidelity}. Non-negativity follows by picking $X_{ij}=0$ for all $i,j\in[r]$ such that $i\neq j$; this is a feasible choice for which the objective function evaluates to zero. Thus, $F_{\operatorname{SDP}}^G(\rho_1, \ldots, \rho_r)\geq 0$.

\medskip 
\underline{Weak orthogonality:}
By recalling~\eqref{eq:ineq-FGSDP-FG-FH} and the orthogonality of average pairwise Holevo fidelity, for a mutually orthogonal tuple of states, the equality $ F_{\operatorname{SDP}}^G(\rho_1, \ldots, \rho_r) =0$ holds, thus proving weak orthogonality.

\medskip 
\underline{Faithfulness:}
The inequality in~\eqref{eq:upper_geometric_SDP} also shows that, when $F_{\operatorname{SDP}}^G(\rho_1, \ldots, \rho_r) =1$, all states are the same by utilizing the faithfulness of $F_{\operatorname{SDP}}$ (see \cref{thm:properties_SDP_fidelity}). To see the reverse implication, suppose that $\rho_i=\rho$ for all $i\in [r]$. In~\cref{def:SDP_geometric_multi} choose $X_{ij}= \rho$ for $i \neq j$, which is a valid candidate since it is self-adjoint and satisfies the SDP condition
\begin{equation}
    \sum_{i=1}^r |i\rangle\!\langle i| \otimes \rho_i + \sum_{i \neq j} |i\rangle\!\langle j| \otimes X_{ij}  = \sum_{i,j=1}^r |i\rangle\!\langle j| \otimes \rho  \geq 0.
\end{equation}
For this choice of $X_{ij}$, the objective function takes the value $1$, which implies $1 \leq F_{\operatorname{SDP}}^G(\rho,\ldots,\rho)$. Using~\eqref{eq:upper_geometric_SDP} along with $F_{\operatorname{SDP}}(\rho_1, \ldots, \rho_r) \leq 1$, we conclude that $F_{\operatorname{SDP}}^G(\rho,\ldots,\rho)=1$, thus proving the faithfulness property.

\medskip 
\underline{Reduction to classical fidelity:}
For commuting states, with  the use of~\eqref{eq:upper_geometric_SDP} and reduction to commuting states for multivariate SDP fidelity (see \cref{thm:properties_SDP_fidelity}), we have 
\begin{equation}
       F_{\operatorname{SDP}}^G(\rho_1, \ldots, \rho_r) \leq F_{\operatorname{SDP}}(\rho_1, \ldots, \rho_r) = \frac{2}{r(r-1)} \sum_{i<j} \sum_{x \in \cX} \sqrt{\rho_i(x) \rho_j(x)}.
       \label{eq:up-bound-F_G_SDP-comm-states}
\end{equation}
For the lower bound, set $X_{ij} = \sum_x \sqrt{\rho_i(x)\rho_j(x)} |x\rangle\!\langle x|$ in~\cref{def:SDP_geometric_multi}. By defining  $M_x \coloneqq \sum_i |i\rangle \otimes \sqrt{\rho_i(x)} |x\rangle$, we conclude that 
\begin{equation}
   \sum_{i=1}^r |i\rangle\!\langle i| \otimes \rho_i + \sum_{i \neq j} |i\rangle\!\langle j| \otimes X_{ij} = \sum_x M_xM_x^\dag \geq 0 ,
\end{equation}
which proves that this choice of $X_{ij}$ is feasible.
Thus, it follows that 
\begin{equation}
    \frac{2}{r(r-1)} \sum_{i<j} \sum_{x \in \cX} \sqrt{\rho_i(x) \rho_j(x)} \leq  F_{\operatorname{SDP}}^G(\rho_1, \ldots, \rho_r),
\end{equation}
which, when combined with \eqref{eq:up-bound-F_G_SDP-comm-states}, proves that $F_{\operatorname{SDP}}^G$ reduces to average pairwise fidelity for commuting states.

\medskip 
\underline{Super-multiplicativity:}
    Let $(X_{ij})_{i,j}$ be such that $X_{ij}=X_{ji} \in \operatorname{Herm}$  and $A \coloneqq \sum_{i=1}^r |i\rangle\!\langle i| \otimes \rho_i + \sum_{i \neq j} |i\rangle\!\langle j| \otimes X_{ij} \geq 0$.
    Then for integer $n >1$ we also have 
    \begin{equation}
     \sum_{i=1}^r |i\rangle\!\langle i| \otimes \rho_i^{\otimes n} + \sum_{i \neq j} |i\rangle\!\langle j| \otimes X_{ij}^{\otimes n} \geq 0,
    \end{equation}
    by applying~\cite[Theorem~3.1]{BlockMatrix} $n-1$ times. This theorem can be easily seen from a different proof. Let
    \begin{equation}
    Y = \sum_{i,j} |i\rangle\!\langle j| \otimes Y_{ij} \geq 0, \qquad Z = \sum_{i,j} |i\rangle\!\langle j| \otimes Z_{ij} \geq 0.
    \end{equation}
    Then $Y\otimes Z \geq 0$, and picking
    \begin{equation}
    M = \sum_{i} |i\rangle\!\langle i| \otimes I \otimes \langle i| \otimes I,
    \end{equation}
    it follows that $M (Y\otimes Z) M^\dag \geq 0$, and one can check that
    \begin{equation}
    M (Y\otimes Z) M^\dag =  \sum_{i,j} |i\rangle\!\langle j| \otimes Y_{ij} \otimes Z_{ij}.
    \end{equation}

    With that, $X_{ij}^{\otimes n}$ is a feasible candidate for $F^G_{\operatorname{SDP}}\!\left(\rho_1^{\otimes n}, \ldots, \rho_r^{\otimes n}\right)$, leading to 
    \begin{equation}
        \frac{2}{r(r-1)} \sum_{i<j} \left( \Tr[X_{ij}]\right)^n \leq F^G_{\operatorname{SDP}}\!\left(\rho_1^{\otimes n}, \ldots, \rho_r^{\otimes n}\right).
    \end{equation}
Then using convexity of the function $x^n$ for $n>1$, we have 
 \begin{equation}
       \left( \frac{2}{r(r-1)} \sum_{i<j} \Tr[X_{ij}]\right)^n\leq F^G_{\operatorname{SDP}}\!\left(\rho_1^{\otimes n}, \ldots, \rho_r^{\otimes n}\right).
    \end{equation}
We conclude the proof by supremizing over $(X_{ij})_{i,j}$ satisfying the required constraints.

\subsection{Proof of Proposition~\ref{prop:mult-geo-SDP-G-alt} (Alternate form of multivariate geometric SDP fidelity)}

\label{app:mult-geo-SDP-G-alt}

To see the expression in \eqref{eq:mult-geo-SDP-G-alt}, recall the expression
for $F^{G}_{\operatorname{SDP}}(\rho_{1},\ldots,\rho_{r})$ in \eqref{eq:SDP-multi-geometric-def}. Let us
define
\begin{equation}
X=\sum_{i,j=1}^{r}|i\rangle\!\langle j|\otimes X_{ij}.
\end{equation}
The constraint in \eqref{eq:SDP-multi-geometric-def} can be rewritten as
\begin{align}
X  &  \geq0,\\
\left(  \Delta_{r}\otimes\operatorname{id}_{d}\right)  \left(  X\right)   &
=\sum_{i=1}^{r}|i\rangle\!\langle i|\otimes\rho_{i}.
\label{eq:state-constraint-S-G-SDP}
\end{align}
Since the constraint $X\geq0$ imposes that $X$ is Hermitian, it follows that
$X_{ij}=X_{ji}^{\dag}$ for all $i,j\in\left[  r\right]  $. The further
constraint that $X_{ij}=X_{ji}\in\operatorname{Herm}$ follows from this and
the constraint $\left(  T_r\otimes\operatorname{id}\right)  \left(  X\right)
=X$ because
\begin{align}
\left(  T_r\otimes\operatorname{id}\right)  \left(  X\right)   &  =\left(
T_r\otimes\operatorname{id}\right)  \left(  \sum_{i,j=1}^{r}|i\rangle\!\langle
j|\otimes X_{ij}\right) \\
&  =\sum_{i,j=1}^{r}T(|i\rangle\!\langle j|)\otimes X_{ij}\\
&  =\sum_{i,j=1}^{r}|j\rangle\!\langle i|\otimes X_{ij}\\
&  =\sum_{i,j=1}^{r}|i\rangle\!\langle j|\otimes X_{ji}.
\end{align}
The sum $2\sum_{i<j}\operatorname{Tr}[X_{ij}]$ can be rewritten as
\begin{align}
2\sum_{i<j}\operatorname{Tr}[X_{ij}]  &  =\sum_{i\neq j}\operatorname{Tr}
[X_{ij}]\\
&  =\sum_{i,j=1}^{r}\operatorname{Tr}[X_{ij}]-\sum_{i=1}^{r}\operatorname{Tr}
[X_{ii}]\\
&  =\sum_{i,j=1}^{r}\operatorname{Tr}[X_{ij}]-\sum_{i=1}^{r}\operatorname{Tr}
[\rho_{i}]\\
&  =\left(  \sum_{i,j=1}^{r}\operatorname{Tr}[X_{ij}]\right)  -r,
\end{align}
where we made use of \eqref{eq:state-constraint-S-G-SDP}. Now observe that
\begin{equation}
\sum_{i,j=1}^{r}\operatorname{Tr}[X_{ij}]=\operatorname{Tr}[\left(
|+\rangle\!\langle+|\otimes I_d\right)  X]
\end{equation}
because
\begin{align}
\operatorname{Tr}[\left(  |+\rangle\!\langle+|\otimes I_d\right)  X]  &
=\operatorname{Tr}\!\left[  \left(  \sum_{i,j=1}^{r}|i\rangle\!\langle
j|\otimes I_d\right)  \left(  \sum_{k,\ell=1}^{r}|k\rangle\!\langle\ell|\otimes
X_{k\ell}\right)  \right] \\
&  =\sum_{i,j,k,\ell=1}^{r}\operatorname{Tr}\!\left[  |i\rangle\!\langle
j|k\rangle\!\langle\ell|\otimes X_{k\ell}\right] \\
&  =\sum_{i,j,k,\ell=1}^{r}\langle j|k\rangle\!\langle\ell|i\rangle
\operatorname{Tr}[X_{k\ell}]\\
&  =\sum_{i,j=1}^{r}\operatorname{Tr}[X_{ji}]\\
&  =\sum_{i,j=1}^{r}\operatorname{Tr}[X_{ij}].
\end{align}
Putting everything above together, we conclude the equality in \eqref{eq:mult-geo-SDP-G-alt}.

\subsection{Proof of Proposition~\ref{prop:secrecy-meas-geo-SDP-SDP-dual} (Alternate form of secrecy measure~$S^{G}_{\operatorname{SDP}}$)}

\label{app:secrecy-meas-geo-SDP-SDP-dual}

The equality in \eqref{eq:secrecy-measure-G-primal} follows from combining
\eqref{eq:def-mult-secrecy-geo} and \eqref{eq:mult-geo-SDP-G-alt}. So it
remains to prove \eqref{eq:secrecy-measure-G-dual}, for which we employ
SDP\ duality theory. Consider that
\begin{align}
&  \sup_{X\geq0}\left\{
\begin{array}
[c]{c}
\operatorname{Tr}[\left(  |+\rangle\!\langle+|\otimes I_d\right)  X]:\\
\left(  \Delta_{r}\otimes\operatorname{id}_{d}\right)  \left(  X\right)
=\sum_{i=1}^{r}|i\rangle\!\langle i|\otimes\rho_{i},\\
\left(  T_r\otimes\operatorname{id}\right)  \left(  X\right)  =X
\end{array}
\right\} \nonumber\\
&  =\sup_{X\geq0}\left\{
\begin{array}
[c]{c}
\operatorname{Tr}[\left(  |+\rangle\!\langle+|\otimes I_d\right)  X]+\\
\inf_{Y,Z\in\operatorname{Herm}}\left\{
\begin{array}
[c]{c}
\operatorname{Tr}\!\left[  Y\left(  \sum_{i=1}^{r}|i\rangle\!\langle
i|\otimes\rho_{i}-\left(  \Delta_{r}\otimes\operatorname{id}_{d}\right)
\left(  X\right)  \right)  \right] \\
+\operatorname{Tr}[Z\left(  \left(  T_r\otimes\operatorname{id}\right)  \left(
X\right)  -X\right)  ]
\end{array}
\right\}
\end{array}
\right\} \\
&  =\sup_{X\geq0}\inf_{Y,Z\in\operatorname{Herm}}\left\{
\begin{array}
[c]{c}
\operatorname{Tr}[\left(  |+\rangle\!\langle+|\otimes I_d\right)
X]+\operatorname{Tr}\!\left[  Y\left(  \sum_{i=1}^{r}|i\rangle\!\langle
i|\otimes\rho_{i}-\left(  \Delta_{r}\otimes\operatorname{id}_{d}\right)
\left(  X\right)  \right)  \right] \\
+\operatorname{Tr}[Z\left(  \left(  T_r\otimes\operatorname{id}\right)  \left(
X\right)  -X\right)  ]
\end{array}
\right\} \\
&  =\sup_{X\geq0}\inf_{Y,Z\in\operatorname{Herm}}\left\{
\begin{array}
[c]{c}
\operatorname{Tr}[\left(  |+\rangle\!\langle+|\otimes I_d\right)
X]-\operatorname{Tr}[\left(  \Delta_{r}\otimes\operatorname{id}_{d}\right)
\left(  Y\right)  X]\\
+\operatorname{Tr}\!\left[  Y\left(  \sum_{i=1}^{r}|i\rangle\!\langle
i|\otimes\rho_{i}\right)  \right] \\
+\operatorname{Tr}[\left(  T_r\otimes\operatorname{id}\right)  \left(  Z\right)
X]-\operatorname{Tr}[ZX]
\end{array}
\right\} \\
&  =\sup_{X\geq0}\inf_{Y,Z\in\operatorname{Herm}}\left\{
\begin{array}
[c]{c}
\operatorname{Tr}\!\left[  Y\left(  \sum_{i=1}^{r}|i\rangle\!\langle
i|\otimes\rho_{i}\right)  \right] \\
+\operatorname{Tr}[\left(  \left(  |+\rangle\!\langle+|\otimes I_d\right)
-\left(  \Delta_{r}\otimes\operatorname{id}_{d}\right)  \left(  Y\right)
+\left(  T_r\otimes\operatorname{id}\right)  \left(  Z\right)  -Z\right)  X]
\end{array}
\right\} \\
&  \leq\inf_{Y,Z\in\operatorname{Herm}}\sup_{X\geq0}\left\{
\begin{array}
[c]{c}
\operatorname{Tr}\!\left[  Y\left(  \sum_{i=1}^{r}|i\rangle\!\langle
i|\otimes\rho_{i}\right)  \right] \\
+\operatorname{Tr}[\left(  \left(  |+\rangle\!\langle+|\otimes I_d\right)
-\left(  \Delta_{r}\otimes\operatorname{id}_{d}\right)  \left(  Y\right)
+\left(  T_r\otimes\operatorname{id}\right)  \left(  Z\right)  -Z\right)  X]
\end{array}
\right\} \\
&  =\inf_{Y,Z\in\operatorname{Herm}}\left\{
\begin{array}
[c]{c}
\operatorname{Tr}\!\left[  Y\left(  \sum_{i=1}^{r}|i\rangle\!\langle
i|\otimes\rho_{i}\right)  \right]  :\\
\left(  |+\rangle\!\langle+|\otimes I_d\right)  -\left(  \Delta_{r}
\otimes\operatorname{id}_{d}\right)  \left(  Y\right)  +\left(  T_r\otimes
\operatorname{id}\right)  \left(  Z\right)  -Z\leq0
\end{array}
\right\} \\
&  =\inf_{Y,Z\in\operatorname{Herm}}\left\{
\begin{array}
[c]{c}
\operatorname{Tr}\!\left[  Y\left(  \sum_{i=1}^{r}|i\rangle\!\langle
i|\otimes\rho_{i}\right)  \right]  :\\
|+\rangle\!\langle+|\otimes I_d+\left(  T_r\otimes\operatorname{id}\right)
\left(  Z\right)  \leq\left(  \Delta_{r}\otimes\operatorname{id}_{d}\right)
\left(  Y\right)  +Z
\end{array}
\right\}
\end{align}
Strong duality follows by picking $X=\sum_{i=1}^{r}|i\rangle\!\langle
i|\otimes\rho_{i}$ in the primal (feasible choice)\ and $Z=0$ and
$Y=2rI_{r}\otimes I_{d}$ in the dual (strictly feasible choice). Indeed, for
the latter, consider that
\begin{align}
\left(  \Delta_{r}\otimes\operatorname{id}_{d}\right)  \left(  2rI_{r}\otimes
I_{d}\right)   &  =2rI_{r}\otimes I_{d}\\
&  =2r\left(  I_{r}-\frac{|+\rangle\!\langle+|}{r}+\frac{|+\rangle\!\langle
+|}{r}\right) \\
&  =2r\left(  I_{r}-\frac{|+\rangle\!\langle+|}{r}\right)  +2|+\rangle\!
\langle+|\\
&  >|+\rangle\!\langle+|,
\end{align}
where we used the fact that $I_{r}-\frac{|+\rangle\!\langle+|}{r}\geq0$.

\subsection{Proof of Proposition~\ref{prop:multi-geo-submult} (Super-multiplicativity of the secrecy measure $S_{\operatorname{SDP}}^{G}$)}

\label{app:multi-geo-submult}

Let $X_{1}$ satisfy the constraints
\begin{align}
X_{1}  & \geq0,\\
\left(  \Delta_{r_{1}}\otimes\operatorname{id}_{d_{1}}\right)  \left(
X_{1}\right)    & =\sum_{i=1}^{r_{1}}|i\rangle\!\langle i|\otimes\rho_{i},\\
\left(  T_{r_{1}}\otimes\operatorname{id}_{d_{1}}\right)  \left(
X_{1}\right)    & =X_{1},
\end{align}
and let $X_{2}$ satisfy the constraints
\begin{align}
X_{2}  & \geq0,\\
\left(  \Delta_{r_{2}}\otimes\operatorname{id}_{d_{2}}\right)  \left(
X_{2}\right)    & =\sum_{j=1}^{r_{2}}|j\rangle\!\langle j|\otimes\sigma_{j},\\
\left(  T_{r_{2}}\otimes\operatorname{id}_{d_{2}}\right)  \left(
X_{2}\right)    & =X_{2}.
\end{align}
Then $X_{1}\otimes X_{2}$ satisfies the constraints for $S_{\operatorname{SDP}
}^{G}((\rho_{i}\otimes\sigma_{j})_{i\in\left[  r_{1}\right]  ,j\in\left[
r_{2}\right]  })$. Indeed,
\begin{align}
X_{1}\otimes X_{2}  & \geq0,\\
\sum_{i=1}^{r_{1}}\sum_{j=1}^{r_{2}}|i\rangle\!\langle i|\otimes
|j\rangle\!\langle j|\otimes\rho_{i}\otimes\sigma_{j}  & \simeq\left(
\sum_{i=1}^{r_{1}}|i\rangle\!\langle i|\otimes\rho_{i}\right)  \otimes\left(
\sum_{j=1}^{r_{2}}|j\rangle\!\langle j|\otimes\sigma_{j}\right)  \\
& =\left(  \Delta_{r_{1}}\otimes\operatorname{id}_{d_{1}}\right)  \left(
X_{1}\right)  \otimes\left(  \Delta_{r_{2}}\otimes\operatorname{id}_{d_{2}
}\right)  \left(  X_{2}\right)  \\
& =\left(  \Delta_{r_{1}}\otimes\operatorname{id}_{d_{1}}\otimes\Delta_{r_{2}
}\otimes\operatorname{id}_{d_{2}}\right)  \left(  X_{1}\otimes X_{2}\right)
\\
X_{1}\otimes X_{2}  & =\left(  T_{r_{1}}\otimes\operatorname{id}_{d_{1}
}\right)  \left(  X_{1}\right)  \otimes\left(  T_{r_{2}}\otimes
\operatorname{id}_{d_{2}}\right)  \left(  X_{2}\right)  \\
& =\left(  T_{r_{1}}\otimes\operatorname{id}_{d_{1}}\otimes T_{r_{2}}
\otimes\operatorname{id}_{d_{2}}\right)  \left(  X_{1}\otimes X_{2}\right)  .
\end{align}
Then it follows that
\begin{align}
& \operatorname{Tr}[\left(  |+\rangle\!\langle+|_{r_{1}}\otimes I_{d_{1}
}\right)  X_{1}]\cdot\operatorname{Tr}[\left(  |+\rangle\!\langle+|_{r_{2}
}\otimes I_{d_{2}}\right)  X_{2}]\nonumber\\
& =\operatorname{Tr}[\left(  \left(  |+\rangle\!\langle+|_{r_{1}}\otimes
I_{d_{1}}\right)  X_{1}\right)  \otimes\left(  \left(  |+\rangle
\!\langle+|_{r_{2}}\otimes I_{d_{2}}\right)  X_{2}\right)  ]\\
& =\operatorname{Tr}[\left(  |+\rangle\!\langle+|_{r_{1}}\otimes I_{d_{1}
}\otimes|+\rangle\!\langle+|_{r_{2}}\otimes I_{d_{2}}\right)  \left(
X_{1}\otimes X_{2}\right)  ]\\
& \leq (r_1 r_2)^2 S_{\operatorname{SDP}}^{G}((\rho_{i}\otimes\sigma_{j})_{i\in\left[
r_{1}\right]  ,j\in\left[  r_{2}\right]  }) .
\end{align}
Since the inequality holds for all feasible $X_{1}$ and $X_{2}$, we conclude
the inequality in \eqref{eq:multi-geo-submult}.

\end{document}